\documentclass[onecolumn]{IEEEtran}

\usepackage{macro}

\title{Quadratically Constrained Myopic Adversarial Channels}
\author{\IEEEauthorblockN{Yihan Zhang\IEEEauthorrefmark{1},
Shashank Vatedka\IEEEauthorrefmark{2}, 
Sidharth Jaggi\IEEEauthorrefmark{3}\IEEEauthorrefmark{1} 
and Anand Sarwate\IEEEauthorrefmark{4}}\\
\IEEEauthorblockA{
\IEEEauthorrefmark{1}Dept.\ of Information Engineering, The Chinese University of Hong Kong\\
\IEEEauthorrefmark{2}Dept.\ of Electrical Engineering, Indian Institute of Technology Hyderabad\\
\IEEEauthorrefmark{3}\rev{School of Mathematics, University of Bristol}\\
\IEEEauthorrefmark{4}Dept.\ of Electrical and Computer Engineering, Rutgers, The State University of New Jersey} \\
\{\href{mailto:zy417@ie.cuhk.edu.hk}{zy417},\href{mailto:jaggi@ie.cuhk.edu.hk}{jaggi}\}@ie.cuhk.edu.hk, 
\href{mailto:shashankvatedka@iith.ac.in}{shashankvatedka@iith.ac.in},
\href{mailto:sid.jaggi@bristol.ac.uk}{sid.jaggi@bristol.ac.uk},
\href{mailto:anand.sarwate@rutgers.edu}{anand.sarwate@rutgers.edu}
}

\begin{document}
\maketitle
\footnotetext[1]{This work was partially funded by a grant from the University Grants Committee of the Hong Kong Special Administrative Region (Project No.\  AoE/E-02/08) and RGC GRF grants 14208315 and 14313116. {This work was done when Shashank Vatedka was at the department of Information Engineering, Chinese University of Hong Kong}. The work of Shashank Vatedka was supported in part by CUHK Direct Grants 4055039 and 4055077.\\
	 {This work was presented in parts at the 2012 International Conference on Signal Processing and Communications, Bengaluru, India~\cite{sarwate-spcom2012}, and the 2018 IEEE International Symposium on Information theory, Vail, CO, USA~\cite{zhang2018quadraticISIT}}.} 
\begin{abstract}
 We study communication in the presence of a jamming adversary where quadratic  power constraints are imposed on the transmitter and the jammer.
 The jamming signal is allowed to be a function of the codebook, and a noncausal but noisy observation of the
 transmitted codeword. 
 For a certain range of the noise-to-signal ratios (NSRs) of the transmitter and the jammer, we are able to characterize the capacity of this channel under deterministic encoding \rev{or} stochastic encoding, i.e., with no common randomness between the encoder/decoder pair.
 For the remaining NSR regimes, we determine the capacity under the assumption of a small amount of common randomness (at most $2\log(n)$ bits in one sub-regime, and at most $\Omega(\rev{n})$ bits  in the other sub-regime) available to the encoder-decoder pair. Our proof techniques involve a novel myopic list-decoding result for achievability, and a Plotkin-type push attack for the converse in a subregion of the NSRs,  both of which which may be of independent interest. We also give bounds on the  secrecy capacity of this channel \rev{assuming that the jammer is simultaneously eavesdropping}.
 
\end{abstract}

\tableofcontents

\section{Introduction and prior work}


Consider a point-to-point communication system where a transmitter, Alice, wants to send a message to a receiver, Bob, through a channel distorted by additive noise.
She does so by encoding the message to a length-$n$ codeword, which is fed into the channel.
Much of traditional communication and information theory has focused on the scenario where the noise is independent of the transmitted signal
and the coding scheme. We study the case where communication takes place in the presence of a malicious jammer (whom we call James) who tries to ensure that Bob is
unable to recover the transmitted message. 
The channel is a discrete-time, real-alphabet channel, and the codeword transmitted by Alice is required to satisfy a quadratic power constraint.
It is assumed that the coding scheme is known to all three parties, and James also observes a noisy version of the transmitted signal (hence the term \emph{myopic}).
The jamming signal is required to satisfy a separate power constraint, but otherwise can be a noncausal function of the noisy observation and the coding scheme.

This problem is part of the general framework of arbitrarily varying channels (AVCs), introduced by Blackwell et al.~\cite{blackwell-avc-1960}.
The quadratically constrained AVC (also called the Gaussian AVC) was studied by Blachman~\cite{blachman-1962}, who gave upper and lower bounds on the capacity of the channel under the assumption that
James observes a noiseless version of the transmitted codeword (a.k.a.\ the \emph{omniscient} adversary). The lower bound used a sphere packing argument similar to the one used to prove the Gilbert-Varshamov (GV) bound for binary linear codes. The upper bound was based on Rankin's upper bound on the number of non-intersecting spherical caps that can be placed on a sphere~\cite{rankin-sphericalcap-1955}. The quadratically constrained AVC is closely related to the sphere-packing problem where the objective is to find the densest arrangement of identical $n$-dimensional balls of radius $\sqrt{nN}$ subject to the constraint that the center of each ball lies within a ball of radius $\sqrt{nP}$.
An exact characterization of the capacity of this problem is not known, though inner~\cite{blachman-1962} and outer bounds~\cite{kabatiansky-1978,cohnzhao-2014} are known.
At the other end of the spectrum,  Hughes and Narayan~\cite{hughes-narayan-it1987}, and later Csisz\'ar and Narayan~\cite{csiszar-narayan-it1991}, studied the problem with an ``oblivious'' James, who knows the codebook, but does not see the transmitted codeword. 
They consider the regime when $P>N$ (it can be shown that no positive throughput is possible when $P<N$). They showed that under an average probability of error metric, the capacity of the oblivious adversarial channel is equal to that of an additive white Gaussian noise (AWGN) channel whose noise variance is equal to the power constraint imposed on James.
These omniscient and oblivious cases are two extreme instances of the general myopic adversary that we study in this paper. 

The oblivious vector Gaussian AVC was studied by Hughes and Narayan~\cite{hughes1988vectorgaussAVC}, and later Thomas and Hughes~\cite{thomas1991exponential} derived bounds on the error exponents for the oblivious Gaussian AVC.
Sarwate and Gastpar~\cite{sarwate2006randomizationGauss} showed that the randomized coding capacity of the oblivious channel is the same under average and maximum error probability constraints.

The work builds on~\cite{sarwate-spcom2012}, which characterized the capacity of this channel under the assumption that James knows a noisy version of the transmitted signal,
but Alice's codebook is shared only with Bob. This can be interpreted as a myopic channel with an unlimited amount of common randomness (or shared secret key, CR) between Alice and Bob.
A related model was studied by Haddadpour et al.~\cite{haddadpour-isit2013}, who assumed that James knows the message, but not the exact codeword transmitted by Alice. 
In this setup, Alice has access to private randomness which is crucially used to pick a codeword for a given message. However, Alice and Bob do not share any common randomness.
Game-theoretic versions of the problems have also been considered in the literature, notably by M\'edard~\cite{medard}, Shafiee and Ulukus~\cite{shafiee} and Baker and Chao~\cite{baker-chao-siam1996}. Shafiee and Ulukus~\cite{shafiee} considered a more general two-sender scenario, while Baker and Chao~\cite{baker-chao-siam1996} studied a multiple antenna version of the problem. 
More recently, Hosseinigoki and Kosut~\cite{hosseinigoki2018capacity} derived the list-decoding capacity of the Gaussian AVC with an oblivious adversary. 
Zhang and Vatedka~\cite{zhang-vatedka-list-dec-real} derived bounds on achievable list-sizes for random spherical and lattice codes.
Pereg and Steinberg~\cite{pereg2018arbitrarily} have analyzed a relay channel where the observation of the destination is corrupted by a power-constrained oblivious adversary. \rev{Beemer et al.~\cite{beemer2019authentication} studied a related problem of authentication against a myopic adversary, where the goal of the decoder is to correctly either decode the message or detect adversarial interference. Zhang et al.~\cite{zhang2020quadratically} also studied a quadratically constrained two-way interference channel with a jamming adversary, where proof techniques similar to ours were used to obtain upper and lower bounds on the capacity.}
\rev{Budkuley et al.~\cite{bdjlsw-myopic_symm-isit2020} gave an improved symmetrization (known as \emph{$ \mathsf{CP} $-symmetrization} where $\mathsf{CP}$ is for \emph{completely positive}) bound for myopic AVCs over \emph{discrete} alphabets. 
The result expands the parameter region where the capacity (without common randomness) is zero. 
The proof is based on a significant generalization of the Plotkin bound in classical coding theory which is proved by Wang et al.~\cite{wbbj-2019-gen_plotkin}.
Dey et al.~\cite{djlsw-2019-myopic-causal} studied, among others, the binary erasure-erasure myopic adversarial channels and gave nontrivial achievability schemes beating the Gilbert--Varshamov bound in the \emph{insufficiently} myopic regime.}

Communication in the presence of a myopic jammer has also received considerable attention in the discrete-alphabet case  (see~\cite{dey-sufficiently-2015} and references therein, and the recent~\cite{zhang2018covert,song2018multipath} on covert communication with myopic jammers).
We would like to draw connections to the bit-flip adversarial problem where communication takes place over a binary channel, and James observes the codeword through a binary symmetric channel (BSC)
with crossover probability $q$. He is allowed to flip at most $np$ bits, where $0<p<1/2$ can be interpreted as his ``jamming power.'' Dey et al.~\cite{dey-sufficiently-2015} showed that
when James is sufficiently myopic, i.e., $q>p$, the capacity is equal to $1-H(p)$. In other words, he can do no more damage than an oblivious adversary. As we will see in the present article,
this is not true for the quadratically constrained case. We will show that as long as the \emph{omniscient list-decoding capacity} for Bob is greater than the \emph{AWGN channel capacity} for James, the capacity is equal to a certain  
\emph{myopic list-decoding capacity} for Bob. In this regime, James cannot uniquely determine the transmitted codeword among exponentially many. As a result no attack strategy by James that ``pushes'' the transmitted codeword to the nearest other codeword is as bad as in the omniscient case since the nearest codeword in general will be different for different choices of the transmitted codeword.

Recent works have also considered communication with simultaneous active and passive attacks~\cite{molavianjazi2009arbitrary,bjelakovic2013capacity,he2013mimo,he2014mimo,notzel2016arbitrarily,wiese2016channel,goldfeld2016arbitrarily}. However, in these works, the eavesdropper and jammer are assumed to be independent entities and the jammer is assumed to be an oblivious adversary. In this work, we derive lower bounds on the capacity of the myopic adversarial channel with an additional wiretap secrecy~\cite{wyner1975wire} constraint, \rev{treating the jammer as an eavesdropper at the same time}. 

\begin{figure} 
\begin{center}
 \includegraphics[width = 0.8\textwidth]{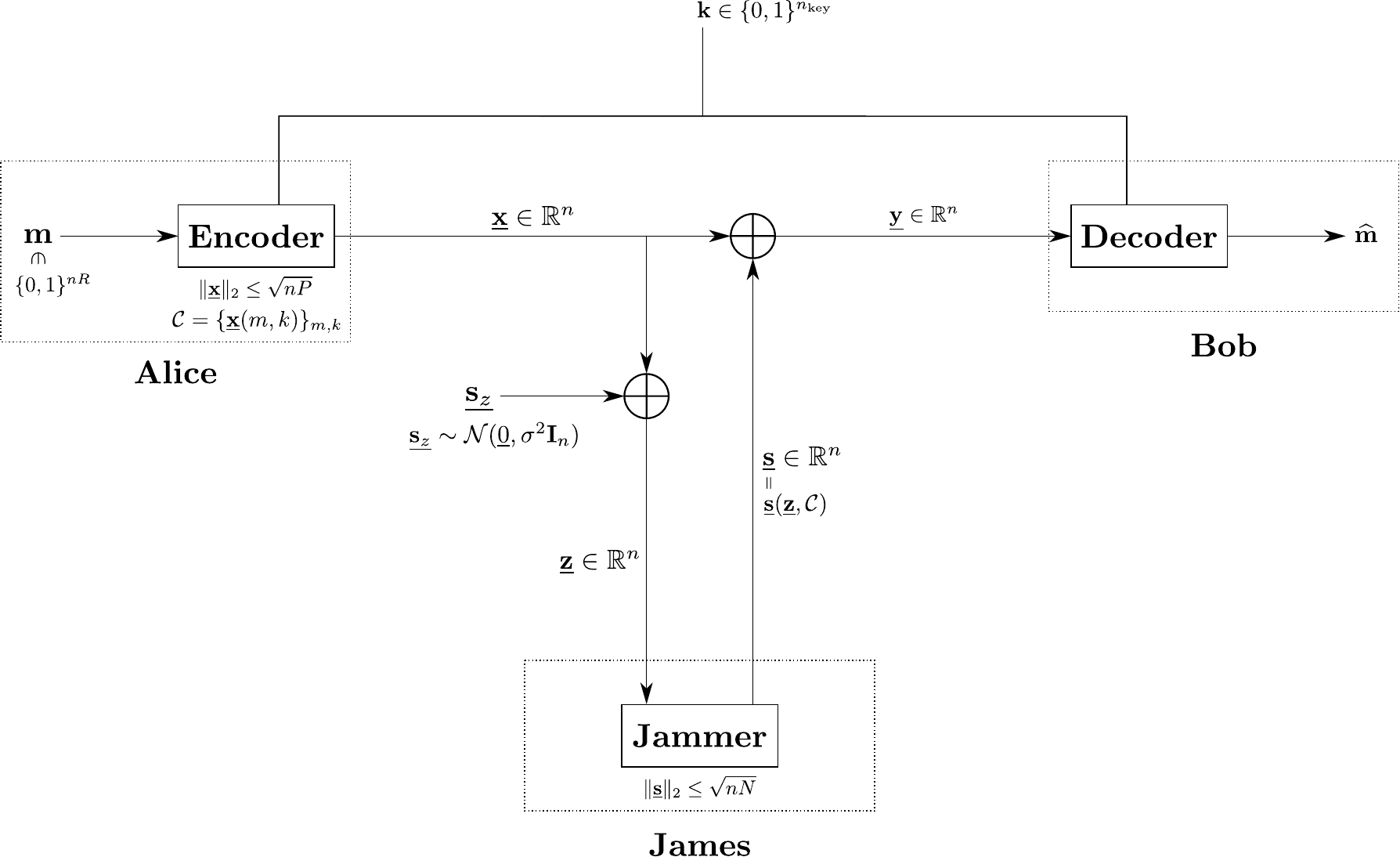}
 \caption{The setup studied in this paper: Alice wants to transmit $nR$-bit message $\bfm$ to Bob. Across the channel, she transmits a codeword $\vbfx$, which is a function of $\bfm$ (and potentially a shared key $\bfk$ of $\Nkey$ bits, though we also study the scenario when $\Nkey=0$). The collection of codewords $\cC$ is called the codebook, and every codeword in the codebook must satisfy a power constraint of $\sqrt{nP}$. The jammer James observes \rev{$\vbfz$} corresponding to the output of an AWGN channe with variance $\sigma^2$. He then chooses a jamming/state sequence $\vbfs$ (satisfying a power constraint of $\sqrt{nN}$) as a noncausal function of $\vbfsz$ and $\cC$. On observing $\vbfy=\vbfx+\vbfs$, Bob must output his estimate $\hat{\bfm}$ of the message $\bfm$ such that the probability of error (averaged over $\bfm$ and $\vbfsz$) vanishes.}
  \label{fig:blockdiagram}
 \end{center}
\end{figure}

Let us now describe the problem we address in this paper.
 The setup is illustrated in Fig.~\ref{fig:blockdiagram}. Alice wants to send a message $\bfm$ to Bob. The message is assumed to be uniformly chosen from $\{ 0,1 \}^{nR}$, where $R>0$ is a parameter called the \emph{rate}. Alice and Bob additionally have $\Nkey$ bits of shared secret key, $\mathbf{k}$ ($ \Nkey $ could be zero --- indeed, some of the major results in this work derive AVC capacity for some NSR regimes when $\Nkey=0$). This key is kept private from James. Alice encodes the message $\bfm$ (using $\mathbf{k}$) to a codeword $\vbfx\in\bR^n$, which is transmitted across the channel. Let $\mathcal{C}$ denote the set of all possible codewords (the codebook).
In this work, we study three types of encoding:
\begin{itemize}
    \item \emph{Deterministic encoding:} $\Nkey=0$ and $\vbfx$ is a deterministic function of $\bfm$
    \item \emph{Stochastic encoding:} $\Nkey=0$, but $\vbfx$ is a function of $\bfm$ and private random bits known only to Alice
    \item \emph{Randomized encoding:} $\Nkey>0$, and $\vbfx$ can be a function of the shared key $\bfk$ and random bits known only to Alice.
\end{itemize}
If the code is non-deterministic, then the \emph{codebook rate} $\Rcode\coloneq\frac{1}{n}\log|\cC|$ could be different from the \emph{message rate} $R$ (which we sometimes simply refer to as the rate). The codebook must satisfy a power constraint of $P>0$, i.e.\ $\Vert\vx\Vert_2\leq \sqrt{nP}$ for all $\vx\in\mathcal{C}$. James sees $\vbfz=\vbfx+\vbfsz$, where $\vbfsz$ is an AWGN with mean zero and variance $\sigma^2$. He chooses a jamming vector $\vbfs\in\bR^n$ as a noncausal function of $\vbfz$, the codebook $\mathcal{C}$, and his private randomness.  The jamming vector is also subject to a power constraint: $\Vert \vs \Vert_2\leq \sqrt{nN}$ for some $N>0$. Bob obtains $\vbfy=\vbfx+\vbfs$, and decodes this to a message $\widehat\bfm$. The message is said to have been conveyed reliably if $\widehat\bfm=\bfm$. The probability of error, $P_e$, is defined as the probability that $\widehat\bfm\neq\bfm$, where the randomness is over the message $\bfm$, the private randomness that Alice uses, the random noise $\vbfsz$, the key $\bfk$, and the private random bits available to James\footnote{An averaging argument shows that the rate cannot be improved even if Bob uses additional private random bits for randomized decoding.}.
In all our code constructions\footnote{An exception is Appendix~\ref{sec:prf_omni_stoch_vs_det_enc}, where we show that private randomness does not increase the capacity of the omniscient adversarial channel. However, we have reason to believe (albeit unsupported  by formal proof) that additional private randomness \emph{may} increase the achievable rate --- this is part of our ongoing investigation.}, we will assume that Alice and Bob may share a secret key, but the mapping from $(\bfm,\bfk)$ to $\vbfx$ is deterministic. In other words, Alice does not possess any source of additional private randomness. Conversely, all our impossibility results are robust to the presence of private randomness at the encoder (since in some AVC scenarios, private randomness is known to boost capacity --- e.g.~\cite{dey-stochastic-2016}) We  study the problem with different amounts of common randomness shared by Alice and Bob but unknown to James, and present results in each case. 

We say that a rate $R>0$ is achievable if there exists a sequence (in increasing $n$) of codebooks for which the probability of error\footnote{See Sec.~\ref{sec:preliminaries} for a formal definition.} goes to zero as $n\to\infty$.
The supremum of all achievable rates is called the capacity of the channel.

We say that a rate $ R>0 $ is achievable with (wiretap) secrecy if there exists a sequence (in increasing $n$) of codebooks for which the probability of error and the mutual information rate $ \frac{1}{n}I(\bfm;\vbfz) $ both go to zero as $n\to\infty$.
The supremum of all achievable rates is called the secrecy capacity of the channel.

\subsection{Organization of the paper}

We give a summary of our results and proof techniques in Sec.~\ref{sec:results_prooftechniques}. The formal statements of the results are presented in Sec.~\ref{sec:results_formalstatements}. The main results are also compactly summarized in Table~\ref{table:rates} and the results with secrecy are tabulated in Table~\ref{table:rates_secrecy}.
\rev{We then discuss the connection between our work and several closely related prior works in Sec.~\ref{sec:comparison_other_work}.}
Notation and preliminaries are described in \rev{Sec.~\ref{sec:notation} and} Sec.~\ref{sec:preliminaries}, \rev{respectively}.
This, \rev{as mentioned,} is followed by ideas and details of the proof techniques \rev{in Sec.~\ref{sec:results_formalstatements}}. 
In Sec.~\ref{sec:result_infinite_CR}, we describe the results for infinite common randomness and give a formal proof of the converse. 
Sec.~\ref{sec:result_linear_sublinearCR} contains the main ideas required to prove our results with linear and logarithmic amounts of common randomness. 
Our results on list-decoding are described in Sec.~\ref{sec:list_decoding_main}, with Theorem~\ref{thm:capacity_noCR} giving the main result.  
\rev{Coming to the no-common randomness regime, we present a technical yet high-level proof sketch of the achievability  and a full proof of the symmetrization converse  in Sec.~\ref{sec:result_zero_CR}. }
Sec.~\ref{sec:myopic_list_decoding_details} contains a detailed proof of \rev{Theorem~\ref{thm:myopic_listdecoding_summary}}, and Sec.~\ref{sec:achievability_suffmyopic} gives the proof of Theorem~\ref{thm:capacity_noCR}. Appendix~\ref{sec:prf_omni_stoch_vs_det_enc} has a note on why private randomness does not improve the capacity if James is omniscient. 
\rev{We transcribe a rigorous proof of a folklore theorem regarding list-decoding in Euclidean space against an omniscient adversary in Sec.~\ref{sec:prf_omniscient_list_capacity}. }
Some of the technical details of proofs appear in the other appendices \rev{(specifically Appendix~\ref{sec:appendix_proofs_basiclemmas} and Appendices~\ref{sec:prf_cap_scaleandbabble}--\ref{sec:proof_achievablerate_thetanbits_secrecy})}.
Frequently used notation is summarized in Table~\ref{tab:notation} in Appendix~\ref{sec:tableofnotation}. {Fig.~\ref{fig:flowchart} is a flowchart outlining steps involved in the proof.}

\section{Overview of results and proof techniques}\label{sec:results_prooftechniques}
\subsection{Overview of results}

We now briefly describe our results and proof techniques.  
It is helpful to visualize our results in terms of the noise-to-signal ratio (NSR), using a $N/P$ (adversarial NSR to Bob) versus $\sigma^2/P$ (random NSR to James) plot similar to the one shown in Fig.~\ref{fig:rateregion_infiniteCR}.\footnote{\rev{Our parameterization makes the parameter regions of interest 
\emph{compact} and concentrated in a bounded region around the origin (rather than scattered or shooting infinitely far away) in the two-dimensional plane spanned by $ \sigma^2/P $ and $ N/P $. }}
In~\cite{sarwate-spcom2012}, it was shown that with an infinite amount of common randomness, the capacity is $\Rld\coloneq\frac{1}{2}\log\frac{P}{N}$ in the red region, and $\Rmyop\coloneq\frac{1}{2}\log\left(\frac{(P+\sigma^2)(P+N)-2P\sqrt{N(P+\sigma^2)}}{N\sigma^2}\right)$ in the blue region.
The capacity is zero in the grey region.

In this article, while the major results are for the case when $\Nkey=0$, along the way we prove anciliary results for the regimes where $\Nkey = \Theta(n)$ and $\Nkey = \Theta(\log n)$.
\begin{itemize}
 \item \rev{\textbf{List decoding.}} We prove a general result for list-decoding in the presence of a myopic adversary. For an omniscient adversary, the list-decoding capacity is $\Rld=\frac{1}{2}\log\frac{P}{N}$. This is a folklore result, but we give a proof of this statement in Appendix~\ref{sec:prf_omniscient_list_capacity} for completeness. When the adversary is myopic, and the encoder-decoder pair shares $\cO(n)$ bits of common randomness, we give achievable rates for list-decoding. This is equal to $\Rld$ for $\frac{\sigma^2}{P}\leq \frac{P}{N}-1$, and is larger than $\Rld$ in a certain regime (depending on the amount of common randomness) where $\frac{\sigma^2}{P}> \frac{P}{N}-1$. The achievable rates are illustrated in Fig.~\ref{fig:rateregion_myopiclist_arg}. With no common randomness, we can achieve $\Rld$ and $\Rmyop$ in the red and blue regions of Fig.~\ref{fig:rateregion_myopiclist_arg_a} respectively. If Alice and Bob share $\Nkey$ bits, then $\Rmyop$ is achievable in a larger region. For instance, if $\Nkey=0.2n$, then the blue region can be expanded to give Fig.~\ref{fig:rateregion_myopiclist_arg_b}. 
 \item \rev{\textbf{Linear CR.}} When common randomness is present, we combine our list-decoding result with~\cite[Lemma 13]{sarwate-thesis} to give achievable rates over the myopic adversarial channel.
 Let us first discuss the case the amount of common randomness is linear in $n$, i.e., $\Nkey=n\Rkey$ for some $\Rkey>0$.
 If $\Rkey\geq  \frac{1}{2}\log\left(1+\frac{P}{\sigma^2}\right)-\Rmyop $, then we are able to give a complete characterization of the capacity of the channel for all values of the NSRs. We can achieve everything in Fig.~\ref{fig:rateregion_infiniteCR}. If  $\Rkey<  \frac{1}{2}\log\left(1+\frac{P}{\sigma^2}\right)-\Rmyop $, then we are able to characterize the capacity in only a sub-region of the NSRs --- This is illustrated in Fig.~\ref{fig:rateregion_rkey_0pt2} and Fig.~\ref{fig:rateregion_rkey_1} for different values of $\Rkey$. In the dotted regions, we only have nonmatching upper and lower bounds. It is worth pointing out that no fixed $\Rkey$ will let us achieve $\Rmyop$ in the entire blue region of Fig.~\ref{fig:rateregion_infiniteCR}. However, for every point in the blue region, there exists a finite value of $\Rkey$ such that $\Rmyop$ is achievable at that point. In other words, an $\Nkey=\Omega(n)$ is sufficient to achieve $\Rmyop$ at every point in the interior of the blue region in Fig.~\ref{fig:rateregion_infiniteCR}.
 \item \rev{\textbf{Logarithmic CR.}} For the $\Nkey = \Theta(\log n)$ case, we are able to find the capacity in the red and blue regions in Fig.~\ref{fig:rateregion_logn}. In the dotted regions, we have nonmatching upper and lower bounds. 
 \item \rev{\textbf{No CR.}} For $\Nkey=0$, we require a more involved approach to find the capacity. We use some of the results on myopic list-decoding in our bounds for the probability of error. We find the capacity in the red, blue and grey regions in Fig.~\ref{fig:rateregion_noCR}, but only have nonmatching upper and lower bounds in the dotted green and white regions. 
 \item \rev{\textbf{Sufficiency of deterministic encoding against omniscient adversaries.}} We show that if James is omniscient, then private randomness at the encoder does not help improve the capacity. This is based on a similar observation made by Dey et al.~\cite{dey-stochastic-2016} for the bit-flip adversarial channel. See Appendix~\ref{sec:prf_omni_stoch_vs_det_enc} for details. 
 \item \rev{\textbf{Wiretap secrecy.}} We use the above results to derive achievable rates under wiretap secrecy constraints. Specifically, we want to ensure that the mutual information rate to James, $ \frac{1}{n}I(\bfm;\vbfz)=o(1) $ in addition to  Bob being able to decode $ \bfm $ reliably.
\end{itemize}
The variation of the regions of the noise-to-signal ratios (NSR) where we can obtain achievable rates is illustrated in Fig.~\ref{fig:rateregion_differentCR}. As seen in the figure, even $\Theta(\log n)$ bits of common randomness is sufficient to ensure that the red and blue regions are expanded. An additional $\Theta(n)$ bits can be used to expand the blue region even further, eventually achieving everything in Fig.~\ref{fig:rateregion_infiniteCR}. 
The rates derived in this paper are compared with prior work in Table~\ref{table:rates}.\footnote{Many prior works (for example~\cite{medard,sarwate-spcom2012} etc.) also consider additional random noise, independent of the jamming noise introduced by James, on the channel from Alice to Bob.
In principle the techniques in this paper carry over directly even to that setting, but for ease of exposition we choose not to present those results.}

\subsection{Proof techniques for converse results}
We begin by outlining the proof techniques used in our converse results. At first sight, it might seem that geometric/sphere packing bounds such as in~\cite{kabatiansky-1978} may be used when Bob's NSR $N/P$ is higher than James's NSR $\sigma^2/P$, since whenever Bob can hope to decode Alice's message, so can James. If Alice's encoder is deterministic, James can therefore infer Alice's transmitted codeword, and thereby ``push'' it to the nearest codeword. However, such a reasoning applies only to deterministic codes, i.e., when Alice does not use any private or common randomness.  We therefore highlight two converse techniques that apply even when Alice's encoder is \emph{not} deterministic.

\subsubsection{Scale-and-babble}
The scale-and-babble attack is a strategy that reduces the channel from Alice to Bob into an AWGN channel. James expends a certain amount of power in cancelling the transmitted signal,
and the rest in adding independent Gaussian noise. Since the capacity of the AWGN channel cannot be increased using common randomness, the scale-and-babble attack gives an upper bound that is valid for all values of $\Nkey$. This technique gives us the rate region illustrated in Fig.~\ref{fig:rateregion_scalebabble_arg}. The capacity is upper bounded by $\Rld$ in the red region, $\Rmyop$ in the blue region, and is zero in the grey region. 

We  remark that the scale-and-babble attack is not an original idea of this work. This proof was suggested by Sarwate~\cite{sarwate-spcom2012}, and is an application of a more general technique proposed by Csisz\'ar and Narayan~\cite{csiszar-narayan-it1988-2} to convert an AVC into a discrete memoryless channel. Nevertheless we give a complete proof to keep the paper self-contained. 

\subsubsection{Symmetrization attacks}
Symmetrization attacks give us upper bounds on the throughput when Alice and Bob do not share a secret key, but hold regardless of whether Alice's encoder uses private randomness or not. We give two attacks for James: 

\begin{itemize}
    \item \emph{A $\vbfz$-aware symmetrization attack:} James picks a codeword $\vbfx'$ from Alice's codebook uniformly at random and independently of $\vbfz$.
He transmits $(\vbfx'-\vbfz)/2$ --- since $\vbfz=\vbfx+\vbfsz$ for some vector $\vbfsz$ with $\cN(0,\sigma^2)$ components, therefore Bob receives $(\vbfx+\vbfx'-\vbfsz)/2$. If $\vbfx\neq\vbfx'$, then Bob makes a decoding error with nonvanishing probability. This attack is inspired by a technique used to prove the Plotkin bound for bit-flip channels. The symmetrization attack lets us prove that the capacity is zero when $\frac{\sigma^2}{P}\leq \rev{\frac{1}{1-N/P} - 2}$ (\rev{Fig.~\ref{fig:rateregion_symmetrization_arg_c}}). The $\vbfz$-aware attack is novel in the context of myopic channels, but is also inspired by similar ideas in~\cite{tongxin-causal-2018}.
\item \emph{A $\vbfz$-agnostic symmetrization argument:} This lets us show that the capacity is zero for $N>P$ (Fig.~\ref{fig:rateregion_symmetrization_arg_b}). James picks a codeword $\vbfx'$ as before but instead transmits $\vbfs = \vbfx'$. Bob receives $\vbfx+\vbfx'$ and we can show that the probability of error is nonvanishing. The $\vbfz$-agnostic symmetrization attack was used by Csisz\'ar and Narayan~\cite{csiszar-narayan-it1991} to show that the capacity of the oblivious adversarial channel is zero for $N>P$.
\end{itemize}

The scale-and-babble attack holds for all values of $\Nkey$ since it involves reducing the channel into an equivalent AWGN channel, and the capacity of the AWGN channel cannot be increased using common randomness.
On the other hand, the symmetrization arguments are not valid when $\Nkey>0$. Indeed, we will show that strictly positive rates can be achieved in the symmetrizable regions with even $\Omega(\log n)$ bits of common randomness. 

Combining the three techniques give us the upper bounds in Fig.~\ref{fig:rateregion_converse_noCR}.

\subsection{Proof techniques for achievability results}
The achievability proofs for the three regimes of $\Nkey$ outlined above involve some common techniques. We now give a high-level description of some of the ideas. Fundamental to the achievability proofs is the concept of list-decoding. In all the achievability proofs, we use random spherical codes $\cC=\{ \vx(m,k):1\leq m\leq 2^{nR},1\leq k\leq 2^{\Nkey} \}$, where each $\vx(m,k)$ is sampled independently and uniformly from the sphere $\cS^{n-1}(0,\sqrt{nP})$ in $\bR^n$ centred at $0$ and comprising of vectors of magnitude $\sqrt{nP}$.  

\subsubsection{Myopic list-decoding}
This is a central idea in our proofs, and a novel contribution of this work. The broad idea is to use myopia to ensure that James is unable to uniquely recover the transmitted codeword. We show that if the codebook rate is sufficiently large, then there are exponentially many codewords that from James's perspective Alice could plausibly have transmitted. Due to this confusion, no attack strategy (by pushing the transmitted $\vbfx$ in the direction of the nearest other codeword $\vbfx'$, since the nearest codeword will in general be different directions for different $\vbfx$) by James is as bad as the one he could instantiate in the omniscient case. We study the list-decoding problem, where instead of recovering the transmitted message uniquely, Bob tries to output a poly$(n)$ sized list that includes the transmitted codeword. Since James is myopic, we could hope to achieve rates greater than the omniscient list-decoding capacity $\Rld$. Even with $\Nkey=0$, we can achieve a higher rate, equal to $\Rmyop$, in the blue region in Fig.~\ref{fig:rateregion_myopiclist_arg_a}. The blue region can be expanded with a larger amount of common randomness, as seen in Fig.~\ref{fig:rateregion_myopiclist_arg_b}. 
We will in fact show that the list-decoding capacity is equal to $C_{\mathrm{myop,rand}}$ \rev{(see Eqn.~\eqref{eq:sarwate_rates} for its definition)} if $\Nkey$ is large enough.

Let us briefly outline the proof techniques. We show that conditioned on $\vbfz$, the transmitted codeword lies in a strip \rev{(informally denoted by $ \strip $ for now)} approximately at a distance $\sqrt{n\sigma^2}$ to $\vbfz$. See Fig.~\ref{fig:geometry_ogs} for an illustration. If the codebook rate exceeds $\frac{1}{2}\log\left(1+\frac{P}{\sigma^2}\right)$, then this strip will contain exponentially many codewords. All these codewords are roughly at the same distance to $\vbfz$ and are therefore nearly indistinguishable from the one actually transmitted. We operate under the assumption of a more powerful adversary who has, in addition to $\vbfz$, access to an oracle. This is a matter of convenience and will greatly simplify our proofs. The oracle reveals an exponential sized subset of the codewords (that includes $\vbfx$) from the strip. We call this the oracle-given set (OGS). We prove that for most codewords in the OGS, no attack vector $\vbfs$ can eventually force a list-size greater than poly$(n)$ as long as the rate is less than $C_{\mathrm{myop,rand}}$. To prove this result, we obtain a bound on the \emph{typical} area of the decoding region $\cS^{n-1}(0,\sqrt{nP})\cap\cB^n(\vbfx+\vbfs,\sqrt{nN})$. We will show that for certain regimes of the noise-to-signal ratios (NSRs), the volume of the decoding region is typically much less than
the worst-case scenario (i.e., had James known $\vbfx$).
This gives us an improvement over the omniscient list-decoding capacity.

\begin{figure}
	\begin{center}
		\includegraphics[width = 0.4\textwidth]{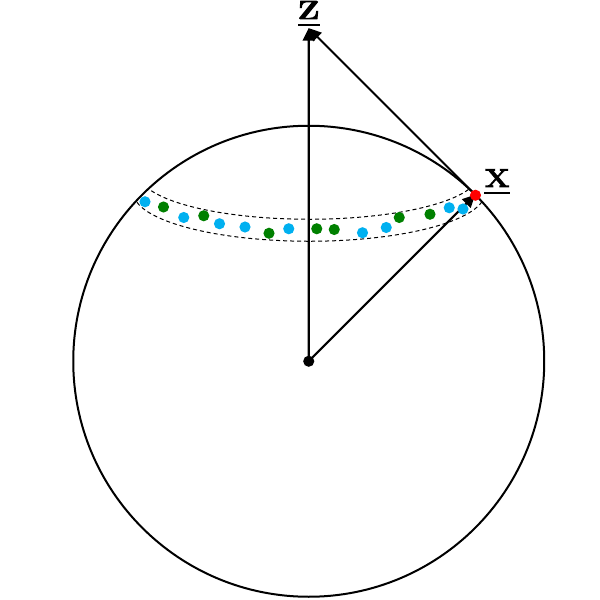}
	\end{center}
	\caption{Illustration of the strip and the oracle-given set (OGS). All codewords in the strip (the collection of red, blue and green points) are roughly at the same distance from $\vbfz$, and hence approximately have the same likelihood of being transmitted. The OGS (illustrated as red and green points) is a randomly chosen subset of the codewords in this strip.}
	\label{fig:geometry_ogs}
\end{figure}

\subsubsection{Reverse list-decoding}
This technique, along with myopic list decoding outlined above, is used to obtain achievable rates in the case where Alice and Bob do not share any common randomness. 
Given an attack vector $\vs$, we say that $\vx'$ confuses $\vx$ if $\vx'$ lies within $\cB^n(\vx+\vs,\sqrt{nN})$.
In list-decoding, we attempt to find the number of codewords that could potentially confuse the transmitted codeword. Our goal in the list-decoding problem is to keep the number of confusable codewords to a minimum. In reverse list-decoding, we ask the opposite question: Given a potentially confusing codeword, how many codewords in the strip could this confuse?

For every codeword $\vx'$ and attack vector $\vs$, we could define the reverse list-size as $|\{ \vx\in \strip : \Vert\vx+\vs-\vx'\Vert_2 \leq\sqrt{nN}\}|$. In other words, this is the number of codewords in the strip which when translated by $\vs$, are confusable with $\vx'$. Our goal is to keep this as small as possible (in fact, poly$(n)$) for all possible attack vectors $\vs$. We show that small reverse list-sizes are guaranteed as long as $\frac{1}{2}\log\frac{P}{N}>\frac{1}{2}\log\left(1+\frac{P}{\sigma^2} \right)$, i.e. in the red and blue regions of Fig.~\ref{fig:rateregion_reverse_list_arg}.

\subsubsection{Going from list-decoding to unique decoding}
Obtaining results for unique decoding uses two different approaches that depends on the amount of common randomness.
\begin{itemize}
 \item \emph{Linear/logarithmic amount of common randomness}: Langberg~\cite{langberg-focs2004} gave a combinatorial technique to convert any list-decodable code (with no common randomness) for a binary channel into a uniquely decodable code of the same rate with $\Omega(\log n)$ bits of common randomness.  This was later generalized by Sarwate~\cite{sarwate-thesis} to arbitrary AVCs, and~\cite{bhattacharya2019sharedrandomness} recently showed that  only $(1+\varepsilon)\log n$ bits suffices (where $ \varepsilon $ denotes the difference between the list decoding capacity and the transmission rate). This combined with our result on myopic list-decoding will give us an achievable rate for reliable communication over the myopic channel.
 \item \emph{No common randomness}: The ideas in myopic list-decoding can be used to show that there are at most poly$(n)$ codewords that can potentially confuse the exponentially many codewords in the OGS. Using reverse list-decoding, we can conclude that each codeword outside the OGS can confuse at most poly$(n)$ codewords in the OGS. Using this, and a ``grid argument'' along the lines of~\cite{dey-sufficiently-2015}, we can show that the probability of decoding error is vanishingly small.
\end{itemize}

\subsubsection{Two uses of the common randomness}
When Alice and Bob have only $\cO(\log n)$  bits of common randomness, $\bfk$ is only used to find the true message from the list using the approach proposed in~\cite{langberg-focs2004,sarwate-thesis}.
However, when Alice and Bob have $\Omega(n)$ bits of shared secret key, there are two different uses for $\bfk$.

Recall from our discussion above that to obtain an improvement over omniscient list-decoding, we must ensure that James is sufficiently confused about the true codeword. If he can recover $\vbfx$ with high probability (w.h.p.), then there would  be no hope of achieving rates greater than $\Rld$. When Alice and Bob share linear amounts of common randomness, all but $(1+\varepsilon)\log n$ bits  is used to ensure that the strip contains exponentially many codewords. This is done by generating $2^{\Nkey - (1+\varepsilon)\log n}$ independent codebooks, each containing $2^{nR}$ codewords. Based on the realization of the key, Alice picks the appropriate codebook for encoding. Since Bob knows the key, he can use this codebook for decoding. However, James does not have access to the shared secret key and to him, all the $2^{nR+\Nkey-(1+\varepsilon)\log n}$ codewords are equally likely to have been transmitted. This ensures that James is sufficiently confused about the true codeword. The remaining $(1+\varepsilon)\log n$  bits are then used to disambiguate the list at the decoder.

\section{Comparison with other works}
\label{sec:comparison_other_work}
We now describe the similarities and differences with three closely related works.

Sarwate~\cite{sarwate-spcom2012} derived the capacity of the myopic adversarial channel with unlimited common randomness. 
\rev{The present paper} is in some sense a continuation of~\cite{sarwate-spcom2012}, with a restriction on the amount of shared secret key. In~\cite{sarwate-spcom2012}, the codebook was privately shared by Alice and Bob. A minimum angle decoder was used,
and the achievable rate was the solution of an optimization problem identical to what we obtain in our analysis of myopic list-decoding. The converse used the scale-and-babble attack that we describe in a subsequent section.
When Alice and Bob share $\Omega(n)$ bits of common randomness, we find that the upper bound is indeed optimal. We were unable to obtain an improved upper bound for smaller amounts of common randomness. However, we do give an improved converse for certain NSRs using symmetrization when there is no common randomness shared by the encoder-decoder pair.

Dey et al.~\cite{dey-sufficiently-2015} studied the discrete myopic adversarial channel. We borrow several proof techniques from their paper to obtain results when $\Nkey=0$. This includes the ideas of blob list-decoding, reverse list-decoding, and the grid argument to prove that the probability of error is vanishingly small even when transmitting at rates above $R_{\mathrm{GV}}$, despite the fact that the jamming pattern may be correlated with the transmission. However, there are several differences between our work and the discrete case. A key difference between this work and~\cite{dey-sufficiently-2015} is our use of myopic list-decoding. A direct extension of the techniques in~\cite{dey-sufficiently-2015} would only let us achieve rates up to the omniscient list-decoding capacity $\frac{1}{2}\log \frac{P}{N}$. The study of myopic list-decoding is novel, and is one of the main contributions of this work. Furthermore, the random variables involved in this work are continuous which introduces several challenges. Several arguments involving union bounds and fixed distances do not go through directly from the discrete setting. We overcome some of these by quantizing the relevant random variables, a standard trick to approximate the continuous variables by discrete ones. In~\cite{dey-sufficiently-2015}, the oracle-given set (OGS) was chosen to be a subset of codewords all at the same distance to the vector received by James. All codewords in the OGS would then have the same posterior probability given James's observation and the OGS. As $\vbfsz$ is Gaussian in our case, such an argument cannot be used. Instead, we choose the OGS from a thin strip. Given $\vbfz$ and the OGS, the codewords are not uniformly distributed. However, we carefully choose the thickness of the strip to be small enough to ensure that this posterior distribution is close to being uniform. Due to the random variables being continuous, dealing with the quadratically constrained case requires a more careful analysis involving many more slackness parameters than the discrete analogue. 

The symmetrization argument we use in Section~\ref{sec:converse_symmetrization} is inspired by the ``scaled babble-and-push'' attack studied  by Li et al.~\cite{tongxin-causal-2018}. This work studies the attack where James generates an independent codeword $\vbfx'$ and transmits $(\vbfx'-\vbfz)/2$.
\rev{A similar idea with optimized parameters allows us to prove} that the capacity is zero in the regime where $ \frac{\sigma^2}{P}<\rev{\frac{1}{1-N/P} - 2} $. 

We could potentially extend this idea, along the lines of the scaled babble-and-push attack~\cite{tongxin-causal-2018} in the following manner. In a subset of codeword indices, James uses the scale-and-babble attack. In the remaining indices $i$, he transmits $(\bfx_i'-\bfz_i)/2$. We could hope to get an improved converse in the regime $1\leq \frac{P}{N}\leq 1+\frac{\sigma^2}{P}\leq \frac{1}{N/P-1}$. We were unsuccessful in analyzing this and it has been left as future work.

\begin{table*} 
\centering
\begin{tabular}{|p{0.09\textwidth}|p{0.33\textwidth}|p{0.40\textwidth}|p{0.06\textwidth}|}
\hline
 {\bf Reference} & {\bf Rate  } & {\bf Level of myopia} & {\bf Common randomness} \\\hline 
 Folklore \newline (Appendix~\ref{sec:prf_omniscient_list_capacity}) & $\Rld\coloneq\frac{1}{2}\log\frac{P}{N}$ for list-decoding & $\sigma^2 = 0$ & $0$  \\\hline 
 Hughes-Narayan~\cite{hughes-narayan-it1987} & $C_{\mathrm{obli},\text{rand}}\coloneq C_{\mathrm{AWGN}}\coloneq\frac{1}{2}\log\left(1+\frac{P}{N}\right)$ & $\sigma^2 = \infty$ & $\infty$ \\\hline 
 Csisz{\'a}r-Narayan~\cite{csiszar-narayan-it1991} & $C_{\mathrm{obli}}\coloneq C_{\mathrm{AWGN}}\mathds1_{\{P\ge N\}}$ & $\sigma^2 = \infty$ & $0$ \\\hline 
 \multirow{2}{0.09\textwidth}{Blachman~\cite{blachman-1962}} & $C_{\mathrm{omni}}\leq R_{\mathrm{Rankin}}\coloneq\frac{1}{2}\log\left(\frac{P}{2N}\right)\mathds1_{\{P\geq 2N\}}$ & \multirow{2}{0.40\textwidth}{$\sigma^2=0$} & \multirow{2}{0.06\textwidth}{$0$} \\
 & $C_{\mathrm{omni}}\geq R_{\mathrm{GV}}\coloneq\frac{1}{2}\log\left(\frac{P^2}{4N(P-N)}\right)\mathds1_{\{P\geq 2N\}}$ & & \\\hline 
 \multirow{2}{0.09\textwidth}{Kabatiansky-Levenshtein~\cite{kabatiansky-1978}} & $C_{\mathrm{omni}}\leq R_{\mathrm{LP}}\coloneq (\alpha\log\alpha - \beta\log\beta)\mathds1_{\{P\geq 2N\}}$,        &\multirow{2}{0.40\textwidth}{$\sigma^2 = 0$}&\multirow{2}{0.06\textwidth}{$0$} \\
 & where $\alpha\coloneq\frac{P+2\sqrt{N(P-N)}}{4\sqrt{N(P-N)}}$, $\beta\coloneq\frac{P-2\sqrt{N(P-N)}}{4\sqrt{N(P-N)}}$  & & \\\hline
 \multirow{3}{0.09\textwidth}{Sarwate~\cite{sarwate-spcom2012}} & $C_{\mathrm{myop},\text{rand}}=\Rld\coloneq\frac{1}{2}\log\frac{P}{N}$ & $\frac{\sigma^2}{P}\le\frac{1}{N/P}-1$ & \multirow{3}{0.06\textwidth}{$\infty$}  \\
 & $C_{\mathrm{myop},\text{rand}}=\Rmyop \coloneq\frac{1}{2}\log\left(\frac{(P+\sigma^2)(P+N)-2P\sqrt{N(P+\sigma^2)}}{N\sigma^2}\right)$ & $\frac{\sigma^2}{P}\geq\max\left\{\frac{1}{N/P}-1, \frac{N}{P}-1\right\}$ & \\
 & $C_{\mathrm{myop},\text{rand}}=0$ & $\frac{\sigma^2}{P}\leq \frac{N}{P}-1$ & \\\hline
 \multirow{5}{0.09\textwidth}{\textbf{This work} (Theorem~\ref{thm:capacity_noCR})} & $C_{\mathrm{myop}}=\Rld$ & $\frac{1}{1-N/P}-1\le\frac{\sigma^2}{P}\le\frac{1}{N/P}-1$ & \multirow{5}{0.06\textwidth}{$0$} \\
 & $C_{\mathrm{myop}}=\Rmyop $ & $\frac{\sigma^2}{P}\ge\max\left\{\frac{1}{N/P}-1,\frac{1}{1-N/P}-1\right\}$ & \\
 & $R_{\mathrm{GV}}\le C_{\mathrm{myop}}\le \Rld$ & $\rev{\frac{1}{1-N/P} - 2}\le\frac{\sigma^2}{P}\le\min\left\{\frac{1}{N/P}-1,\frac{1}{1-N/P}-1\right\}$ & \\
 & $R_{\mathrm{GV}}\le C_{\mathrm{myop}}\le \Rmyop $ & $\max\left\{\frac{1}{N/P}-1,\rev{\frac{1}{1-N/P} - 2}\right\}\le\frac{\sigma^2}{P}\le\frac{1}{1-N/P}-1,\frac{N}{P}\le1$ & \\
 & $C_{\mathrm{myop}}= 0$ & $\frac{\sigma^2}{P}\le\rev{\frac{1}{1-N/P} - 2}$ or $\frac{N}{P}\ge1$ & \\\hline
 \multirow{4}{0.09\textwidth}{\textbf{This work} (Lemma~\ref{lemma:achievablerate_Olognbits})}
 & $C_{\mathrm{myop}}=\Rld$ & $\frac{\sigma^2}{P}\le\frac{1}{N/P}-1$ & \multirow{4}{0.06\textwidth}{$\Theta(\log n)$} \\
 & $C_{\mathrm{myop}}=\Rmyop $ & $\frac{\sigma^2}{P}\ge\max\left\{\frac{1}{N/P}-1,4\frac{N}{P}-1\right\}$ & \\
 & $\Rld\le C_{\mathrm{myop}}\le \Rmyop $ & $\max\left\{\frac{1}{N/P}-1,\frac{N}{P}-1\right\}\le\frac{\sigma^2}{P}\le 4\frac{N}{P}-1$ & \\
 & $C_{\mathrm{myop}}=0$ & $\frac{\sigma^2}{P}\leq \frac{N}{P}-1$ & \\\hline
 \multirow{4}{0.09\textwidth}{\textbf{This work} (Lemma~\ref{lemma:achievablerate_thetanbits})}
 & $C_{\mathrm{myop}}=\Rld$ & $\frac{\sigma^2}{P}\le\frac{1}{N/P}-1$ & \multirow{4}{0.06\textwidth}{$n\Rkey$} \\
 & {$C_{\mathrm{myop}}=\Rmyop $} & $\frac{\sigma^2}{P}\ge\max\left\{\frac{1}{N/P}-1,\frac{N}{P}-1\right\},$ & \\
 & & $\text{ and } \Rkey>\frac{1}{2}\log\left( 1+\frac{P}{\sigma^2} \right)-\Rmyop$ & \\
 & {$\Rld\leq C_{\mathrm{myop}}\leq \Rmyop$} & $\frac{\sigma^2}{P}\ge\max\left\{\frac{1}{N/P}-1,\frac{N}{P}-1\right\},$ & \\
 & & $\text{ and } \Rkey<\frac{1}{2}\log\left( 1+\frac{P}{\sigma^2} \right)-\Rmyop$ & \\
 & $C_{\mathrm{myop}}=0$ & $\frac{\sigma^2}{P}\leq \frac{N}{P}-1$ & \\\hline
\end{tabular} 
\caption{Summary of results for the adversarial channel with quadratic constraints. Here, $\mathds1_{\{P\geq N\}}$ is $1$ when $P\geq N$ and zero otherwise. Also, $ C_J\coloneq \frac{1}{2}\log\left(1+\frac{P}{\sigma^2}\right) $.  Rates known in prior work are plotted in Fig.~\ref{fig:achievable_rate_prior}. \rev{Note that in any entry in the third column ``\textbf{Level of myopia}'', all regimes of $ \sigma^2/P $ as a function of $ N/P $ are disjoint. In particular, in the last four entries of that column, only one of the cases $ \sigma^2/P \le 1/(N/P) - 1 $ and $ \sigma^2/P \le N/P - 1 $ can occur for any fixed value of $ N/P $.}}
\label{table:rates}
\end{table*}

\begin{table*} 
	\centering
	\begin{tabular}{|p{0.10\textwidth}|p{0.44\textwidth}|p{0.37\textwidth}|p{0.06\textwidth}|}
		\hline
		{\bf Reference} & {\bf Rate  } & {\bf Level of myopia} & {\bf Common randomness} \\\hline
		\multirow{2}{0.10\textwidth}{Goldfeld et al.~\cite{goldfeld2016arbitrarily}} & $ C_{\mathrm{obli,sec}} \geq \max\left\{\frac{1}{2}\log \left( 1+\frac{P}{N} \right)-\frac{1}{2}\log \left( 1+\frac{P}{\sigma_{\mathrm{eve}}^2} \right),0\right\}$ & {$ P>N $. Eavesdropper and jammer are independent of each other, and $ \sigma^2=\infty $. Here, $ \sigma^2_{\mathrm{eve}} $ denotes noise variance to the eavesdropper.} & 0\\
		& $C_{\mathrm{obli,sec}}=0  $ & $ P\leq N $ and $ \sigma^{2}=\infty $ &
		\\\hline
		\multirow{4}{0.10\textwidth}{\textbf{This work}  (Lemma~\ref{lemma:capacity_noCR_secrecy})}
		& $C_{\mathrm{myop,sec}}\geq \Rld-C_J$ & $\frac{1}{1-N/P}-1\le\frac{\sigma^2}{P}\le\frac{1}{N/P}-1$ & \multirow{4}{0.06\textwidth}{$0$} \\
		& $C_{\mathrm{myop,sec}}\geq \Rmyop-C_J $ & $\frac{\sigma^2}{P}\ge\max\left\{\frac{1}{N/P}-1,\frac{1}{1-N/P}-1\right\}$ & \\
		& $C_{\mathrm{myop,sec}}= 0$ & $\frac{\sigma^2}{P}\le\rev{\frac{1}{1-N/P} - 2}$ or $\frac{N}{P}\ge1$ &  \\
		& $ C_{\mathrm{myop,sec}}\ge R_{GV}-C_J$ & otherwise & \\\hline
		\multirow{3}{0.10\textwidth}{\textbf{This work}  (Lemma~\ref{lemma:achievablerate_Olognbits_secrecy})}
		& $C_{\mathrm{myop}}=0$ & $\frac{\sigma^2}{P}\leq \frac{N}{P}-1$ & \multirow{3}{0.06\textwidth}{$\Theta(\log n)$} \\
		& $C_{\mathrm{myop,sec}}\geq\Rmyop-C_J $ & $\frac{\sigma^2}{P}\ge\max\left\{\frac{1}{N/P}-1,4\frac{N}{P}-1\right\}$ & \\
		& $C_{\mathrm{myop,sec}}\geq\Rld-C_J$ & otherwise&
		\\\hline
		\multirow{3}{0.10\textwidth}{\textbf{This work}  (Lemma~\ref{lemma:achievablerate_thetanbits_secrecy})}
		& $C_{\mathrm{myop}}=0$ & $\frac{\sigma^2}{P}\leq \frac{N}{P}-1$ & \multirow{3}{0.06\textwidth}{$n\Rkey$} \\
		& {$C_{\mathrm{myop,sec}}\geq\min\{\Rmyop, \Rmyop-C_J+\Rkey\} $} & $\frac{\sigma^2}{P}\ge\max\left\{\frac{1}{N/P}-1,\frac{N}{P}-1\right\},$ & \\
		& & $\text{ and } \Rkey>\frac{1}{2}\log\left( 1+\frac{P}{\sigma^2} \right)-\Rmyop$ & \\
		& & $\text{ and } \Rkey<\frac{1}{2}\log\left( 1+\frac{P}{\sigma^2} \right)-\Rmyop$ & \\
		& $C_{\mathrm{myop,sec}}\geq\min\{\Rld, \Rld-C_J+\Rkey\}$ & otherwise &
		\\\hline
	\end{tabular} 
	\caption{Summary of results for the adversarial channel with quadratic constraints  and wiretap secrecy. Here, $\mathds1_{\{P\geq N\}}$ is $1$ when $P\geq N$ and zero otherwise. Also, $ C_J\coloneq \frac{1}{2}\log\left(1+\frac{P}{\sigma^2}\right) $. Prior works mostly considered the case where the eavesdropper is independent of the jammer and observes a degraded version of $ \vbfx $, while the jammer must choose his signals obliviously of $ \vbfx $ and $ \bfm $.}
	\label{table:rates_secrecy}
\end{table*}

\begin{figure} 
 \begin{center}
  \includegraphics[width=0.5\textwidth]{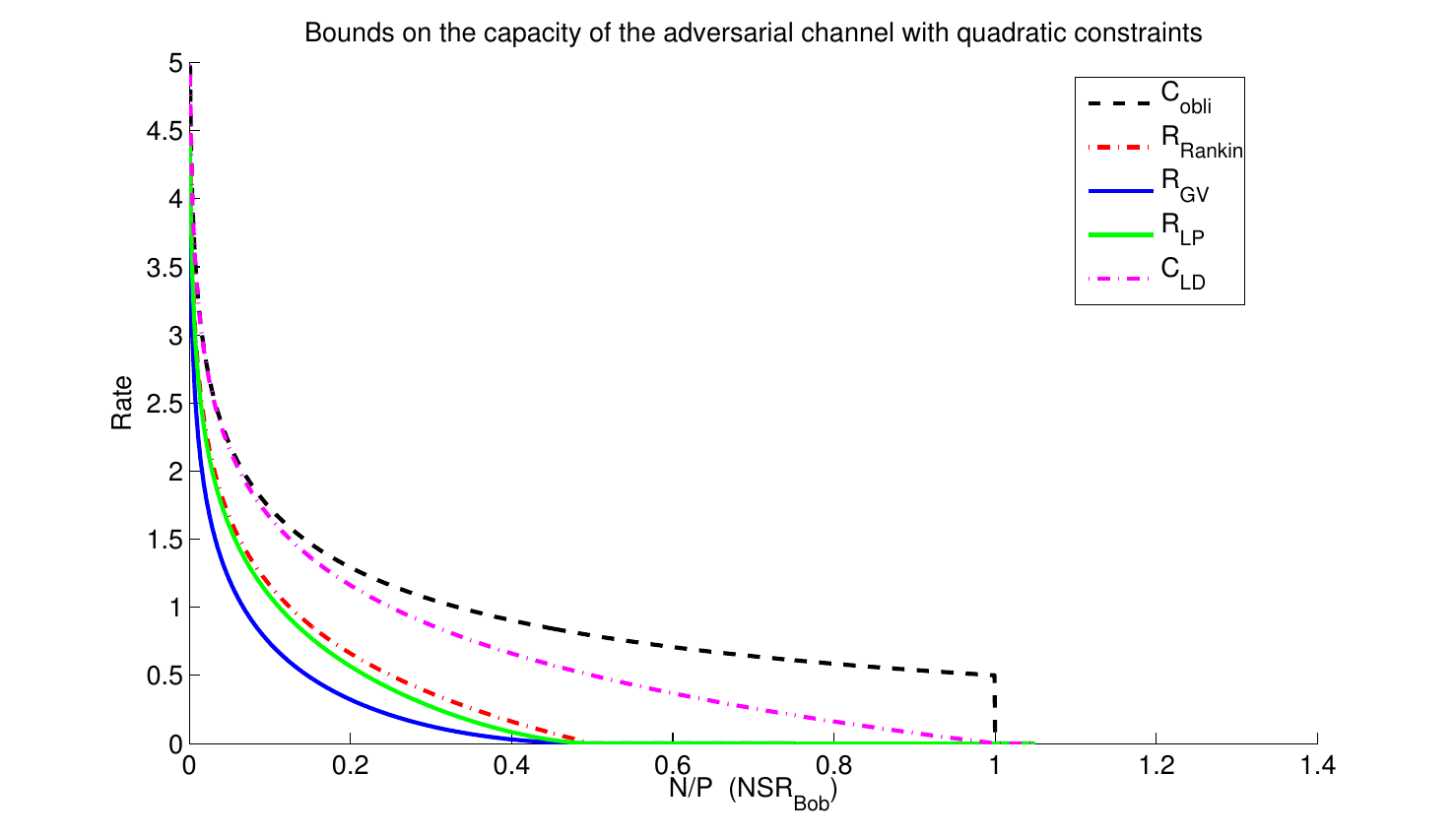}
  \caption{Achievable rates for the quadratically constrained adversarial channel --- prior work.}
  \label{fig:achievable_rate_prior}
 \end{center}
\end{figure}

\section{Notation}\label{sec:notation}
Random variables are denoted by lower case Latin letters in boldface, e.g., $\bfm$. Their realizations are denoted by corresponding letters in plain typeface, e.g., $m$. Vectors of length $n$, where $n$ is the block-length, are denoted by lower case Latin letters with an underline, e.g., $\vbfx,\vbfs,\vx,\vs$, etc. The $i$th entry of a vector is denoted by a subscript $i$, e.g., $\vbfx_i,\vbfs_i,\vx_i,\vs_i$, etc. Matrices are denoted by capital Latin/Greek letters in boldface, e.g., $\bfI,\mathbf{\Sigma}$, etc. 

Sets are denoted by capital Latin letters in calligraphic typeface, e.g., $\cC,\cI$, etc. In particular, an $(n-1)$-dimensional sphere in $n$-dimensional Euclidean space centered at $\vx$ of radius $r$ is denoted by
\[\cS^{n-1}(\vx,r)=\{\vy\in\bR^n:\|\vy\|_2=r\}.\]
An $n$-dimensional ball in Euclidean space centered at $\vx$ of radius $r$ is denoted by
\[\cB^n(\vx,r)=\{\vy\in\bR^n:\|\vy\|_2\le r\}.\]
As shown in Figure~\ref{fig:cap}, an $(n-1)$-dimensional cap centered at $\vx$ of radius $r$ living on an $(n-1)$-dimensional sphere of radius $r'$ is denoted by
\begin{align*}
 \C^{n-1}(\vx,r,r')&=\{\vy\in\cS^{n-1}(O,r'):\|\vy-\vx\|_2\le r\}\\
 &=\cB^n(\vx,r)\cap\cS^{n-1}(O,r').
\end{align*}
\begin{figure} 
    \centering
    \mbox{}\hfill
    \begin{subfigure}[t] {.42\textwidth}
        \centering
        \includegraphics[width=0.9\linewidth]{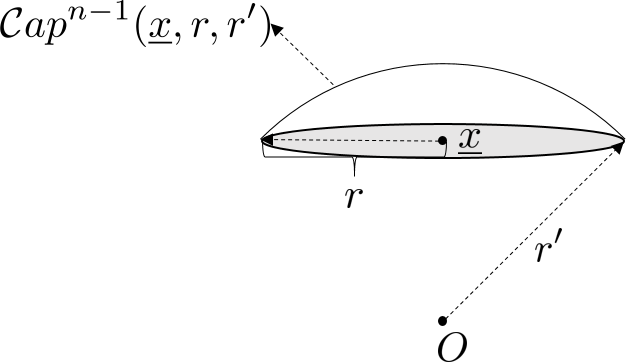}
        \caption{An $(n-1)$-dimensional cap $\C^{n-1}(\vx,r,r')$ centered at $\vx$ of radius $r$ living on an $(n-1)$-dimensional sphere of radius $r'$.}
        \label{fig:cap}
    \end{subfigure}
    \hfill
    \begin{subfigure}[t] {.42\textwidth}
        \centering
        \includegraphics[width=0.9\linewidth]{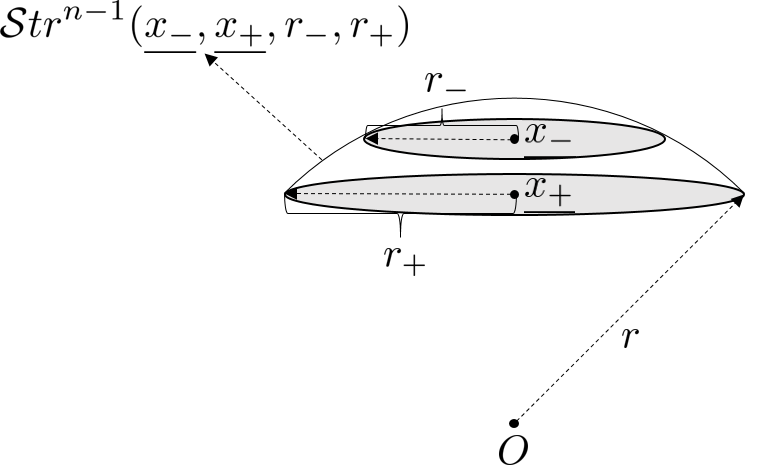}
        \caption{An $(n-1)$-dimensional strip $\strip^{n-1}(\underline{x_-},\underline{x_+},r_-,r_+)$ centered at $\underline{x_-}$ and $\underline{x_+}$ of radii $r_-$ and $r_+$.}
        \label{fig:strip}
    \end{subfigure}
    \hfill\mbox{}
    \caption{The geometry of caps and strips.}
    \label{fig:my_label}
\end{figure}
As shown in Figure~\ref{fig:strip}, an $(n-1)$-dimensional strip centered at $\underline{x_-}$ and $\underline{x_+}$ of radii $r_-$ and $r_+$ is denoted by
\begin{align*}
    \strip^{n-1}(\underline{x_-},\underline{x_+},r_-,r_+)=&\{\vx\in\cS^{n-1}(r):\|\vx-\underline{x_-}\|_2\ge r_-,\;\|\vx-\underline{x_+}\|_2\le r_+\}\\
    =&\cB^n(\underline{x_-},r_-)^c\cap\cB^n(\underline{x_+},r_+)\cap\cS^{n-1}(O,r),
\end{align*}
where $r$ satisfies $\sqrt{r^2-r_-^2}-\sqrt{r^2-r_+^2}=\Vert\underline{x_-}-\underline{x_+}\Vert$. An $n$-dimensional shell centered at $\vx$ of inner radius $r_{\mathrm{in}}$ and outer radius $r_{\mathrm{out}}$, where $r_{\mathrm{out}}>r_{\mathrm{in}}$, is denoted by 
\begin{align*}
    \sh^n(\vx,r_{\mathrm{in}},r_{\mathrm{out}})=&\sh^n\left(\vx,\frac{r_{\mathrm{in}}+r_{\mathrm{out}}}{2}\pm\frac{r_{\mathrm{out}}-r_{\mathrm{in}}}{2}\right)\\
    =&\cB^n(\vx,r_{\mathrm{out}})\backslash\cB^n(\vx,r_{\mathrm{in}}).
\end{align*}
Let $\vol(\cdot)$ denote the Lebesgue volume of a Euclidean body  and let  $\area(\cdot)$ denote its Lebesgue area of an Euclidean surface. For $M\in\bZ_{>0}$, we let $[M]$ denote the set of first $M$ positive integers $\{1,2,\cdots,M\}$. 

The probability mass function (p.m.f.) of a discrete random variable $\bfx$ or a random vector $\vbfx$ is denoted by $p_{\bfx}$ or $p_{\vbfx}$. Here with a slight abuse of notation, we use the same to denote the probability density function (p.d.f.) of $\bfx$ or $\vbfx$ if they are continuous.
If every entry of $\vbfx$ is independently and identically distributed (i.i.d.) according to $p_{\bfx}$, then we write $\vbfx\sim p_{\bfx}^{\tn}$. In other words,
\[p_{\vbfx}(\vx)=p_{\bfx}^{\tn}(\vx)\coloneq\prod_{i=1}^np_{\bfx}(\vx_i).\]
 Let $\unif(\Omega)$ denote the uniform distribution over some probability space $\Omega$. Let $\cN(\underline{\mu},\mathbf\Sigma)$ denote the $n$-dimensional Gaussian distribution with mean vector $\underline{\mu}$ and covariance matrix $\mathbf\Sigma$. 

The indicator function is defined as, for any $\cA\subseteq\Omega$ and $x\in\Omega$, 
\[\one_{\cA}(x)=\begin{cases}1,&x\in \cA\\0,&x\notin \cA\end{cases}.\]
At times, we will slightly abuse notation by saying that $\one_{\mathsf{A}}$ is $1$ when relation $\mathsf{A}$ is satisfied and zero otherwise.
We use standard Bachmann-Landau (Big-Oh) notation for asymptotic functions. All logarithms are to the base two.

We use $H(\cdot)$ to denote interchangeably Shannon entropy and differential entropy; the exact meaning  will usually be clear from context.

\section{Preliminaries}\label{sec:preliminaries}
\noindent\textbf{Arbitrarily Varying Channel (AVC).} A channel with a state controlled by an adversarial jammer is called an AVC in the literature. James's jamming strategy is a (potentially probabilistic) map which, based on his observation, constructs an \emph{attack vector} $\vs$ satisfying his \emph{maximum power constraint} $\|\vs\|_2\le\sqrt{nN}$,
\begin{equation*}
    \begin{array}{llll}
        \mathrm{Jam}:&\bR^n&\to&\bR^n\\
        &\vz&\mapsto&\vs
    \end{array}.
\end{equation*}

\noindent\textbf{Code.} A \emph{deterministic encoder} is a deterministic map which encodes a \emph{message} to a \emph{codeword} of length $n$, where $n$ is called \emph{block-length} or the \emph{number of channel uses}, satisfying Alice's \emph{maximum power constraint} $\|\vx(m)\|_2\le\sqrt{nP}$,
\begin{equation*}
    \begin{array}{llll}
        \enc:&\{0,1\}^{nR}&\to&\bR^n\\
        &m&\mapsto&\vx(m)
    \end{array},
\end{equation*}
where $R$ is the \emph{rate} of the system. Alice uses her encoder to encode the set of messages $\{0,1\}^{nR}$ and get a \emph{codebook} $\{\vx(m)\}_{m=1}^{2^{nR}}$ which is simply the collection of codewords. 

A \emph{deterministic decoder} is a deterministic function which maps Bob's observation to a reconstruction of the message,
\begin{equation*}
    \begin{array}{llll}
        \dec:&\bR^n&\to&\{0,1\}^{nR}\\
        &\vy&\mapsto&{\widehat m}
    \end{array}.
\end{equation*}
An \emph{$(n,R,P,N)$ deterministic code} $\cC$ is a deterministic encoder-decoder pair $(\enc,\dec)$. Sometimes we also slightly abuse the notation and call the set of codewords $\{\vx(m)\}_{m=1}^{2^{nR}}$ a code.

We distinguish between three types of codes:
\begin{itemize}
 \item \emph{Deterministic codes}: The encoder is a deterministic map from $\{0,1\}^{nR}$ to $\bR^n$, and the decoder is a deterministic map from $\bR^n$ to $\{0,1\}^{nR}$.
 \item \emph{Stochastic codes}: The encoder and decoder are allowed to use private randomness. If Alice and Bob have $n_{A}$ and $n_{B}$ bits of private randomness respectively, then a stochastic encoder is a map
 \begin{equation*}
    \begin{array}{llll}
        \enc:&\{0,1\}^{nR}\times \{0,1 \}^{n_A}&\to&\bR^n\\
        &(m,k_A)&\mapsto&\vx(m,k)
    \end{array},
\end{equation*}
while the decoder is a map
\begin{equation*}
    \begin{array}{llll}
        \dec:&\bR^n\times \{0,1 \}^{n_B}&\to&\{0,1\}^{nR}\\
        &(\vy, k_B)&\mapsto&{\widehat m}
    \end{array}.
\end{equation*}
Here, $k_A$ is known only to Alice while $k_B$ is known only to Bob.
\item \emph{Randomized codes}: Alice and Bob share $\Nkey$ bits of common randomness, which is kept secret from James. They may additionally have $n_A$ and $n_B$ bits of private randomness respectively.
The encoder is a map
\begin{equation*}
    \begin{array}{llll}
        \enc:&\{0,1\}^{nR}\times \{0,1 \}^{n_A}\times\{0,1 \}^{\Nkey}&\to&\bR^n\\
        &(m,k_A,k)&\mapsto&\vx(m,k_A,k)
    \end{array},
\end{equation*}
while the decoder is a map
\begin{equation*}
    \begin{array}{llll}
        \dec:&\bR^n\times \{0,1 \}^{n_B}\times\{0,1 \}^{\Nkey} &\to&\{0,1\}^{nR}\\
        &(\vy, k_B,k) &\mapsto&{\widehat m}
    \end{array}.
\end{equation*}
Here, $k$ is known to Alice and Bob, but not to James. The private randomness $k_A$ is known only to Alice, while $k_B$ is known only to Bob.
\end{itemize}
In all our code constructions, we will not use any private randomness, i.e., $n_A=n_B=0$. However, our converse results are true for all values of $n_A$ and $n_B$.

\noindent\textbf{Probability of error.} A decoding error occurs when Bob's reconstruction does not match Alice's message. The \emph{average} (over messages) \emph{probability of error} $\pe^{\text{avg}}$ (also denoted by $\pe$ for notational brevity in this paper) is defined as
\[\pe= \rev{\sup_{\vbfs }} \p{}(\widehat\bfm\ne\bfm)
    =\rev{\sup_{\vbfs }}\sum_{m=1}^{2^{nR}}\p{}(\widehat\bfm\ne\bfm,\bfm=m)
    =\rev{\sup_{\vbfs }}\frac{1}{2^{nR}}\sum_{m=1}^{2^{nR}}\p{}(\widehat\bfm\ne m|\bfm=m).\]
where the probability is taken over the private randomness in the encoder-decoder pair, the common randomness shared  by Alice and Bob, the noise to James, and any additional randomness he may use in choosing $\vbfs$.
\rev{The maximization is taken over all (potentially stochastic) functions $ \vbfs\colon\bR^n\to\cB^n(0,\sqrt{nP}) $ which map James's observation $ \vbfz $  to a jamming sequence $ \vbfs(\vbfz) $.\footnote{The output of the jammer can also depend on the (potentially stochastic) codebook  used by Alice and Bob, and his own private randomness. However, it cannot depend on the common randomness shared only by the encoder-decoder pair. We omit these dependences in the notation for brevity.}}
For deterministic codes, 
\begin{align*}
    &\p{}(\widehat\bfm\ne m|\bfm=m)\\
    =&\intg{\cB^n(0,\sqrt{nN})}{}\intg{\bR^n}{}p_{\vbfz|\vbfx}(\vz|\vx(m))p_{\vbfs|\vbfz}(\vs|\vz)\one_{\{\dec(\vx(m)+\vs)\ne m\}}\diff\vz\diff\vs.
\end{align*}
For stochastic codes,
\begin{align*}
    &\p{}(\widehat\bfm\ne m|\bfm=m)\\
    =&\frac{1}{2^{n_B}}\sum_{k_B=1}^{2^{n_B}}\intg{\cB^n(0,\sqrt{nN})}{}\intg{\bR^n}{}\frac{1}{2^{n_A}}\sum_{k_A=1}^{2^{n_A}}p_{\vbfz|\vbfx}(\vz|\vx(m,k_A))p_{\vbfs|\vbfz}(\vs|\vz)\one_{\{\dec(\vx(m,k_A)+\vs,k_B)\ne m\}}\diff\vz\diff\vs.
\end{align*}
For randomized codes,
\begin{align*}
    &\p{}(\widehat\bfm\ne m|\bfm=m)\\
    =&\frac{1}{2^{n_B}}\sum_{k_B=1}^{2^{n_B}}\intg{\cB^n(0,\sqrt{nN})}{}\intg{\bR^n}{}\frac{1}{2^{n_A}}\sum_{k_A=1}^{2^{n_A}}\frac{1}{2^{\Nkey}}\sum_{k=1}^{2^{\Nkey}}p_{\vbfz|\vbfx}(\vz|\vx(m,k_A,k))p_{\vbfs|\vbfz}(\vs|\vz)\one_{\{\dec(\vx(m,k_A,k)+\vs,k_B,k)\ne m\}}\diff\vz\diff\vs.
\end{align*}


\noindent\textbf{Rate and capacity.} A rate $R$ is said to be \emph{achievable} if there exists a sequence of $(n,R,P,N)$ codes $\cC^{(n)}$ labelled by block-length $n$ such that each code in the sequence has rate $R^{(n)}$ at least $R$ and average probability of error $\pe^{(n)}$ vanishing in $n$, i.e., 
\[\forall n,\;R^{(n)}\ge R,\quad\text{and}\quad\lim_{n\to\infty}\pe^{(n)}=0.\]
The \emph{capacity} $C$ of a communication system is the supremum of all achievable rates.

\medskip
\noindent\textbf{List decoding}
%
%
%
%
\begin{definition}
	Fix $R>0$ and $\Nkey\geq 0$.
	A codebook $\cC=\{ \vx(m,k):m\in [2^{nR}], k\in [2^{\Nkey}] \}$  is said to be $(P,N,L)$-list-decodable at rate $R$ with $\Nkey$ bits of common randomness if
	\begin{itemize}
		\item $\Vert \vx(m,k) \Vert_2 \leq \sqrt{nP}$ for all $m,k$; and
		\item for all possible randomized functions $\vbfs\coloneq \vbfs(\cC,\rev{\vbfx})$ satisfying $\p{}(\Vert\vbfs\Vert_2\leq\sqrt{nN})=1$, we have
		\[
		\p{}(|\cB^n({\vbfx}+\vbfs,\sqrt{nN})\cap \cC^{(\bfk)}|>L) = o(1),
		\]
		where $\cC^{(k)}\coloneq \{\vx(m,k): m\in [2^{nR}]\}$ \rev{and the shorthand notation $ \vbfx = \vbfx(\bfm,\bfk) $ is the (potentially stochastic) encoding of $ \bfm $ under common randomness $\bfk$}. In the above equation, the averaging is over the randomness in\footnote{Note that we are using an average probability of error in our work. This is different from a maximum probability of error, where the list-size is required to be less than or equal to $L$ for every codeword. On the other hand, we are satisfied with this being true for all but a vanishingly small fraction of the codewords.} $\bfm,\bfk,\vbfsz$ and $\vbfs$.
	\end{itemize}
	A rate $R$ is said to be achievable for $(P,N,L)$-list-decoding with $\Nkey$ bits of common randomness if there exist sequences of codebooks (in increasing $n$) that are $(P,N,L)$-list-decodable.
	The list-decoding capacity is the supremum over all achievable rates.
	\label{def:listdecodablecode}
\end{definition}

\rev{\begin{remark}
An error in list decoding can only occur if the list is too big, i.e., larger than $ L $. 
\end{remark}}

\begin{remark}
While sticking with maximum power constraint and average probability of error, 
\rev{there also exist other jamming power constraints and error criteria in the literature.}
\rev{They are not always equivalent  and }
we  strictly distinguish \rev{them in this remark}. 
\begin{enumerate}
    \item Maximum vs. average power constraint for James. If we use an average power constraint for James, then no positive rate is achievable under both maximum and average probability of error. This is because James can focus only on a small constant fraction of codewords and allocate large enough power to completely corrupt them to ensure Bob has no hope to decode them. Then Bob's (maximum/average) probability of error is bounded away from zero while James's jamming strategy still satisfies his average power constraint. We will therefore only consider maximum power constraint on James. 
    \item Maximum vs. average probability of error. As we know, in information theory, an achievability result under maximum probability of error is stronger than that under average one, while it is the other way round for a converse result. In adversarial channel model, if we adopt maximum probability of error, notice that it suffices for James to corrupt the transmission of only one message. Thus we may assume that James knows the transmitted message a priori (but not necessarily the transmitted codeword if Alice uses stochastic encoding where a message potentially corresponds to a bin of codewords). 
\end{enumerate}
More formal justification of these criteria and their effect on capacity are given by Hughes and Narayan~\cite{hughes-narayan-it1987}. The authors defined the notion of $\lambda$-capacity for different criteria.
\end{remark}

\noindent\textbf{Probability distributions of interest.} Alternatively, the system can be described using probability distributions. We assume the messages are uniformly distributed, i.e., $p_{\bfm}=\unif([2^{nR}])$. Given the message to be transmitted, the codewords are distributed according to $p_{\vbfx|\bfm}$. Notice that it is not necessarily a 0-1 distribution, i.e., the codeword may not be entirely determined by the message, as Alice may have access to some private randomness and use \emph{stochastic encoding}. Each message may be associated to many codewords and Alice samples one codeword according to the distribution $p_{\vbfx|\bfm}$ and transmits it. James receives a corrupted version $\vbfz$ of $\vbfx$ through an AWGN channel specified by 
\[p_{\vbfz|\vbfx}(\vz|\vx)=p_{\bfz|\bfx}^{\tn}(\vz|\vx)=p_{\bfz|\bfx}^{\tn}(\vx+(\vz-\vx)|\vx)=p_{\bfsz}^{\tn}(\vz-\vx),\]
where $p_{\bfsz}^{\tn}=\cN(0,\sigma^2\bfI_n)$. Based on his observation, James designs an attack vector $\vbfs$ according to his jamming strategy specified by $p_{\vbfs|\vbfz}$. Again, notice that it is not necessarily a 0-1 distribution as James may have access to private randomness and the output of his jamming may not be deterministic. Then Bob receives $\vbfy$ which is the sum of $\vbfx$ and $\vbfs$. In particular, $\vbfy$ is a deterministic function of the codeword transmitted in the main channel and the attack vector added to it, i.e.,
\[p_{\vbfy|\vbfx,\vbfs}(\vy|\vx,\vs)=\one_{\{\vy=\vx+\vs\}}.\]
Based on his observation, Bob reconstructs $\widehat\bfm$ using a (potentially stochastic) decoder specified by $p_{\widehat\bfm|\vbfy}$.

\noindent\textbf{Area and volume.} The area of an $(n-1)$-dimensional Euclidean sphere of radius $r$ is given by
\begin{fact}
\begin{equation*}
\area(\cS^{n-1}(\cdot,r))=\frac{2\pi^{n/2}}{\Gamma(n/2)}r^{n-1}.
\end{equation*}	
\label{fact:areaofcap}
\end{fact}
The area of an $(n-1)$-dimensional cap centered at $\vx$ of radius $r$ living on an $(n-1)$-dimensional sphere of radius $r'$ can be lower bounded by the volume of an $(n-1)$-dimensional ball centered at $\vx$ of radius $r$ since the intersection of an \rev{$n$}-dimensional \rev{ball} and an $(n-1)$-dimensional hyperplane is an $(n-1)$-dimensional ball, as shown in Figure~\ref{fig:apx}, i.e., 
\begin{figure} 
    \centering
    \includegraphics[width = 0.4\textwidth]{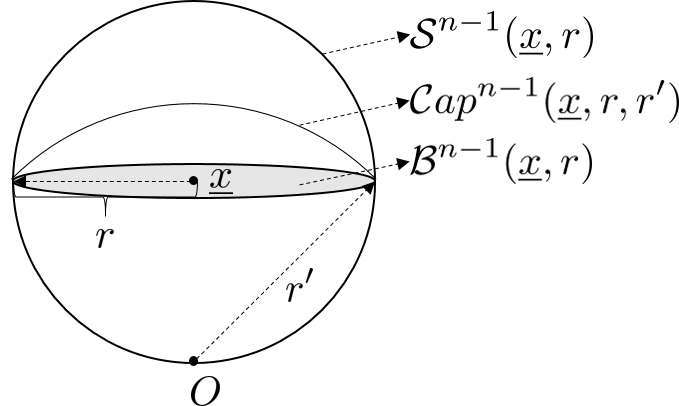}
    \caption{Approximation of surface area. The surface area of a cap $\C^{n-1}(\vx,r,r')$ is upper bounded by that of a sphere $\cS^{n-1}(\vx,r)$. It is lower bounded by the volume of a lower dimensional ball $\cB^{n-1}(\vx,r)$ since the intersection of an $(n-1)$-dimensional hyperplane parallel to the bottom of the cap passing through $\vx$ and the ball $\cB^n(O,r')$ whose surface the cap lives on is an $(n-1)$-dimensional ball $\cB^{n-1}(\vx,r)$.}
    \label{fig:apx}
\end{figure}
\begin{fact}
\[\area(\C^{n-1}(\vx,r,r'))\ge\vol(\cB^{n-1}(\vx,r)).\]	
	\label{fact:areacap_lowerbound_ball}
\end{fact}
The area of a cap can also be upper bounded by a sphere of the same radius, i.e.,
\begin{fact}
\[\area(\C^{n-1}(\vx,r,r'))\le\frac{1}{2}\area(\cS^{n-1}(\vx,r))\le\area(\cS^{n-1}(\vx,r)).\]
\label{fact:areacap_upperbound_sphere}
\end{fact}

The volume of an $n$-dimensional Euclidean ball of radius $r$ is given by
\begin{fact}
\[\vol(\cB^{n}(\cdot,r))=\frac{\pi^{n/2}}{\Gamma(n/2+1)}r^n.\]	
	\label{fact:vol_nball}
\end{fact}
More facts about high-dimensional geometry can be found in the notes by Ball~\cite{ball-convex-geom-notes-1997}.

\noindent\textbf{Error event decomposition.} We will frequently apply the following fact to decompose various decoding error events.
\begin{fact}\label{fact:error_event_decompo_lem}
For any two events $\cA$ and  $\cB$, we have $\p{}(\cA)\le\p{}(\cB)+\p{}(\cA|\cB^c)$.
\end{fact}

\noindent\textbf{Basic tail bounds.} Standard tail bounds (see, for instance, the monograph by Boucheron et al.~\cite{boucheron-etal-concentration}) for Gaussians and $\chi^2$-distributions are used at times throughout this paper.
\begin{fact}\label{fact:gaussian_tail}
If $\bfg\sim\cN(0,\sigma^2)$, then $\p{}(|\bfg|\ge\varepsilon)\le2\exp\left(-\frac{\varepsilon^2}{2\sigma^2}\right)$.
\end{fact}
\begin{fact}\label{fact:gaussian_norm}
If $\vbfg\sim\cN(0,\sigma^2\bfI_n)$, then $\|\vbfg\|_2^2$ has (scaled) $\chi^2$-distribution and
\[\p{}(\|\vbfg\|_2^2\ge n\sigma^2(1+\varepsilon))\le\exp\left(-\frac{\varepsilon^2}{4}n\right),\]
\[\p{}(\|\vbfg\|_2^2\le n\sigma^2(1-\varepsilon))\le\exp\left(-\frac{\varepsilon^2}{2}n\right),\]
\[\p{}(\|\vbfg\|_2^2\notin n\sigma^2(1\pm\varepsilon))\le2\exp\left(-\frac{\varepsilon^2}{4}n\right).\]
\end{fact}

The following lemma is proved in Appendix~\ref{sec:appendix_proofs_basiclemmas}.
\begin{lemma}\label{lemma:beta_tail}
Fix $\zeta>0$ and $\vb\in\bR^n$. If $\vbfa$ is isotropically distributed on the sphere $\cS^{n-1}(0,\|\va\|_2)$, then
\[\p{}(|\langle\vbfa,\vb\rangle|>n\zeta)\le2^{-\frac{(n-1)n^2\zeta^2}{2\|\va\|_2^2\|\vb\|_2^2}}.\]
\end{lemma}

We now state a general lemma that will be useful in proving list-decoding results. This essentially says that if codewords are chosen i.i.d. according to a uniform distribution, then the probability that more than $\cO(n^2)$ codewords lie within any set of sufficiently small volume is super-exponentially decaying in $n$. A proof of this lemma for discrete case appeared in Langberg's~\cite[Lemma\ 2.1]{langberg-oblivious-channels-2006} paper. A proof can be found in Appendix~\ref{sec:appendix_proofs_basiclemmas}. 
\begin{lemma}\label{lemma:sup_exp_ld}
Suppose $A\subseteq\bR^n$ Lebesgue measurable and $\cC=\{\vx(m)\}_{m=1}^{2^{nR}}\subseteq A$ contains $2^{nR},R>0$ uniform samples from $A$. If $V\subseteq A$ is Lebesgue measurable and for some small enough constant $\nu>0$ 
\begin{align}
p\coloneq\p{}(\vbfx\in V)=\frac{\mu(V)}{\mu(A)}\le2^{-(R+\nu)n},\label{eqn:prob_bound_implies_expt_bound}
\end{align}
where $\vbfx$ is sampled uniformly at random from $A$ and $\mu(\cdot)$ denotes the Lebesgue measure of a measurable subset of $\bR^n$, then there exists some constant $C=C(c,\nu,R)>0$, s.t.,
\[\p{}(|V\cap\cC|\ge cn^2)\le 2^{-Cn^3}.\]
\end{lemma}

\rev{\begin{remark}
The above lemma can be generalized in a straightforward manner to the case where the measure is, instead of Lebesgue measure, the unique translation-invariant measure on the sphere, i.e., the Haar measure. 
This variant will be the version that we invoke later in the proof. 
\end{remark}}

\begin{remark}
The above lemma indicates that the number of codewords falling into $V$ is small (say at most $\mathrm{poly}(n)$) with high probability as long as we can get an exponentially small upper bound on the expected number of such codewords \rev{(note that the condition given by Eqn.~\eqref{eqn:prob_bound_implies_expt_bound} implies $ \e{}(|V\cap\cC|)=p2^{nR}{\le}2^{-\nu n} $)}.
\end{remark}

\section{Formal statements of main results}\label{sec:results_formalstatements}

We now formally state the results in this paper with respect to decreasing values of $\Nkey$.
We start with the case when there is an unlimited amount of common randomness available between Alice and Bob.

\begin{lemma}
	[\cite{sarwate-spcom2012}]
	The capacity of the myopic adversarial channel with an unlimited amount of common randomness is
	\begin{equation}
	C_{\mathrm{myop,rand}} = \begin{cases}
	\Cld \coloneq \frac{1}{2}\log\frac{P}{N}, & \frac{\sigma^2}{P}\le\frac{1}{N/P}-1 \\
	\Rmyop  \coloneq \frac{1}{2}\log\left(\frac{(P+\sigma^2)(P+N)-2P\sqrt{N(P+\sigma^2)}}{N\sigma^2}\right), & \frac{\sigma^2}{P}\ge\max\left\{\frac{1}{N/P}-1,\frac{N}{P}-1\right\}\\
	0, & \frac{\sigma^2}{P} \le \frac{N}{P}-1
	\end{cases}.
	\label{eq:sarwate_rates}
	\end{equation}
	\rev{These are summarized in Fig.~\ref{fig:rateregion_infiniteCR}. }
	\label{lemma:sarwate_thm}
\end{lemma}
\rev{\begin{proof}
The achievability was proved in~\cite{sarwate-spcom2012} and the converse is proved in Sec.~\ref{sec:proof_scalebabble}. 
See Sec.~\ref{sec:result_infinite_CR} for details. 
\end{proof}}

\rev{\begin{remark}
The capacities in Eqn.~\eqref{eq:sarwate_rates} are continuous at the boundaries of different parameter regimes. 
Indeed, it is easy to check that 
\begin{align}
\Rmyop = \begin{cases}
\Cld, & \frac{\sigma^2}{P}=\frac{1}{N/P}-1 \\
0, & \frac{\sigma^2}{P} = \frac{N}{P}-1
\end{cases}. \notag 
\end{align}
In fact, in all our results, whenever we have characterizations on both sides of a certain boundary, the capacities are continuous at the boundary. 
\end{remark}}

\rev{\begin{remark}
It is worth noting that $ \Cld\le\Rmyop $ in all parameter regimes. 
\end{remark}}

We now discuss two possibilities: (1) $\Nkey$ is $\Theta(n)$, and (2) $\Nkey$ is $\Theta(\log n)$.
The rate given by Lemma~\ref{lemma:sarwate_thm} is an upper bound on the capacity in both cases, and we  show that this is also achievable in a subregion of the NSRs.
The proof will involve a myopic list-decoding argument, which we state below. We will combine this with with a known technique~\cite{langberg-focs2004,sarwate-thesis,bhattacharya2019sharedrandomness} which uses $\Theta(\log n)$ bits of common randomness to disambiguate the list and give us a result for unique decoding.

\begin{theorem}
	For $(P,N,\cO(n^2))$-list-decoding, the capacity is lower bounded as follows
	\[
	C_{\mathrm{myop,LD}} \geq \begin{cases}
	\Rmyop, & \text{ if } \frac{\sigma^2}{P}\ge\max\left\{\frac{1}{N/P}-1,\frac{N}{P}-1\right\} \text{ and }  \Rmyop +\Rkey > \frac{1}{2}\log \left( 1+\frac{P}{\sigma^2} \right)\\
	\Rld, & \text{ otherwise}
	\end{cases}.
	\]
	These are summarized in Fig.~\ref{fig:rateregion_myopiclist_arg}.
	\label{thm:myopic_listdecoding_summary}
\end{theorem}
\rev{\begin{proof}
See Sec.~\ref{sec:myop_list_dec_sketch} for a proof sketch and  Sec.~\ref{sec:myopic_list_decoding_details} for details. 
\end{proof}}
The rate $\Rld$ is achievable even in the presence of an omniscient adversary. A major contribution of this work is in showing that myopia indeed does help, and we can obtain a higher rate
of $\Rmyop$ in a certain regime. 
It is interesting to note that even when $\Nkey=0$, the myopic list-decoding capacity is nonzero for sufficiently large values of $\sigma^2/P$.
Furthermore, increasing rates of common randomness ($\Rkey$) help us  achieve higher list-decoding rates as seen from Fig.~\ref{fig:rateregion_myopiclist_arg}.

The above theorem is crucially used in proving the following result for linear amount of common randomness:
\begin{lemma}
	Fix any $ \Rkey>0 $.
	If Alice and Bob share $n\Rkey$ bits of common randomness, then the capacity is
	\[
	C_{\mathrm{myop}} = \begin{cases}
	\Rmyop, & \text{ if } \frac{\sigma^2}{P}\ge\max\left\{\frac{1}{N/P}-1,\frac{N}{P}-1\right\}\text{ and } \Rkey>\frac{1}{2}\log\left( 1+\frac{P}{\sigma^2} \right)-\Rmyop \\
	\Rld, & \frac{\sigma^2}{P}\le\frac{1}{N/P}-1\\
	0, & \frac{\sigma^2}{P}\leq \frac{N}{P}-1
	\end{cases}.
	\]
	Furthermore, 
	\[
	\Rld\leq C_{\mathrm{myop}}\leq \Rmyop \quad\text{ if }\frac{\sigma^2}{P}\ge\max\left\{\frac{1}{N/P}-1,\frac{N}{P}-1\right\},\text{ and } \Rkey<\frac{1}{2}\log\left( 1+\frac{P}{\sigma^2} \right)-\Rmyop .
	\]
	\label{lemma:achievablerate_thetanbits}
\end{lemma}
\rev{\begin{proof}
See Sec.~\ref{sec:achievability_thetan}. 
\end{proof}}
The rates achievable for different values of the NSRs are illustrated in Fig.~\ref{fig:rateregion_rkey_0pt2} and Fig.~\ref{fig:rateregion_rkey_1} for different values of $\Rkey$.
Our scheme achieves capacity in the red and blue regions of Fig.~\ref{fig:rateregion_rkey_0pt2} and~\ref{fig:rateregion_rkey_1}. In the white dotted region, $\Rld$ is a lower bound on the capacity, as guaranteed by Lemma~\ref{lemma:achievablerate_thetanbits}. However, we can achieve $\Rmyop$ with an infinite amount of common randomness. Hence, there is a small gap between the upper and lower bounds in this region. In the cyan dotted region, $\Rld<0$ and our lower bound is trivial, while the converse says that the capacity is upper bounded by $\Rmyop$.  

As remarked earlier, we only require $ \Theta(\log n) $ bits of common randomness to disambiguate the list (in other words, go from a list-decodable code to a uniquely decodable code). The additional $ n\Rkey-\Theta(\log n) $ bits are used for additional randomization at the encoder to ``confuse'' James. As is evident from Figures~\ref{fig:rateregion_rkey_0pt2} and~\ref{fig:rateregion_rkey_1}, and the following lemma, larger values of $ \Rkey $ can guarantee that $ \Rmyop $ is achievable in a larger range of NSRs. 

Using Theorem~\ref{thm:myopic_listdecoding_summary} for $\Rkey=0$, we can show Lemma~\ref{lemma:achievablerate_Olognbits}. Note that when $\Rkey=0$, the condition  $\Rmyop +\Rkey > \frac{1}{2}\log \left( 1+\frac{P}{\sigma^2}\right)$ reduces to $\frac{\sigma^2}{P}\geq 4\frac{N}{P}-1$.
\begin{lemma}
	When $\log n< \Nkey  =\cO(\log n)$, the capacity of the myopic adversarial channel is:
	\[
	C_{\mathrm{myop}} = \begin{cases}
	\Rmyop, & \text{ if } \frac{\sigma^2}{P}\ge\max\left\{\frac{1}{N/P}-1,4\frac{N}{P}-1\right\} \\
	\Rld, & \frac{\sigma^2}{P}\le\frac{1}{N/P}-1\\
	0, & \frac{\sigma^2}{P}\leq \frac{N}{P}-1
	\end{cases}.
	\]
	Furthermore, 
	\[
	\Rld\leq C_{\mathrm{myop}}\leq \Rmyop \quad\text{ if }\max\left\{\frac{1}{N/P}-1,\frac{N}{P}-1\right\} \le  \frac{\sigma^2}{P}\le 4\frac{N}{P}-1.
	\]
	These results are summarized in Fig.~\ref{fig:rateregion_logn}.
	\label{lemma:achievablerate_Olognbits}
\end{lemma}
\rev{\begin{proof}
See Sec.~\ref{sec:achievability_thetalogn}.
\end{proof}}

These results are illustrated in Fig.~\ref{fig:rateregion_logn}. We have matching upper and lower bounds in the red, blue, and grey regions. As in the previous subsection, there is a nontrivial gap between the upper and lower bounds in the green and white regions.

When Alice and Bob do not have access to a shared secret key, the achievability proofs in Sec.~\ref{sec:achievability_thetan} and \ref{sec:achievability_thetalogn} are not valid, and for certain values of the NSRs, tighter converses can be obtained. In this scenario, we will prove the following result.
\begin{theorem}
	The capacity of the myopic adversarial channel with no common randomness is given by
	\[
	C_{\mathrm{myop}} = \begin{cases}
	\Rld, & \text{ if } \frac{1}{1-N/P}-1\leq \frac{\sigma^2}{P}\leq \frac{1}{N/P}-1\\
	\Rmyop, & \text{ if } \frac{\sigma^2}{P} \geq \max\left\{ \frac{1}{1-N/P}-1,\frac{1}{N/P}-1 \right\}\\
	0, &\text{ if }\frac{\sigma^2}{P}\le \rev{\frac{1}{1-N/P} - 2} \text{ or }\frac{N}{P}\geq 1
	\end{cases}.
	\]
	In the other regimes, we have
	\[
	R_{\mathrm{GV}} \le C_{\mathrm{myop}} \le \begin{cases}
	\Rld, & \text{ if } \rev{\frac{1}{1-N/P} - 2}\leq \frac{\sigma^2}{P}\leq \min\left\{ \frac{1}{N/P}-1,\frac{1}{1-N/P}-1\right\}\\
	\Rmyop, &\text{ if } \max\left\{\frac{1}{N/P}-1,\rev{\frac{1}{1-N/P} - 2}\right\}\le\frac{\sigma^2}{P}\le \frac{1}{1-N/P}-1\text{ and }\frac{N}{P}\le1.
	\end{cases}
	\]
	\rev{These are summarized in Fig.~\ref{fig:rateregion_noCR}.}
	\label{thm:capacity_noCR}
\end{theorem}
\rev{\begin{proof}
The achievability follows by combining the myopic list-decoding lemma (Theorem~\ref{thm:myopic_listdecoding_summary} which is proved in Sec.~\ref{sec:myopic_list_decoding_details}) and a trick used in~\cite{djl-2019-myopic-tit} for myopic list disambiguation.
We adapt the trick to the quadratically constrained case dealt with in this paper and the proof is presented in Sec.~\ref{sec:achievability_suffmyopic}. 

The converse follows by combining the scale-and-babble converse (proved in Sec.~\ref{sec:proof_scalebabble}) for Lemma~\ref{lemma:sarwate_thm} and a symmetrization converse proved in Sec.~\ref{sec:converse_symmetrization}. 

See Sec.~\ref{sec:zero_CR_sketch} for an overview of the proof structure. 
\end{proof}}

\subsubsection{Wiretap secrecy}\label{sec:wiretap_secrecy}

We will show that a simple modification can guarantee secrecy when James also wants to eavesdrop on the message.
In addition to achieving a vanishingly small probability of error at the decoder, we wish to ensure that the information leaked to James is vanishingly small, \emph{i.e.}, $\frac{1}{n} I(\vbfx;\vbfz)\to 0 $ as $ n\to \infty $. This can be easily guaranteed by a wiretap code, and we briefly describe the modifications required to ensure this.

%
%
%
%
%
%

When Alice and Bob share infinite common randomness, $ nC_{\mathrm{myop,rand}} $ bits of secret key can be used as a one-time pad. In fact, the one-time pad guarantees perfect secrecy: $ I(\bfm;\vbfz)=0 $ for all $ n $. In this case, the secrecy capacity can be completely characterized and is equal to $ C_{\mathrm{myop,rand}} $.

\begin{lemma}
	The secrecy capacity of the myopic adversarial channel with an unlimited amount of common randomness is
	\begin{equation}
	C_{\mathrm{myop,rand,sec}} = \begin{cases}
	\Cld  & \frac{\sigma^2}{P}\le\frac{1}{N/P}-1 \\
	\Rmyop , & \frac{\sigma^2}{P}\ge\max\left\{\frac{1}{N/P}-1,\frac{N}{P}-1\right\}\\
	0, & \frac{\sigma^2}{P} \le \frac{N}{P}-1
	\end{cases}.
	\label{eq:sarwate_rates_secrecy}
	\end{equation}
	\label{lemma:sarwate_thm_secrecy}
\end{lemma}

In the regime where $ \Nkey=\Theta(n) $, we use a binning scheme analogous to~\cite{wyner1975wire,leung1978gaussian,kang2010wiretap}. Our approach is identical to~\cite{kang2010wiretap}, and we only give a sketch of the proof in Appendix~\ref{sec:proof_achievablerate_thetanbits_secrecy}.

Let $ C_J\coloneq \frac{1}{2}\log\left(1+\frac{P}{\sigma^2}\right) $.
The following rates are achievable:
\begin{lemma}
	If Alice and Bob share $n\Rkey$ bits of common randomness, then the secrecy capacity is
	\[
	C_{\mathrm{myop,sec}}  \begin{cases}
	\geq \min \{\Rmyop,\Rmyop-C_J+\Rkey\}, & \text{ if } \frac{\sigma^2}{P}\ge\max\left\{\frac{1}{N/P}-1,\frac{N}{P}-1\right\}\text{ and } \Rkey>\frac{1}{2}\log\left( 1+\frac{P}{\sigma^2} \right)-\Rmyop \\
	=0, & \frac{\sigma^2}{P}\leq \frac{N}{P}-1\\
	\geq \min\{ \Rld,\Rld-C_J+\Rkey\}, & \text{otherwise.}
	\end{cases}.
	\]
	
	\label{lemma:achievablerate_thetanbits_secrecy}
\end{lemma}

The results of Lemma~\ref{lemma:achievablerate_thetanbits_secrecy} are pictorially summarized in Figs.~\ref{fig:rateregion_rkey_0pt2_secrecy} and \ref{fig:rateregion_rkey_1_secrecy}. Positive rates are guaranteed in the red and blue regions, as in Figs.~\ref{fig:rateregion_rkey_0pt2} and \ref{fig:rateregion_rkey_1}. 

The analysis above easily extends to the case where $ \Nkey=\Theta(\log n) $ and $ \Nkey=0 $. We can obtain the following results.

\begin{lemma}
	When $\Nkey = \Theta(\log n)$, the secrecy capacity of the myopic adversarial channel is:
	\[
	C_{\mathrm{myop,sec}}  \begin{cases}
	\geq \Rmyop-C_J, & \text{ if } \frac{\sigma^2}{P}\ge\max\left\{\frac{1}{N/P}-1,4\frac{N}{P}-1\right\} \\
	=0, & \frac{\sigma^2}{P}\leq \frac{N}{P}-1\\
	\geq \Rld-C_J, & \text{otherwise.}
	\end{cases}.
	\]
	These results are summarized in Fig.~\ref{fig:rateregion_rkey_logn_secrecy}.
	\label{lemma:achievablerate_Olognbits_secrecy}
\end{lemma}

\begin{lemma}
	The secrecy capacity of the myopic adversarial channel with no common randomness is given by
	\[
	C_{\mathrm{myop,sec}}  \begin{cases}
	\geq \Rld-C_J, & \text{ if } \frac{1}{1-N/P}-1\leq \frac{\sigma^2}{P}\leq \frac{1}{N/P}-1\\
	\geq \Rmyop-C_J, & \text{ if } \frac{\sigma^2}{P} \geq \max\left\{ \frac{1}{1-N/P}-1,\frac{1}{N/P}-1 \right\}\\
	=0, &\text{ if }\frac{\sigma^2}{P}\le \rev{\frac{1}{1-N/P} - 2}
	\text{ or }\frac{N}{P}\geq 1\\
	\geq R_{\mathrm{GV}}-C_J, &\text{otherwise}
	\end{cases}.
	\]
	\label{lemma:capacity_noCR_secrecy}
\end{lemma}
The above results are pictorially depicted in Fig.~\ref{fig:rateregion_rkey_nocr_secrecy}.

We now discuss the proof techniques used to derive our main results.

\begin{figure}
	\centering
	\begin{subfigure}[t]{.49\textwidth}
		\centering
		\includegraphics[height = 0.37\textheight]{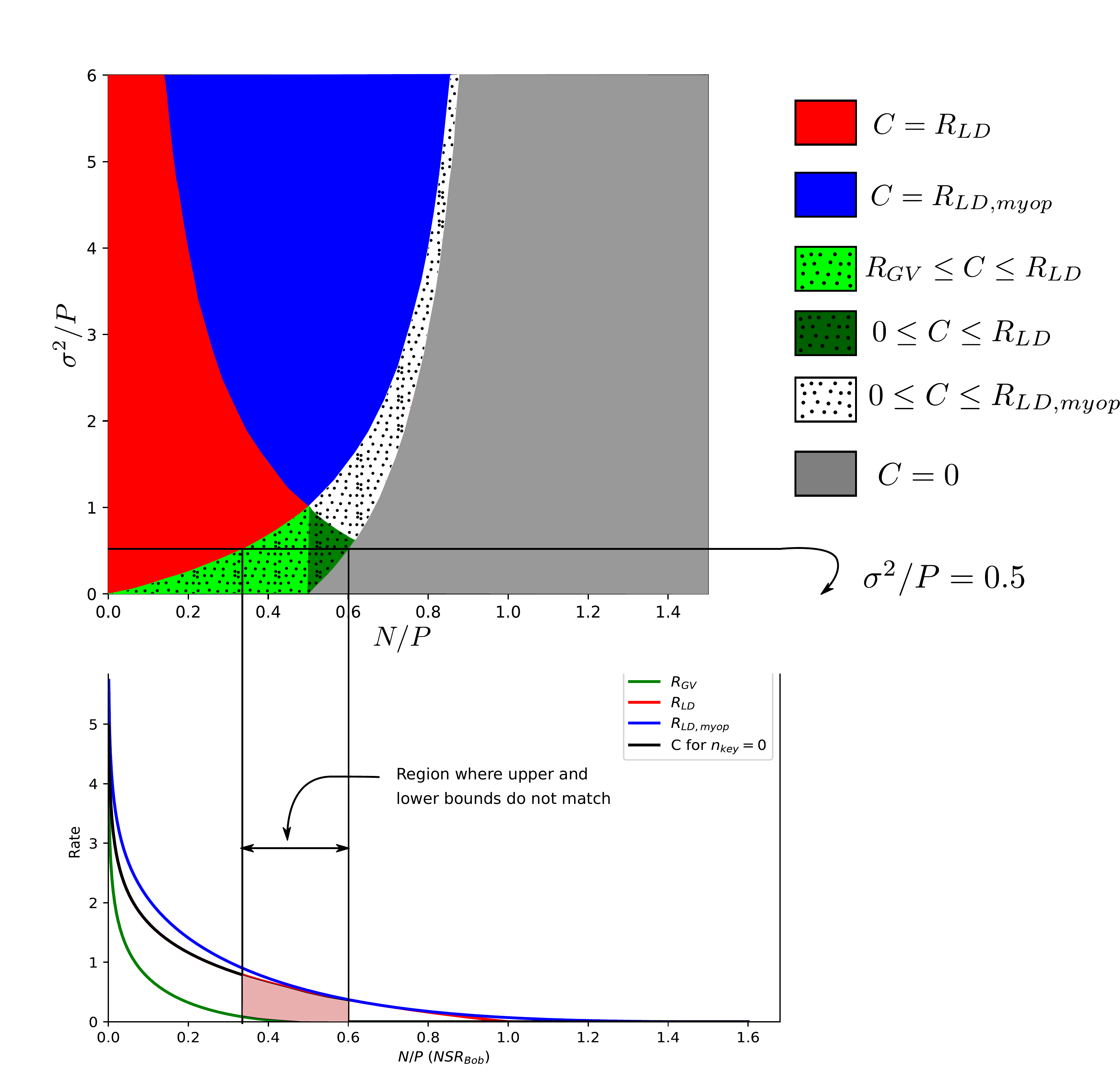}
		\subcaption{Capacity bounds for the myopic channel with no common randomness. Here, $R_{\mathrm{GV}}\coloneq\frac{1}{2}\log\left(\frac{P^2}{4N(P-N)}\right)\mathds1_{\{P\geq 2N\}}$.}
		\label{fig:rateregion_noCR}
	\end{subfigure}
	\hfill
	\begin{subfigure}[t]{.49\textwidth} 
		\centering
		\includegraphics[height = 0.37\textheight]{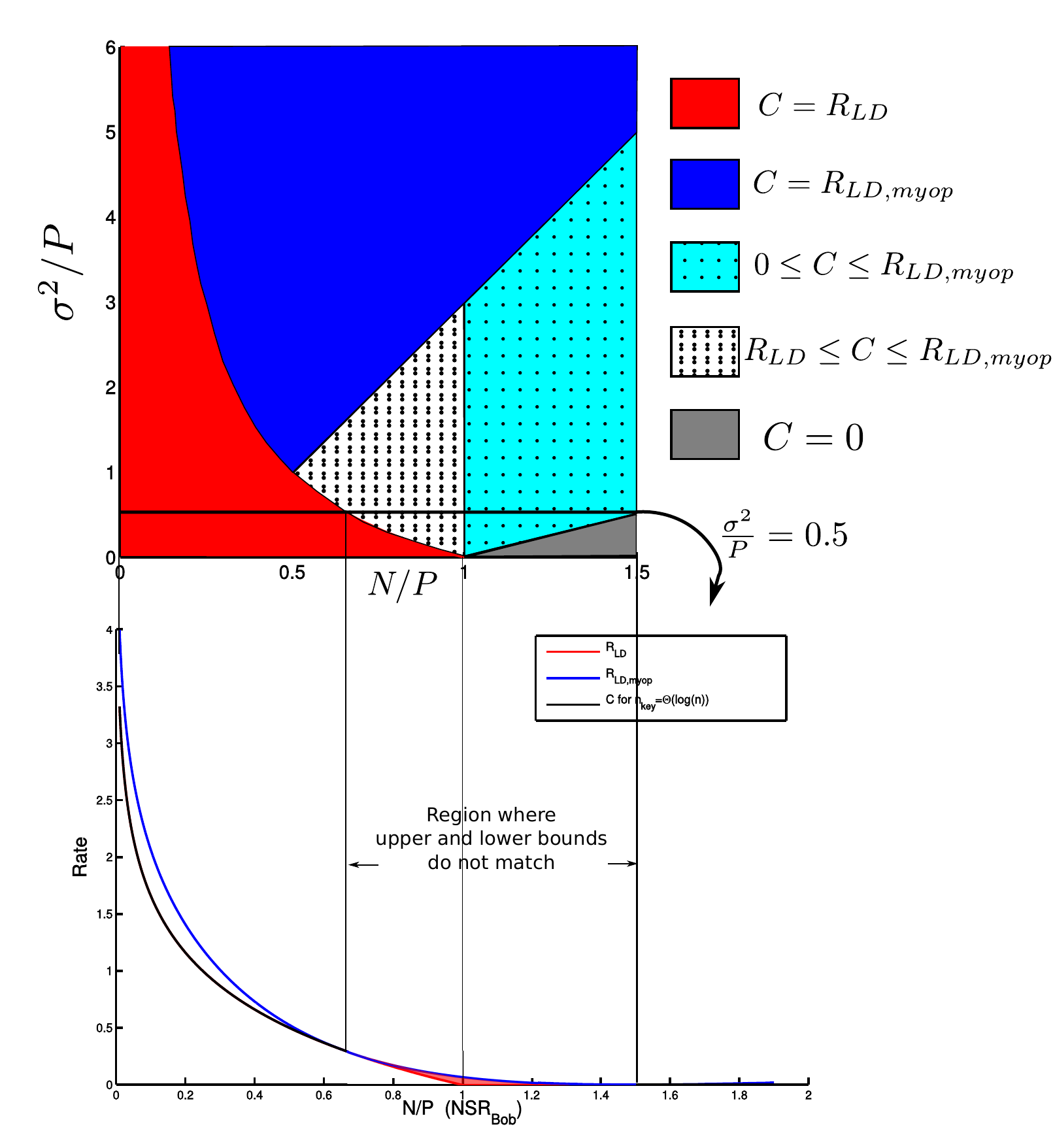}
		\subcaption{Capacity bounds with $\Nkey = \Theta(\log n)$. We have a complete characterization in the solid regions but nonmatching upper and lower bounds in the dotted regions.}
		\label{fig:rateregion_logn}
	\end{subfigure}
	\\
	\begin{subfigure}[t]{.49\textwidth} 
		\centering
		\includegraphics[height = 0.37\textheight]{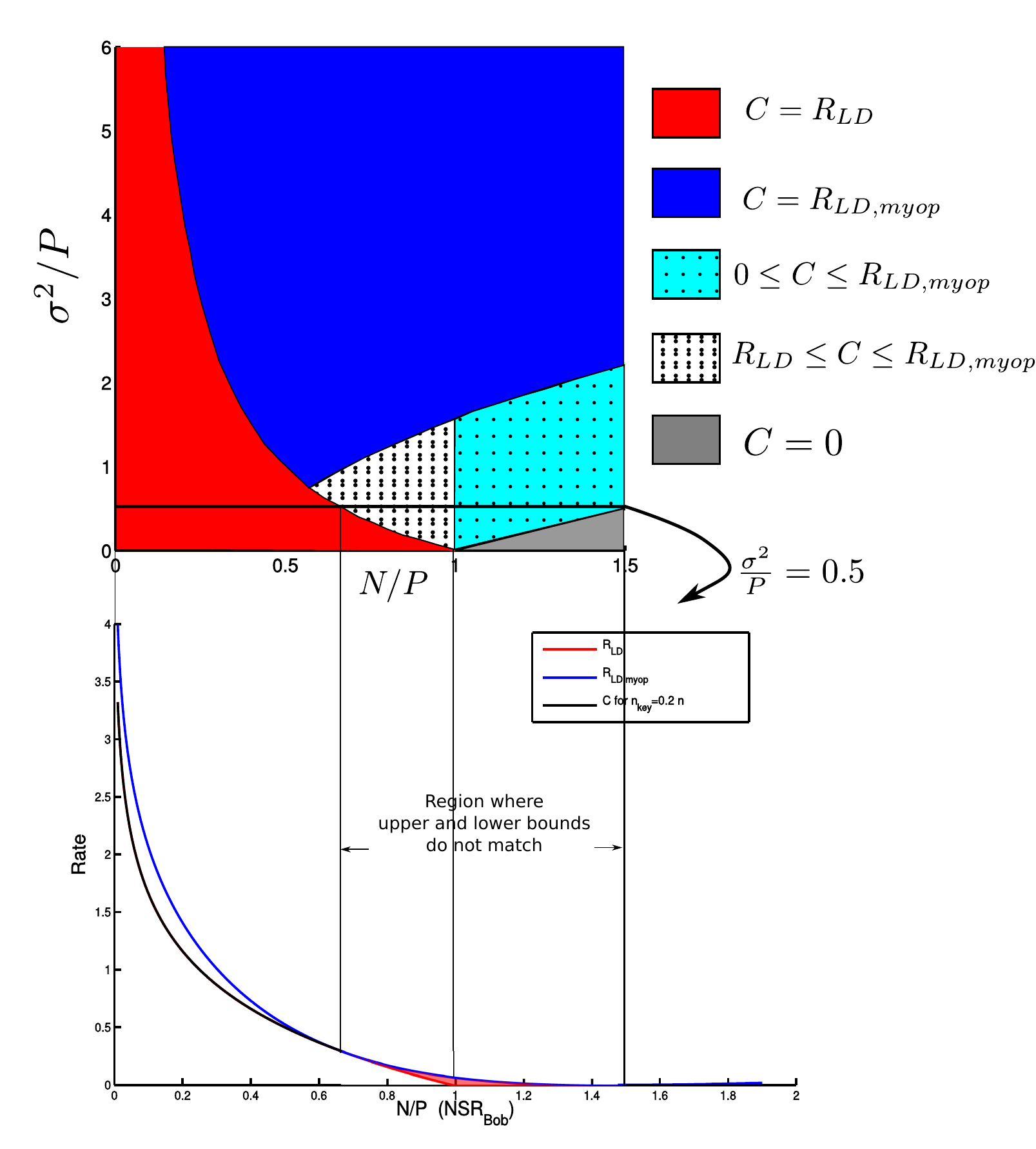}
		\caption{An example of our results when $\Nkey=\Theta(n)$: Capacity bounds with $\Nkey=0.2n$. }
		\label{fig:rateregion_rkey_0pt2}
	\end{subfigure}
	\hfill
	\begin{subfigure}[t]{.49\textwidth} 
		\centering
		\includegraphics[height = 0.37\textheight]{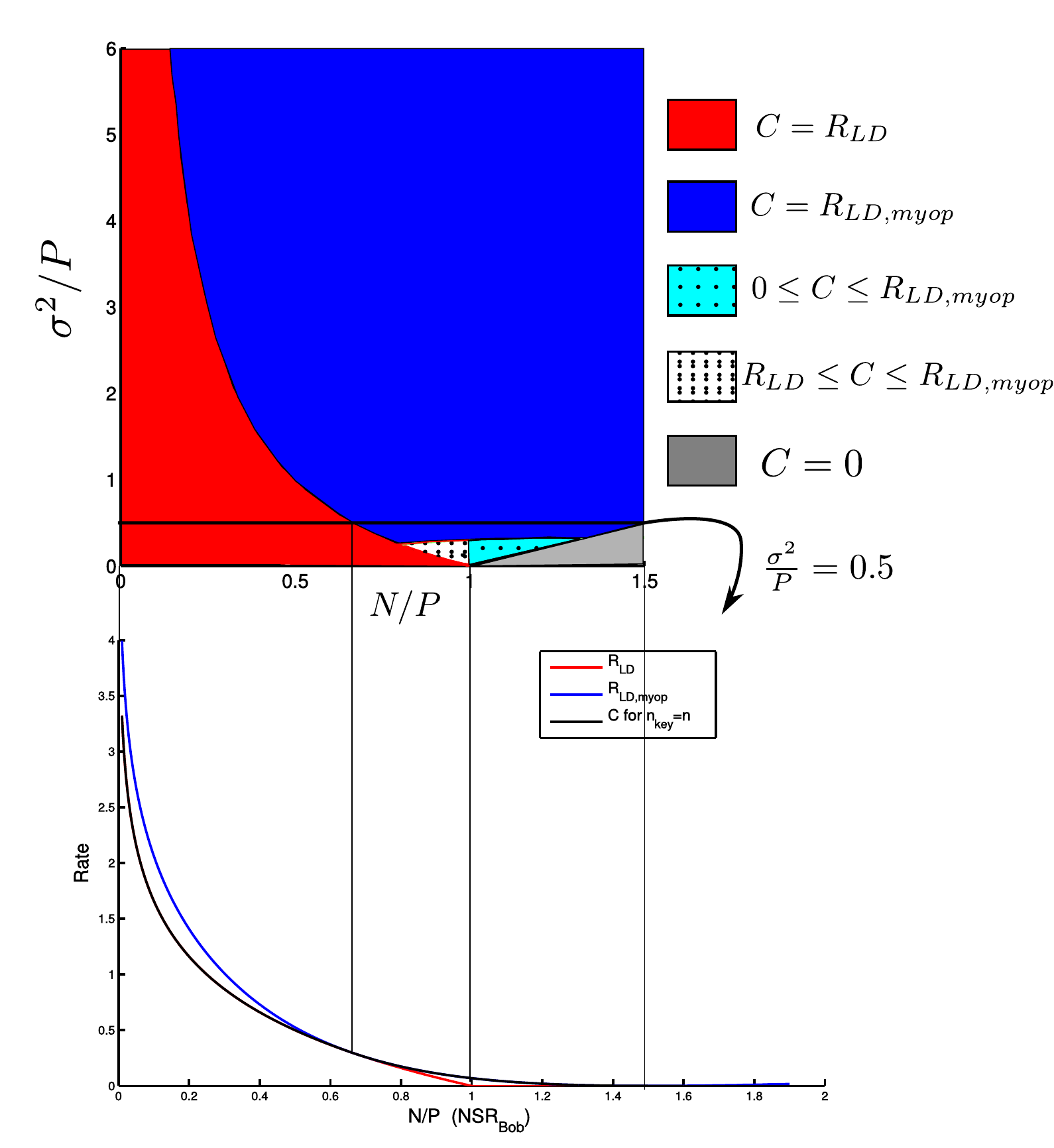}
		\caption{An example of our results when $\Nkey=\Theta(n)$: Capacity bounds for $\Nkey=n$.}
		\label{fig:rateregion_rkey_1}
	\end{subfigure}
	\caption{Capacity bounds for different values of $ \Nkey $.}
\end{figure}


\begin{figure*} 
	\centering
	\begin{subfigure} {0.5\textwidth}
		\centering
		\includegraphics[height=0.25\textheight]{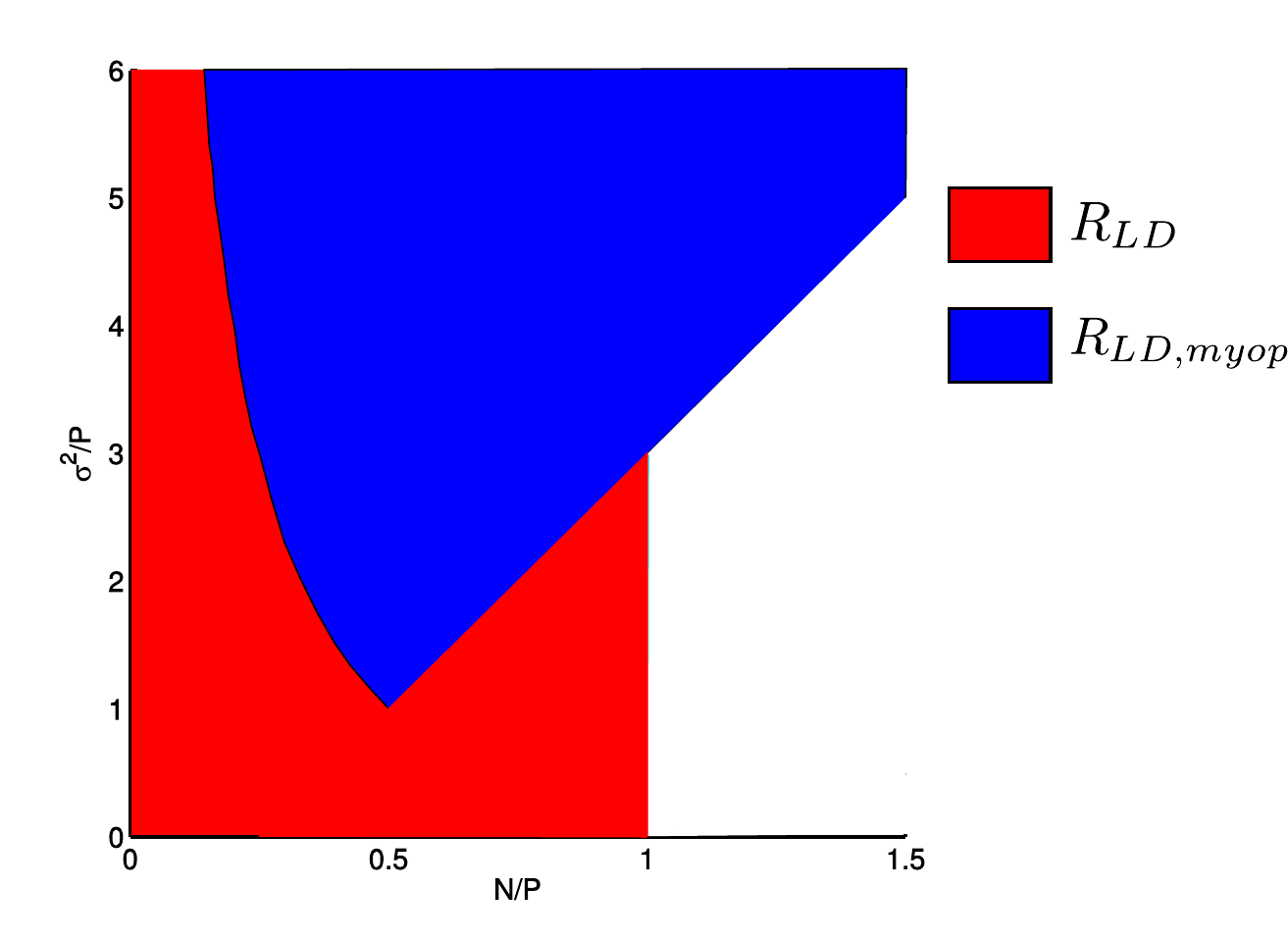}
		\caption{}
		\label{fig:rateregion_myopiclist_arg_a}
	\end{subfigure}
	~ 
	\begin{subfigure} {0.5\textwidth}
		\centering
		\includegraphics[height=0.25\textheight]{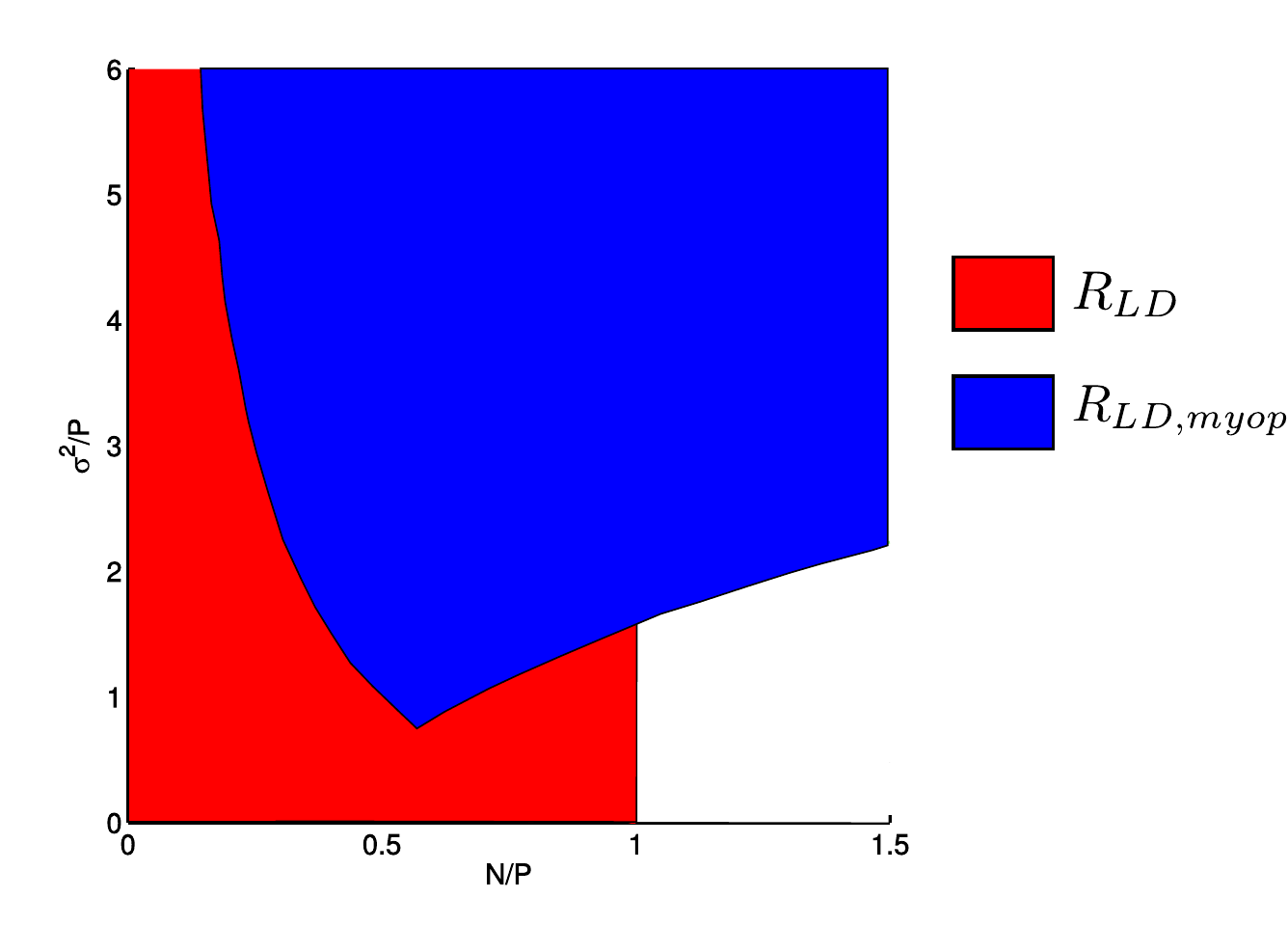}
		\caption{}
		\label{fig:rateregion_myopiclist_arg_b}
	\end{subfigure}
	\caption{Achievable rates for myopic list-decoding for \ref{fig:rateregion_myopiclist_arg_a} $\Nkey=0$ and \ref{fig:rateregion_myopiclist_arg_b} $\Nkey=0.2n$ bits of common randomness.}
	\label{fig:rateregion_myopiclist_arg}
\end{figure*}
\begin{figure} 
	\centering
	\includegraphics[height=0.25\textheight]{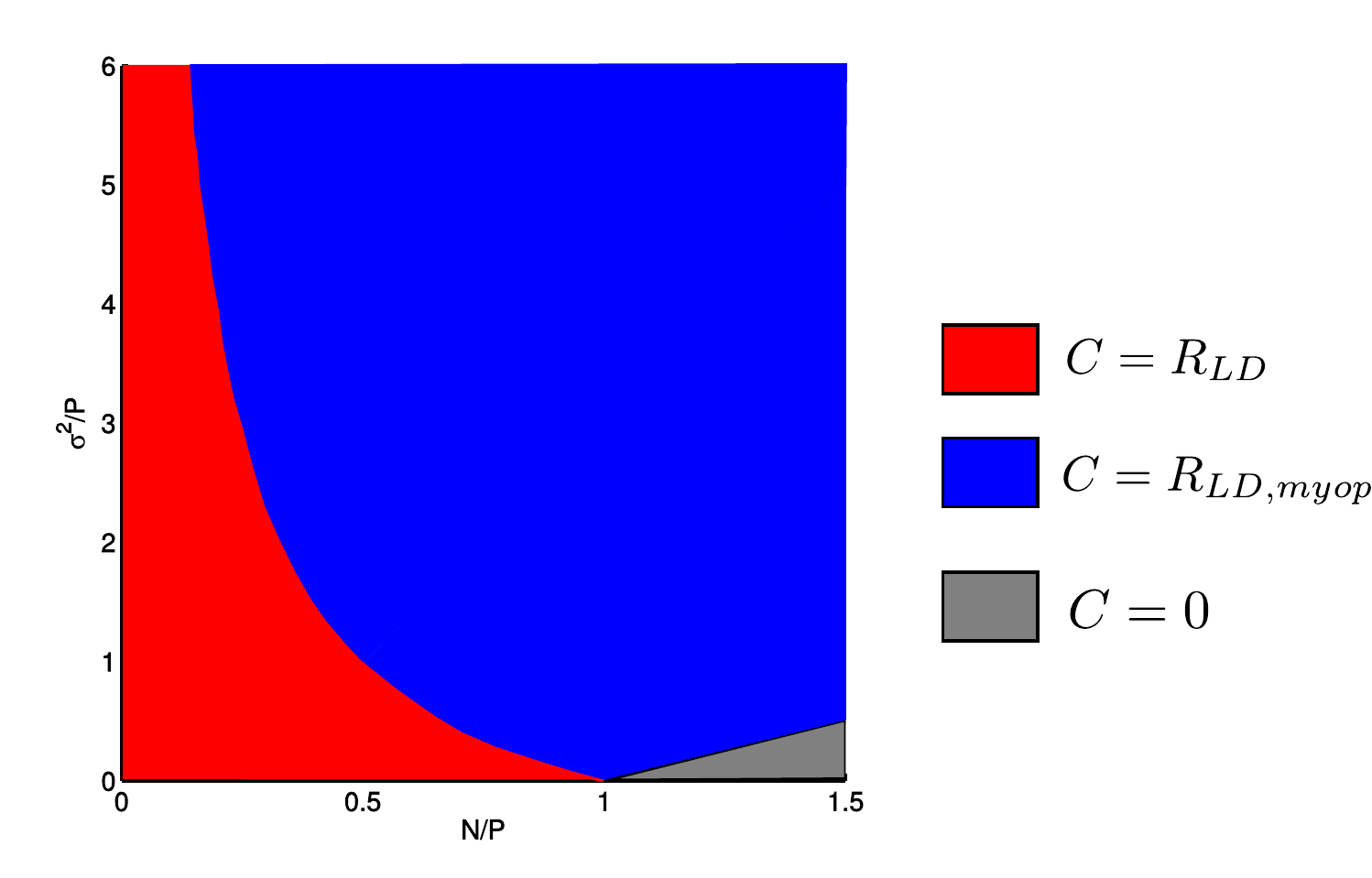}
	\caption{Capacity of the myopic adversarial channel for $\Nkey=\infty$~\cite{sarwate-spcom2012}. The $x$-axis denotes the NSR from Alice to Bob (with noise from James) and the $y$-axis denotes NSR from Alice to James (with AWGN added).}
	\label{fig:rateregion_infiniteCR}
\end{figure}
\begin{figure} 
	\centering
	\includegraphics[height=0.25\textheight]{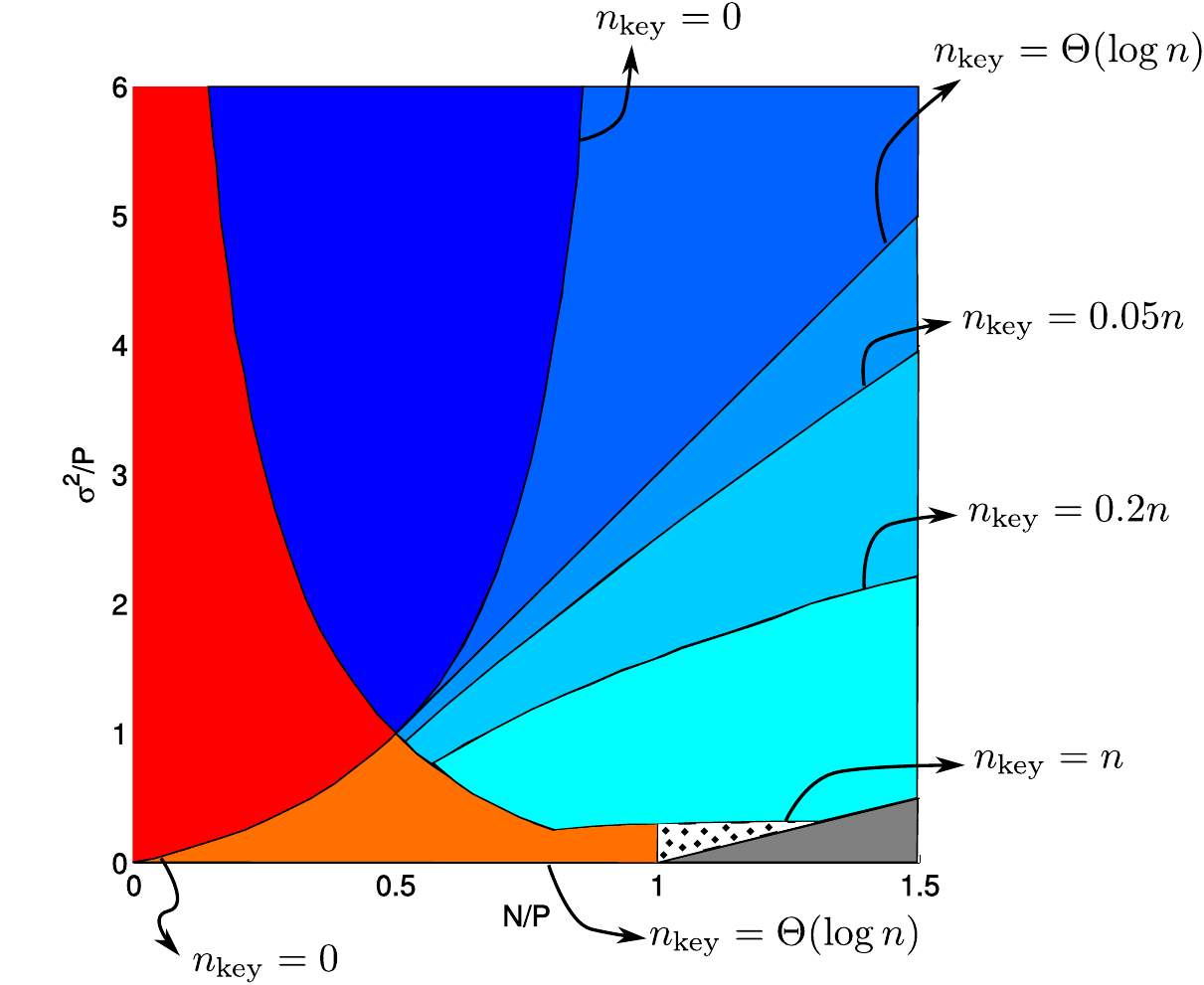}
	\caption{Expansion of the achievable rate regions for the myopic adversarial channel with different amounts of common randomness. A rate of $\Rld$ is achievable in the red and orange regions, whereas $\Rmyop$ is achievable in the regions with different shades of blue. With $\Omega(n)$ bits of common randomness, a rate of at least $\Rld$ is achievable whenever $P>N$.}
	\label{fig:rateregion_differentCR}
\end{figure}

\begin{figure} 
	\centering
	\includegraphics[height=0.25\textheight]{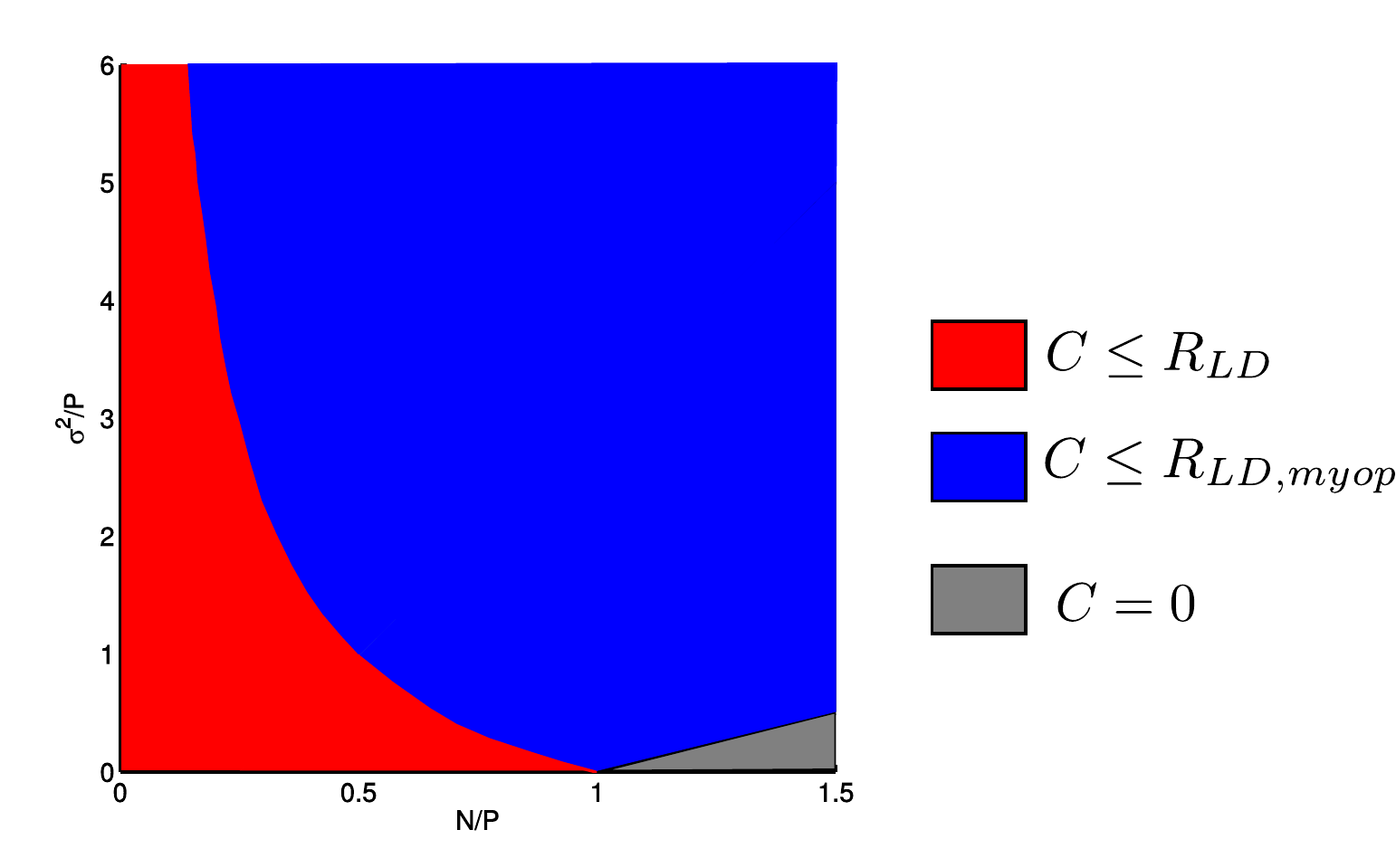}
	\caption{Upper bounds on capacity as obtained from the scale-and-babble attack. These outer bounds are valid for all values of $\Nkey$. }
	\label{fig:rateregion_scalebabble_arg}
\end{figure}

\begin{figure*} 
	\centering
	\begin{subfigure}[t] {0.33\textwidth}
		\centering
		\includegraphics[width = \linewidth]{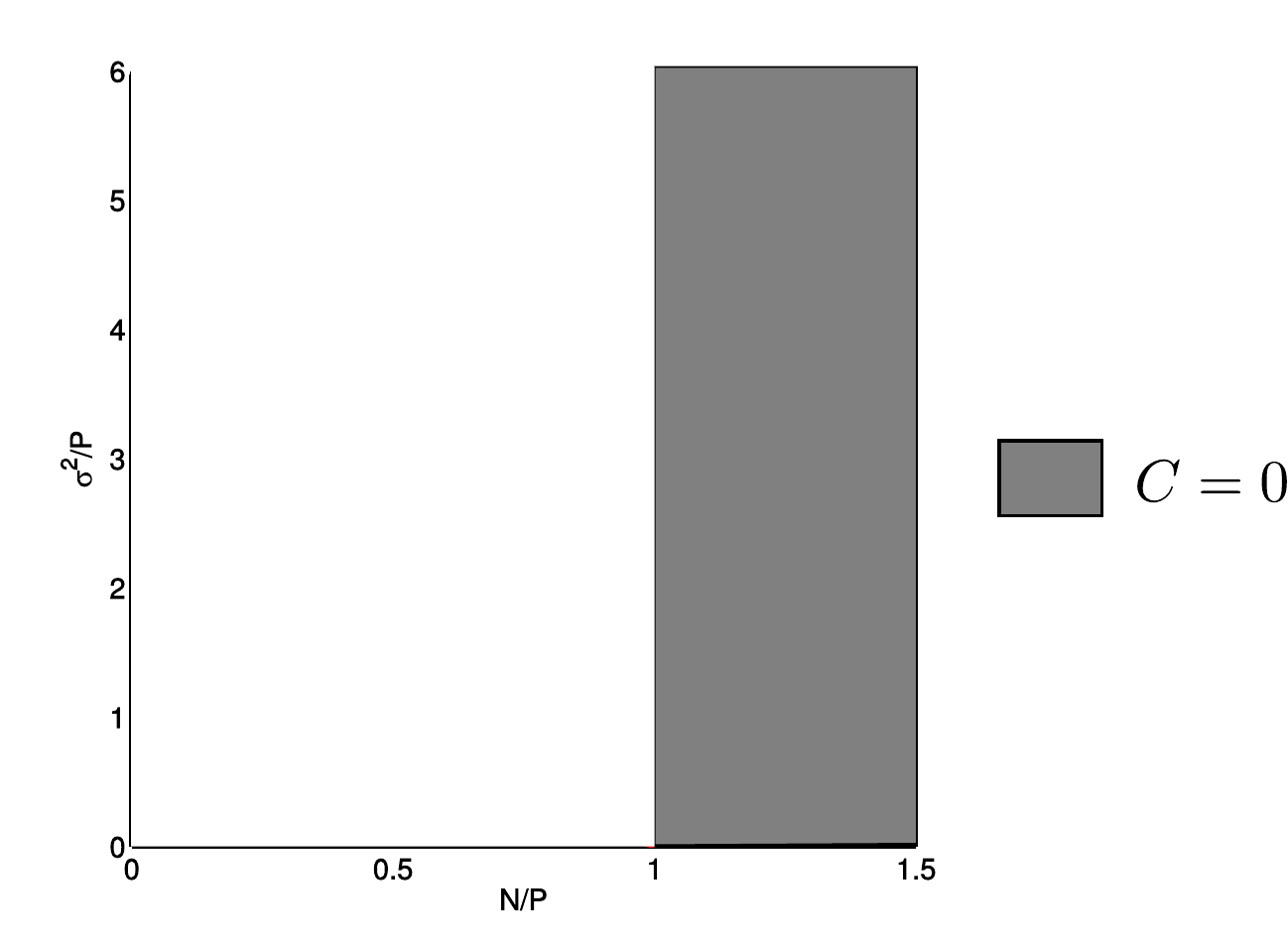}
		\caption{}
		\label{fig:rateregion_symmetrization_arg_a}
	\end{subfigure}
	~ 
	\begin{subfigure}[t] {0.33\textwidth}
		\centering
		\includegraphics[width = \linewidth]{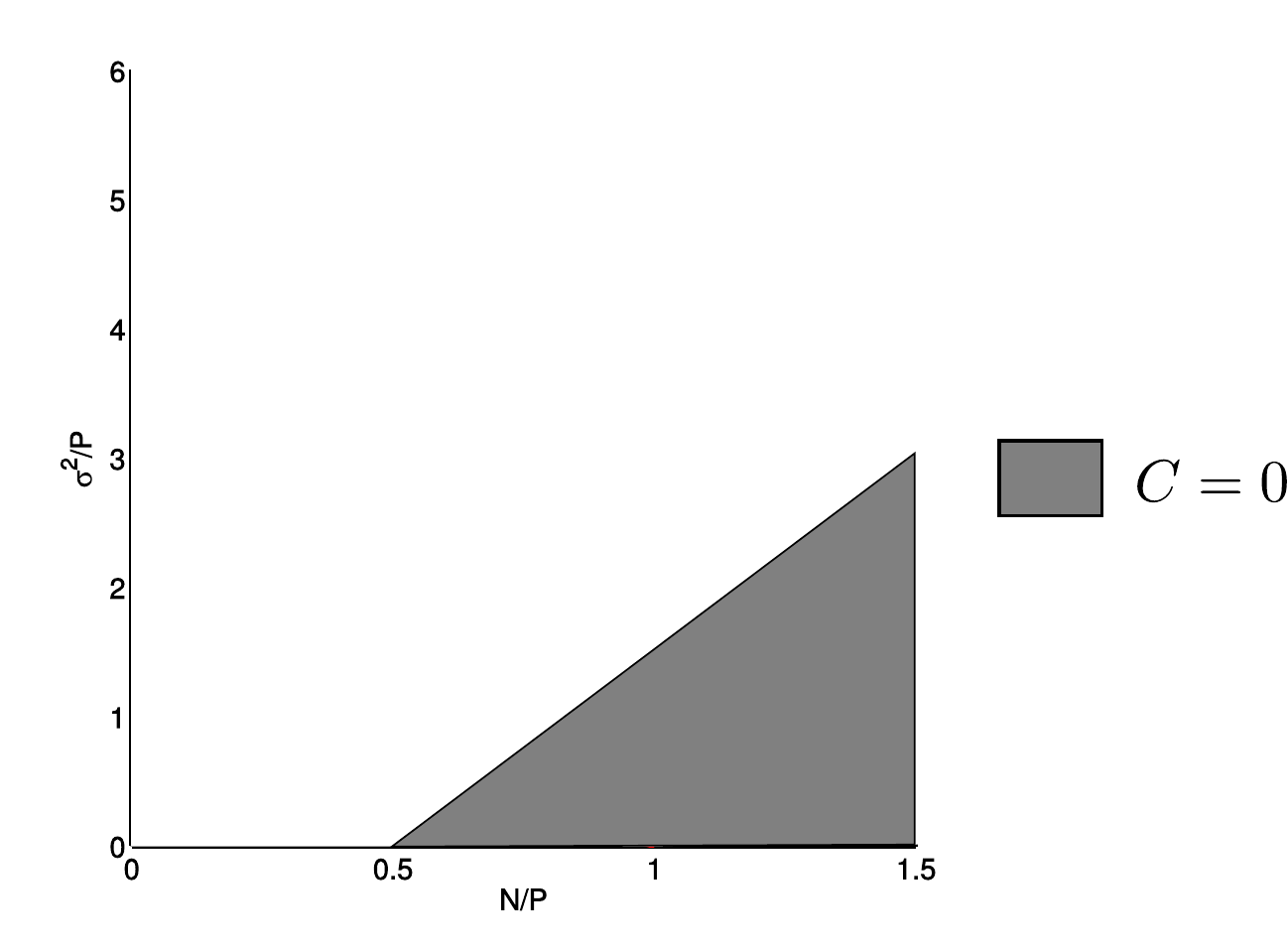}
		\caption{}
		\label{fig:rateregion_symmetrization_arg_b}
	\end{subfigure}
	~
	\begin{subfigure}[t] {0.33\textwidth}
		\centering
		\includegraphics[width = \linewidth]{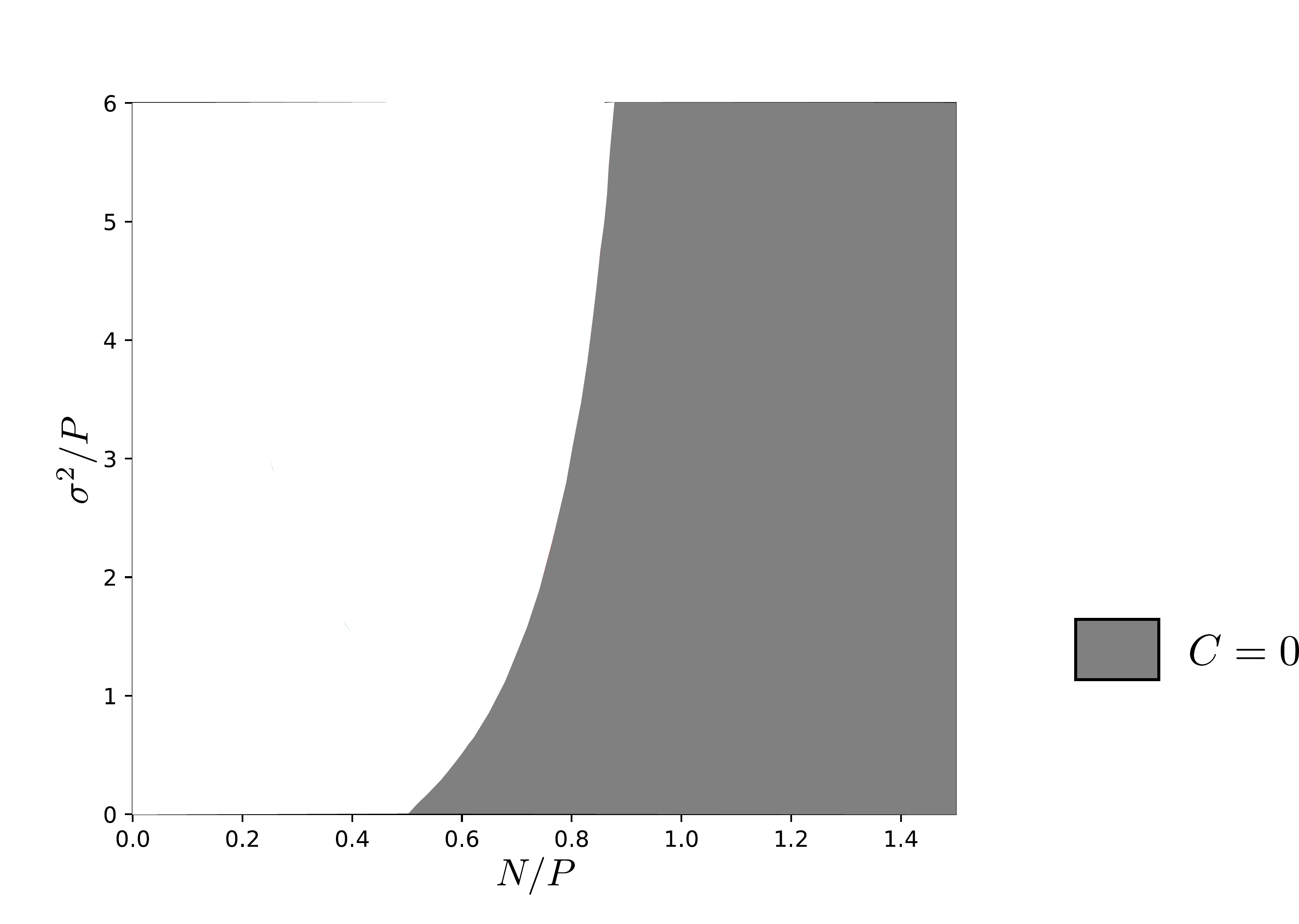}
		\caption{}
		\label{fig:rateregion_symmetrization_arg_c}
	\end{subfigure}
	\caption{Regions where the \ref{fig:rateregion_symmetrization_arg_a} $\vbfz$-agnostic symmetrization argument, \ref{fig:rateregion_symmetrization_arg_b} $\vbfz$-aware symmetrization argument  and \ref{fig:rateregion_symmetrization_arg_c} improved  $\vbfz$-aware symmetrization argument gives zero capacity. Note that these are valid only for $\Nkey=0$. However, they continue to hold even if the encoder has access to private randomness.}
	\label{fig:rateregion_symmetrization_arg}
\end{figure*}
\begin{figure} 
	\centering
	\includegraphics[height = 0.25\textheight]{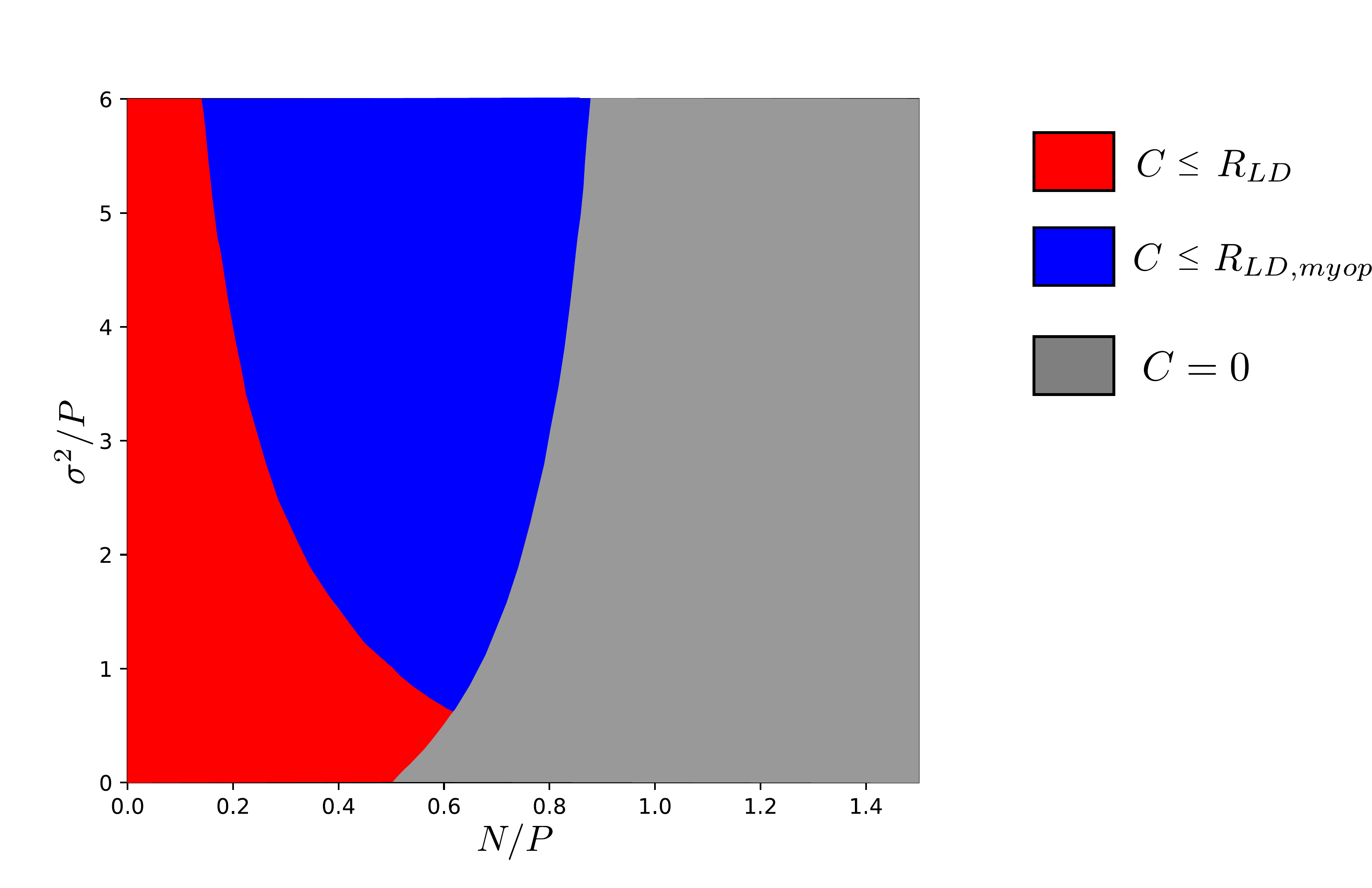}
	\caption{Upper bounds on capacity for $\Nkey=0$. This is obtained by combining the bounds obtained using the scale-and-babble attack in Fig.~\ref{fig:rateregion_scalebabble_arg} with the symmetrization arguments in Fig.~\ref{fig:rateregion_symmetrization_arg}}
	\label{fig:rateregion_converse_noCR}
\end{figure}

\begin{figure} 
	\centering
	\includegraphics[height=0.25\textheight]{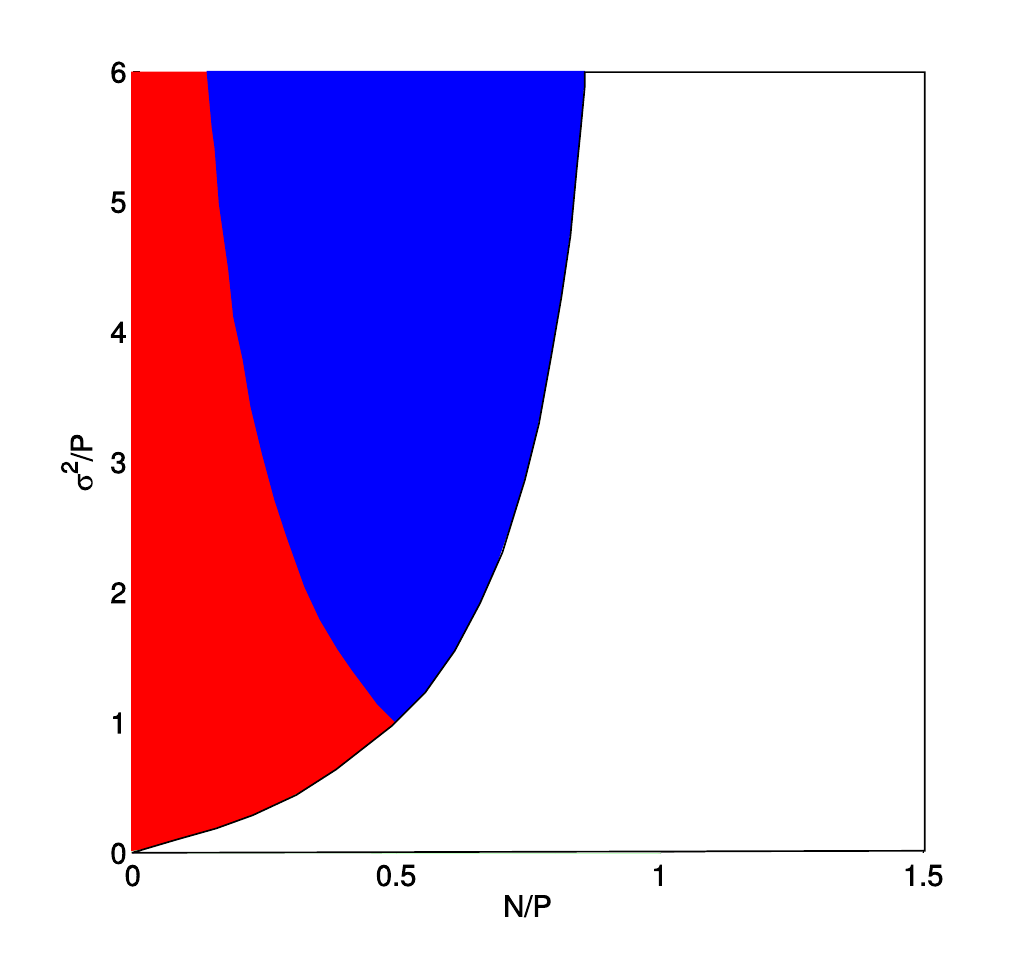}
	\caption{Currently, our reverse list-decoding arguments, an integral part of our proof techniques for $\Nkey=0$, are only valid in the red and blue regions, which present bottlenecks in our analysis of the scenario when no common randomness is available. As ongoing work, we are trying to expand the region where reverse list-decoding is possible.}
	\label{fig:rateregion_reverse_list_arg}
\end{figure}

\begin{figure} 
	\begin{center}
		\includegraphics[width=0.5\textwidth]{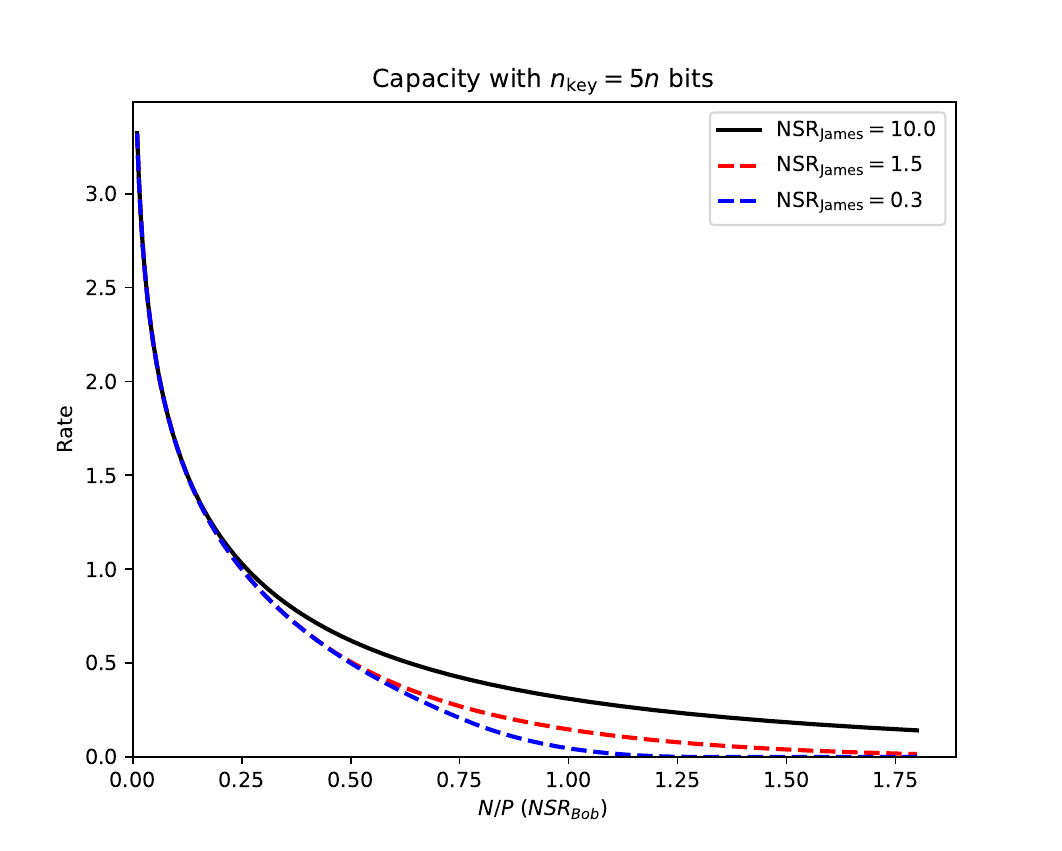}
		\caption{Capacity of the myopic adversarial channel with $\Omega(n)$ bits of common randomness. Here $\mathrm{NSR}_{\mathrm{James}}\coloneq \sigma^2/P$.}
		\label{fig:achievable_rate_omegan}
	\end{center}
\end{figure}
\begin{figure} 
	\begin{center}
		\includegraphics[height=0.25\textheight]{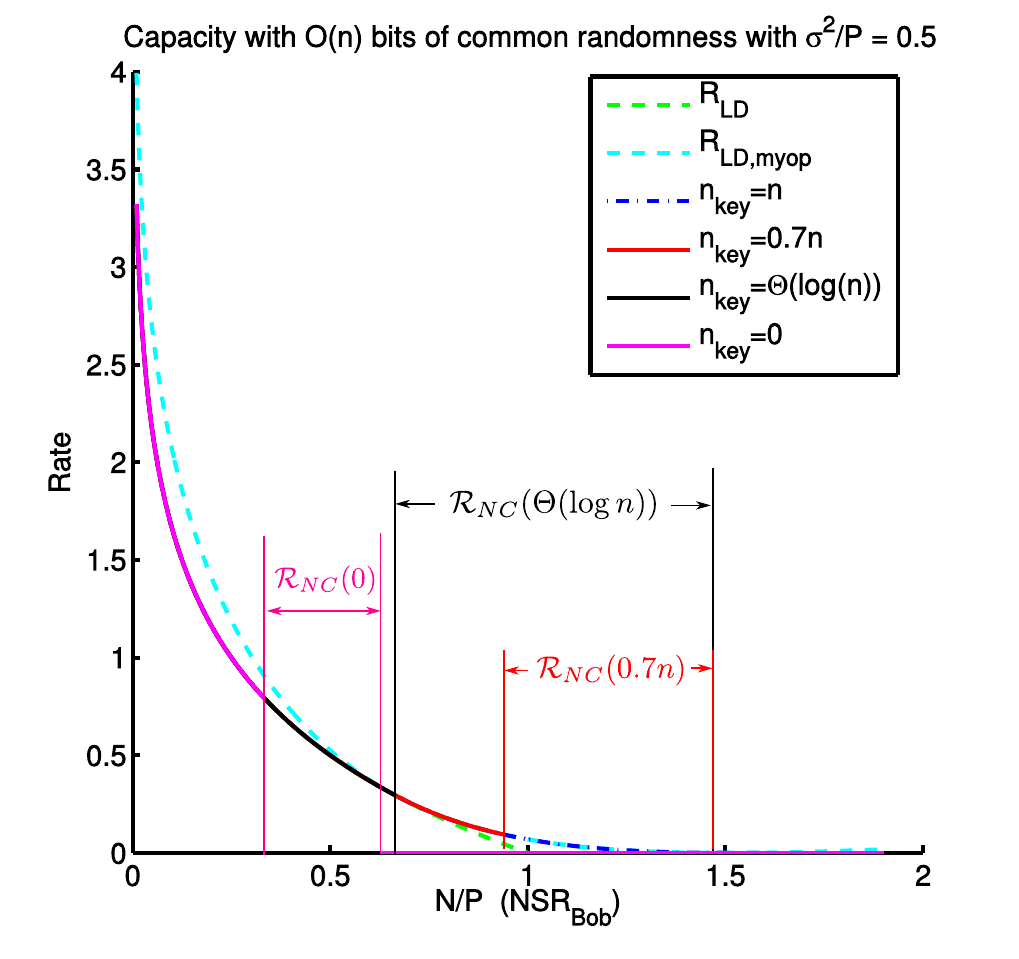}
		\caption{Capacity of the myopic adversarial channel with different amounts of common randomness. Here, $\cR_{NC}(\Nkey)$ denotes the region where our upper and lower bounds do not meet.}
		\label{fig:achievable_rate_diffkey}
	\end{center}
\end{figure}
\begin{figure} 
	\begin{center}
		\includegraphics[height=0.25\textheight]{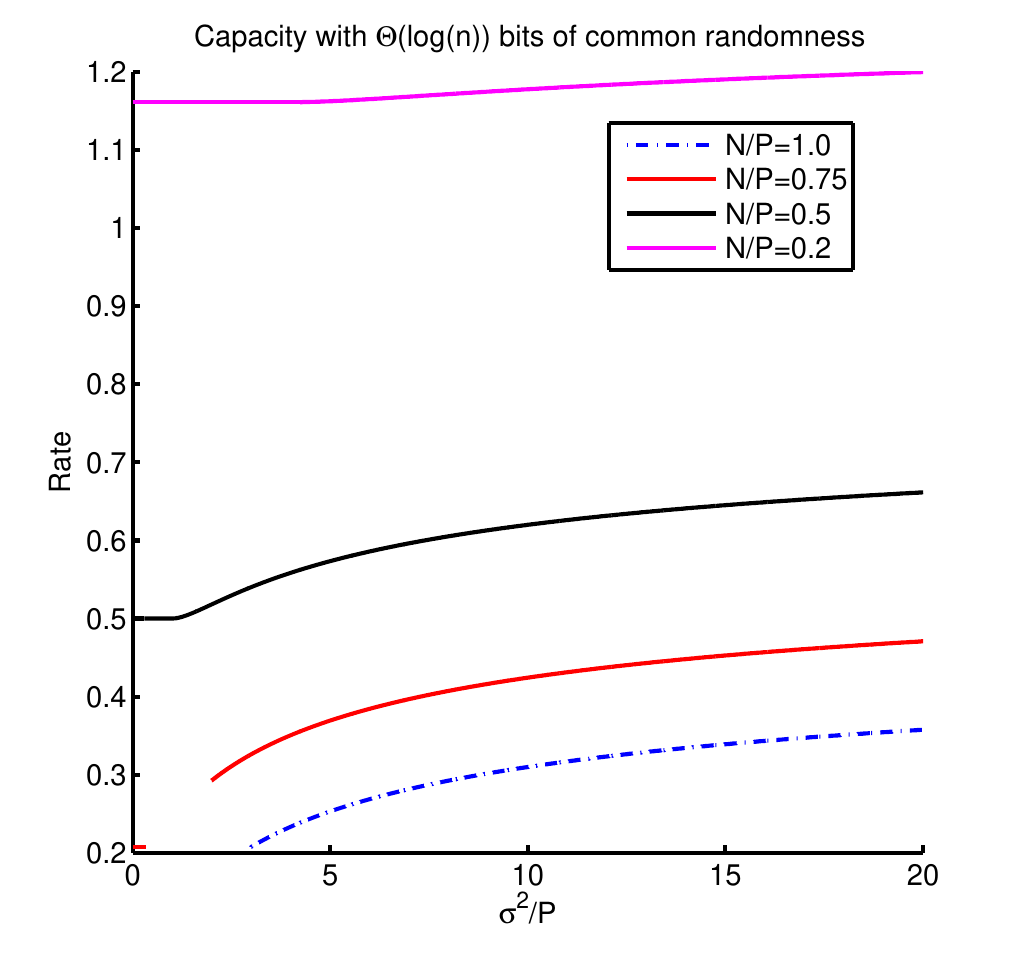}
		\caption{Capacity of the myopic adversarial channel as a function of $\sigma^2/P$.}
		\label{fig:achievable_rate_fixed_N}
	\end{center}
\end{figure}

\begin{figure}
	\begin{subfigure}[t]{.49\textwidth}
		\begin{center}
			\includegraphics[width=\linewidth]{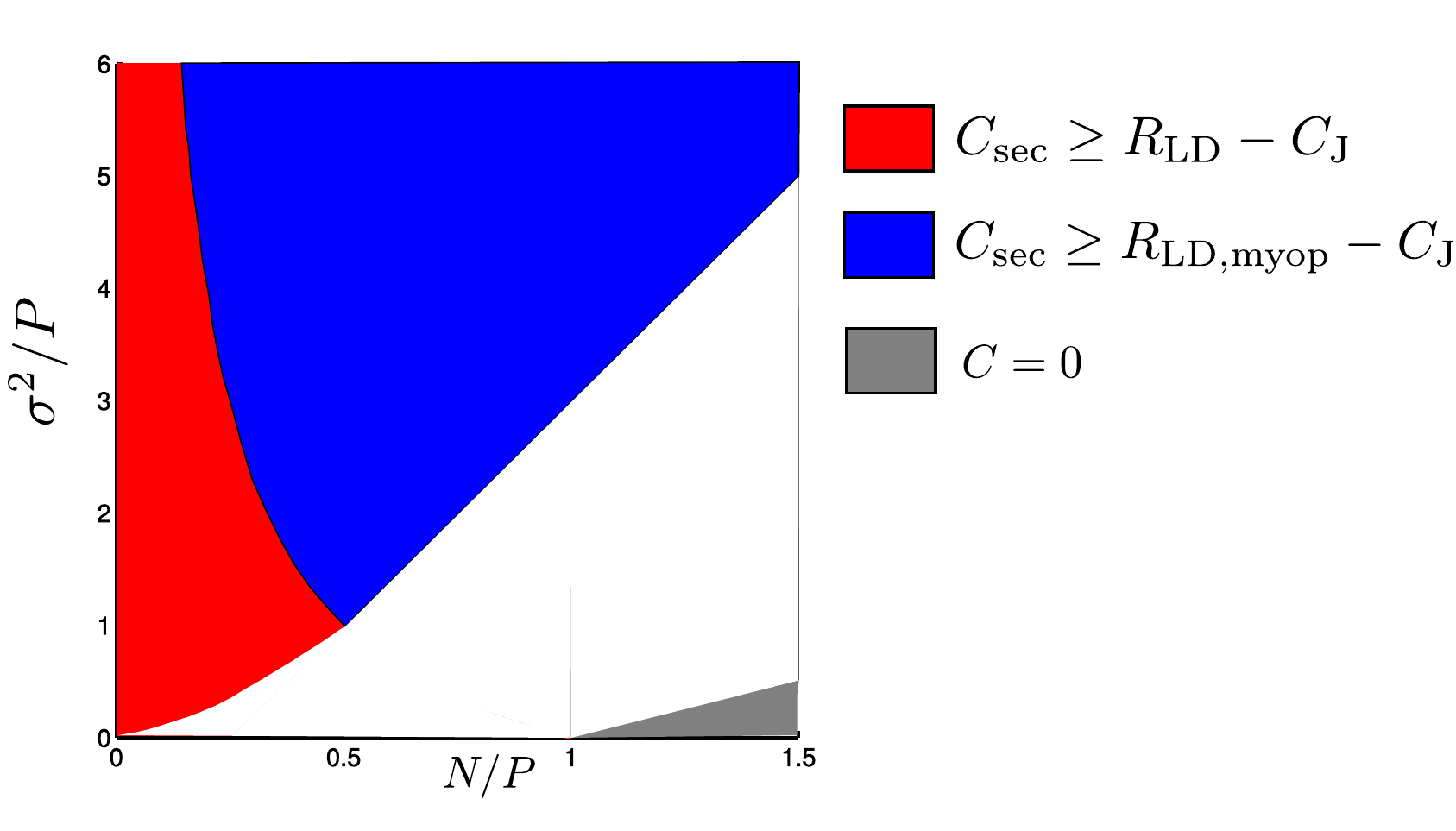}
			\caption{Achievable rates with secrecy for different noise-to-signal ratios when Alice and Bob share $ \cO(\log n) $ bits of secret key. Positive rates can be achieved in the red and blue regions.}
			\label{fig:rateregion_rkey_logn_secrecy}
		\end{center}
	\end{subfigure}
	\hfill
	\begin{subfigure}[t]{.49\textwidth}
		\begin{center}
			\includegraphics[width=\linewidth]{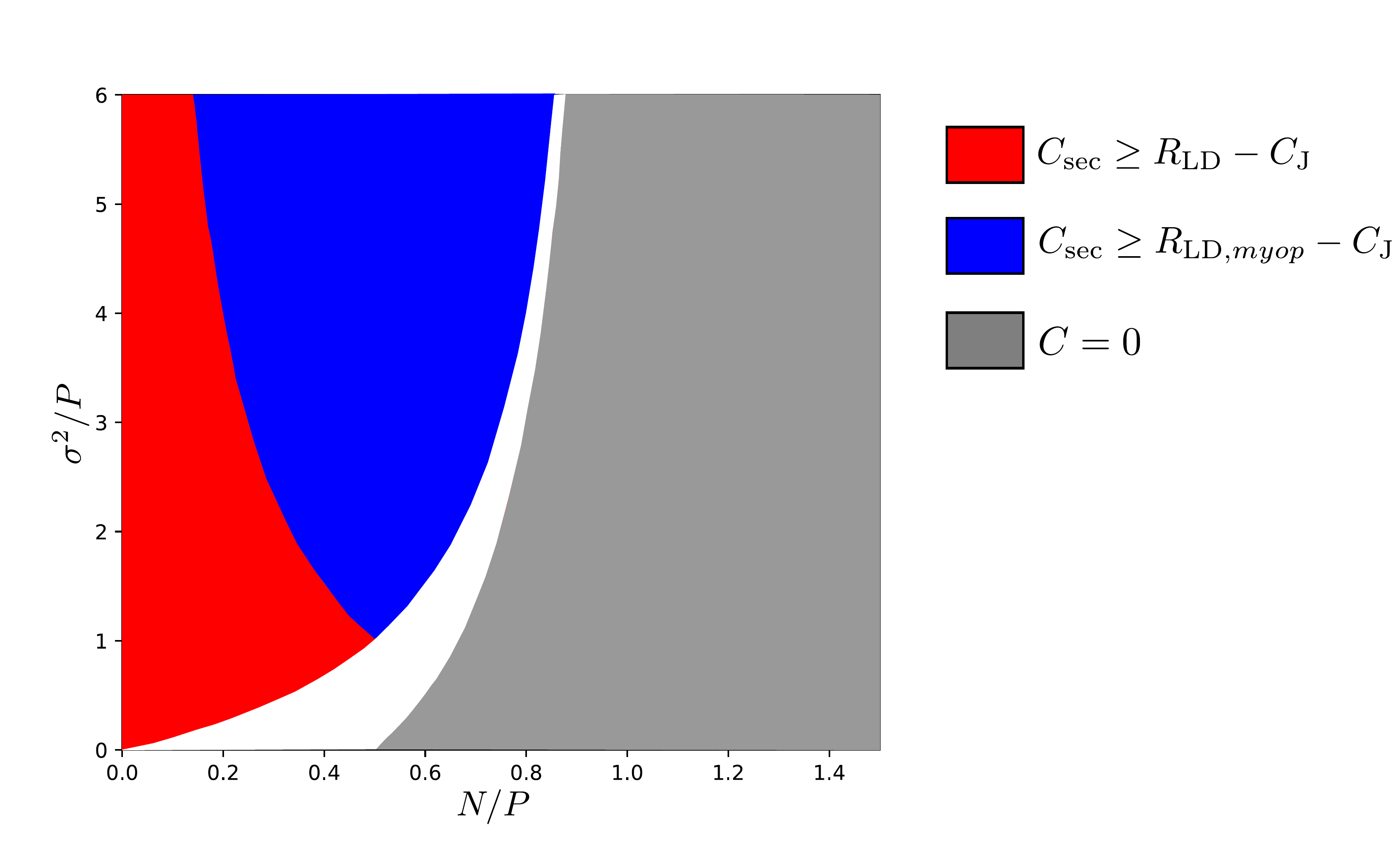}
			\caption{Achievable rates with secrecy for different noise-to-signal ratios when Alice and Bob do not share a secret key. Positive rates can be achieved in the red and blue regions.}
			\label{fig:rateregion_rkey_nocr_secrecy}
		\end{center}
	\end{subfigure}
	\\
	\begin{subfigure}[t]{.49\textwidth}
		\begin{center}
			\includegraphics[width=\linewidth]{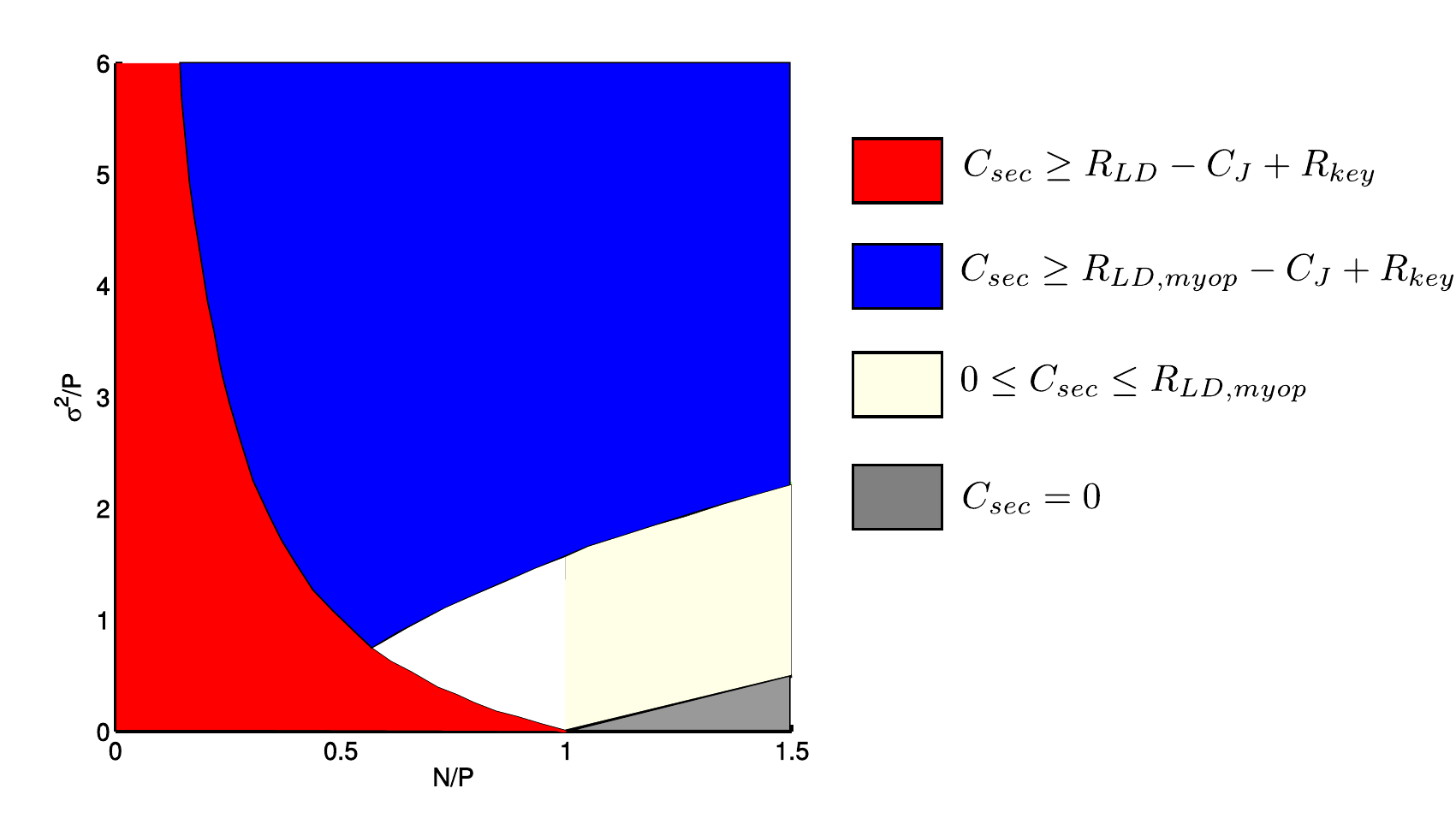}
			\caption{Achievable rates with secrecy for different noise-to-signal ratios when Alice and Bob share $ 0.2n $ bits of secret key. Positive rates can be achieved in the red and blue regions.}
			\label{fig:rateregion_rkey_0pt2_secrecy}
		\end{center}
	\end{subfigure}
	\hfill
	\begin{subfigure}[t]{.49\textwidth}
		\begin{center}
			\includegraphics[width=\linewidth]{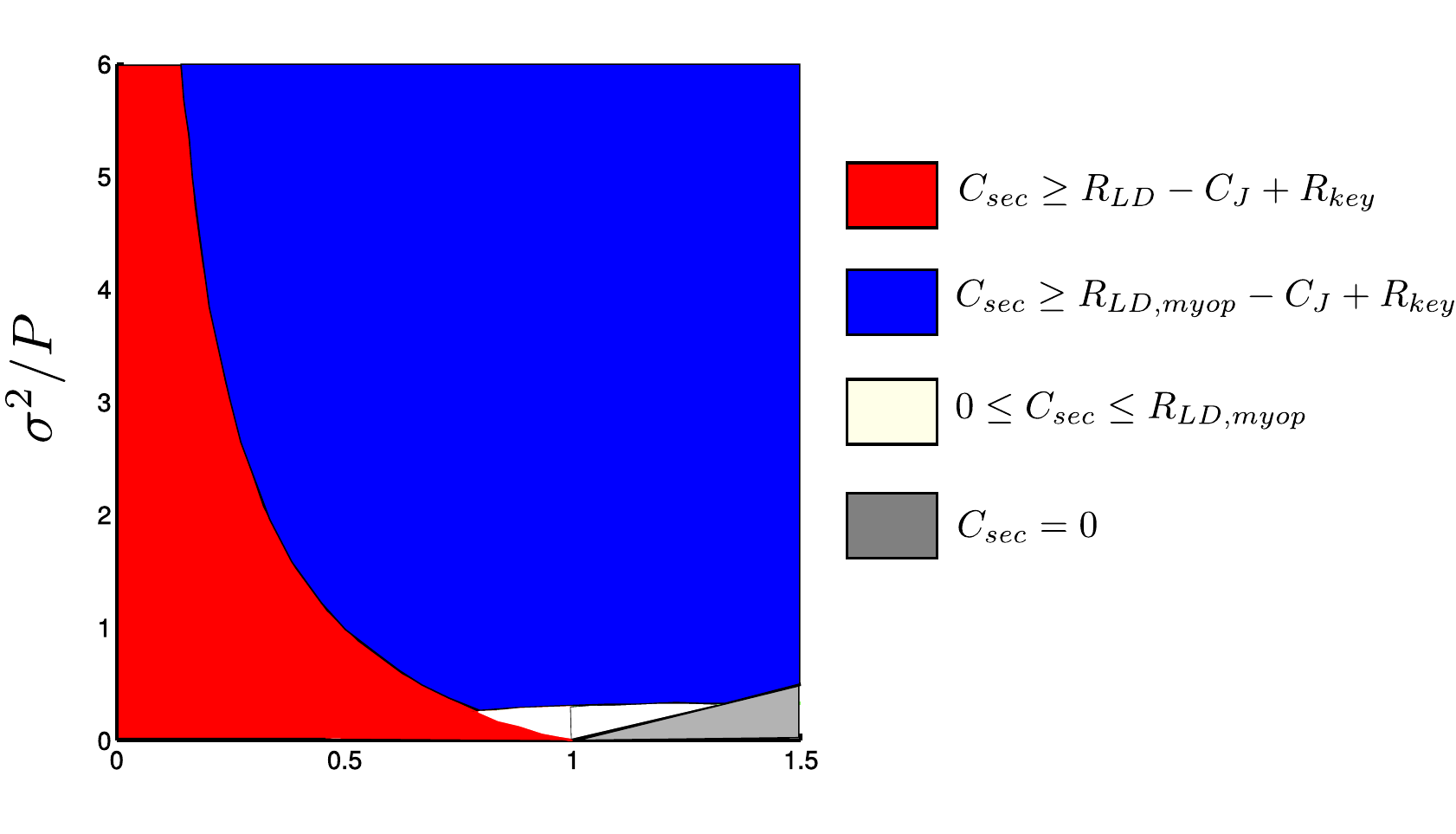}
			\caption{Achievable rates with secrecy for different noise-to-signal ratios when Alice and Bob share $ n $ bits of secret key. Positive rates can be achieved in the red and blue regions.}
			\label{fig:rateregion_rkey_1_secrecy}
		\end{center}
	\end{subfigure}
	\caption{Achievable rates with secrecy for different values of $ \Nkey $.}
\end{figure}


\section{Lemma~\ref{lemma:sarwate_thm}: Capacity of the myopic adversarial channel when $\Rkey = \infty$ }\label{sec:result_infinite_CR}
%



\rev{In this section, we prove Lemma~\ref{lemma:sarwate_thm}.}
The achievability part uses a random coding argument and can be found in~\cite{sarwate-spcom2012}. The converse can be proved by specifying an attack for James that instantiates an AWGN channel from Alice to Bob having the capacity
given by \eqref{eq:sarwate_rates}. We now give a proof of the converse (which was omitted in~\cite{sarwate-spcom2012}). The bounds that we obtain in the following subsection are tight. Coupled with the achievability in~\cite{sarwate-spcom2012}, we have a complete characterization of the capacity. This also gives us some insight as to what optimal attack strategies for James might be, and hence give a detailed argument.

\subsection{Proof of converse: ``scale-and-babble'' attack}\label{sec:proof_scalebabble}
\rev{In this subsection, we prove the converse part of Lemma~\ref{lemma:sarwate_thm}.}
The converse involves what we call a ``scale-and-babble'' attack strategy for James. This attack essentially converts the myopic channel into an equivalent AWGN channel.
Since the capacity of the AWGN channel cannot be increased using common randomness, this gives us an upper bound on the capacity for all values of $\Nkey$. 
We make no claim about the originality of this proof, as the strategy is well-known (and maybe dates back to Blachman~\cite{blachman-1962}, as suggested in~\cite{sarwate-spcom2012}. This was also implicitly used in~\cite{medard}). But we nevertheless provide the details to keep the paper self-contained.

The strategy for James is the following: He uses a certain fraction of his power to subtract a negatively scaled version of his observation; the remainder of his power is used to add AWGN.
Specifically,
\begin{equation}
 \vbfs=\pmb\beta(-\alpha\vbfz+\vbfg)=\pmb\beta(-\alpha(\vbfx+\vbfsz)+\vbfg), \notag 
\end{equation}
where\footnote{Here, $\pmb\beta$ is a factor introduced to handle atypicality in the noise. As we show subsequently, the value of $\pmb\beta$ is $1$ with high probability.}
\[\pmb\beta=\begin{cases}
1,&\|-\alpha\vbfz+\vbfg\|_2\le\sqrt{nN}\\
\frac{\sqrt{nN}}{\|-\alpha\vbfz+\vbfg\|_2}=:\pmb\beta',&\ow
\end{cases},\]
$\alpha>0$ is a constant to be optimized subject to $N-\alpha^2(P+\sigma^2)\ge0$, i.e., $\alpha\le\sqrt{N/(P+\sigma^2)}$. Also, $\vbfg\sim\cN(0,\gamma^2\bfI_n)$ with $\alpha^2(P+\sigma^2)+\frac{\gamma^2}{1-\varepsilon}=N$, for a small $\varepsilon>0$. Therefore, $\gamma^2=(N-\alpha^2(P+\sigma^2))(1-\varepsilon)$, and $\p{}(\pmb\beta\ne1)=\p{}(\|-\alpha\vbfz+\vbfg\|_2>\sqrt{nN})=2^{-\Omega(n)}$. Then
\begin{align}
\vbfy=\vbfx+\vbfs=\begin{cases}
    \vbfx-\alpha(\vbfx+\vbfsz)+\vbfg=(1-\alpha)\vbfx-\alpha \vbfsz+\vbfg,&\|-\alpha\vbfz+\vbfg\|_2\le\sqrt{nN}\\
    \vbfx+\pmb\beta'(-\alpha(\vbfx+\vbfsz)+\vbfg)=\left(1-\pmb\beta'\alpha\right)\vbfx-\pmb\beta'\alpha\vbfsz+\pmb\beta'\vbfg,&\ow
    \end{cases}.
{\label{eq:scale_babble_channel}}
\end{align}
    The scaling factor $\pmb\beta$ is introduced to make the attack vector satisfy the power constraint.
Note that the channel above is not exactly an AWGN channel. However, the probability that $\pmb\beta \neq 1$ is exponentially small in $n$. We have the following claim, formally proved in Appendix~\ref{sec:prf_cap_scaleandbabble}.
\begin{claim}
Fix $0<\varepsilon<1$, and let $\alpha$ be nonnegative. If $\gamma^2= (N-\alpha^2(P+\sigma^2))(1-\varepsilon)$, then 
 the capacity $C_{\mathrm{AWGN}}$ of the channel \eqref{eq:scale_babble_channel} from $\vbfx$ to $\vbfy$ is upper bounded as follows
 \begin{equation}
  C_{\mathrm{AWGN}}<\frac{1}{2}\log\left(1+\frac{(1-\alpha)^2P}{\alpha^2\sigma^2+\gamma^2}\right).
 \label{eq:cap_awgn_converse}
 \end{equation}
 \label{claim:cap_scaleandbabble}
\end{claim}

Using Claim~\ref{claim:cap_scaleandbabble}, we have
\[C_{\mathrm{myop,rand}}<\frac{1}{2}\log\left(1+\frac{(1-\alpha)^2P}{\alpha^2\sigma^2+\gamma^2}\right)=\frac{1}{2}\log\left(1+\frac{(1-\alpha)^2P}{N-\alpha^2P}\right) + g(\varepsilon,P,N,\sigma^2,\alpha),\]
where $g(\varepsilon,P,N,\sigma^2,\alpha)\to 0$ as $\varepsilon\to 0$. Since this holds for every $\varepsilon>0$, and every $0< \alpha \leq\sqrt{N/(P+\sigma^2)}$, we can say that
\begin{equation}
 C_{\mathrm{myop,rand}}< \min_{0< \alpha \leq\sqrt{N/(P+\sigma^2)}}\frac{1}{2}\log\left(1+\frac{(1-\alpha)^2P}{N-\alpha^2P}\right).
\label{eq:opt_conv}
\end{equation}

Denote $f(\alpha)\coloneq\frac{(1-\alpha)^2P}{N-\alpha^2P}$ and  $R(\alpha)\coloneq\frac{1}{2}\log(1+f(\alpha))$. The minimum points of $f(\alpha)$ are $\alpha=N/P$ \rev{and $\alpha = 1$}.
\begin{enumerate}
  \item $N/P<1$ (See Figure~\ref{fig:opt1} with $P/N=1.1$).
    \begin{figure} 
        \centering
        \begin{subfigure} {.45\textwidth}
            \centering
            \includegraphics[width=.9\textwidth]{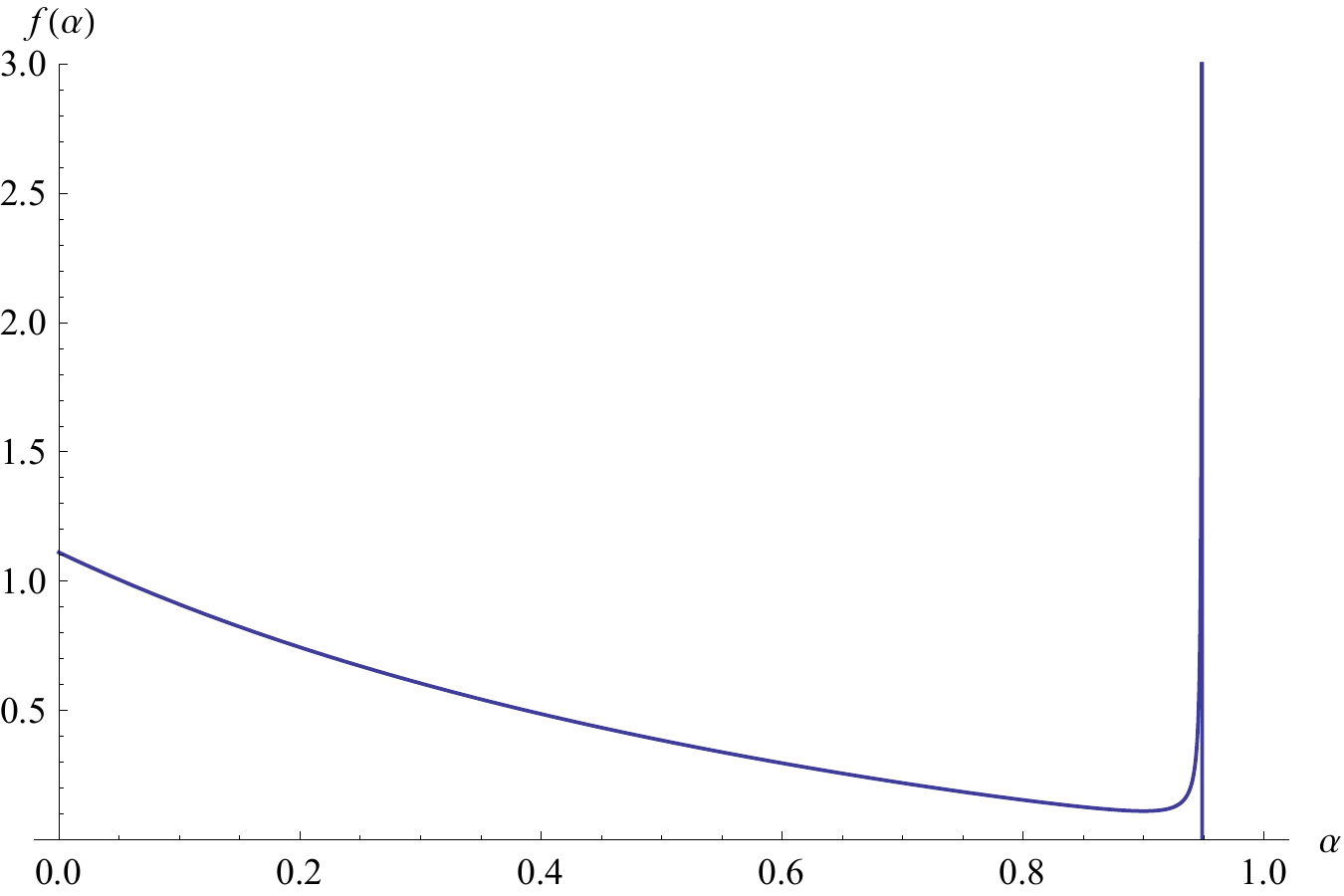}
            \caption{$f(\alpha)$ with $N/P=0.9$.}
            \label{fig:opt1}
        \end{subfigure}
        ~
        \begin{subfigure} {.45\textwidth}
            \centering
            \includegraphics[width=.9\textwidth]{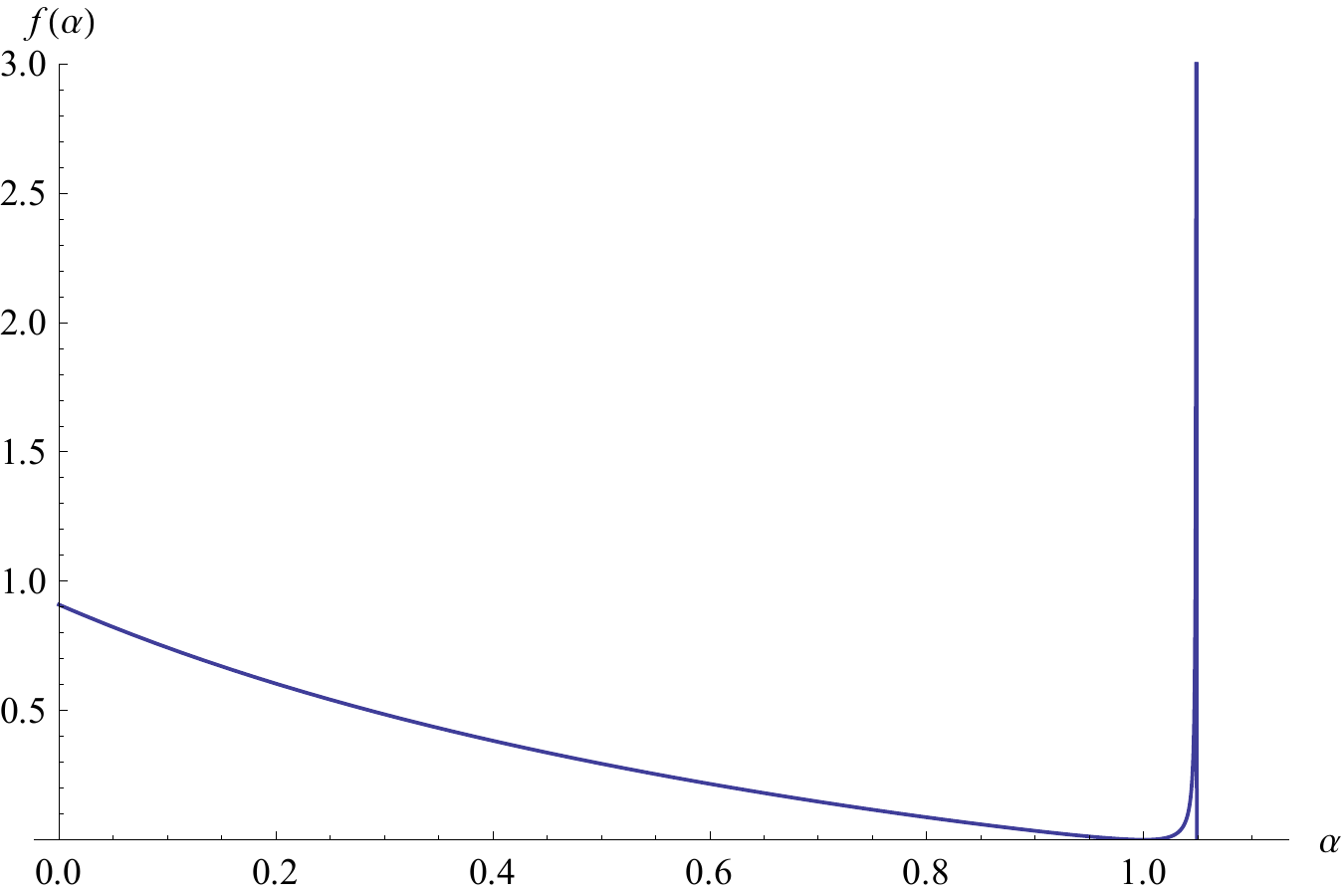}
            \caption{$f(\alpha)$ with $N/P=1.1$.}
            \label{fig:opt2}
        \end{subfigure}
        \caption{Optimization of $f(\alpha)$.}
        \label{fig:opt}
    \end{figure}
    \begin{enumerate}
        \item If $N/P<\sqrt{N/(P+\sigma^2)}$, i.e., $\frac{\sigma^2}{P}<\frac{1}{N/P}-1$, then
        \[\min_{0<\alpha\leq\sqrt{N/(P+\sigma^2)}}R(\alpha)=\frac{1}{2}\log(1+f(N/P))=\frac{1}{2}\log\frac{P}{N}.\]
        \item If $N/P\ge\sqrt{N/(P+\sigma^2)}$, i.e., $\frac{\sigma^2}{P}\ge\frac{1}{N/P}-1$, then
        \begin{align*}
            \min_{0<\alpha\leq\sqrt{N/(P+\sigma^2)}}R(\alpha)=&\frac{1}{2}\log\left(1+f\left(\sqrt{\frac{N}{P+\sigma^2}}\right)\right)\\
            =&\frac{1}{2}\log\left(\frac{(P+\sigma^2)(P+N)-2P\sqrt{N(P+\sigma^2)}}{N\sigma^2}\right)\eqcolon \Rmyop .
        \end{align*}
    \end{enumerate}
    Notice that if $\sigma\to\infty$, the channel becomes oblivious. It can be directly verified that
    \[\lim_{\sigma\to\infty}\frac{1}{2}\log\left(\frac{(P+\sigma^2)(P+N)-2P\sqrt{N(P+\sigma^2)}}{N\sigma^2}\right)=\frac{1}{2}\log\left(1+\frac{P}{N}\right),\]
    consistent with the oblivious capacity of quadratically constrained channels.
    \item $N/P<1$ (See Figure~\ref{fig:opt2} with $P/N=0.9$).
    \begin{enumerate}
        \item If $\sqrt{N/(P+\sigma^2)}\ge1$, i.e., $\frac{\sigma^2}{P}\le\frac{N}{P}-1$, then
        \[\min_{0<\alpha\leq\sqrt{N/(P+\sigma^2)}}R(\alpha)=\frac{1}{2}\log(1+f(1))=0.\]
        \item If $\sqrt{N/(P+\sigma^2)}<1$, i.e., $\frac{\sigma^2}{P}>\frac{N}{P}-1$, then
        \[\min_{0<\alpha\leq\sqrt{N/(P+\sigma^2)}}R(\alpha)=\frac{1}{2}\log\left(1+f\left(\sqrt{\frac{N}{P+\sigma^2}}\right)\right)=\Rmyop .\]
    \end{enumerate}
\end{enumerate}
The upper bound on the capacity given by the scale-and-babble attack is shown in Figure~\ref{fig:r_opt}.
\begin{figure} 
    \centering
    \includegraphics[width = 0.5\textwidth]{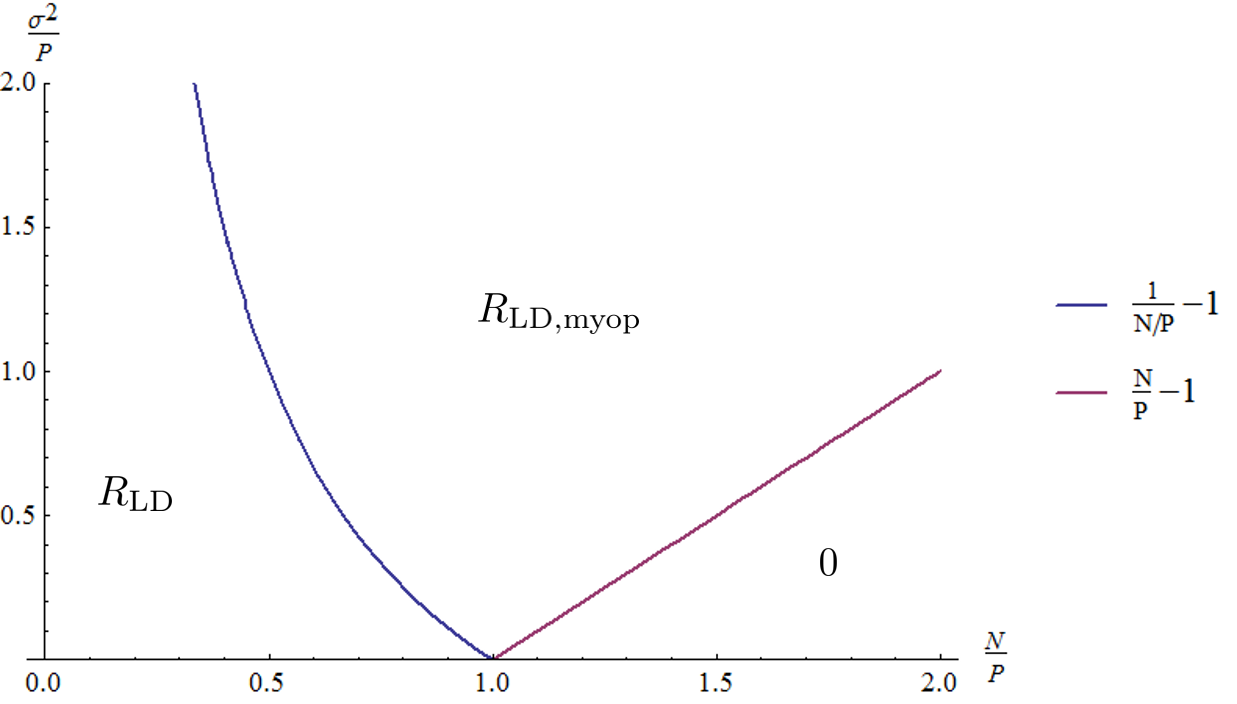}
    \caption{The upper bound of capacity given by scale-and-babble attack.}
    \label{fig:r_opt}
\end{figure}

The scale-and-babble attack converts the adversarial channel into an equivalent AWGN channel.
The capacity of the point-to-point AWGN channel cannot be increased using private/common randomness. Therefore, the upper bounds obtained using this technique
hold regardless of whether deterministic/stochastic/randomized codes are used, and regardless of the amount of common randomness shared by the encoder and decoder\footnote{This is due to the fact that common randomness/stochastic encoding does not increase the capacity of a memoryless channel. This in turn can be derived from Fano's inequality with common randomness (see Appendix~\ref{sec:prf_cap_scaleandbabble} for more details).}.

\section{Lemmas~\ref{lemma:achievablerate_thetanbits} and~\ref{lemma:achievablerate_Olognbits}: Linear and sublinear amounts of common randomness}\label{sec:result_linear_sublinearCR}

Our approach in these two regimes will involve a myopic list-decoding argument, which we will prove next. We will combine this with a known technique~\cite{langberg-focs2004,sarwate-thesis,bhattacharya2019sharedrandomness} which uses $\Theta(\log n)$ bits of common randomness to disambiguate the list and give us a result for unique decoding\footnote{For the case $ \sigma^2=0 $,~\cite{bhattacharya2019sharedrandomness} also showed that $ \log n $ bits are necessary to achieve a rate equal to list decoding capacity with unique decoding.}.
Before we state the main results, we take a brief detour to discuss classical and myopic list-decoding.

\subsection{List-decoding}\label{sec:list_decoding_main}

Consider the quadratically constrained adversarial channel model where James observes $\vbfz,$ which is a noisy copy of $\vbfx$. In the list-decoding problem, the decoder is not required to recover the transmitted message exactly but can instead output a (small) list of messages with the guarantee that the true message is in the list. We are typically interested in list-sizes that are constant or grow as a low-degree polynomial function of the blocklength.  



As a warm-up, let us consider the omniscient adversary (i.e., when $ \sigma^2=0 $). We have the following folk theorem\footnote{As will be evident from the proof, the omniscient list-decoding capacity is shown to be $ \frac{1}{2}\log \frac{P}{N} $ under a (stronger) maximum probability of error criterion.}.
\begin{lemma}
 Let $\sigma = 0$ and $\Nkey=0$. If $\varepsilon\coloneq\frac{1}{2}\log\frac{P}{N}-R>0$, then $R$ is achievable for $\left(P,N,\Omega\left(\frac{1}{\varepsilon}\log\frac{1}{\varepsilon}\right)\right)$-list-decoding.
 If $R>\frac{1}{2}\log\frac{P}{N}$, then no sequence of codebooks of rate $R$ is  $(P,N,n^{\cO(1)})$-list-decodable.
 Therefore, for polynomial list-sizes, the omniscient list-decoding capacity is equal to $\frac{1}{2}\log\frac{P}{N}$.
\label{lemma:listdecoding_omniscient}
\end{lemma}
Although the above result is well-known, we are not aware of a reference with a formal proof of this statement. For completeness, and to keep the paper self-contained, we give
a proof in Appendix~\ref{sec:prf_omniscient_list_capacity}.

\subsection{Theorem~\ref{thm:myopic_listdecoding_summary}: List-decoding with a myopic adversary}
\label{sec:myop_list_dec_sketch}
\rev{In this subsection, we provide a proof outline of Theorem~\ref{thm:myopic_listdecoding_summary} and forward the detailed proof to Sec.~\ref{sec:myopic_list_decoding_details}. }

In the case where $\sigma>0$, we can achieve a higher list-decoding rate for certain values of $N/P$ and $\sigma^2/P$.
We show that when noise level to James is large enough (he is ``sufficiently myopic''), he is unable to exactly determine $\vbfx$ from $\vbfz$. Conditioned on $\vbfz$, the transmitted codeword lies in a thin strip 
which is roughly $\sqrt{n\sigma^2}$ away from $\vbfz$. If $R$ is large enough, then the strip will contain exponentially many codewords,
and James cannot distinguish the true codeword from the others. Since we use a random code, these codewords are roughly uniformly distributed over the strip.
An effective value of $\vbfs$ for one codeword on the strip (in the sense of ensuring maximum confusion for Bob) may be ineffective for most of the remaining codewords. 
As a result, there is no single direction where James can align $\vbfs$ in order to guarantee the level of confusion that Bob could have if he were omniscient.
This is what will let us achieve a higher rate.

We will consider the case $\Nkey=n\Rkey$, for some $\Rkey\geq 0$. Even the case when $\Rkey = 0$ is non-trivial -- see Figure~\ref{fig:rateregion_myopiclist_arg_a} for an illustration of the achievable rate. 
\begin{customthm}{11}[Restatement of Theorem~\ref{thm:myopic_listdecoding_summary}]
 For $(P,N,\cO(n^2))$-list-decoding, the capacity is lower bounded as follows
 \[
  C_{\mathrm{myop,LD}} \geq \begin{cases}
                          \Rmyop, & \text{ if } \frac{\sigma^2}{P}\ge\max\left\{\frac{1}{N/P}-1,\frac{N}{P}-1\right\} \text{ and }  \Rmyop +\Rkey > \frac{1}{2}\log \left( 1+\frac{P}{\sigma^2} \right)\\
                          \Rld, & \text{ otherwise}
                         \end{cases}.
 \]
 These are summarized in Fig.~\ref{fig:rateregion_myopiclist_arg}.
\label{thm:myopic_listdecoding}
\end{customthm}


We now provide a sketch of the proof, but relegate the details to Sec.~\ref{sec:myopic_list_decoding_details}.

\subsubsection{Proof sketch}\label{sec:proofsketch_myopicld}
First, observe that $\Rld$ is achievable as long as $N<P$. This is true since $\Rld$ is achievable even with an omniscient adversary (Lemma~\ref{lemma:listdecoding_omniscient}). The nontrivial step is in showing that a higher rate of $\Rmyop$ is achievable in a certain regime of the NSRs.
We will prove the achievability using random spherical codes. The $2^{n(R+\Rkey)}$ codewords are sampled independently and uniformly at random from $\cS^{n-1}(0,\sqrt{nP})$. {This is  partitioned randomly into $ 2^{n\Rkey} $ codebooks, each containing $ 2^{nR} $ codewords. The value of the shared key determines which of the $ 2^{n\Rkey} $ codebooks is used for transmission. Since Bob has access to the key, he has to decode one of $ 2^{nR} $ codewords. We analyze the probability of error taking James's point of view. To James, one of $2^{n(R+\Rkey)}$ codewords is chosen at random,} and the code is $(P,N,\cO(n^2))$-list-decodable with high probability if no attack vector $ \vs $ can force a list-size of $ \Omega(n^2) $ for a nonvanishing fraction of the codewords. 

Conditioned on $\vbfz$, the true codeword $\vbfx$ lies in a thin strip at distance approximately $\sqrt{n\sigma^2}$ to $\vbfz$. We show Lemma~\ref{lemma:exp_cw_strip} that as long as the codebook rate $R+\Rkey$ is greater than $\frac{1}{2}\log\left(1+\frac{P}{\sigma^2}\right)$, this strip (with high probability) contains exponentially many codewords, thereby causing sufficient confusion for James. This condition effectively limits the values of the NSRs where $\Rmyop$ is achievable. 

For ease of analysis, we assume that James has access to an oracle, which reveals a particular subset of $2^{\varepsilon n}$ codewords from the strip. This set is guaranteed to contain the true codeword. Clearly, the oracle only makes James more powerful, and any result that holds in this setting also continues to be valid when there is no oracle. The codewords in the oracle-given set (OGS) are all independent (over the randomness in the codebook) and are approximately uniformly (which we call \emph{quasi-uniformly}) distributed over the strip. See Lemma~\ref{lemma:quasiuniformity} for a formal statement. The codewords outside the oracle-given set are independent of these, and are uniformly distributed over the sphere. {From the point of view of James, the true codeword is quasi-uniformly distributed over the OGS. }

We fix an attack vector $\vs$, and bound the probability that this forces a list-size greater than $L$ for a significant fraction of the codewords in the oracle-given set. 
To do so, we find the typical area of the decoding region $\cB^n(\vbfx+\vs,\sqrt{nN})\cap \cS^{n-1}(0,\sqrt{nP})$ by computing the typical norm of $\vbfy$. This decoding region is a cap, whose area is maximized when the radius of the cap (see Sec.~\ref{sec:notation}) is $\sqrt{nN}$. This would be the result of James's attack if he were omniscient. However, due to the randomness in $\vbfsz$ and his uncertainty about $\vbfx$, the typical radius is considerably less than $\sqrt{nN}$.
It is this reduction in the typical radius that helps us achieve rates above $\Rld$. 
The value of the typical radius is the solution of an optimization problem, which is identical (under a change of variables) to the one we obtain when analyzing the scale-and-babble attack in Sec.~\ref{sec:proof_scalebabble}. {See Fig.~\ref{fig:blob_overlap} for an illustration. This is proved in Sec.~\ref{sec:myop_ld}, with some of the calculations appearing in subsequent subsections.}

\begin{figure*}
\begin{center}
	\includegraphics[width = 0.5\textwidth]{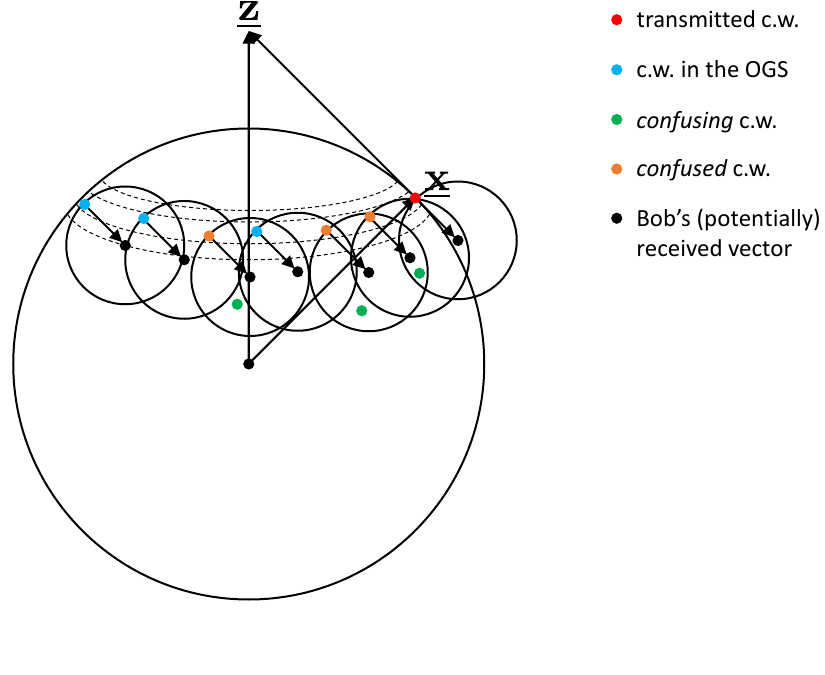}
\end{center}
\caption{Intuition behind myopic list-decoding. From the point of view of James, the true codeword is quasi-uniformly distributed over the OGS. The best-case scenario for James would be to ``push'' the true codeword towards the origin. However, for any fixed attack vector $ \vbfs, $ only a small fraction of the codewords in the OGS lead to a large overlap with $ \cS^{n-1}(0,\sqrt{nP}) $ (which could lead to a large list-size). For a randomly chosen codeword from the OGS, the overlap of the decoding ball with $ \cS^{n-1}(0,\sqrt{nP}) $ is small with high probability.}
\label{fig:blob_overlap}
\end{figure*}

With an upper bound on the typical decoding volume, we can bound the probability that there are more than $L$ codewords in the decoding region. We separately handle the codewords within and outside the oracle-given set. The probability that a fixed $\vs$ causes list-decoding failure for a significant fraction of codewords  in the oracle-given set is found to decay superexponentially in $n$. We complete the proof using a covering argument for $\vs$ and taking a union bound over all representative attack vectors $\vsq$.

\subsection{Achievable rates using $\Theta(n)$ bits of CR}\label{sec:achievability_thetan}
\rev{In this subsection, we prove Lemma~\ref{lemma:achievablerate_thetanbits}. }

The following lemma by Sarwate\cite{sarwate-thesis} (originally proved by Langberg~\cite{langberg-focs2004} for the bit-flip channel, later generalized and improved for more general AVCs and constant list sizes in~\cite{bhattacharya2019sharedrandomness}) says that a list-decodable code can be converted to a uniquely decodable code with an additional $\cO(\log n)$ bits of common randomness.
Although the lemma was stated in the context of discrete AVCs with deterministic encoding, the proof goes through even without these restrictions. We can use our list-decodable code as a black box in the following lemma to go from a list-decodable code to a uniquely decodable code.

For an arbitrary AVC $\cW$, we define an $(n,R,L,\varepsilon)$-list-decodable code as one which has blocklength $n$, message rate $R$, and 
achieves a list-size of $L$ with probability $1-\varepsilon$.

\begin{lemma}[Lemma 13,~\cite{sarwate-thesis}]
Suppose we have a deterministic $(n,R,L,\varepsilon)$-list-decodable code for an AVC $\cW$. If the encoder-decoder pair shares $\Nkey$ bits of common randomness, then there exists a blocklength-$n$ code of rate $R-\frac{\Nkey}{2n}$ such that the decoder can recover the transmitted message with probability $1-\varepsilon-\varepsilon'$, where
\[
 \varepsilon' \coloneq \frac{2nLR}{\Nkey2^{\Nkey/2}}.
\]
 \label{lemma:sarwate_listdecoding}
\end{lemma}
The above lemma says that an additional $\Nkey=2\log(nL)$ bits of common randomness is sufficient to disambiguate the list. 
If $L=n^{\cO(1)}$, then the $\Nkey$ required is only logarithmic, and the penalty in the rate $\frac{\Nkey}{2n}$ is vanishing in $n$.

We can therefore use this with our myopic list-decoding result in Theorem~\ref{thm:myopic_listdecoding_summary} to obtain achievable rates. Combining this with the converse in Lemma~\ref{lemma:sarwate_thm}, we obtain Lemma~\ref{lemma:achievablerate_thetanbits}

\subsection{Achievable rates with $\Theta(\log n)$ bits of CR}\label{sec:achievability_thetalogn}
\rev{In this subsection, we prove Lemma~\ref{lemma:achievablerate_Olognbits}. }

Lemma~\ref{lemma:sarwate_listdecoding} says that $(1+\varepsilon)\log n$ bits of common randomness is sufficient to disambiguate the list. Using this with Theorem~\ref{thm:myopic_listdecoding_summary} for $\Rkey=0$, we have Lemma~\ref{lemma:achievablerate_Olognbits}. Note that when $\Rkey=0$, the condition  $\Rmyop +\Rkey > \frac{1}{2}\log \left( 1+\frac{P}{\sigma^2}\right)$ reduces to $\frac{\sigma^2}{P}\geq 4\frac{N}{P}-1$.
%

\section{No common randomness}\label{sec:result_zero_CR}


We now discuss the basic ideas required to obtain Theorem~\ref{thm:capacity_noCR}.
\subsection{Proof sketch}\label{sec:zero_CR_sketch}
The proof involves two parts:
\begin{itemize}
 \item The upper bounds are obtained using Lemma~\ref{lemma:sarwate_thm}, and symmetrization arguments described in Sec.~\ref{sec:converse_symmetrization}.
 \item The achievability involves a combination of list-decoding, reverse list-decoding, and the grid argument. We give a high-level description below. For the rigorous proof, see Sec.~\ref{sec:achievability_suffmyopic}.
\end{itemize}
The achievability proof uses several ideas from our discussion on myopic list-decoding in Sec.~\ref{sec:proofsketch_myopicld}.
As discussed in Sec.~\ref{sec:proofsketch_myopicld}, the transmitted codeword lies in a strip.
We can only prove our result in the sufficiently myopic case, i.e., when the strip contains exponentially many codewords. This is possible only if $R>\frac{1}{2}\log\left(1+\frac{P}{\sigma^2}\right)$. Just as in the proof of myopic list-decoding, we assume that James has access to an oracle-given set. Our goal is to 
show that there exists no attack strategy for James that would cause decoding failure for a significant fraction of codewords in the oracle-given set.

See Fig.~\ref{fig:blob_overlap}.
Let us fix an $\vs$.
We say that a codeword $\vx'$ confuses $\vx$ if $\cB^n(\vx+\vs,\sqrt{nN})$ contains $\vx'$. From James's perspective, the codewords in the oracle-given set are (approximately\footnote{We add this qualifier since all points in the strip are not equidistant from $\vbfz$.}) equally likely to have been transmitted.
We say that decoding fails if the attack vector chosen by James causes $\vbfx$ to be confused with another codeword. To analyze this, we study the effect of $\vs$ simultaneously over all codewords in the oracle-given set. The set $\bigcup_{m\in\ogs}\cB^n(\vbfx(m)+\vs,\sqrt{nN})$ forms a ``blob''. We show that the probability that the blob contains more than $n^4$ confusing codewords\footnote{The degree of the polynomial here is not important, but is chosen to be $4$ for convenience. What matters is that the list-size grows polynomially in $n$.} is vanishingly small. This ensures that there are only a polynomial number of codewords that could potentially confuse the exponentially many codewords in the oracle-given set. 

We then find the number of codewords $\vbfx(m)$ in the OGS that could potentially be confused by a fixed $\vbfx'$. We call this reverse list-decoding, and show that each $\vbfx'$ can confuse only polynomially many codewords in the OGS. Our reverse list-decoding argument holds if $\frac{\sigma^2}{P}\geq \frac{1}{1-N/P}-1$.

We combine the blob list-decoding and reverse list-decoding results and give a combinatorial argument to show that the probability of a fixed $\vs$ causing a decoding error for a significant fraction of codewords in the OGS is super-exponentially decaying in $n$. 
{
We categorise the error into two types:
\begin{itemize}
	\item Type I: The ``confusing'' codeword does not lie within the OGS. Blob list-decoding and reverse list-decoding in this case are studied in Sec.~\ref{sec:type_i}.
	\item Type II: The ``confusing'' codeword lies within the OGS. This is studied in Sec.~\ref{sec:type_ii}.
\end{itemize}
}
 We then use a standard covering argument for $\vs$ and show that the average probability of decoding error is also vanishing in $n$. This will complete the proof of Theorem~\ref{thm:capacity_noCR}.

\subsection{An improved converse using symmetrization}\label{sec:converse_symmetrization}
\rev{In this section, we prove the symmetrization part of the converse of Theorem~\ref{thm:capacity_noCR}.}

When the encoder is a deterministic map from the set of messages to $\bR^n$, we can give a better upper bound for certain values of the NSRs. This attack is based on the scaled babble-and-push attack designed by Li et al.~\cite{tongxin-causal-2018} for the quadratically constrained channel with a causal adversary.
The basic idea in a symmetrization argument is to make sure that Bob is equally confused between the actual transmitted codeword
and a random codeword independently chosen by James. Bob will then be unable to distinguish between the two codewords and therefore makes an error with nonvanishing probability.
\begin{lemma}
 If $\Nkey=0$, then 
 $C_{\mathrm{myop}}=0$ when $\rev{\frac{\sigma^2}{P}\le\frac{1}{1-N/P} - 2}$.
\label{lemma:converse_symmetrization}
\end{lemma}
\begin{proof}
\rev{We first present two \emph{suboptimal} jamming strategies referred to as \emph{$\vbfz$-agnostic symmetrization} and \emph{$\vbfz$-aware symmetrization}. 
They are simple and natural strategies and give respectively the following bounds \emph{inferior} to the one claimed in Lemma~\ref{lemma:converse_symmetrization}.
\begin{enumerate}
	\item\label{itm:bound_z_agnostic} $\vbfz$-agnostic symmetrization: If $\Nkey=0$, then  $ C_{\mathrm{myop}}=0 $ when $ N\geq P $.
	\item\label{itm:bound_z_aware} $\vbfz$-aware symmetrization: If $\Nkey=0$, then  $ C_{\mathrm{myop}}=0 $ when $ \frac{\sigma^2}{P}<4\frac{N}{P}-2 $.
\end{enumerate}
We then slightly modify $\vbfz$-aware symmetrization and present an optimal symmetrization-type attack.
The analysis follows verbatim that of $\vbfz$-aware symmetrization by changing some coefficients. 
The bound given by such an improved symmetrization subsumes and extends those given by $\vbfz$-agnostic/-aware symmetrization.}

\paragraph{\rev{$\vbfz$-agnostic symmetrization}}
The first part (\Cref{itm:bound_z_agnostic}) is considerably simpler, and involves a \emph{$\vbfz$-agnostic} symmetrization argument. If $N\geq P$, then James can mimic Alice.
A simple attack strategy is the following: He generates a message $\bfm'$ uniformly at random, and independently of everything else.
Using the same encoding strategy that Alice uses; $ \bfm' $ is mapped to a codeword $\vbfs=\vbfx'$ which is transmitted. Bob receives $\vbfx+\vbfx'$,
and unless $\bfm'=\bfm$, he will be unable to determine whether Alice sent $\bfm$ or $\bfm'$. Therefore, with probability $1-2^{-nR}$ he is unable to decode the correct message, and this is true for all $R>0$.
Therefore, the capacity is zero when $N\geq P$.

\paragraph{\rev{$\vbfz$-aware symmetrization}\label{par:z_aware_symm}}
To prove the second part (\Cref{itm:bound_z_aware}), when $\frac{\sigma^2}{P}<4\frac{N}{P}-2$, we give a \emph{$\vbfz$-aware} symmetrization attack. 
Here, James picks a random codeword $\vbfx'$ uniformly from the codebook
and ``pushes'' $\vbfz$ to the midpoint of $\vbfz$ and $\vbfx'$. Bob is then unable to distinguish between $\vbfx$ and $\vbfx'$, and will therefore make an error with nonvanishing probability.
Specifically,
 James samples $\bfm'\sim p_{\bfm},\vbfx'\sim p_{\vbfx|\bfm}$, both independently of Alice, and sets 
\begin{align}
\vbfs=\frac{1}{2}\pmb\beta(\vbfx'-\vbfz)=\frac{1}{2}\pmb\beta(\vbfx'-\vbfx-\vbfsz),
\label{eqn:z_aware_s_def}
\end{align}
where 
\begin{equation}
\pmb\beta=\begin{cases}
1,&\left\|\frac{1}{2}(\vbfx'-\vbfz)\right\|_2\le\sqrt{nN}\\
\frac{\sqrt{nN}}{\left\|\frac{1}{2}(\vbfx'-\vbfz)\right\|_2}=:\pmb\beta',&\ow
\end{cases},
\label{eq:symmetrized_channel}
\end{equation}
such that 
\[\vbfy=\vbfx+\vbfs=\begin{cases}
\vbfx+\frac{1}{2}(\vbfx'-\vbfx-\vbfsz)=\frac{1}{2}(\vbfx'+\vbfx)-\frac{1}{2}\vbfsz,&\left\|\frac{1}{2}(\vbfx'-\vbfz)\right\|_2\le\sqrt{nN}\\
\vbfx+\frac{1}{2}\pmb\beta'(\vbfx'-\vbfx-\vbfsz)=\left(1-\frac{1}{2}\pmb\beta'\right)\vbfx+\frac{1}{2}\pmb\beta'\vbfx'-\frac{1}{2}\pmb\beta'\vbfsz,&\ow
\end{cases}.\]
We introduce $\pmb\beta$ merely to ensure that the attack vector always satisfies James's power constraint. We don't care much about the second case in Equation~\eqref{eq:symmetrized_channel}, since the probability of the second case goes to zero. Hence, even if Bob may be able to decode the message in the second case,  we show that in the first case his probability of error is going to be bounded away from zero. Assume we operate at a rate $R$. 
Define  $ \vbfz' \coloneq \vbfx' + \vbfsz $, $ \vbfs'\coloneq \frac{1}{2}(\vbfx - \vbfz') $ and $ \vbfy' \coloneq \vbfx' + \vbfs' = \vbfx' + \frac{1}{2}(\vbfx - \vbfz') = \vbfx' + \frac{1}{2}(\vbfx - \vbfx' - \vbfsz) = \frac{1}{2}(\vbfx' + \vbfx) - \frac{1}{2}\vbfsz = \vbfy $.
The probability of error can be lower bounded by
\begin{align*}
    P_e=&\p{}(\widehat\bfm\ne\bfm)\\
    \ge&\p{}(\widehat\bfm\ne\bfm, \vbfx'\ne\vbfx, \|\vbfs\|_2\le\sqrt{nN}, \|\vbfs'\|_2\le\sqrt{nN} )\\
    =&\p{}(\vbfx'\ne\vbfx, \|\vbfs\|_2\le\sqrt{nN}, \|\vbfs'\|_2\le\sqrt{nN})\p{}(\widehat\bfm\ne\bfm|\vbfx'\ne\vbfx, \|\vbfs\|_2\le\sqrt{nN}, \|\vbfs'\|_2\le\sqrt{nN})\\
    \ge&\frac{1}{2}\p{}(\vbfx'\ne\vbfx, \|\vbfs\|_2\le\sqrt{nN}, \|\vbfs'\|_2\le\sqrt{nN}),
\end{align*}
where the last inequality comes from the following argument. Suppose that James has enough power to push the channel output $\vbfy$ to $(\vbfx+\vbfx')/2$ even when $\vbfsz=0$, and that Bob knew that his observation $\vbfy$ is the average of $\vbfx$ and $\vbfx'$\footnote{Notice that there may exist other pairs of codewords with the same average.}. In this case, Bob cannot distinguish whether $\vbfx$ or $\vbfx'$ was transmitted and his probability of decoding error is no less than $1/2$. Note that, in the myopic case we are considering, Bob's observation $\vbfy=(\vbfx+\vbfx')/2-\vbfsz/2$ also contains a (scaled) random noise component other than the average of two codewords. The noise is completely random and independent of everything else, hence it does not provide Bob with any information of $\bfm$ and a decoding error will occur still with probability at least $1/2$.
However, there is one more caveat. 
The output $ \vbfy $ when $ \vbfx $ was transmitted by Alice and $ \vbfx' $ was sampled by James coincide with the output $ \vbfy' $ when $ \vbfx' $ was transmitted by Alice and $ \vbfx $ was sampled by James. 
If $ \vbfs' $ violates James's power constraint, then Bob immediately knows that the output is not $ \vbfy' $, $ \vbfx $ is the genuine codeword and  $ \vbfx' $ is a spoofing codeword. 
Hence, to ensure that Bob is fooled  by $ \vbfx $ and $ \vbfx' $, it had better be the case that $ \vbfs' $ satisfies his power constraint as well.

The probability $\p{}(\vbfx'\ne\vbfx, \|\vbfs\|_2\le\sqrt{nN}, \|\vbfs'\|_2\le\sqrt{nN})$ can be bounded as follows. 
\begin{align}
\p{}(\vbfx'\ne\vbfx, \|\vbfs\|_2\le\sqrt{nN}, \|\vbfs'\|_2\le\sqrt{nN}) =& 
\p{}(\|\vbfs\|_2\le\sqrt{nN}, \|\vbfs'\|_2\le\sqrt{nN}) - \p{}(\|\vbfs\|_2\le\sqrt{nN}, \|\vbfs'\|_2\le\sqrt{nN},\vbfx'=\vbfx) \notag\\
\ge& \p{}(\|\vbfs\|_2\le\sqrt{nN}, \|\vbfs'\|_2\le\sqrt{nN}) - \p{}(\vbfx'=\vbfx). \notag
\end{align}
Apparently, $ \p{}(\vbfx'=\vbfx) = 2^{-nR}\to0 $. 
It now remains to lower bound the first term.
\begin{align}
&\p{}(\|\vbfs\|_2\le\sqrt{nN}, \|\vbfs'\|_2\le\sqrt{nN})  \notag \\
=& \p{}\left( \left\|\frac{1}{2}(\vbfx' - \vbfz)\right\|_2\le\sqrt{nN},\left\|\frac{1}{2}(\vbfx - \vbfz')\right\|_2\le\sqrt{nN} \right) \notag \\
=& \p{}( \|\vbfx' - \vbfx - \vbfsz\|_2\le2\sqrt{nN},\|\vbfx - \vbfx' - \vbfsz\|_2\le2\sqrt{nN} ) \notag \\
=& \p{} ( \|\vbfx - \vbfx'\|_2^2 + \|\vbfsz\|_2^2 + 2\langle\vbfx - \vbfx',\vbfsz \rangle\le4nN, \|\vbfx - \vbfx'\|_2^2 + \|\vbfsz\|_2^2 - 2\langle\vbfx - \vbfx',\vbfsz \rangle\le4nN) \notag \\
\ge& \p{}(\|\vbfx - \vbfx'\|_2^2\le2nP(1+\delta_1),\|\vbfsz\|_2^2\le n\sigma^2(1+\delta_2),|\langle \vbfx - \vbfx',\vbfsz\rangle|\le n\delta_3 ) \label{eqn:set_delta123} \\
\ge& 1 - \p{}(\|\vbfx - \vbfx'\|_2^2>2nP(1+\delta_1)) - \p{}(\|\vbfsz\|_2^2> n\sigma^2(1+\delta_2)) - \p{}(|\langle \vbfx - \vbfx',\vbfsz\rangle|> n\delta_3), \label{eqn:symm_threeterms}
\end{align}
where in Eqn.~\eqref{eqn:set_delta123} we assume $ 2P+\sigma^2 = 4N - \varepsilon<4N $ for some constant $ \varepsilon>0 $ and we set $ \delta_1 \coloneq \frac{\varepsilon}{6P}, \delta_2\coloneq \frac{\varepsilon}{3\sigma^2},\delta_3 \coloneq \varepsilon/6 $. 
The first term in Eqn.~\eqref{eqn:symm_threeterms} can be bounded using Markov's inequality.
Specifically,
\begin{align}
\e{}(\|\vbfx - \vbfx'\|_2^2) =& \e{}(\|\vbfx\|_2^2) + \e{}(\|\vbfx'\|_2^2) - 2\e{}(\langle \vbfx,\vbfx'\rangle) \notag \\
=& \e{}(\|\vbfx\|_2^2) + \e{}(\|\vbfx'\|_2^2) - 2\sum_{i = 1}^n\e{}(\vbfx_i\vbfx'_i) \notag \\
=&\e{}(\|\vbfx\|_2^2) + \e{}(\|\vbfx'\|_2^2) - 2\sum_{i = 1}^n\e{}(\vbfx_i)\e{}(\vbfx'_i) \label{eq:iid1} \\
=&\e{}(\|\vbfx\|_2^2) + \e{}(\|\vbfx'\|_2^2) - 2\sum_{i = 1}^n(\e{}(\vbfx_i))^2 \label{eq:iid2} \\
\le& 2nP. \label{eqn:used_next_part}
\end{align}
where Equation~\eqref{eq:iid1} and Equation~\eqref{eq:iid2} follow since $\vbfx'$ and $\vbfx$ are i.i.d. By Markov's inequality, we have
\[\p{}(\|\vbfx - \vbfx'\|_2^2>2nP(1+\delta_1)) \le \frac{2nP}{2nP(1+\delta_1)} = \frac{1}{1+\delta_1}.\]
The second term of Eqn.~\eqref{eqn:symm_threeterms} follows from  $ \chi^2 $ tail bound (Fact~\ref{fact:gaussian_norm}).
\begin{align}
\p{}(\|\vbfsz\|_2^2> n\sigma^2(1+\delta_2)) \le& \exp(-\delta_2^2n/4). \notag
\end{align}
Since $ \langle \vbfx - \vbfx',\vbfsz\rangle\sim\cN(0,\|\vbfx-\vbfx'\|_2^2\sigma^2\bfI_n) $, by Gaussian tal bound (Fact~\ref{fact:gaussian_tail}),
\begin{align*}
\p{}(|\langle \vbfx - \vbfx',\vbfsz\rangle|> n\delta_3)\le&2\exp\left(-\frac{(n\delta_3)^2}{2\|\vbfx - \vbfx'\|_2^2\sigma^2}\right)  \\
\le& 2\exp\left(-\frac{n^2\delta_3^2}{4nP\sigma^2}\right) \\
=& 2\exp\left(-\frac{n\delta_3^2}{4P\sigma^2}\right). 
\end{align*}
Finally, we have
\begin{align}
P_e\ge& \frac{1}{2}(1 - \p{}(\|\vbfx - \vbfx'\|_2^2>2nP(1+\delta_1)) - \p{}(\|\vbfsz\|_2^2> n\sigma^2(1+\delta_2)) - \p{}(|\langle \vbfx - \vbfx',\vbfsz\rangle|> n\delta_3) -  \p{}(\vbfx'=\vbfx)) \notag \\
\ge& \frac{1}{2}\left(1 - \frac{1}{1+\delta_1} - \exp(-\delta_2^2n/4) - 2\exp\left(-\frac{n\delta_3^2}{4P\sigma^2}\right) - 2^{-nR} \right) \notag \\
=& \frac{1}{2}\left( \frac{\varepsilon/6P}{1 + \varepsilon/6P} - \exp\left(-\frac{n\varepsilon^2}{36\sigma^4}\right) - 2\exp\left(-\frac{n\varepsilon^2}{144P\sigma^2}\right) - 2^{-nR} \right) \notag \\
\to& \frac{\varepsilon/6P}{2(1 + \varepsilon/6P)}, \notag
\end{align}
which is bounded away from zero.
Thus no positive rate is achievable when $\frac{\sigma^2}{P}<4\frac{N}{P}-2$.

\rev{\paragraph{Improved $\vbfz$-aware symmetrization}
Finally, we modify the previous $\vbfz$-aware symmetrization by optimizing the coefficients in front of $ \vbfx' $ and $ \vbfz $ in the design of $ \vbfs $ (Eqn.~\eqref{eqn:z_aware_s_def}). 

Let
\begin{align}
 \vbfs = \alpha \vbfz + \beta \vbfx' +  \vbfg , \label{eqn:improved_z_aware_s_def}
\end{align}
where $\alpha<0,\beta>0$ are to be determined momentarily, $ \vbfg\sim\cN(0,\gamma^2\bfI_n) $ for some $\gamma>0$ to be determined later, and $ \vbfx' $ is a random codeword sampled uniformly from Alice's codebook.\footnote{Strictly speaking, as in Eqn.~\eqref{eqn:z_aware_s_def}, we should also multiply $\vbfs$ by a normalization factor $ \pmb\beta $. It ensures that $\vbfs$ satisfies James's power constraint with probability \emph{one}.
The way to handle it is precisely the same as in the previous part~\ref{par:z_aware_symm} and we omit the technical details in this part.}

Under the choice of $ \vbfs $ defined in Eqn.~\eqref{eqn:improved_z_aware_s_def}, Bob receives
\begin{align}
\vbfy =& \vbfx + \vbfs = \vbfx + \alpha\vbfz + \beta\vbfx' + \vbfg = (1+\alpha)\vbfx + \beta\vbfx' + \alpha\vbfsz + \vbfg. \label{eqn:improved_z_aware_received} 
\end{align}
We observe the following two points from Eqn.~\eqref{eqn:improved_z_aware_received}.
Firstly, to ``symmetrize'' the channel from Alice to Bob, James had better set $ 1+\alpha = \beta $.
This ensures that Bob has no idea whether $ \vbfx $ or $ \vbfx' $ was transmitted even if he somehow magically knew the value of $ \alpha\vbfsz+\vbfg $.
Secondly, to save his power, James had better set $ \gamma = 0 $, that is, not add additional Gaussian noise in $ \vbfs $. 
Therefore we set $ \vbfs = \alpha\vbfz + (1+\alpha)\vbfx' $ where $ \alpha<0 $ and $ 1+\alpha>0 $, i.e.,  $ \alpha>-1 $.

We now evaluate $ \frac{1}{n}\e{}(\|\vbfs\|_2^2) $ and contrast it with James's power constraint $N$.
\begin{align}
\e{}(\|\vbfs\|_2^2) =& \e{}(\|\alpha\vbfz + (1+\alpha)\vbfx'\|_2^2) \notag \\
=& \e{}(\| \alpha\vbfx  + (1+\alpha)\vbfx' + \alpha\vbfsz\|_2^2) \notag \\
=& \alpha^2\e{}(\|\vbfx\|_2^2) + (1+\alpha)^2\e{}(\|\vbfx'\|_2^2) + \alpha^2\e{}(\|\vbfsz\|_2^2) + 2\alpha(1+\alpha)\e{}(\langle \vbfx,\vbfx'\rangle) + \alpha^2\e{}(\langle \vbfx,\vbfsz\rangle) + \alpha(1+\alpha)\e{}(\langle \vbfx',\vbfsz\rangle)\notag \\
\le&\alpha^2\cdot nP+(1+\alpha)^2\cdot nP+\alpha^2\cdot n\sigma^2 + 0 + 0\label{eqn:improved_z_aware_same_calc}  \\
=& n(2P+\sigma^2)\alpha^2 + 2nP\alpha + nP. \label{eqn:improved_z_aware_opt}
\end{align}
Eqn.~\eqref{eqn:improved_z_aware_same_calc} follows from the same calculation as in Eqn.~\eqref{eqn:used_next_part}.
Minimizing Eqn.~\eqref{eqn:improved_z_aware_opt} over $\alpha\in(-1,0) $ (so as to minimize the amount of power James spent), we obtain the minimizer
\begin{align}
\alpha_* = -\frac{P}{2P+\sigma^2},\;\beta_* = 1+\alpha_* = \frac{P+\sigma^2}{2P+\sigma^2}. \notag
\end{align}
The above calculation can be directly substituted into the previous part (\ref{par:z_aware_symm}).
This implies
that $ C_{\mathrm{myop}} = 0 $ as long as the power James spent in transmitting $ \vbfs $ defined in Eqn.~\eqref{eqn:z_aware_s_def} is at most $\sqrt{nN}$.
That is, the RHS of Eqn.~\eqref{eqn:improved_z_aware_opt} evaluated at $ \alpha = \alpha_* $ is at most $N$: $ (2P+\sigma)^2\alpha_*^2 + 2P\alpha_*+P\le N $.
This reduces to the condition $ \frac{\sigma^2}{P}\le\frac{1}{1-N/P} - 2 $,
as promised in Lemma~\ref{lemma:converse_symmetrization}.
}
\end{proof}

\begin{remark}
The argument above generalizes the Plotkin bound (via the Cauchy--Schwarz inequality -- see, e.g., Li et al.~\cite{tongxin-causal-2018}) to scenarios with  additional randomness in $\vbfsz$.
\end{remark}

\section{Myopic list-decoding}\label{sec:myopic_list_decoding_details}

We now describe our coding scheme and prove that it achieves the rate in \rev{Theorem~\ref{thm:myopic_listdecoding_summary} (restated in Theorem~\ref{thm:myopic_listdecoding_summary})}. In Sec.~\ref{sec:myopic_list_decoding_scheme}, we formally describe the scheme. The proof of Theorem~\ref{thm:myopic_listdecoding_summary} proceeds by analyzing various error events which are formally defined in Sec.~\ref{sec:myopic_list_decoding_errorevents}. The proof is outlined in Sec.~\ref{sec:myop_ld_p_e}, and the probabilities of the various error events are analyzed in the following subsections.

\subsection{Coding scheme}\label{sec:myopic_list_decoding_scheme}
\noindent\textbf{Codebook construction.} We use random spherical codes. Before the communication, Alice samples $2^{n(R+\Rkey)}$ codewords $\{\vx(m,k): m\in[ 2^{nR}], k\in[ 2^{n\Rkey}]\}$ independently and uniformly at random from the sphere $\cS^{n-1}\left(0,\sqrt{nP}\right)$. Once sampled, the codebook is fixed and revealed to every party: Alice, Bob and James. Notice that all codewords satisfy Alice's power constraint $\|\vx\|_2\le\sqrt{nP}$. We define $\cC^{(k)}\coloneq \{\vx(m,k): m\in[ 2^{nR}]\}$ to be the $k$th codebook. We also distinguish the message rate $R$ from the codebook rate $\Rcode\coloneq R+\Rkey$.

\noindent\textbf{Encoder.} Let $k\in [2^{n\Rkey}]$ be the realization of the secret key shared by Alice and Bob. Alice sends $\vx(m,k)$ if she wants to transmit message $m$ to Bob.

\noindent\textbf{Decoder.} Bob uses a minimum distance decoder. Having received $\vy$, he outputs $\widehat m$ such that the corresponding codeword $\vx(\widehat{m},k)$ is the nearest (in Euclidean distance) one in $\cC^{(k)}$ to his observation, i.e., 
\[\widehat m=\argmin{ m'\in\{0,1\}^{nR}}\|\vx( m',k)-\vy\|_2.\]

\rev{\begin{remark}
We emphasize that the above coding scheme is designed for the original \emph{unique} decoding problem.
To approach the proof of unique decodability, we have to go through a novel notion of list decoding referred to as myopic list decoding\footnote{See Section~\ref{sec:myop_ld_p_e} below for what it means for a code to be non-myopic-list-decodable.} as the title of this section suggests.
However, myopic list decoding appears only as a proof technique and neither Bob nor James really performs a step of myopic list decoding. 
Note that the above decoder for unique decoding will not be used until Section~\ref{sec:achievability_suffmyopic}. 
\end{remark}}


For the rest of this section, we fix two quantities: $\varepsilon $ is a small positive constant independent of $n$, and $\delta$ is a parameter that decays as $\Theta((\log n)/n)$.
\rev{The latter parameter $ \delta $ is used in Eqn.~\eqref{eq:defn_str_z} to parameterize the thickness of each strip $ \strip^{n-1}(\vzq,i) $. }

\subsection{The strips and the oracle-given set (OGS)}\label{sec:myopiclistdecoding_strip_ogs}

Let $L=3n^2$.
To simplify the proof, we prove the achievability part under a more powerful adversary who has access to an oracle in addition to $\vbfz$. The oracle reveals a random subset of $2^{\varepsilon n}$ codewords
that contains the transmitted codeword and others that are all at approximately the same distance to $\vbfz$. We call it an oracle-given set, denoted $\ogs(\vbfz,\vbfx)$. Conditioned on James's knowledge, the transmitted codeword is independent of all codewords outside the oracle-given set. We now describe the rule that assigns a pair $(\vbfx,\vbfz)$ to an OGS.

Choose any optimal covering $\cZ$ of $\sh^{n}(0,\sqrt{n(P+\sigma^2)(1\pm \varepsilon)})$ such that $\min_{\vz'\in\cZ}\Vert \vz-\vz'\Vert_2\leq \sqrt{n\delta_{\cZ}}$ for all $\vz$ in the shell. 
The size of such a covering can be bounded as follows.
\begin{equation}
    |\cZ|\le\left(\frac{\vol(\cB^n(0,\sqrt{n(P+\sigma^2)(1+\varepsilon)}+\sqrt{n\delta_\cZ}))}{\vol(\cB^n(0,\sqrt{n\delta_\cZ}))}\right)^{1+o(1)}=\left(\frac{\sqrt{(P+\sigma^2)(1+\varepsilon)}+\sqrt{\delta_\cZ}}{\sqrt{\delta_\cZ}}\right)^{n(1+o(1))}\eqcolon c_{\varepsilon,\delta_\cZ}^n.
    \label{eq:covering_size_bound_z}
\end{equation}
Given $\vz$, let $\vzq\coloneq \arg\min_{\vz'\in\cZ}\Vert \vz-\vz'\Vert_2$ denote the closest point to $\vz$ in $\cZ$ (a.k.a. the \emph{quantization} of $\vz$). For each $\vzq\in\cZ$, and $i\in \{ -\varepsilon /\delta+1,\ldots,\varepsilon /\delta \}$, define the $ i $-th strip
\begin{equation}
 \strip^{n-1}(\vzq,i) \coloneq  \cS^{n-1}(0,\sqrt{nP}) \cap \sh^n(\vzq,\sqrt{n\sigma^2(1+(i-1)\delta)},\sqrt{n\sigma^2(1+{i}\delta)})
 \label{eq:defn_str_z}
\end{equation}
to be the set of all points on the coding sphere at a distance of at least $\sqrt{n\sigma^2(1+(i-1)\delta)}$ but at most $\sqrt{n\sigma^2(1+i\delta)}$ away from $\vzq$. 
\rev{It is not hard to see that the union of the strips is the whole power sphere:} $\bigcup_{\vzq}\bigcup_i\strip^{n-1}(\vzq,i)  = \cS^{n-1}(0,\sqrt{nP})$.\footnote{\rev{Indeed, to see this, note that the quantization $ \vzq\in\cZ $ essentially ranges over all directions. By making the quantization level $ \delta_\cZ $ sufficiently fine compared to the thickness parameters $\varepsilon$ and $ \delta $ of the strips, one is able to cover the whole sphere using strips.}}
Let $\cC \coloneqq \{ \vx{(m,k)}:m\in[2^{nR}],k\in[2^{n\Rkey}]\}$ denote the codebook. Define
\begin{equation}
 \mstr(\vzq,i)\coloneq \{ (m,k): \vx{(m,k)} \in \strip^{n-1}(\vzq,i) \}
 \label{eq:defn_Ms}
\end{equation}
to be the set of indices of the codewords that lie in $\strip^{n-1}(\vzq,i)$. We partition this set of indices into blocks of size $2^{n\varepsilon }$ each (except perhaps the last block) in the lexicographic order of $(m,k)$. Let $\{\ogs^{(j)}(\vzq,i)\}_{j=1}^\ell$ denote the partition, where $\ell\coloneq\lceil|\mstr(\vzq,i)|/2^{n\varepsilon}\rceil$. Each of these blocks constitutes an oracle-given set. If $(m,k)$ corresponding to the transmitted codeword $\vx(m,k)$ lies in the $\mu$th block $\ogs^{(\mu)}(\vzq,\lambda)$ of the partition of the $\lambda$th strip $\mstr(\vzq,\lambda)$ for some $\lambda\in\{-\varepsilon/\delta+1,\cdots,\varepsilon/\delta\}$ and $\mu\in[\ell]$, then the oracle reveals $\ogs(\vzq,\vx)\coloneq\ogs^{(\mu)}(\vzq,\lambda)$ to James.

\begin{remark}
It is important to note that all sets defined above (strips, OGSs, etc.) are designed a priori, before communication takes place. 
\end{remark}

\subsection{Error events}\label{sec:myopic_list_decoding_errorevents}
Define 
\begin{align*}
    \cL^{(k)}(\vx(m),\vs)\coloneq&\{w\in[2^{nR}]:\vx(w,k)\in\cB^n(\vx(m,k)+\vs,\sqrt{nN})\cap\cC^{(k)}\}\\
    =&\{w\in[2^{nR}]\colon \|\vx(w,k)-\vx(m,k)-\vs\|_2\le\sqrt{nN}\}
\end{align*}
for $m\in[2^{nR}],k\in[2^{n\Rkey}],\vs\in\cB^n(0,\sqrt{nN})$. Recall that \rev{by the proof sketch in Section~\ref{sec:myop_list_dec_sketch},} to prove the existence of a \rev{myopic} list-decodable code, \rev{we want to show that with high probability over the randomness in codebook selection, the key shared by encoder-decoder and the channel from Alice to James, only a vanishing fraction of codewords in the OGS (say, $ 2^{n\varepsilon/4} $ out of $ 2^{n\varepsilon} $ codewords in the OGS) have list size larger than $L$ under \emph{some} attack vector by James.}
Formally, we want to show 
\begin{definition}[\rev{Myopic list-decodability}]\label{def:myop_list_dec}
\rev{A code ensemble $ \left\{\cC^{(\bfk)}\right\} $ with common randomness $ \bfk $ shared by Alice and Bob is said to be myopic list-decodable if }
\[
 \p{}\left(\exists\vs\in\cB^n(0,\sqrt{nN}),\;|\{(m,\bfk)\in\ogs(\vbfz,\vbfx):|\cL^{(\bfk)}(\vbfx(m),\vs)|>L\}|>2^{n(\varepsilon-h(\varepsilon,\tau,\delta_\cS,\delta_\cZ)/2)}\right) =o(1),
\]
where $0<h(\varepsilon,\tau,\delta_\cS,\delta_\cZ)<2\varepsilon$ is a vanishing function in each of its variables. In particular, we can take $h(\varepsilon)=\frac{3}{2}\varepsilon$ by setting  $\tau$, $\delta_\cS$ and $\delta_\cZ$ to be suitable functions of $\varepsilon$.
\end{definition}
\rev{Here and throughout the rest of this paper, we often write $ \vbfx $  as a shorthand for $ \vbfx(\bfm,\bfk) $ where $\bfm$ is a uniform message, $ \bfk $ is a random shared key and for any given pair $ (\bfm,\bfk) $ the corresponding codeword $ \vbfx(\bfm,\bfk) $ is chosen uniformly from the power sphere $ \cS^{n-1}(0,\sqrt{nP}) $.}
To analyze the above probability, we define a number of error events. 

Choose any optimal covering $\cS$ of $\cB^n(0,\sqrt{nN})$ such that $\min_{\vs'\in\cS}\Vert \vs-\vs'\Vert_2\leq \sqrt{n\delta_{\cS}}$ for all $\vs$ in the ball. Given $\vs$, let $\vsq\coloneq \arg\min_{\vs'\in\cS}\Vert \vs-\vs'\Vert_2$ denote the closest point to $\vs$ in $\cS$ (a.k.a. the quantization of $\vs$). By similar calculation to Equation~\eqref{eq:covering_size_bound_z}, we have
\begin{equation}
    |\cS|\le\left(\frac{\sqrt{N}+\sqrt{\delta_\cS}}{\sqrt{\delta_\cS}}\right)^{n(1+o(1))}\eqcolon c_{\delta_\cS}^n.
    \label{eq:covering_size_bound_s}
\end{equation}

We say that list-decoding fails if any of the following events occurs. We will  ultimately show that the probability of failure is negligible. The error events that we analyze are listed below:
\begin{itemize}
 \item[$\eatyp$:] James's observation $\vbfz$ behaves atypically, or equivalently, the noise $\vbfsz$ to James behaves atypically.
\begin{align}
    \eatypi\coloneq& \{ \|\vbfsz\|_2\notin\sqrt{n\sigma^2(1\pm\varepsilon)} \}.
  \label{eq:defn_eatypi}\\
  \eatypii\coloneq& \{ |\cos(\angle_{\vbfx,\vbfsz})|\ge\varepsilon \}.
  \label{eq:defn_eatypii}\\
  \eatypiii\coloneq& \{ \|\vbfz\|_2\notin \sqrt{n(P+\sigma^2)(1\pm\varepsilon)} \}.
  \label{eq:defn_eatypiii}
\end{align}
 Hence the error event $\eatyp$ is the union of the above three events.
 \begin{equation}
  \eatyp\coloneq \eatypi\cup \eatypii \cup\eatypiii.
  \label{eq:defn_eatyp}
 \end{equation}
 \item[$\estr$:] One of the strips $\{\strip^{n-1}(\vzq,i)\}_i$ contains fewer than $2^{3\varepsilon  n}$ codewords.
 \begin{align}
     \estr(i)\coloneq& \{|\mstr(\vzq,i)|<2^{3\varepsilon  n}\}.
  \label{eq:defn_estr_i}\\
  \estr\coloneq& \bigcup_i\estr(i).
  \label{eq:defn_estr}
 \end{align}
 \rev{Note that for a fixed $ \vzq $, the randomness in error events defined in Eqn.~\eqref{eq:defn_estr_i} and \eqref{eq:defn_estr} comes from codebook construction, message selection and common randomness. That is, the number of message-key pairs in each $\mstr(\vzq,i)$ is a random variable.}
 \item[$\eorcl$:]\label{item:defn_eorcl} Since the number of messages need not be an interger multiple of $ 2^{n\varepsilon} $, the last OGS may be substantially smaller than the others, and hence could have a higher probability of error. But the probability that the transmitted codeword happens to fall into the last set is small. Call $\eorcl$ the event that the message corresponding to the transmitted codeword belongs to the last block $\ogs^{(\ell)}(\vzq,\lambda)$ of the partition of $\mstr(\vzq,\lambda)$. 
 \begin{equation}
     \eorcl\coloneq\{\rev{\pmb\mu}=\ell\}.
     \label{eq:defn_eorcl}
 \end{equation}
 \rev{In the above definition (Eqn.~\eqref{eq:defn_eorcl}), the randomness in the random variable $ \pmb\mu $ results from the uniform selection of message-key pair $ (\bfm,\bfk) $.}
 \item[$\er$:] Assume that none of $\eatyp,\estr,\eorcl$ occurs. For any $(m,k)\in\ogs^{(j)}(\vzq,i)$ and $\vsq\in\cS$, let $\sqrt{n\cdot r(m,\vsq)}$ be the radius of the cap $\cB^n(\vx(m,k)+\vsq,\sqrt{nN}+\sqrt{n\delta_\cS})\cap\cS^{n-1}(0,\sqrt{nP})$ (which is the list decoding region) of $\vx(m,k)$ under $\vsq$. We will show that $\bfr\coloneq\bfr(m,\vsq)$   concentrates around a certain typical value $\ropt\coloneq\ropt(\vsq)=\e{}(\bfr)$, which is substantially smaller than $ \sqrt{nN} $. Let $ \er(m,\vsq) $ denote the event that  $\bfr(m,\vsq)$  is atypical, \emph{i.e.}, $ \bfr $ is significantly larger than its expectation $ \ropt $.  
 \begin{align}
     \er(m,\vsq)\coloneq&\{\bfr(m,\vsq)>\ropt(\vsq)(1+f_{11}(\varepsilon,\delta_{\cS}))\},
     \label{eqn:def_er}
 \end{align}
 for some small  function $f_{11}(\varepsilon,\delta_\cS)$ to be determined later.
 
 \item[$\esq$:] Assume that none of $\eatyp,\estr,\eorcl$ occurs. Define
 \begin{align}
     \psi(\vzq,i,j,\vsq)\coloneq&|\{(m,k)\in\ogs^{\rev{(j)}}(\vzq,\rev{i}):\bfr(m,\vsq)>\ropt(\vsq)(1+f_{11}(\varepsilon,\delta_{\cS}))\}|\notag\\
     =&\sum_{(m,k)\in\ogs^{(j)}(\vzq,i)}\one_{\{\bfr(m,\vsq)>\ropt(\vsq)(1+f_{11}(\varepsilon,\delta_{\cS}))\}}\notag\\
     =&\sum_{(m,k)\in\ogs^{(j)}(\vzq,i)}\one_{\er(m,\vsq)}\label{eqn:def_psi}
 \end{align}
 to be the number of messages in the OGS whose encodings have  atypical list-decoding radii. 
 Let $ \esq(\vzq,i,j) $ be the event  that there are  more than $ n^2 $ such messages in the OGS.
 \begin{align}
     \esq(\vzq,i,j)\coloneq&\{\psi(\vzq,i,j,\vsq)> n^2\}. \label{eqn:def_esq}
 \end{align}
 
 \item[$\eerr$:]  Given that none of $\eatyp,\estr,\eorcl$ occurs, there exists an attack vector $\vsq\in\cS$ that results in a list-size greater than $L$ for at least one codeword in the oracle-given set. For each $k,\vzq,i,j,\vsq,$ define
   \begin{align}
   \chi(\vzq,i,j,\vsq)
   \coloneqq&|\{(m,k)\in\ogs^{\rev{(j)}}(\vzq,\rev{i}):|\cL^{(k)}(\vx(m),\vsq)|>L\}|\notag\\
   =& \sum_{(m,k)\in \ogs^{(j)}(\vzq,i)}\one_{\{ |\cL^{(k)}(\vx(m),\vsq)|>L \}}
       \label{eq:defn_no_of_cw_with_large_list}
   \end{align}
  to be the number of codewords in an oracle-given set that result in a large list-size when perturbed by $\vsq$. Then, we define
  \begin{align}
      \eerr(\vzq,i,j,\vsq)\coloneq& \left\{\left|\left\{(m,k)\in\ogs^{\rev{(j)}}(\vzq,\rev{i})\colon  |\cL^{(k)}(\vx(m),\vsq)|>L\right\}\right|>n^2+1 \right\} \notag\\
      =&\{ \chi(\vzq,i,j,\vsq)>n^2+1  \} 
      \label{eq:defn_too_many_cw_with_large_list}
  \end{align}
  to be the event that there exists codewords that can be perturbed by $\vsq$ to give a large list-size. The error event $\eerr$  is defined as
  \begin{equation}
   \eerr\coloneq \bigcup_{\vzq}\bigcup_{i}\bigcup_j\bigcup_{\vsq}\eerr(\vzq,i,j,\vsq).
   \label{eq:defn_eerr}
  \end{equation}
\end{itemize}

\subsection{The probability of myopic list-decoding error -- proof of Theorem~\ref{thm:myopic_listdecoding_summary}}\label{sec:myop_ld_p_e}
We now prove Theorem~\ref{thm:myopic_listdecoding_summary} by upper bounding the probability that an error occurs in myopic list-decoding. 
Let \[\cL^{(k)}(\vx(m),\vsq)\coloneq\{w\in[2^{nR}]:\vx(w,k)\in\cB^n(\vx(m,k)+\vsq,\sqrt{nN}+\sqrt{n\delta_\cS})\cap\cC^{(k)}\}.\]
We can decompose the failure probability in the following manner using \rev{Fact}~\ref{fact:error_event_decompo_lem}.
\begin{equation}
\begin{aligned}
    &\p{}\left(\exists \vs\in\cB^n(0,\sqrt{nN}),\; |\{(m,\bfk)\in\ogs(\vbfz,\vbfx):|\cL^{(\bfk)}(\vbfx(m),\vs)|>L\}|>2^{n(\varepsilon-h(\varepsilon,\tau,\delta_\cS,\delta_\cZ)/2)}\right)\\
    \le&\p{}\left(\exists \vsq\in\cS,\; |\{(m,\bfk)\in\ogs(\vbfzq,\vbfx):|\cL^{(\bfk)}(\vbfx(m),\vsq)|>L\}|>2^{n(\varepsilon-h(\varepsilon,\tau,\delta_\cS,\delta_\cZ)/2)}\right)\\
    \le&\p{}\left(\eatyp\cup\estr\cup\eorcl\cup\eerr\right)\\
    =&\p{}\left(\eatyp\cup\bigcup_i\estr(i)\cup\eorcl\cup\bigcup_{\vzq}\bigcup_{i}\bigcup_j\bigcup_{\vsq}\eerr(\vzq,i,j,\vsq)\right)\\
    \le&1-(1-\p{}(\eatyp))\\
    &\quad\cdot\left(1-\sum_i\p{}(\estr(i)|\eatyp^c)\right)\\
    &\quad\cdot(1-\p{}(\eorcl|\eatyp^c\cap\estr^c))\\
    &\quad\cdot\left(1-\sum_{\vzq}\sum_{i}\sum_j\sum_{\vsq}\p{}(\eerr(\vzq,i,j,\vsq)|\eatyp^c\cap\estr^c\cap\eorcl^c)\right).
\end{aligned}
\label{eq:decomp_pe_myop_ld}
\end{equation}

It is therefore sufficient to show that each of the error terms is vanishing in $n$.

The analysis of $\eatyp$, $\estr$ and $\eorcl$ follows from somewhat standard concentration inequalities which are formally justified in Section~\ref{sec:e_z}, Section~\ref{sec:e_str} and Section~\ref{sec:e_orcl}, respectively. Notice that in the analysis of $\estr$, James is said to be \emph{sufficiently myopic} if, given his observation $\vbfz$, his uncertainty set, i.e., any strip $\strip^{n-1}(\vbfz,i)$, contains at least exponentially many codewords. This is guaranteed if $R+\Rkey>\frac{1}{2}\log\left(1+\frac{P}{\sigma^2}\right)$.

Much of the complication of our work is devoted to the analysis of $\eerr(\vzq,i,j,\vsq)$. We further factorize it into sub-events and treat them separately. Define $\cE\coloneq\eatyp\cup\estr\cup\eorcl$. 
Fix $\vzq$, $i$, $j$ and $\vsq$. We are able to show that 
\begin{align*}
    \p{}(\eerr(\vzq,i,j,\vsq)|\cE^c)\le&2^{-\Omega(n^3)},
\end{align*}
which allows us to take a union bound over exponentially many objects. To this end, we need to further decompose the error event $\eerr(\vzq,i,j,\vsq)$ in a careful manner.

Since the codewords all lie on a sphere, the effective decoding region is equal to $ \cB(\vbfx(m,\bfk)+\vsq,\sqrt{nN}+\sqrt{n\delta_{\cS}})\cap \cS^{n-1}(0,\sqrt{nP}) $. 
We will first show in Lemma~\ref{lem:er} that for most codewords in the oracle-given set, the area of the effective decoding region under any fixed $ \vsq $ is not too large.{ The list-sizes for the remaining codewords can be controlled using two-step list-decoding argument and grid argument in Lemma~\ref{lemma:listfailure_two_types_err}.}  Specifically, 
\[
\p{}(\eerr(\vzq,i,j,\vsq)|\cE^c) \leq \p{}(\esq|\cE^{c}) + \p{}(\eerr(\vzq,i,j,\vsq)|\cE^c\cap \esq^c).
\]
We first compute the explicit value of the typical volume of list-decoding region.
Using Chernoff's bound,  we show in Lemma~\ref{lem:esq} that $ \p{}(\esq|\cE^{c})\leq 2^{-\Omega(n^3)} $. The second term can be shown to also be $ 2^{-\Omega(n^3)} $. This implies that for any given attack vector $ \vsq $, the probability that $ \vsq $ can force a large list-size for any  codewords is superexponentially small. To complete the proof, we take a union bound over $ \vzq $, the strips, the OGSs and $\vsq$.

The whole bounding procedure (including myopic list-decoding in this section and unique decoding in Section~\ref{sec:achievability_suffmyopic}) is depicted in Figure~\ref{fig:flowchart}.
\begin{figure} 
    \centering
    \includegraphics[height = 0.95\textheight]{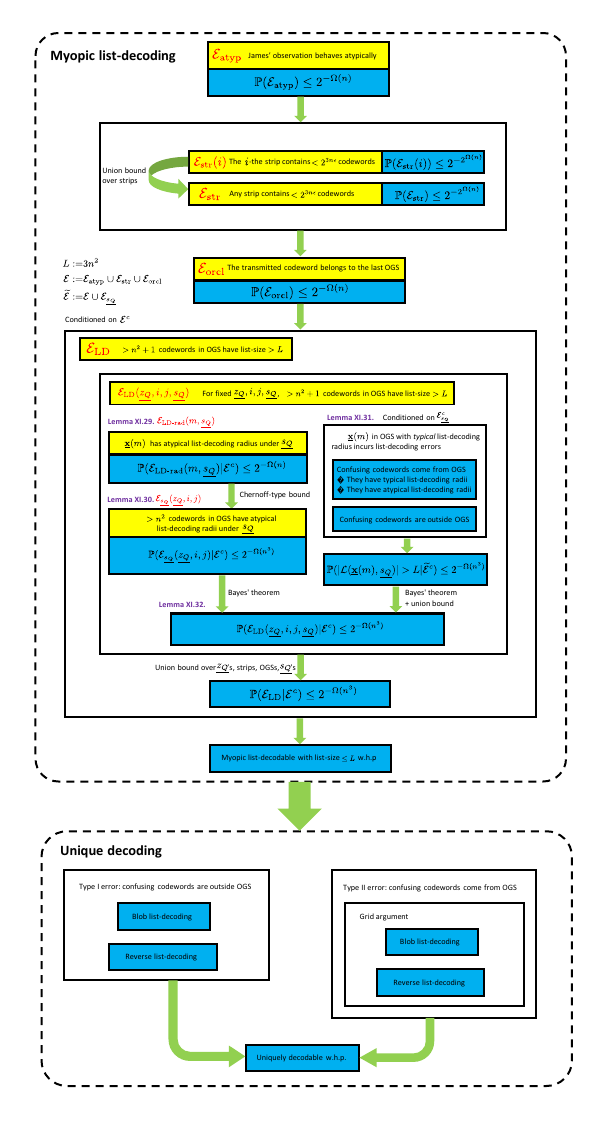}
    \caption{A flowchart describing the procedure of bounding error probability. Some notation is simplified for ease of drawing.}
    \label{fig:flowchart}
\end{figure}

\subsection{Event $\eatyp$: Analysis of atypical behaviour of James's observation}\label{sec:e_z}
\rev{Though, as already mentioned in Section~\ref{sec:myopic_list_decoding_errorevents}, we still  use the shorthand notation $ \vbfx $ to denote $ \vbfx(\bfm,\bfk) $, results in this subsection in fact hold regardless of the distribution of $ \vbfx $ and in particular the readers, if they want, can take $ \vbfx $ to be any fixed vector $\vx$ on $ \cS^{n-1}(0,\sqrt{nP}) $. }

From Fact~\ref{fact:gaussian_norm}, we know the probability that $\vbfsz,\vbfx$ are not jointly typical vanishes as $n\to\infty$. Specifically, the AWGN to James is independent of everything else, hence has norm concentrating around $\sqrt{n\sigma^2}$ and is appproximately orthogonal to $\vbfx$.
\[
 \p{}(\eatypi)=\p{}(\|\vbfsz\|_2\notin\sqrt{n\sigma^2(1\pm\varepsilon)})\le2\exp(-\varepsilon^2n/4)\eqcolon2^{-f_1(\varepsilon)n}.
\]
Since $ \vbfs_z $ is AWGN independent of $ \vbfx $, the average dot product between the two vectors is zero. We can further bound the probability that the cosine of the angle between the two exceeds $ \varepsilon $.
\begin{align}
    \p{}(\eatypii)=&\p{}(|\cos(\angle_{\vbfx,\vbfsz})|\ge\varepsilon)\notag\\
    =&\p{}\left(\left|\frac{\langle\vbfx,\vbfsz\rangle}{\|\vbfx\|_2\|\vbfsz\|_2}\right|\ge\varepsilon\right)\notag\\
    =&\p{}\left(\left|\frac{\langle\underline{e_1},\vbfsz\rangle}{\|\vbfsz\|_2}\right|\ge\varepsilon\right)\label{eq:assumption}\\
    =&\p{}(|\vbfsz_1|\ge\varepsilon\|\vbfsz\|_2)\notag\\
    \le&\p{}(|\vbfsz_1|\ge\varepsilon\|\vbfsz\|_2,\;\|\vbfsz\|_2\in\sqrt{n\sigma^2(1\pm\varepsilon)})+\p{}(\|\vbfsz\|_2\notin\sqrt{n\sigma^2(1\pm\varepsilon)})\notag\\
    \le&\p{}(|\vbfsz_1|\ge\varepsilon\sqrt{n\sigma^2(1-\varepsilon)})+\p{}(\|\vbfsz\|_2\notin\sqrt{n\sigma^2(1\pm\varepsilon)})\notag\\
    \le&2\exp(-\varepsilon^2(1-\varepsilon)n/2)+2^{-f_1(\varepsilon)n}\notag\\
    \eqcolon&2^{-f_2(\varepsilon)n},\notag
\end{align}
where in Equation~\eqref{eq:assumption}, without loss of generality, we assume $\vbfx/\|\vbfx\|_2=\underline{e_1}$, where $\underline{e_1}=(1,0,\cdots,0)^T$ is the unit vector along the first dimension.
Notice that 
\[\|\vbfz\|_2^2=\|\vbfx+\vbfsz\|_2^2=\|\vbfx\|_2^2+\|\vbfsz\|_2^2+2\langle\vbfx,\vbfsz\rangle.\]
Choose the smallest $\varepsilon_1\coloneq\varepsilon_1(\varepsilon)$ that satisfies  $n(P+\sigma^2)(1\pm\varepsilon)\subset nP+n\sigma^2(1\pm\varepsilon_1)\pm2\sqrt{nP}\sqrt{n\sigma^2(1+\varepsilon_1)}\varepsilon_1$. We can also concentrate James's observation. In the following, the probability is computed with respect to the product distribution of $ (\vbfx,\vbfs_z) $ since they are independent.
\begin{align*}
    \p{}(\eatypiii)=&\p{}(\|\vbfz\|_2^2\notin n(P+\sigma^2)(1\pm\varepsilon))\\
    \leq&\p{}(\|\vbfz\|_2^2\notin nP+n\sigma^2(1\pm\varepsilon_1)\pm\sqrt{nP}\sqrt{n\sigma^2(1+\varepsilon_1)}\varepsilon_1)\\
    \le&\p{}(\|\vbfsz\|_2^2\notin n\sigma^2(1\pm\varepsilon_1))+\p{}(|\langle\vbfx,\vbfsz\rangle|\ge\sqrt{nP}\sqrt{n\sigma^2(1+\varepsilon_1)}\varepsilon_1)\\
    =&\p{}(\|\vbfsz\|_2^2\notin n\sigma^2(1\pm\varepsilon_1))+\p{}(|\langle\vbfx,\vbfsz\rangle|\ge\sqrt{nP}\sqrt{n\sigma^2(1+\varepsilon_1)}\varepsilon_1,\;\|\vbfsz\|_2^2\in n\sigma^2(1\pm\varepsilon_1))\\
    &+\p{}(|\langle\vbfx,\vbfsz\rangle|\ge\sqrt{nP}\sqrt{n\sigma^2(1+\varepsilon_1)}\varepsilon_1,\;\|\vbfsz\|_2^2\notin n\sigma^2(1\pm\varepsilon_1))\\
    \le&2\p{}(\|\vbfsz\|_2^2\notin n\sigma^2(1\pm\varepsilon_1))+\p{}(|\langle\vbfx,\vbfsz\rangle|\ge\|\vbfx\|_2\|\vbfsz\|_2\varepsilon_1)\\
    \le&2\cdot2^{-f_1(\varepsilon_1)n}+2^{-f_2(\varepsilon_1)n}\\
    \eqcolon&2^{-f_3(\varepsilon)n}.
\end{align*}
Therefore,
\[
 \p{}(\eatyp) \leq \p{}(\eatypi) + \p{}(\eatypii)+ \p{}(\eatypiii) \leq 2^{-f_1(\varepsilon)n}+2^{-f_2(\varepsilon)n}+2^{-f_3(\varepsilon)n}\eqcolon 2^{-nf_{\mathrm{atyp}}(\varepsilon)},
\]
where $f_{\mathrm{atyp}}(\varepsilon)$ is positive as long as $\varepsilon>0$ and $\lim_{\varepsilon\downarrow0}f_{\mathrm{atyp}}(\varepsilon)=0$.

\subsection{Event $\estr$: Number and distribution of codewords in strips: Exponential behaviour and quasi-uniformity}\label{sec:e_str}
The intersection of $\sh^{n}(\vbfzq,\sqrt{n\sigma^2(1\pm\varepsilon)})$ with $\cS^{n-1}({0},\sqrt{nP})$ forms a thick strip $\strip^{n-1}(\vbfzq)\coloneq\bigcup_i\strip^{n-1}(\vbfzq,i)$, i.e., the union of all thin strips, which we study next. 
For ease of incoming calculations, let us first translate the slacks in the distances from the strips to $\vbfzq$, i.e., $\varepsilon$ and $\delta$ as afore-defined, to slacks in the radii $\sqrt{n\rstr}$ of strips, i.e, $\rho$ and $\tau$, respectively. Recall that the set of strips is defined as follows
\begin{equation}
    \strip^{n-1}(\vbfzq,i)=\cS^{n-1}(0,\sqrt{nP})\cap\sh^n(\vbfzq,\sqrt{n\sigma^2(1+(i-1)\delta)},\sqrt{n\sigma^2(1+{i}\delta)}),\;\forall i\in\{-\varepsilon/\delta+1,\cdots,\varepsilon/\delta\}.
    \label{eq:strip_param_epsilon_delta}
\end{equation}
Now we write it in a slightly different form
\begin{equation}
    \strip^{n-1}(\vbfzq,i)=\C^{n-1}(\cdot,\sqrt{n\rstr(1+i\tau)},\sqrt{nP})\backslash\C^{n-1}(\cdot,\sqrt{n\rstr(1+(i-1)\tau)},\sqrt{nP}),\;\forall i\in\{-\rho/\tau+1,\cdots,\rho/\tau\}.
    \label{eq:strip_param_rho_tau}
\end{equation}
Note that $\rho/\tau=\varepsilon/\delta$.
Define $d_i\coloneq\sqrt{n\sigma^2(1+{i}\delta)}$ and $\rstri\coloneq\sqrt{n\rstr(1+i\tau)}$, for any $i\in\{-\varepsilon/\delta+1,\cdots,\varepsilon/\delta\}$. Then by Heron's formula,
\[\frac{1}{2}\|\vzq\|_2\rstri=\sqrt{s(s-d_i)(s-\|\vzq\|_2)(s-\sqrt{nP})},\]
where \[s=\frac{1}{2}(d_i+\|\vzq\|_2+\sqrt{nP}).\]
Solving the equation, we have
\[\rstri=\frac{2}{\|\vzq\|_2}\sqrt{s(s-d_i)(s-\|\vzq\|_2)(s-\sqrt{nP})}.\]
It follows that $\rho$ and $\tau$ only differ by a constant factor from $\varepsilon$ and $\delta$, respectively.

As mentioned, codewords are \emph{almost} uniformly distributed in the strip from James's point of view.  Given $\vbfz$, we now characterize the \emph{quasi-uniformity} in terms of $\tau$. Define quasi-uniformity factor
\begin{equation}
 \Delta(\tau) \coloneqq \sup_{\vz\in\sh^n(0,\sqrt{n(P+\sigma^2)(1\pm\varepsilon)})}\max_i\sup_{\vx^{(1)},\vx^{(2)}\in \strip^{n-1}(\vzq,i)} \frac{p_{\vbfx|\vbfz}(\vx^{(1)}|\vzq)}{p_{\vbfx|\vbfz}(\vx^{(2)}|\vzq)},
 \label{eq:defn_quasiuniformity}
\end{equation}
\rev{where the conditional density is determined by the joint law $ p_{\vbfx,\vbfz} $ where $ \vbfx\sim \unif(\cS^{n-1}(0,\sqrt{nP}))$, $ \vbfz = \vbfx + \vbfsz $ and $ \vbfsz\sim\cN(0,\sigma^2\bfI_n) $. 

We will show that under appropriate choice of parameters, the above ratio is small maximized over $ \vx^{(1)} ,\vx^{(2)} $ in $ \strip^{n-1}(\vzq,i) $.
By our random  code construction, each codeword indeed follows the distribution $ \vbfx(m,k)\sim \unif(\cS^{n-1}(0,\sqrt{nP})) $ for any $(m,k)$.
Hence the ratio remains small if $ \vx^{(1)} ,\vx^{(2)} $ are respectively replaced by $ \vx(m_1,k_2),\vx(m_2,k_2) $ for some $ (m_1,k_1),(m_2,k_2) $ in $ \mstr(\vzq,i) $.}

\begin{lemma}[Quasi-uniformity]
 For appropriate choices of the small constant $\rho$ and $\tau=\Theta((\log n)/n)$, conditioned on $\eatyp^c$, we have
 $
  \Delta(\tau) = \cO(\mathrm{poly}(n)).
 $
 \label{lemma:quasiuniformity}
\end{lemma}

\begin{proof}

As shown in Figure~\ref{fig:quasi_unif}, obviously, for fixed $\vz$ and $i$, the inner supremum is achieved by a point $\vx^-$ on the upper boundary (closer to $\vz$) of the strip and a point $\vx^+$ on the lower boundary (further from $\vz$) of the strip. Calculations in Appendix~\ref{sec:prf_quasiuniformity} show that
\[\sup_{\vx^{(1)},\vx^{(2)}\in \strip^{n-1}(\vz,i)} \frac{p_{\vbfx|\vbfz}(\vx^{(1)}|\vzq)}{p_{\vbfx|\vbfz}(\vx^{(2)}|\vzq)}=\exp\left(\frac{\|\vz\|_2+\sqrt{n\delta_\cZ}}{\sigma^2}\frac{2n\rstr\tau}{\sqrt{n(P-r_-)}+\sqrt{n(P-r_+)}}\right),\]
\rev{where $ r_-,r_+ $ are defined in Eqn.~\eqref{eqn:def_r_plus_minus}.}
Then, conditioned on $\eatyp^c$, the quasi-uniformity factor is upper bounded by
\begin{align}
    \Delta(\tau)\le&\sup_{\vz\in\sh^n(0,\sqrt{n(P+\sigma^2)(1\pm\varepsilon)})}\max_i\exp\left(\frac{\|\vz\|_2+\sqrt{n\delta_\cZ}}{\sigma^2}\frac{2n\rstr\tau}{\sqrt{n(P-r_-)}+\sqrt{n(P-r_+)}}\right)\notag\\
    \le&\exp\left(\frac{\sqrt{n(P+\sigma^2)(1+\varepsilon)}+\sqrt{n\delta_{\cZ}}}{\sigma^2}\frac{2n\rstr\tau}{\sqrt{n(P-\rstr(1-\tau))}+\sqrt{n(P-\rstr(1+\tau))}}\right)\notag\\
    \le&\exp\left( \frac{\sqrt{(P+\sigma^2)(1+\varepsilon)}+\sqrt{\delta_{\cZ}}}{\sigma^2} \frac{2n\tau\frac{P\sigma^2(1+\varepsilon)}{(P+\sigma^2)(1-\varepsilon)}}{\sqrt{P-\frac{P\sigma^2(1+\varepsilon)(1-\tau)}{(P+\sigma^2)(1-\varepsilon)}} + \sqrt{P-\frac{P\sigma^2(1+\varepsilon)(1+\tau)}{(P+\sigma^2)(1-\varepsilon)}} } \right),\label{eqn:substitute_rstr}
\end{align}
where $r_-=\rstr(1-\tau)$ and $r_+=\rstr(1+\tau)$. Eqn.~\eqref{eqn:substitute_rstr} follows since the bound is increasing in $\rstr$. Bounds on $\rstr$ can be obtained as follows. In the triangle $\Delta xzO$, we have
\[\frac{1}{2}\|\vz\|_2\sqrt{n\bfrstr}=\frac{1}{2}\|\vbfx\|_2\|\vbfsz\|\sin(\angle_{\vbfx,\vbfsz}),\]
which implies 
\[\bfrstr=\frac{P\|\vbfsz\|_2(1-\cos(\angle_{\vbfx,\vbfsz})^2)}{\|\vz\|_2}.\]
Conditioned on $\eatyp^c$,
\begin{align}
    \frac{P\sigma^2(1-\varepsilon)}{(P+\sigma^2)(1-\varepsilon)}\le\rstr\le\frac{P\sigma^2(1+\varepsilon)}{P\sigma^2(1-\varepsilon)}.
    \label{eqn:bound_rstr}
\end{align}
We have $\Delta(\tau)=\cO(\poly(n))$ by taking $\tau=\Theta((\log n)/n)$.
\begin{figure} 
    \centering
    \includegraphics[width = 0.4\textwidth]{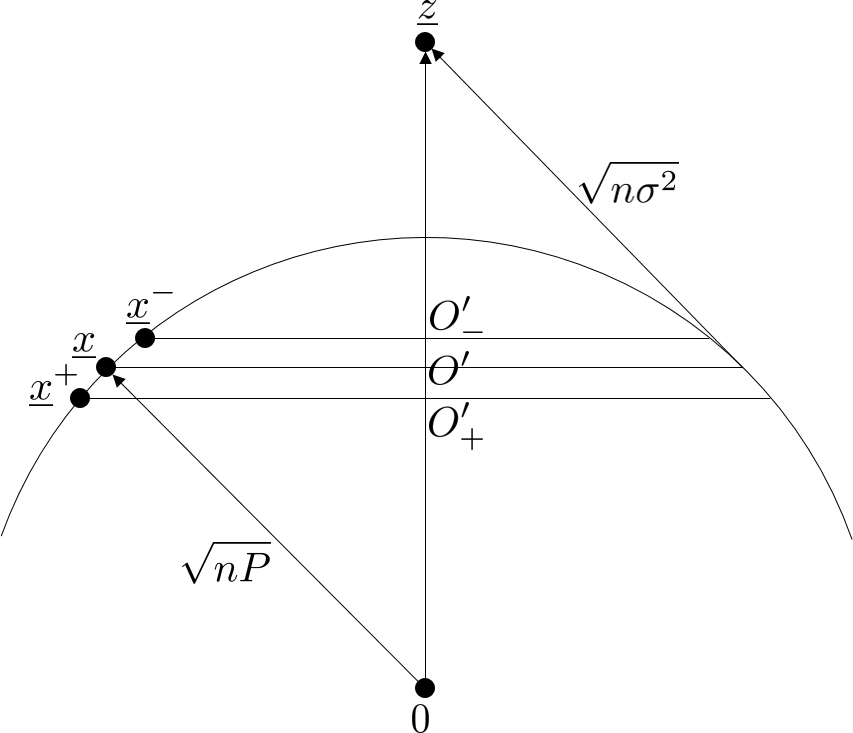}
    \caption{For any James's observation $\vz$, a thin strip containing the transmitted codeword $\vx$ is constructed on the coding sphere $\cS^{n-1}(0,\sqrt{nP})$. Typically, the geometry is shown in above figure. From James's perspective, given $\vz$, codewords in the strip are approximately equally likely to be transmitted by Alice. The quasi-uniformity is defined as the maximum deviation of probability of codewords in the strip. Notice that any codeword at the same latitude has exactly the same probability. Codewords on the upper (respectively lower) boundary of the strip, say $\vx^-$ (respectively $\vx^+$), are most (respectively least) likely in the strip to be transmitted. For small enough ($\cO((\log n)/n)$) thickness of the strip, the quasi-uniformity factor is a polynomial in $n$.}
    \label{fig:quasi_unif}
\end{figure}
\end{proof}

Next, we show that if the codebook rate is large enough, then with high probability (over the randomness in the codebook generation) every strip will contain exponentially many codewords.
\begin{lemma}[Exponentially many codewords in strips]
 Let $\Rcode>\frac{1}{2}\log\left(1+\frac{P}{\sigma^2}\right)$. Then, from James's perspective, he is confused with exponentially many codewords in the strip with probability doubly exponentially close to one.
 \begin{align*}
     \p{}(\estr|\eatyp^c)=&\p{}\left(\left.\bigcup_i\estr(i)\right|\eatyp^c\right)\\
     =&\p{}\left(\left.\exists i,\;|\mstr(\vzq,i)|\leq 2^{3\varepsilon n}\right|\eatyp^c\right)\\
     =&\p{}\left(\left.\exists i,\;|\strip^{n-1}(\vzq,i)\cap \cC|\leq 2^{3\varepsilon n}\right|\eatyp^c\right)\\
     \leq& 2^{-2^{\Omega(n)}}.
 \end{align*}
\label{lemma:exp_cw_strip}
\end{lemma}
\begin{remark}
Note that James's uncertainty set contains codewords from the whole codebook $\cC$, not only from $\cC^{(\bfk)}$ for some particular $\bfk$, since the shared key is assumed to be kept secret from James.
\end{remark}
\begin{proof}
First, in Appendix~\ref{sec:prf_strip}, we show that, for any typical $\vzq$ and $i$,
\[
 \e{}\left(\left.|\strip^{n-1}(\vzq,i)\cap \cC|\right|\eatyp^c\right)\ge2^{4\varepsilon n}.
\]
Note that the random variable $|\strip^{n-1}(\vzq,i)\cap\cC|$ can be written as a sum of a bunch of independent indicator variables
\[|\strip^{n-1}(\vzq,i)\cap\cC|=\sum_{(m,k)\in[2^{n(R+\Rkey)}]}\one_{\{\vbfx(m,k)\in\strip^{n-1}(\vzq,i)\}},\]
or in slightly different notation
\[|\mstr(\vzq,i)|=\sum_{(m,k)\in[2^{n(R+\Rkey)}]}\one_{\{(m,k)\in\mstr(\vzq,i)\}}.\]
Thus by Chernoff bound, 
\[
 \p{}(\estr(i)|\eatyp^c)=\p{}(|\strip^{n-1}(\vzq,i)\cap \cC|\leq 2^{3\varepsilon n}|\eatyp^c) \leq 2^{-2^{\Omega(n)}}.
\]
Note that there are at most $2\rho/\tau=\cO(n/\log n)$ many $i$'s. Lemma~\ref{lemma:exp_cw_strip} is then obtained by taking a union bound over all $i$'s. 
\end{proof}

\subsection{Event $\eorcl$: Transmitted codeword falls into the last block}\label{sec:e_orcl}
Conditioned on $\vbfz$ and the OGS, the transmitted codeword is quasi-uniformly distributed over the strip that contains the OGS. Given $\eatyp^c$ and $\estr^c$, there are at least $2^{3\varepsilon n}$ many codewords in each strip. Also, notice that each OGS (except perhaps the last one) is of size $2^{n\varepsilon}$.
Therefore, the probability over $ \bfm,\bfk $ that $\eorcl$ occurs can be bounded as follows.
\[
 \p{}(\eorcl|\eatyp^c\cap\eorcl^c) \leq \frac{1}{\ell}\Delta(\tau)=\frac{1}{\lceil|\mstr(\vzq,i)|/2^{n\varepsilon}\rceil}\Delta(\tau)\le\frac{1}{2^{3\varepsilon n}/2^{n\varepsilon}}\Delta(\tau)=2^{-2\varepsilon n}\Delta(\tau).
\]

\subsection{Event $\eerr$: Existence of a ``bad" attack vector}\label{sec:myop_ld}
At first, we fix $\vzq\in\cZ,i\in\{-\varepsilon/\delta+1,\cdots,\varepsilon/\delta\},j\in[\ell],\vsq\in\cS$. We will show that the probability that the list-size is greater than $L$ is superexponentially small in $n$. We will finally use a quantization argument and take a union bound over $\vzq,i,j,\vsq$ to show that $\p{}(\eerr|\cE^c)=o(1)$.

For any $(m,\bfk)\in\ogs^{(j)}(\vzq,i)$, to prove that 
$\p{}( |\cL^{(\bfk)}(\vbfx(m),\vsq)|>L |\cE^c) $
is superexponentially decaying, we will find the typical value of $\sqrt{n\bfr}$, where $\bfr\coloneq\bfr(m,\vsq)$ is  the normalized radius of the list-decoding region $\cB^n(\vbfx(m,\bfk)+\vsq,\sqrt{nN}+\sqrt{n\delta_\cS})\cap\cS^{n-1}(0,\sqrt{nP})$ (which is nothing but a cap), and use this to obtain an upper bound on the probability that
the list-size is large. 

Recall that the typical radius is defined as $ \ropt(\vsq)\coloneq\e{}(\bfr(m,\vsq)) $. This will be obtained as the solution to an optimization problem~\eqref{eq:opt_myop_ld} and actually corresponds to the worst-case $ \vs $ that James can choose for the given OGS. 
{Recall that $ \er(m,\vsq) $ denotes the event that the the radius of the list-decoding region $\cB^n(\vbfx(m,\bfk)+\vsq,\sqrt{nN}+\sqrt{n\delta_\cS})\cap\cS^{n-1}(0,\sqrt{nP})$ is not typical, }  \emph{i.e.} that $ \bfr $ is much larger than $ \ropt $.  
Let $ \cJ $ denote the event that $ (\vbfzq,\vbfsq)=(\vzq ,\vsq ) $ and the transmitted codeword lies in $ \ogs^{(j)}(\vzq ,i) $.
In Section~\ref{sec:listfailure_fixed_s_ogs_k_m}, we will show the following:
\begin{lemma}
Fix $\vzq$, $i$, $j$ and $\vsq$.  There exists $ f_{11}(\varepsilon,\delta_{\cS}) $ satisfying $ f_{11}(\varepsilon,\delta_{\cS})\to 0 $ as $ \varepsilon\to 0 $  such that for every $ (m,k) $ in the OGS,
$$ \p{}(\er(m,\vsq)|\cE^c\cap\cJ) = \p{}\left(\left.\bfr(m,\vsq)>\ropt(\vsq)(1+f_{11}(\varepsilon,\delta_{\cS}))\right|\cE^c\cap\cJ\right)\le2^{-f_9(\varepsilon,\eta,\delta_\cS,\delta_\cZ) n},$$ where $\cE=\eatyp\cup\estr\cup\eorcl$ and $f_9(\varepsilon,\eta,\delta_\cS,\delta_\cZ)$ can be taken as $\frac{3}{2}\varepsilon$ by choosing proper $\eta$, $\delta_\cS$ and $\delta_\cZ$.
\label{lem:er}
\end{lemma}

Recall that $ \esq $ is the event  that the radius of the list-decoding region $ \bfr$ is greater than $\ropt(1+f_{11}(\varepsilon,\delta_{\cS})) $ for more than $ n^2 $ codewords in the OGS. Since codewords in OGS are quasi-uniformly distributed and independent,  using Chernoff-type bound similar to Lemma~\ref{lemma:sup_exp_ld} and the above lemma, we get that

\begin{lemma}\label{lem:esq}
Fix $\vzq$, $i$, $j$ and $\vsq$. Then
\begin{align*}
    \p{}(\esq(\vzq,i,j)|\cE^c\cap\cJ)=&\p{}\left(\left.|\{(m,\bfk)\in\ogs^{(j)}(\vzq,i)\colon \bfr(m,\vsq)>\ropt(\vsq)(1+f_{11}(\varepsilon,\delta_\cS))\}|>n^2\right|\cE^c\cap\cJ\right)\\
    \le&\p{}\left(\left.\sum_{(m,k)\in\ogs^{(j)}(\vzq,i)}\one_{\er(m,\vsq)}>n^2\right|\cE^c\cap\cJ\right)\\
    \le&2^{-\Omega(n^3)}.
\end{align*}
\end{lemma}

In Section~\ref{sec:listfailure_two_types_err}, we will show the following:
\begin{lemma}
Fix $\vzq$, $i$, $j$ and $\vsq$. For any $(m,\bfk)\in \ogs^{(j)}(\vzq,i)$ for which $ \bfr $ is typical, we have
 \[
  \p{}( |\cL^{(\bfk)}(\vbfx(m),\vsq)|>L |\wt\cE^c\cap\cJ) \leq 2^{-\Omega(n^3)},
 \]
 where $L$ is set to be $3n^2$ and $\wt\cE$ denotes $\cE\cup\esq(\vzq,i,j)$.
\label{lemma:listfailure_two_types_err}
\end{lemma}

Lemmas~\ref{lem:er} and \ref{lemma:listfailure_two_types_err} allow us to conclude that the probability that there exist $n^2+1$ codewords with list-sizes greater than $ L $ is superexponentially small. 
Recall that $ \cJ $ denotes the realization of James's knowledge, i.e., the event that $ (\vbfz_Q,\vbfs_Q)=(\vzq ,\vsq ) $ and $ (\bfm,\bfk)\in \ogs^{(j)}(\vzq ,i) $.
\begin{lemma}\label{lem:list_dec_error}
 Fix $\vzq$, $i$, $j$ and $\vsq$. Then
 \begin{align*}
     \p{}(\eerr(\vzq,i,j,\vsq)|\cE^c\cap\cJ)\leq& 2^{-\Omega(n^3)}.
 \end{align*}
\label{lemma:listfailure_fixed_s_ogs}
\end{lemma}
\begin{proof}
Using \rev{Fact}~\ref{fact:error_event_decompo_lem} we obtain:
\begin{align*}
\p{}(\eerr(\vzq,i,j,\vsq)|\cE^c\cap\cJ)\leq&\p{}(\eerr(\vzq,i,j,\vsq)|\wt\cE^c\cap\cJ)+\p{}(\esq(\vzq,i,j)|\cE^c\cap\cJ).
\end{align*}
The first term is superexponentially small by  Lemma~\ref{lemma:listfailure_two_types_err} and union bound. The probabilities are computed with respect to the quasiuniform (due to the conditioning) distribution of the codewords in the OGS over the strip, and the uniform distribution of the remaining codewords over the sphere (conditioned on the OGS, they are independent of James's observations).
\begin{align}
    &\p{}\left(\left.\eerr(\vzq,i,j,\vsq)\right|\wt\cE^c\cap\cJ\right)\notag\\
    =&\p{}\left(\left.\chi(\vzq,i,j,\vsq)>n^2+ 1\right|\widetilde{\cE}^c\cap\cJ\right)\notag\\
    =&\p{}\left(\left.\sum_{(m,\bfk)\in \ogs^{(j)}(\vzq,i)}\one_{\{ |\cL^{(\bfk)}(\vbfx(m),\vsq)|>L \}}>n^2+1\right|\widetilde{\cE}^c\cap\cJ\right)\notag\\
    =&\p{}\left(\left.\sum_{(m,\bfk)\in \ogs^{(j)}(\vzq,i)}\one_{\{ |\cL^{(\bfk)}(\vbfx(m),\vsq)|>L \}}\left(\one_{\er(m,\vsq)} + \one_{\er(m,\vsq)^c}\right)>n^2+1\right|\widetilde{\cE}^c\cap\cJ\right)\notag\\
    \le&\p{}\left(\left.\sum_{(m,\bfk)\in \ogs^{(j)}(\vzq,i)}\one_{\{ |\cL^{(\bfk)}(\vbfx(m),\vsq)|>L \}}\one_{\er(m,\vsq)}>n^2\right|\widetilde{\cE}^c\cap\cJ\right)\notag\\
    &+\p{}\left(\left.\sum_{(m,\bfk)\in \ogs^{(j)}(\vzq,i)}\one_{\{ |\cL^{(\bfk)}(\vbfx(m),\vsq)|>L \}}\one_{\er(m,\vsq)^c}>1\right|\widetilde{\cE}^c\cap\cJ\right)\notag\\
    \le&\p{}\left(\left.\sum_{(m,\bfk)\in \ogs^{(j)}(\vzq,i)}\one_{\er(m,\vsq)}>n^2\right|\widetilde{\cE}^c\cap\cJ\right)\notag\\
    &+\p{}\left(\left.\sum_{(m,\bfk)\in \ogs^{(j)}(\vzq,i)}\one_{\{ |\cL^{(\bfk)}(\vbfx(m),\vsq)|>L \}}\one_{\er(m,\vsq)^c}\ge1\right|\widetilde{\cE}^c\cap\cJ\right)\notag\\
    =&\p{}\left(\left.\psi(\vzq,i,j,\vsq)> n^2\right|\widetilde{\cE}^c\cap\cJ\right) \notag\\
    &+\p{}\left(\left.\exists(m,\bfk)\in \ogs^{(j)}(\vzq,i),\;\bfr(m,\vsq)>\ropt(\vsq)(1+f_{11}(\varepsilon,\delta_{\cS})),\;|\cL^{(\bfk)}(\vbfx(m),\vsq)|>L\right|\widetilde{\cE}^c\cap\cJ\right)\notag\\
    \le&\p{}\left(\left.\esq(\vzq,i,j)\right|\widetilde{\cE}^c\cap\cJ\right)+2^{n\varepsilon}\p{}\left(\left.|\cL^{(\bfk)}(\vbfx(m),\vsq)|>L\right|\widetilde{\cE}^c\cap\cJ\right)\notag\\
    \le&0+2^{n\varepsilon}2^{-\Omega(n^3)}\label{eqn:list_dec_error_typ}\\
    =&2^{-\Omega(n^3)},\notag
\end{align}
where Eqn.~\eqref{eqn:list_dec_error_typ} follows from Lemma~\ref{lemma:listfailure_two_types_err}.

The second term is bounded by Lemma~\ref{lem:esq}. This finishes the proof of Lemma~\ref{lem:list_dec_error}.
\end{proof}

There are at most $2^{\cO(n)}$ many $\vzq$'s in the covering $\cZ$ of $\vz$'s. For each $\vzq$, there are at most $2\varepsilon/\delta=\Theta(n/\log n)$ strips. For fixed $\vzq$ and $i$, a loose upper bound on the number of oracle-given sets in the strip $\strip^{n-1}(\vzq,i)$ is $2^{n(\Rcode-\varepsilon)}$. We are required to quantize $\vs$ using a finite covering of $\cB^n(0,\sqrt{nN})$. The steps mimic the proof of Lemma~\ref{lemma:listdecod_omn_achievability}, and we omit the details. The argument for $\p{}(\eerr|\cE^c)=o(1)$ follows by taking a union bound over 
\[2^{\cO(n)}\cdot\Theta(n/\log n)\cdot2^{n(\Rcode-\varepsilon)}\cdot2^{\cO(n)}=2^{\cO(n)}\]
many configurations of the four-tuple $\vzq$, $i$, $j$ and $\vsq$. This completes the basic ingredients needed to obtain Theorem~\ref{thm:myopic_listdecoding_summary}.

All that remains is to prove Lemma~\ref{lem:er} and Lemma~\ref{lemma:listfailure_two_types_err}.

\subsection{Proof of Lemma~\ref{lem:er}}\label{sec:listfailure_fixed_s_ogs_k_m}
In this section, we will upper bound the average area of the list-decoding region over James's uncertainty in the OGS, and show that with high probability it will not exceed the typical value largely.

Let us begin by taking a closer look at the geometry. Let $\vbfx$ be quasi-uniformly distributed over the strip $\strip^{n-1}(\vzq,i)$. For reasons that will be clear in the subsequent calculations, we decompose $\vbfx$ and $\vsq$ into sums of vectors parallel and orthogonal to $\vzq$, 
\[
    \vbfx=\vbfxq^\|+\vbfxq^\perp,\quad\vsq=\vsq^\|+\vsq^\perp,
\]
where $\ve^\|$ denotes the unit vector along $\vzq$ and 
\[
\begin{array}{rlrl}
    \vbfxq^\|\coloneq&\sqrt{n\pmb\alpha_x}\ve^\|,\quad&\vbfxq^\perp\coloneq&\sqrt{n\pmb\beta_x}\vbfe^\perp,\\
    \vsq^\|\coloneq&-\sqrt{n\alpha_s}\ve^\|,\quad&\vsq^\perp\coloneq&\sqrt{n\beta_s}\ves^\perp, \\
    \sqrt{n\pmb\alpha_x}=&\frac{\langle\vbfx,\vzq\rangle}{\|\vzq\|_2},\quad&-\sqrt{n\alpha_s}=&\frac{\langle\vsq,\vzq\rangle}{\|\vzq\|_2}. 
\end{array}
\]
\rev{\begin{remark}
In this remark, we clarify the notation in the above decomposition. Since $ \vzq $ is given, the unit vector $ \ve^\| $ along $ \vzq $ is fixed (hence in non-boldface).
Since $ \vbfx $ is quasi-uniformly distributed over $ \strip^{n-1}(\vzq,i) $, the perpendicular component of $ \vbfx $ is isotropically distributed in an annulus centered at 0, perpendicular to $ \vzq $ and hence its direction $ \vbfe^\perp $ is also isotropically distributed (therefore  in boldface).
Furthermore, since $ \strip^{n-1}(\vzq,i) $ has certain thickness, the angle between $ \vbfx $ and $ \vzq $ gets a slight variation, so $ \pmb\alpha_x,\pmb\beta_x $ are in fact random variables (hence in boldface) as well.

On the other hand, since $\vzq,\vsq$ are both fixed, both components of $\vsq$ are fixed and in particular the direction $\ves^\perp $ of its perpendicular component is in non-boldface.
\end{remark}}

\rev{Note that $ \vbfx $ is on $ \cS^{n-1}(0,\sqrt{nP}) $ and thus $ \pmb\alpha_x + \pmb\beta_x = P $.}
\rev{Note also that} $\vs$, and thus $\vsq$, should satisfy James's power constraint, i.e,  $\alpha_s+\beta_s\le N$. 

This decomposition will bring us analytic ease to compute the list-size, which is equivalent to computing the radius $\sqrt{n\bfr}$ of the myopic list-decoding region. It is noted that, for any $(m,k)\in\ogs^{(j)}(\vzq,i)$, the list-decoding region of $\vx(m,k)$ under $\vsq$  is nothing but a cap 
\begin{align*}
    &\cB^n(\vbfyq,\sqrt{nN}+\sqrt{n\delta_{\cS}})\cap\cS^{n-1}(0,\sqrt{nP})\\
    =&\cB^n(\vbfx(m,\bfk)+\vsq,\sqrt{nN}+\sqrt{n\delta_{\cS}})\cap\cS^{n-1}(0,\sqrt{nP})\\
    =&\C^{n-1}(\cdot,\sqrt{n\bfr},\sqrt{nP}),
\end{align*}
where $\vbfyq\coloneq\vbfx+\vsq$. 

Notice that, averaged over the codebook generation, $\vbfe^\perp$ is uniformly distributed\footnote{Technically speaking, this is indeed the case only when we decompose $\vbfx$ with respect to $\vz$, but not its quantization $\vzq$. It is not exactly, yet still approximately true when we take the quantization $\vzq$ of $\vz$. This quantization error will be taken into account via Lemma~\ref{lemma:quasi_isotropy}.} over the unit \rev{sphere} $\cS^{n-2}(0,1)$ \rev{orthogonal to $\vzq$}. We will use Lemma~\ref{lemma:beta_tail} to bound the tail of inner product of $\vbfxq^\perp$ and $\vzq$.

Heuristically, in expectation, without quantization, we can compute the scale of each component of $\vbfx$ and $\vsq$. A glimpse at the geometry (shown in Figure~\ref{fig:geom_decomp}) immediately gives us the following relations:
\begin{figure} 
    \centering
    \includegraphics[width = 0.4\textwidth]{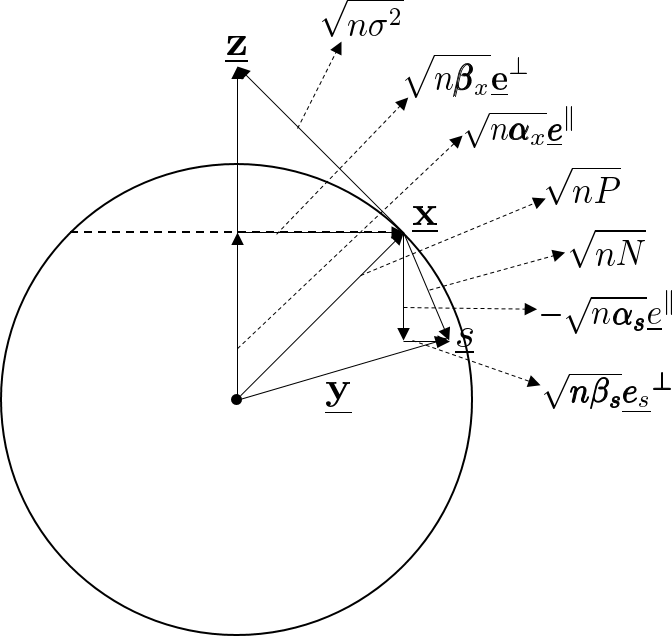}
    \caption{In expectation, the noise to James $\vbfsz$ is orthogonal to the codeword $\vbfx$ and of length $\sqrt{n\sigma^2}$. Fix a legitimate attack vector $\vs$ and decompose $\vbfx$ and $\vs$ into sums of components orthogonal and parallel to James's observation $\vbfz$. Each component in above figure can be concentrated and is robust to quantization errors.}
    \label{fig:geom_decomp}
\end{figure}
\begin{equation}
    \frac{\sqrt{n\pmb\alpha_x}}{\sqrt{nP}}=\frac{\sqrt{nP}}{\sqrt{n(P+\sigma^2)}}\implies\pmb\alpha_x=\frac{P^2}{P+\sigma^2},
    \label{eq:expectation_alpha_x}
\end{equation}
\begin{equation}
    \frac{\sqrt{n\pmb\beta_x}}{\sqrt{n\sigma^2}}=\frac{\sqrt{nP}}{\sqrt{n(P+\sigma^2)}}\implies\pmb\beta_x=\frac{P\sigma^2}{P+\sigma^2},
    \label{eq:expectation_beta_x}
\end{equation}
both of which  follow from similarity of triangles.
In fact, the lengths of both components are well concentrated. Recall that conditioned on $ \cE^c\cap\cJ $, the transmitted codeword is quasi-uniformly distributed over the strip. For any $ 0<\eta<1 $, we have
\begin{align}
    &\p{}\left(\left.\pmb\alpha_x\notin \frac{P^2}{P+\sigma^2}(1\pm\eta)\right|\cE^c\cap\cJ\right)\notag\\
    =&\p{}\left(\left.\frac{\langle\vbfx,\vzq\rangle}{\|\vzq\|_2}\notin\sqrt{n\frac{P^2}{P+\sigma^2}(1\pm\eta)}\right|\cE^c\cap\cJ\right)\notag\\
    =&\p{}\left(\left.\langle\vbfx,\vzq\rangle\notin\sqrt{n\frac{P^2}{P+\sigma^2}(1\pm\eta)}\|\vzq\|_2\right|\cE^c\cap\cJ\right)\notag\\
    \le&\p{}\left(\left.\langle\vbfx,\vzq\rangle\notin\sqrt{n\frac{P^2}{P+\sigma^2}(1\pm\eta)}(\|\vz\|_2\mp\sqrt{n\delta_{\cZ}})\right|\cE^c\cap\cJ\right)\notag\\
    \le&\p{}\left(\left.\langle\vbfx,\vzq\rangle\notin\sqrt{n\frac{P^2}{P+\sigma^2}(1\pm\eta)}(\sqrt{n(P+\sigma^2)(1\mp\varepsilon)}\mp\sqrt{n\delta_{\cZ}})\right|\cE^c\cap\cJ\right)\label{eq:norm_z}\\
    \le&\p{}(\langle\vbfx,\vzq\rangle\notin nP(1\pm f_4(\varepsilon,\eta,\delta_\cZ))|\cE^c\cap\cJ)\notag\\
    \eqcolon&2^{-f_5(\varepsilon,\eta,\delta_\cZ)n},\label{eq:alpha_x}
\end{align}
where Inequality~\eqref{eq:norm_z} follows from that $\|\vz\|_2\in\sqrt{n(P+\sigma^2)(1\pm\varepsilon)}$ since we condition on $\cE^c\cap\cJ$,
and Inequality~\eqref{eq:alpha_x} follows from the following calculations.  
\begin{align}
    &\p{}(\langle\vbfx,\vzq\rangle\notin nP(1\pm f_4(\varepsilon,\eta,\delta_\cZ))|\cE^c\cap\cJ)\notag\\
    =&\p{}(\langle\vbfx,\vz+\vze\rangle\notin nP(1\pm f_4(\varepsilon,\eta,\delta_\cZ))|\cE^c\cap\cJ)\label{eq:z_eq_z_ze}\\
    =&\p{}(\langle\vbfx,\vz\rangle+\langle\vbfx,\vze\rangle\notin nP(1\pm f_4(\varepsilon,\eta,\delta_\cZ))|\cE^c\cap\cJ)\notag\\
    \le&\p{}(\langle\vbfx,\vz\rangle\notin nP(1\pm f_4(\varepsilon,\eta,\delta_\cZ))\mp\|\vbfx\|_2\|\vze\|_2|\cE^c\cap\cJ)\label{eq:cauchy_schwarz}\\
    \le&\p{}(\langle\vbfx,\vz\rangle\notin nP(1\pm f_4(\varepsilon,\eta,\delta_\cZ))\mp n\sqrt{P\delta_\cZ}|\cE^c\cap\cJ)\notag\\
    =&\p{}(\langle\vbfx,\vbfx+\vbfsz\rangle\notin nP(1\pm f_4(\varepsilon,\eta,\delta_\cZ)\mp\sqrt{P\delta_\cZ})|\cE^c\cap\cJ)\notag\\
    \le&\p{}(\langle\vbfx,\vbfsz\rangle+\|\vbfx\|_2^2\notin nP(1\pm f_4'(\varepsilon,\eta,\delta_\cZ))|\cE^c\cap\cJ)\notag\\
    =&\p{}(|\langle\vbfx,\vbfsz\rangle|>nP f_4'(\varepsilon,\eta,\delta_\cZ)|\cE^c\cap\cJ)\notag\\
    =&\p{}\left(\left.|\langle\underline{e_1},\underline{\bfe_{s_z}}\rangle|>\frac{nPf_4'(\varepsilon,\eta,\delta_\cZ)}{\sqrt{nP}\|\vbfsz\|_2}\right|\cE^c\cap\cJ\right)\notag\\
    \le&\p{}\left(\left.|\langle\underline{e_1},\underline{\bfe_{s_z}}\rangle|>\frac{nPf_4'(\varepsilon,\eta,\delta_\cZ)}{\sqrt{nP}\sqrt{n\sigma^2(1+\varepsilon)}}\right|\cE^c\cap\cJ\right)\label{eq:normalize}\\
    \le&2^{-\frac{Pf_4'(\varepsilon,\eta,\delta_\cZ)^2}{2\sigma^2(1+\varepsilon)}(n-1)}\label{eq:concentrate_ip_eone_esz}\\
    \eqcolon&2^{-nf_5(\varepsilon,\eta,\delta_\cZ)},
\end{align}
where in the above chain of (in)equalities, we use the following facts.
\begin{enumerate}
    \item In~\eqref{eq:z_eq_z_ze}, we write $\vzq=\vz+\vze$, where $\vze$ denotes the decomposition error which has norm at most $\sqrt{n\delta_\cZ}$ by the choice of the covering $\cZ$.
    \item Inequality~\eqref{eq:cauchy_schwarz} follows from Cauchy--Schwarz inequality $-\|\vbfx\|_2\|\vze\|_2\le\langle\vbfx,\vze\rangle\le\|\vbfx\|_2\|\vze\|_2$.
    \item In Inequality~\eqref{eq:normalize}, $\|\vbfsz\|_2\in\sqrt{n\sigma^2(1\pm\varepsilon)}$ since we condition on $\cE^c\cap\cJ$.
    \item Inequality~\eqref{eq:concentrate_ip_eone_esz} is a straightforward application of Lemma~\ref{lemma:beta_tail}.
\end{enumerate}
It follows that with probability at least $1-2^{-f_5(\varepsilon,\eta,\delta_\cZ)n}$ the scale of the perpendicular component is also concentrated around its expected value,
\[\pmb\beta_x\in P-\frac{P^2}{P+\sigma^2}(1\pm\eta)=\frac{P^2}{P+\sigma^2}(1\mp\eta/\sigma^2)\eqcolon\frac{P^2}{P+\sigma^2}(1\mp\eta_1),\]
All the above concentration is over the randomness in the channel between Alice and James and the codebook generation.

We are now ready to compute the expected value of the radius $\sqrt{n\bfr}$ of the list-decoding region and concentrate it. To this end, let us first do some rough calculations to see what we should aim for. Loosely speaking, we expect the following quantity 
\begin{align}
    \langle\vbfx,-\vsq\rangle=&\langle\sqrt{n\pmb\alpha_x}\ve^\|+\sqrt{n\pmb\beta_x}\vbfe^\perp,\sqrt{n\alpha_s}\ve^\|-\sqrt{n\beta_s}\ves^\perp\rangle\notag\\
    =&\langle\sqrt{n\pmb\alpha_x}\ve^\|,\sqrt{n\alpha_s}\ve^\|\rangle-\langle\sqrt{n\pmb\beta_x}\vbfe^\perp,\sqrt{n\beta_s}\ves^\perp\rangle\label{eq:crossing_terms_vanish}\\
    =&n\sqrt{\pmb\alpha_x\alpha_s}-n\sqrt{\pmb\beta_x\beta_s}\langle\vbfe^\perp,\ves^\perp\rangle\notag
\end{align}
to satisfy
\begin{equation}
    \e{}(\langle\vbfx,-\vsq\rangle)\approx nP\sqrt{\frac{\alpha_s}{P+\sigma^2}}.
    \label{eq:e_inner_prod}
\end{equation}
This is because:
\begin{enumerate}
    \item In Equation~\eqref{eq:crossing_terms_vanish}, by decomposition, crossing terms vanish.
    \item When there is no quantization error in $\vz$, it holds that $\e{}(\langle\vbfe^\perp,\ves^\perp\rangle)=0$ as $\vbfe^\perp$ is isotropically distributed on a ring $\cS^{n-2}(0,1)$. By previous heuristic calculations (see Equation~\eqref{eq:expectation_alpha_x}), we expect $\pmb\alpha_x$ to be roughly $\frac{P^2}{P+\sigma^2}$.
\end{enumerate}
These indicate that Equation~\eqref{eq:e_inner_prod} should be reasonably correct modulo some small error terms, bounded below in Lemma~\ref{lemma:quasi_isotropy}.

\begin{remark}
Notice that, for a fixed $\beta_x$, we actually know exactly the distribution of $\langle\sqrt{n\beta_x}\vbfe^\perp,\sqrt{n\beta_s}\ves^\perp\rangle$. Indeed, in $\bR^n$, given a fixed unit vector $\ve$ and a random unit vector $\vbfe$ isotropically distributed on the unit sphere $\cS^{n-1}(0,1)$, $\frac{|\langle\vbfe,\ve\rangle|^2+1}{2}$ follows Beta distribution $\Beta\left(\frac{n-1}{2},\frac{n-1}{2}\right)$. Plugging this into our calculation, indeed, we get that $\langle\sqrt{n\beta_x}\vbfe^\perp,\sqrt{n\beta_s}\ves^\perp\rangle$ has mean $0$. However, this result does not bring us any analytic advantage. Rather, in the analysis, when caring about concentration, we  use Lemma~\ref{lemma:beta_tail} to  approximate the tail of this distribution. 
\end{remark}

Although given $\vz$, $\vbfe^\perp$ is perfectly isotropic on the unit ring $\cS^{n-2}(0,1)$, it is not exactly the case when we are working with $\vzq$. It is, however, probably approximately correct (PAC). Specifically, this issue can be fixed by the following lemma. As shown in Figure~\ref{fig:quasi_isotropy}, recall that we denote by $\vbfx=\vbfxq^\|+\vbfxq^\perp$ the decomposition with respect to $\vzq$. Also, denote by $\vbfx=\vbfx^\|+\vbfx^\perp$ the decomposition with respect to $\vz$. Define the error vectors as $\vbfxe\coloneq\vbfxq^\perp-\vbfx^\perp$ and $\vse\coloneq\vsq^\perp-\vs^\perp$.
\begin{figure} 
\centering
\begin{subfigure}[t] {.49\textwidth}
    \centering
    \includegraphics[height = 0.3\textheight]{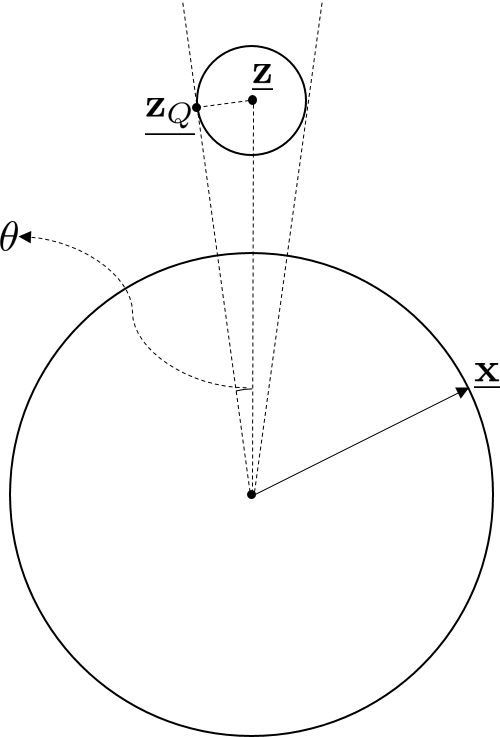}
    \caption{For the ease of union bounds over James's observation, any given $\vbfz$ is quantized to $\vbfzq$ using the $\sqrt{n\delta_\cZ}$-net $\cZ$. The quantization error introduced by covering will have an impact on the rest of the analysis. We measure this using the angular distance between $\vbfz$ and $\vbfzq$. The maximum angle $\theta$ is given by the geometry shown in the above figure. $\theta$ is small given a typical $\vbfz$ due to the covering property of $\cZ$. We will keep track how errors $\theta$ are propagated.}
    \label{fig:quant_z}
\end{subfigure}
\hfill
\begin{subfigure}[t] {.49\textwidth}
    \centering
    \includegraphics[height = 0.3\textheight]{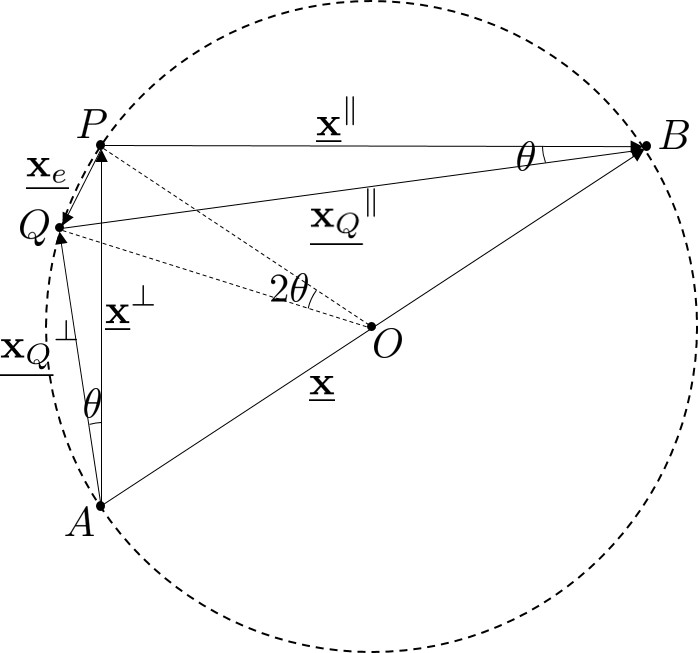}
    \caption{Ideally we would like to decompose everything into directions along and perpendicular to $\vbfz$. However, what we really work with is the quantized version $\vbfzq$ of $\vbfz$. This introduces approximation errors that we must keep track of when computing the parallel and perpendicular components of the vectors we are interested in. A comparison of the geometry of decomposition with respect to $\vbfzq$ and $\vbfz$ is shown in the above figure. In particular, the error in the perpendicular component $\vbfx^\perp$ of $\vbfx$ stemming from the quantization of $\vbfz$ can be bounded in terms of the angular quantization error $\theta$.}
    \label{fig:quasi_isotropy}
\end{subfigure}
\caption{The geometry of the propagation of the quantization error of $\vbfz$.  }
\label{fig:quasi_isotropy_quant_z}
\end{figure}
\begin{lemma}\label{lemma:quasi_isotropy}
Fix $\zeta>0$. Also fix $\vz$ and $\vs$, and thereby $\vzq$ and $\vsq$. Let $\zeta'\coloneq\zeta-\sqrt{P\delta_\cS}-\sqrt{\frac{NP\delta_\cZ}{(P+\sigma^2)(1-\varepsilon)}}-\sqrt{\frac{P\delta_\cS\delta_\cZ}{(P+\sigma^2)(1-\varepsilon)}}$. Then
\[\p{}(|\langle\vbfxq^\perp,\vsq^\perp\rangle|\ge{n\zeta}|\cE^c\cap\cJ)\le2^{-\frac{(n-2)\zeta'^2}{2NP}}.\]
\end{lemma}
\begin{proof}
We write
\begin{align*}
    &\p{}(|\langle\vbfxq^\perp,\vsq^\perp\rangle|\ge{n\zeta}|\cE^c\cap\cJ)\\
    =&\p{}(|\langle\vbfx^\perp+\vbfxe,\vs^\perp+\vse\rangle|\ge{n\zeta}|\cE^c\cap\cJ)\\
    \le&\p{}(|\langle\vbfx^\perp,\vs^\perp\rangle|+|\langle\vbfx^\perp,\vse\rangle|+|\langle\vbfxe,\vs^\perp\rangle|+|\langle\vbfxe,\vse\rangle|\ge{n\zeta}|\cE^c\cap\cJ)\\
    \le&\p{}(|\langle\vbfx^\perp,\vs^\perp\rangle|+\|\vbfx^\perp\|_2\|\vse\|_2+\|\vbfxe\|_2\|\vs^\perp\|_2+\|\vbfxe\|_2\|\vse\|_2\ge{n\zeta}|\cE^c\cap\cJ)\\
    \le&\p{}(|\langle\vbfx^\perp,\vs^\perp\rangle|+n\sqrt{P\delta_\cS}+\sqrt{nN}\|\vbfxe\|_2+\sqrt{n\delta_\cS}\|\vbfxe\|_2\ge{n\zeta}|\cE^c\cap\cJ).
\end{align*}
It then suffices to upper bound $\|\vbfxe\|_2$. Write $\vzq=\vz+\vze$ where $\vze$ denotes the quantization error for $\vz$ with respect to the covering $\cZ$. As $\vze$ is at most $\sqrt{n\delta_{\cZ}}$, for maximum $\theta$ shown in Figure~\ref{fig:quant_z}, we have $\sin(\theta)=\frac{\|\vze\|_2}{\|\vz\|_2}\le\frac{\sqrt{n\delta_{\cZ}}}{\|\vz\|_2}$. Notice that $\angle QAP=\angle QBP=\angle QOP/2=\theta$. Consider the triangle $\Delta QOP$. We have $\frac{|\overline{QP}|/2}{|\overline{QO}|}=\frac{\|\vbfxe\|_2/2}{\|\vbfx\|_2/2}=\sin(\theta)\le\frac{\sqrt{n\delta_\cZ}}{\|\vz\|_2}$, i.e., $\|\vbfxe\|_2\le\frac{\sqrt{nP}}{\sqrt{n(P+\sigma^2)(1-\varepsilon)}}\sqrt{n\delta_\cZ}=\sqrt{\frac{nP\delta_\cZ}{(P+\sigma^2)(1-\varepsilon)}}$. Now Lemma~\ref{lemma:beta_tail} can be applied.
\begin{align*}
    &\p{}(|\langle\vbfxq^\perp,\vsq^\perp\rangle|\ge n\zeta)\\
    \le&\p{}\left(|\langle\vbfx^\perp,\vs^\perp\rangle|\ge n\zeta-n\sqrt{P\delta_\cS}-n\sqrt{\frac{NP\delta_\cZ}{(P+\sigma^2)(1-\varepsilon)}}-n\sqrt{\frac{P\delta_\cS\delta_\cZ}{(P+\sigma^2)(1-\varepsilon)}}\right)\\
    \le&2^{-\frac{(n-2)n^2\zeta'^2}{2\|\vbfx^\perp\|_2^2\|\vs^\perp\|_2^2}}\\
    \le&2^{-\frac{(n-2)\zeta'^2}{2NP}}.
\end{align*}
This completes the proof for Lemma~\ref{lemma:quasi_isotropy}.
\end{proof}

Now we give a concentrated version of Equation~\eqref{eq:e_inner_prod}.
\begin{align}
    &\p{}\left(\left.\langle-\vbfx,\vsq\rangle\notin nP\sqrt{\frac{\alpha_s}{P+\sigma^2}(1\pm\varepsilon)}\right|\cE^c\cap\cJ\right)\label{eqn:inprod_x_sq}\\
    =&\p{}\left(\left.\langle\sqrt{n\pmb\alpha_x}\ve^\|,\sqrt{n\alpha_s}\ve^\|\rangle-\langle\sqrt{n\pmb\beta_x}\vbfe^\perp,\sqrt{n\beta_s}\ves^\perp\rangle\notin nP\sqrt{\frac{\alpha_s}{P+\sigma^2}(1\pm\varepsilon)}\right|\cE^c\cap\cJ\right)\notag\\
    =&\p{}\left(\left.\langle\sqrt{n\pmb\alpha_x}\ve^\|,\sqrt{n\alpha_s}\ve^\|\rangle-\langle\sqrt{n\pmb\beta_x}\vbfe^\perp,\sqrt{n\beta_s}\ves^\perp\rangle\notin nP\sqrt{\frac{\alpha_s}{P+\sigma^2}(1\pm\varepsilon)},\;\pmb\alpha_x\in\frac{P^2}{P+\sigma^2}(1\pm \eta)\right|\cE^c\cap\cJ\right)\notag\\
    &+\p{}\left(\left.\langle\sqrt{n\pmb\alpha_x}\ve^\|,\sqrt{n\alpha_s}\ve^\|\rangle-\langle\sqrt{n\pmb\beta_x}\vbfe^\perp,\sqrt{n\beta_s}\ves^\perp\rangle\notin nP\sqrt{\frac{\alpha_s}{P+\sigma^2}(1\pm\varepsilon)},\;\pmb\alpha_x\notin\frac{P^2}{P+\sigma^2}(1\pm \eta)\right|\cE^c\cap\cJ\right)\notag\\
    \le&\p{}\left(\left.\langle\sqrt{n\pmb\beta_x}\vbfe^\perp,\sqrt{n\beta_s}\ves^\perp\rangle\notin nP\sqrt{\frac{\alpha_s}{P+\sigma^2}(1\pm \eta)}-nP\sqrt{\frac{\alpha_s}{P+\sigma^2}(1\pm\varepsilon)}\right|\cE^c\cap\cJ\right)\label{eqn:term_beta_x}\\
    &+\p{}\left(\left.\pmb\alpha_x\notin\frac{P^2}{P+\sigma^2}(1\pm \eta)\right|\cE^c\cap\cJ\right).\label{eqn:use_alpha_x_bd}
\end{align}
By \eqref{eq:alpha_x}, the second term \eqref{eqn:use_alpha_x_bd} is at most $2^{-f_5(\varepsilon,\eta,\delta_\cZ)n}$. The first term \eqref{eqn:term_beta_x} can be bounded as follows.
\begin{align}
    &\p{}\left(\left.|\langle\sqrt{n\pmb\beta_x}\vbfe^\perp,\sqrt{n\beta_s}\ves^\perp\rangle|>{nP\sqrt{\frac{\alpha_s}{P+\sigma^2}}}(\sqrt{1+\eta}+\sqrt{1+\varepsilon})\right|\cE^c\cap\cJ\right)\notag\\
    \le&\p{}(|\langle\sqrt{n\pmb\beta_x}\vbfe^\perp,\sqrt{n\beta_s}\ves^\perp\rangle|>nf_8(\varepsilon,\tau)|\cE^c\cap\cJ)\notag\\
    =&\p{}(|\langle\vbfxq^\perp,\vsq^\perp\rangle|\ge nf_8(\varepsilon,\eta)|\cE^c\cap\cJ)\notag\\
    \le&2^{-\frac{(n-2)f_8'(\varepsilon,\eta,\delta_\cS,\delta_\cZ)^2}{2NP}},\label{eq:app_quasi_isotropy}
\end{align}
where inequality~\eqref{eq:app_quasi_isotropy} follows from Lemma~\ref{lemma:quasi_isotropy} and $f_8'(\varepsilon,\eta,\delta_\cS,\delta_\cZ)\coloneq f_8(\varepsilon,\eta)-\sqrt{P\delta_\cS}-\sqrt{\frac{NP\delta_\cZ}{(P+\sigma^2)(1-\varepsilon)}}-\sqrt{\frac{P\delta_\cS\delta_\cZ}{(P+\sigma^2)(1-\varepsilon)}}$. Thus combining bounds \eqref{eq:alpha_x} and \eqref{eq:app_quasi_isotropy} on terms \eqref{eqn:use_alpha_x_bd} and \eqref{eqn:term_beta_x} respectively, the probability \eqref{eqn:inprod_x_sq} is well bounded by
\[2^{-\frac{(n-2)f_8'(\varepsilon,\eta,\delta_\cS,\delta_\cZ)}{2NP}}+2^{-f_5(\varepsilon,\eta,\delta_\cZ)n}\eqcolon2^{-f_9(\varepsilon,\eta,\delta_\cS,\delta_\cZ)n}.\]

The above results allow us to compute the typical length of $\vbfyq$ under the translation of the prescribed $\vsq$. We can do so by writing the ``angular correlation" between $\vbfx$ and $\vsq$ in analytic and geometric ways separately. Specifically, it follows from the above concentration result that with probability at least $1-2^{-f_9(\varepsilon,\eta,\delta_\cS,\delta_\cZ)n}$,
\[\cos(\angle_{-\vbfx,\vsq})=\frac{\langle-\vbfx,\vsq\rangle}{\|\vbfx\|_2\|\vsq\|_2}\in\frac{nP\sqrt{\frac{\alpha_s}{P+\sigma^2}(1\pm\varepsilon)}}{\sqrt{nP}\sqrt{nN}}=\sqrt{\frac{P\alpha_s}{N(P+\sigma^2)}(1\pm\varepsilon)}.\]
On the other hand, by law of cosines, we have (with probability one)
\[\cos(\angle_{-\vbfx,\vsq})=\frac{\|\vbfx\|_2^2+\|\vsq\|_2^2-\|\vbfyq\|_2^2}{2\|\vbfx\|_2\|\vsq\|_2}=\frac{nP+nN-\|\vbfyq\|_2^2}{2\sqrt{nP}\sqrt{nN}}.\]
It immediately follows that
\[\|\vbfyq\|_2^2=n(P+N)-2nP\sqrt{\frac{\alpha_s}{P+\sigma^2}(1\pm\varepsilon)},\]
with probability at least $1-2^{-f_9(\varepsilon,\eta,\delta_\cS,\delta_\cZ)n}$. 

Denote by $\sqrt{n\bfr}$ the radius of the intersection $\cB^n(\vbfyq,\sqrt{nN}+\sqrt{n\delta_\cS})\cap\cS^{n-1}(0,\sqrt{nP})$. Staring at the triangle $\Delta OAO'$ shown in Figure~\ref{fig:geom_decomp_r}, we know that, with probability at least $1-2^{-f_9(\varepsilon,\eta,\delta_\cS,\delta_\cZ)n}$,
\[\left(\frac{\sqrt{n\bfr}}{\sqrt{nP}}\right)^2=(\sin(\angle AOO'))^2=1-(\cos(\angle AOO'))^2,\]
i.e.,
\begin{align*}
    \frac{\bfr}{P}=&1-\left(\frac{nP+\|\vbfyq\|_2^2-(\sqrt{nN}+\sqrt{n\delta_\cS})^2}{2\sqrt{nP}\|\vbfyq\|_2}\right)^2\\ 
    \implies\bfr=&\frac{N-P\frac{\alpha_s}{P+\sigma^2}(1\pm\varepsilon)+f_{10}(\varepsilon,\delta_\cS)}{P+N-2P\sqrt{\frac{\alpha_s}{P+\sigma^2}(1\pm\varepsilon)}}P, 
\end{align*}
where $f_{10}(\varepsilon,\delta_\cS)\coloneq-\frac{(\delta_\cS+2\sqrt{N\delta_\cS})^2}{4P}+\left(1-\sqrt{\frac{\alpha_s}{P+\sigma^2}(1-\varepsilon)}\right)(\delta_\cS+2\sqrt{N\delta_\cS})$. That is to say, we have
\begin{equation}
    \p{}\left(\left.\bfr\notin\frac{N-P\frac{\alpha_s}{P+\sigma^2}}{P+N-2P\sqrt{\frac{\alpha_s}{P+\sigma^2}}}P(1\pm f_{11}(\varepsilon,\delta_\cS))\right|\cE^c\cap\cJ\right)\le2^{-f_9(\varepsilon,\eta,\delta_\cS,\delta_\cZ)n}.
    \label{eq:er}
\end{equation}
\begin{figure} 
    \centering
    \includegraphics[width = 0.4\textwidth]{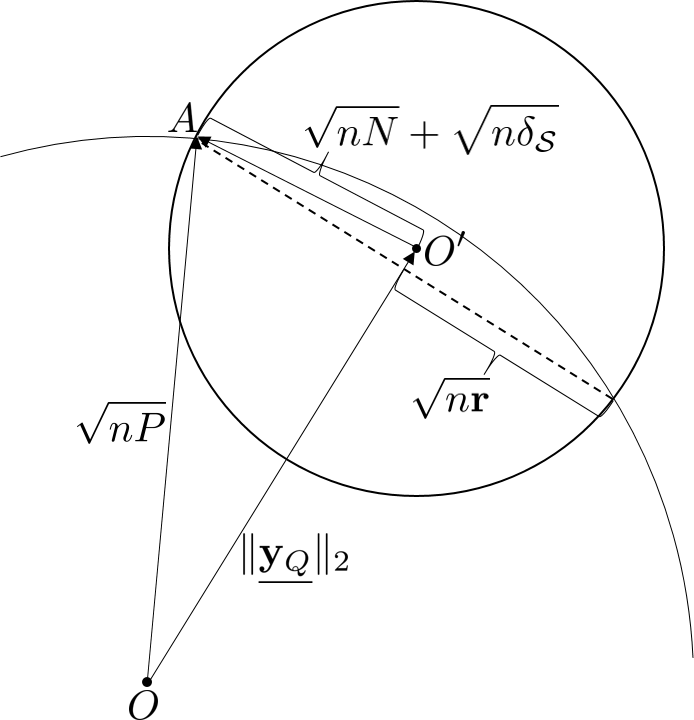}
    \caption{Fix a legitimate attack vector $\vsq$. We compute the expected radius of the list-decoding region, which is the cap $\C^{n-1}(\vbfyq,\sqrt{n\bfr},\sqrt{nP})=\C^{n-1}(\vbfx+\vsq,\sqrt{n\bfr},\sqrt{nP})$ shown in above figure, over codewords in the strip.}
    \label{fig:geom_decomp_r}
\end{figure}

From James's perspective, he aims to maximize the above quantity to confuse Bob to the largest extent. He will take the jamming strategy corresponding to the optimal solution of the following optimization problem\footnote{Although James could possibly choose other strategies, the proof still goes through by taking a union bound over all possible type classes of $ \vbfs $ (corresponding to all possible attack strategies). Since there are only polynomially many of them, and none of them can be as bad as the worst-case strategy, the arguments still hold.}. The average (over the randomness in $ \vbfz $) worst-case (over $ \alpha_s $) value of $ \bfr $ obtained in this manner is what we call $ \ropt $.
\begin{equation}
    \begin{array}{ll}
        \min_{\alpha_s} & \frac{P+N-2P\sqrt{\frac{\alpha_s}{P+\sigma^2}}}{N-P\frac{\alpha_s}{P+\sigma^2}} \\
         \text{subject to} & 0\le\alpha_s\le N \\
         & \frac{\sigma^2}{P}\ge\frac{1}{1-N/P}-1\\
         & P,N,\sigma^2\ge0
    \end{array}
    \label{eq:opt_myop_ld}
\end{equation}
\begin{remark}
Notice that the objective function of the above optimization problem~\eqref{eq:opt_myop_ld} is essentially the same as that of the optimization~\eqref{eq:opt_conv} in the scale-and-babble converse argument under the map $\alpha\mapsto\sqrt{\frac{\alpha_s}{P+\sigma^2}}$.
\end{remark}
Solving this optimization problem and combining it with Lemma~\ref{lemma:listfailure_two_types_err} which will be proved in Section~\ref{sec:listfailure_two_types_err}, we get that, if the code operates at a rate
\begin{equation}
    R<\begin{cases}
    \frac{1}{2}\log\frac{P}{N}\eqcolon C_{\mathrm{LD}},&\frac{1}{1-N/P}-1\le\frac{\sigma^2}{P}\le\frac{1}{N/P}-1\\
    \frac{1}{2}\log\left(\frac{(P+\sigma^2)(P+N)-2P\sqrt{N(P+\sigma^2)}}{N\sigma^2}\right)\eqcolon \Rmyop ,&\frac{\sigma^2}{P}\ge\max\left\{\frac{1}{1-N/P}-1,\frac{1}{N/P}-1\right\}
    \end{cases},
    \label{eq:opt_myop_ld_sol}
\end{equation}
then no matter towards which direction James is going to push the transmitted codeword, a vast majority (an exponentially close to one fraction) of codewords in the strip have small list-sizes (at most a low-degree polynomial in $n$). 

In conclusion,
\[
\p{}\left(\bfr>\ropt(1+ f_{11}(\varepsilon,\delta_\cS)|\cE^c\cap\cJ\right)\le2^{-f_9(\varepsilon,\eta,\delta_\cS,\delta_\cZ)n}.
\]
where $ \ropt  $ is the average list-decoding radius obtained by James choosing the worst-case attack vector minimizing the rate corresponding to the optimization problem~\eqref{eq:opt_myop_ld}.
One can  choose $\eta,\delta_\cS$ and $\delta_\cZ$ so that $ f_9(\varepsilon,\eta,\delta_\cS,\delta_\cZ) = \frac{3}{2} \varepsilon $.  This completes the proof of Lemma~\ref{lem:er}.

\subsection{Proof of Lemma~\ref{lemma:listfailure_two_types_err}}\label{sec:listfailure_two_types_err}
In this section, we will show that myopic list-decoding succeeds with high probability conditioned on everything behaves typically (which is true as we have analyzed in previous sections).

Let $\ropt$ denote the optimal solution of the above optimization. In what follows, we will prove that the probability that there are too many (more than $L=3n^2$) codewords in the list-decoding region is super-exponentially small. According to where the codewords in the list-decoding region come from, the list-decoding error can be divided into two types. If the confusing codewords come from the OGS, then they are quasi-uniformly distributed on the strip that the OGS belongs to, given James observation and extra information revealed to him. Otherwise, if the confusing codewords are outside the OGS, then they are uniformly distributed by the codebook generation. 
\begin{enumerate}
    \item Confusing codewords in the list-decoding region come from OGS. We further subdivide these confusing codewords into two types: those which have a typical $ \bfr $ and those which do not. 

    \begin{itemize}
    	\item Confusing codeword has atypical $ \bfr $: Conditioned on $ \esq^c $, there are at most $ n^2 $ codewords with atypical $ \bfr $. The distribution of these codewords is hard to obtain, and we will pessimistically assume that all these codewords are included in the list.
    	\item Confusing codeword has typical $ \bfr $:
    	The codewords with a typical $ \bfr $ are all independent but no longer uniformly distributed over the strip given $\esq^c$. However, the distribution is almost uniform. For any set $ \cA \subset \bR^n $, we have $$ \p{}(\vbfx(m,\bfk)\in\cA|\esq^c\cap\cE^c\cap\cJ)\leq \frac{\p{}(\vbfx(m,\bfk)\in\cA|\cE^c\cap\cJ)}{\p{}(\esq^c|\cE^c\cap\cJ)} = \p{}(\vbfx(m,\bfk)\in\cA|\cE^c\cap\cJ) (1+o(1)).  $$
    	Therefore, this conditioning does not significantly affect our calculations.
    	Conditioned on bad events aforementioned not happening the average probability that a codeword falls into the list-decoding region is
    	\begin{align*}
    	&\p{}(m'\in\cL^{(\bfk)}(\vbfx(m),\vsq)\cap\ogs^{(j)}(\vzq,i)|\wt\cE^c\cap\cJ)\\
    	\le&\frac{\area(\C^{n-1}(\cdot,\sqrt{n\ropt(1+ f_{11}(\varepsilon,\delta_\cS) )},\sqrt{nP}))}{\area(\strip^{n-1}(\vzq,i))}\cdot\Delta(\tau)(1+o(1))\\
    	=&\sqrt{\frac{\ropt(1+ f_{11}(\varepsilon,\delta_\cS) )}{\rstr}}2^{n\left(-\frac{1}{2}\log\left(\frac{P}{\ropt}\right)+\frac{1}{2}\log\left(\frac{P}{\rstr}\right)+ f_{11}(\varepsilon,\delta_\cS)\right)}\cdot\Delta(\tau)(1+o(1))\\
    	\le&\sqrt{\frac{\ropt(1+ f_{11}(\varepsilon,\delta_\cS) )}{\frac{P\sigma^2(1-\varepsilon)}{(P+\sigma^2)(1+\varepsilon)}}}2^{n\left(-\frac{1}{2}\log\left(\frac{P}{\ropt}\right)+\frac{1}{2}\log\left(1+\frac{P}{\sigma^2}\right) + 2\varepsilon + f_{11}(\varepsilon,\delta_\cS)\right)}\cdot\Delta(\tau)(1+o(1)),
    	\end{align*}
    	which is exponentially small. Then by similar calculations to Lemma~\ref{lemma:sup_exp_ld}, we have
    	\begin{equation}
    	\p{}(|\cL^{(\bfk)}(\vbfx(m),\vsq)\cap\ogs^{(j)}(\vzq,i)|>2n^2|\wt\cE^c\cap\cJ)\le2^{-\Omega(n^3)}.
    	\label{eq:myopic_ld_type_i}
    	\end{equation}
    \end{itemize}
     
    \item Confusing codewords in the list-decoding region do not belong to the OGS.
    \begin{align*}
    &\e{}\left(\left.|\cL^{(\bfk)}(\vbfx(m),\vsq)\backslash\ogs^{(j)}(\vzq,i)|\right|\wt\cE^c\cap\cJ\right)\\
    \le&\frac{\area(\C^{n-1}(\cdot,\sqrt{n\ropt(1+ f_{11}(\varepsilon,\delta_\cS) )},\sqrt{nP}))}{\area(\cS^{n-1}(0,\sqrt{nP}))}2^{n\Rcode}\\
    =&\sqrt{\frac{P}{\ropt(1+ f_{11}(\varepsilon,\delta_\cS) )}}2^{n\left(\Rcode-\frac{1}{2}\log\left(\frac{P}{\ropt}\right)+ f_{11}(\varepsilon,\delta_\cS) \right)},
    \end{align*}
    which is exponentially small if $\Rcode$ is below the threshold. Then immediately by Lemma~\ref{lemma:sup_exp_ld},
    \begin{equation}
    \p{}(|\cL^{(\bfk)}(\vbfx(m),\vsq)\backslash\ogs^{(j)}(\vzq,i)|>n^2|\wt\cE^c\cap\cJ)\le2^{-\Omega(n^3)}.
    \label{eq:myopic_ld_type_ii}
    \end{equation}
    
\end{enumerate}
Taking into account two types of error in Equation~\eqref{eq:myopic_ld_type_i} and Equation~\eqref{eq:myopic_ld_type_ii}, respectively, we get
\begin{align}
    &\p{}(|\cL^{(\bfk)}(\vbfx(m),\vsq)|>L|\wt\cE^c\cap\cJ)\notag\\
    \le&\p{}(|\cL^{(\bfk)}(\vbfx(m),\vsq)\cap\ogs^{(j)}(\vzq,i)|>2n^2|\wt\cE^c\cap\cJ)\notag\\
    &+\p{}(|\cL^{(\bfk)}(\vbfx(m),\vsq)\backslash\ogs^{(j)}(\vzq,i)|>n^2|\wt\cE^c\cap\cJ)\notag\\
    \le&2^{-\Omega(n^3)}. \label{eq:listfailure_fixed_s_ogs_k_m}
\end{align}
This completes the proof of Lemma~\ref{lem:er}.

In the following section, we will argue that Bob will enjoy a vanishing probability of error below the threshold given by optimization~\eqref{eq:opt_myop_ld}, which matches the converse in Section~\ref{sec:proof_scalebabble} in the corresponding region.

\section{Achievability in the sufficiently myopic regime -- from myopic list-decoding to unique decoding}\label{sec:achievability_suffmyopic}
\rev{In this section, given the myopic list-decodability results proved in Sec.~\ref{sec:myopic_list_decoding_details}, we provide the proof of the second half (unique decodability) of the achievability part of Theorem~\ref{thm:capacity_noCR} and hence finish the achievability proof. }

We first sketch the roadmap to proving that Bob can uniquely decode $m$ with high probability.
From James's point of view, $\vbfx$ is quasi-uniformly distributed over the strip. \rev{Loosely speaking, we will say that a message $m_1$ confuses a message $m_2$ if Bob declares the his estimate to be $m_1$ when the actual message is $m_2$. } The probability of error is small if 
\begin{itemize}
    \item the total number of messages (call them \emph{confusing codewords}) that can confuse \emph{any} message in the OGS is small (say $\poly(n)$) with probability super-exponentially close to one; and
    \item any message can only confuse a small number (say $\poly(n)$) of messages in the OGS (call them \emph{confused codewords}) with probability super-exponentially close to one.
\end{itemize}
The first statement follows from a \emph{blob list-decoding} argument and second follows from a \emph{reverse list-decoding} argument. Technically, as we have seen in the analysis of myopic list-decoding error, we have to analyze the decoding error for two cases -- the case where the confusing codewords come from the OGS and the case where they are outside the OGS.
\rev{Note that these two cases are distinguished.}
Details are elaborated in Section~\ref{sec:type_i} and Section~\ref{sec:type_ii}. 

\subsection{Type I error}\label{sec:type_i}
For type I error, confusing codewords come from $[2^{nR}]\backslash\ogs^{(j)}(\vzq,i)$. We will prove that there are at most polynomially many (out of exponentially many) codewords in the OGS which can be erroneously decoded due to the confusion with some message outside the OGS.

Recall the definition of $\cE\coloneq\eatyp\cup\estr\cup\eorcl$, which is the union of several error events. Conditioned on $\cE$, we have that simultaneously $\vbfz$ behaves typically, the strip contains a large number of codewords, and the transmitted codewords does not fall into the last oracle-given set which can potentially have too small size.
\begin{lemma}\label{lemma:type_i}
\begin{align*}
    &\p{}\left(\left.\exists\vzq,\exists i,\exists j,\exists\vsq,\;|\{m\in\ogs^{(j)}(\vzq,i):\exists m'\in[2^{nR}]\backslash\ogs^{(j)}(\vzq,i),\;\vbfx(m')\in\cL^{(\bfk)}(\vbfx(m),\vsq)\}|\ge n^2\right|\cE^c\right)\le&2^{-\Omega(n^3)}.
\end{align*}
\end{lemma}

\begin{proof}
We prove the lemma using a two-step list-decoding argument. 
\rev{The idea behind this type of argument is along the lines of~\cite{djl-2019-myopic-tit}.}
Fix $\vzq$, $i$, $j$ and $\vsq$. Notice that $\{\vbfx(m)\in\cC^{(\bfk)}:m\in[2^{nR}]\backslash\ogs^{(j)}(\vzq,i)\}$ are independently and uniformly distributed.

It is a folklore (Appendix~\ref{sec:prf_omniscient_list_capacity}) in the literature that a spherical code $\cC$ of rate $\frac{1}{2}\log\frac{P}{N}-\varepsilon$ are $(P,N,\widetilde\cO(1/\varepsilon))$-list-decodable (with exponential concentration), thus are also $(P,N,\cO(n^2))$-list-decodable (with super-exponential concentration by Lemma~\ref{lemma:sup_exp_ld}).

\subsubsection{Myopic blob list-decoding}\label{sec:blob_ld_i}
Let $ \cX(\vsq)\coloneq \{ \vbfx(m'):m'\in\ogs^{(j)}(\vzq,i), \text{ and }\bfr\rev{(m',\vsq)} \rev{<} \ropt(1+f_{11}(\varepsilon))  \} $  be the set of all codewords in the OGS having typical list-decoding region. 
\rev{Since as shown in Lemma~\ref{lem:er}, there is only an exponentially small fraction of codewords in the OGS that do not fall into $ \cX(\vsq) $,}
it suffices to prove that a \rev{$1-o(1)$} fraction of the codewords in $ \cX(\vsq) $ can with high probability be decoded uniquely for every $ \vsq $.
\rev{In fact, we will show that, out of exponentially many codewords in $ \cX(\vsq) $, there are only polynomially many codewords that will incur decoding errors.}
Define
\[\blob=\bigcup_{\vbfx\in \cX(\vsq)}\cB^n(\vbfx+\vsq,\sqrt{nN}+\sqrt{n\delta_\cS}).\] 
Let us fix a realization $ \cJ $ of $ (\vzq,i,j,\vsq) $. 
If $ \wt{\cE}=\cE\cup \esq(\vzq,i,j) $, then
the number of codewords in the blob is expected to be
\begin{align}
    &\e{}\left(\left.|\blob\cap(\cC^{(\bfk)}\backslash\ogs^{(j)}(\vzq,i))|\right|\wt\cE^c\cap\cJ\right)\notag\\
    \le&\frac{\area(\C^{n-1}(\cdot,\sqrt{n\ropt(1+ f_{11}(\varepsilon,\delta_\cS) )},\sqrt{nP}))2^{n\varepsilon}}{\area(\cS^{n-1}(0,\sqrt{nP}))}2^{nR}\notag\\
    \le&\frac{\area(\cS^{n-1}(\cdot,\sqrt{n\ropt(1+ f_{11}(\varepsilon,\delta_\cS) )}))2^{n\varepsilon}}{\area(\cS^{n-1}(0,\sqrt{nP}))}2^{nR}\notag\\
    =&\frac{\sqrt{n\ropt(1+ f_{11}(\varepsilon,\delta_\cS) )}^{n-1}}{\sqrt{nP}^{n-1}}2^{n\varepsilon}2^{nR}\notag\\
    =&\sqrt{\frac{P}{\ropt(1+ f_{11}(\varepsilon,\delta_\cS) )}}2^{n\left(R-\frac{1}{2}\log\left(\frac{P}{\ropt(1+ f_{11}(\varepsilon,\delta_\cS) )}\right)+\varepsilon\right)}\notag\\
    \le&\sqrt{\frac{P}{\ropt(1+ f_{11}(\varepsilon,\delta_\cS) )}}2^{n\left(R-\frac{1}{2}\log\left(\frac{P}{\ropt}\right)+ f_{11}(\varepsilon,\delta_\cS) +\varepsilon\right)},\label{eqn:log_ineq_ub}
\end{align}
which is exponentially small. The last inequality \eqref{eqn:log_ineq_ub} follows since $\log(1+x)\le x$ for $x\le1$. Then by Lemma\ref{lemma:sup_exp_ld}, the actual number of codewords exceeds $n^2$ with probability at most $2^{-\Omega(n^3)}$.

\subsubsection{Reverse list-decoding}\label{sec:reverse_ld_i}
Conditioned on $ \esq^{\rev{c}} $ and $ \cJ $, the codewords in $ \cX(\vsq) $ are independent but not uniformly distributed over the strip. However, the distribution is almost uniform and does not affect our calculations except for adding a $ (1+o(1)) $ term. More precisely, for any set $ \cA \subset \bR^n $, we have $$ \p{}(\vbfx(m,\bfk)\in\cA|\esq^c\cap\cJ)\leq \frac{\p{}(\vbfx(m,\bfk)\in\cA)}{\p{}(\esq)} = \p{}(\vbfx(m,\bfk)\in\cA) (1+o(1)).  $$
The expected number of codewords corresponding to messages in OGS translated by $\vsq$ lying in the ball $\cB^n\left(\vbfx(m'),\sqrt{nN}+\sqrt{n\delta_\cS}\right)$ for any $m'\in[2^{nR}]\backslash\ogs^{(j)}(\vzq,i)$ is 
\begin{align}
    &\e{}\left(\left.|\cB^n(\vbfx(m'),\sqrt{nN}+\sqrt{n\delta_\cS})\cap(\{\vbfx(m)\}_{m\in\ogs^{(j)}(\vzq,i)}+\vsq)|\right|\wt\cE^c\cap\cJ\right)\notag\\
    =&\e{}\left(\left.|\cB^n(\vbfx(m')-\vsq,\sqrt{nN}+\sqrt{n\delta_\cS})\cap\{\vbfx(m)\}_{m\in\ogs^{(j)}(\vzq,i)}|\right|\wt\cE^c\cap\cJ\right)\notag\\
    \le&\frac{\area(\cS^{n-1}(\cdot,\sqrt{nN}+\sqrt{n\delta_\cS}))}{\area(\strip^{n-1}(O_-',O_+',\sqrt{nr_-},\sqrt{nr_+}))}2^{n\varepsilon}(1+o(1))\notag\\
    =&\frac{\area(\cS^{n-1}(\cdot,\sqrt{nN}+\sqrt{n\delta_\cS}))}{\area(\C^{n-1}(O_+',\sqrt{nr_+},\sqrt{nP}))-\area(\C^{n-1}(O_-',\sqrt{nr_-},\sqrt{nP}))}2^{n\varepsilon}(1+o(1))\notag\\
    \le&\frac{\area(\cS^{n-1}(\cdot,\sqrt{nN}+\sqrt{n\delta_\cS}))}{\vol(\cB^{n-1}(O_+',\sqrt{nr_+}))-\area(\cS^{n-1}(O_-',\sqrt{nr_-}))}2^{n\varepsilon}(1+o(1))\notag\\
    \le&\sqrt{\frac{N}{\rstr}}2^{n\left(-\frac{1}{2}\log\left(\frac{P}{N+\delta_\cS+2\sqrt{N\delta_\cS}}\right)+\frac{1}{2}\log\left(\frac{P}{\rstr}\right)+\varepsilon\right)}\frac{(1+o(1))}{2^{(n-1)\frac{1}{2}\log(1+\tau)-\Theta(\log n)}-2^{(n-1)\frac{1}{2}\log(1-\tau)}}\notag\\
    =&\sqrt{\frac{N}{\rstr}}2^{n\left(-\frac{1}{2}\log\frac{P}{N}-\frac{1}{2}\log\left(1-\frac{\delta_\cS+2\sqrt{N\delta_\cS}}{N+\delta_\cS+2\sqrt{N\delta_\cS}}\right)+\frac{1}{2}\log\left(\frac{P}{\rstr}\right)+\varepsilon\right)}\frac{(1+o(1))}{2^{(n-1)\frac{1}{2}\log(1+\tau)-\Theta(\log n)}-2^{(n-1)\frac{1}{2}\log(1-\tau)}}\notag\\
    {\le}&\sqrt{\frac{N}{\rstr}}2^{n\left(-\frac{1}{2}\log\frac{P}{N}+\frac{1}{2}\log\left(\frac{P}{\rstr}\right)+2\frac{\delta_\cS+2\sqrt{N\delta_\cS}}{N+\delta_\cS+2\sqrt{N\delta_\cS}}+\varepsilon\right)}\frac{(1+o(1))}{2^{(n-1)\frac{1}{2}\log(1+\tau)-\Theta(\log n)}-2^{(n-1)\frac{1}{2}\log(1-\tau)}}\label{eqn:log_ineq_lb}\\
    \le&\sqrt{\frac{N(P+\sigma^2)(1+\varepsilon)}{P\sigma^2(1-\varepsilon)}}2^{n\left(-\frac{1}{2}\log\frac{P}{N}+\frac{1}{2}\log\left(1+\frac{P}{\sigma^2}\right)+2\frac{\delta_\cS+2\sqrt{N\delta_\cS}}{N+\delta_\cS+2\sqrt{N\delta_\cS}}+3\varepsilon\right)}\frac{(1+o(1))}{2^{(n-1)\frac{1}{2}\log(1+\tau)-\Theta(\log n)}-2^{(n-1)\frac{1}{2}\log(1-\tau)}}  ,\notag
\end{align}
where in Eqn.~\eqref{eqn:log_ineq_lb} we use $\log(1-x)\ge-2x$ for small enough $x>0$.
The above quantity is exponentially small according to the sufficient myopia assumption. Thus the actual number is at most $n^2$ with probability at least $1-2^{-\Omega(n^3)}$.

\subsubsection{Union bound} By Section~\ref{sec:blob_ld_i} and Section~\ref{sec:reverse_ld_i}, for any $\vzq$, $i$, $j$ and $\vsq$, there are at most $n^2$ messages from $\ogs^{(j)}(\vzq,i)$ satisfying the condition in the lemma with probability at most $2^{-\Omega(n^3)}$. Finally, a union bound over all assumptions we have made completes the proof.
\end{proof}

\subsection{Type II error}\label{sec:type_ii}
For type II error, confusing codewords come from $\ogs^{(j)}(\vzq,i)$. We will prove that there are at most polynomially many codewords which are erroneously decoded due to confusion with another message in the same OGS. Once again we will only analyze the probability of error only for codewords having typical list-decoding volume.

\begin{lemma}\label{lemma:type_ii}
\begin{align*}
    &\p{}\bigg(\exists\vzq,\exists i,\exists j,\exists\vsq,\;|\{m\in\ogs^{(j)}(\vzq,i):\exists m'\in\ogs^{(j)}(\vzq,i)\setminus\{m\},\;\vbfx(m')\in\cL^{(\bfk)}(\vx(m),\vsq)\}|\ge\frac{\poly(n)}{2^{n\varepsilon/2}}\bigg|\cE^c\cap\cJ\bigg)\le2^{-\Omega(n)}.
\end{align*}
\end{lemma}
\begin{proof}
Notice that for different $m\in\ogs^{(j)}(\vzq,i)$, the events $\{\exists m'\in[2^{nR}]\backslash\ogs^{(j)}(\vzq,i),\;\vbfx(m')\in\cB^n(\vbfx(m)+\vsq,\sqrt{nN}+\sqrt{n\delta_\cS})\}$ are not independent. This issue is resolved by arranging messages in OGS into a $2^{n\varepsilon/2}\times2^{n\varepsilon/2}$ square matrix $\bfM$ lexicographically and applying blob list-decoding and reverse list-decoding to any row or column, denoted $\cR$, of $\bfM$. Again, using Lemma~\ref{lemma:sup_exp_ld}, it suffices to bound the expected blob list-size and reverse list-size from above by some exponentially small quantity.

\subsubsection{Blob list-decoding}\label{sec:blob_ld_ii}
The expected number of codewords in the intersection of the blob and the codewords corresponding to the OGS is 
\begin{align*}
    &\e{}\left(\left.|\blob\cap(\ogs^{(j)}(\vzq,i)\backslash\cR)|\right|\wt\cE^c\cap\cJ\right)\\
    \le&\frac{\area(\C^{n-1}(\cdot,\sqrt{n\ropt(1+ f_{11}(\varepsilon,\delta_\cS) )},\sqrt{nP}))2^{n\varepsilon}}{\area(\strip^{n-1}(O_-',O_+',\sqrt{nr_-},\sqrt{nr_+}))}2^{n\varepsilon/2}\Delta(\tau)(1+o(1))\\
    \le&\sqrt{\frac{\ropt(1+ f_{11}(\varepsilon,\delta_\cS) )}{\rstr}}2^{n\left(-\frac{1}{2}\log\left(\frac{P}{\ropt(1+ f_{11}(\varepsilon,\delta_\cS) )}\right)+\frac{1}{2}\log\left(\frac{P}{\rstr}\right)+3\varepsilon/2\right)}\frac{\Delta(\tau)(1+o(1))}{2^{(n-1)\frac{1}{2}\log(1+\tau)-\Theta(\log n)}-2^{(n-1)\frac{1}{2}\log(1-\tau)}}\\
    \le&\sqrt{\frac{\ropt(1+ f_{11}(\varepsilon,\delta_\cS) )}{\rstr}}2^{n\left(-\frac{1}{2}\log\left(\frac{P}{\ropt}\right)+\frac{1}{2}\log\left(\frac{P}{\rstr}\right)+ f_{11}(\varepsilon,\delta_\cS) +3\varepsilon/2\right)}\frac{\Delta(\tau)(1+o(1))}{2^{(n-1)\frac{1}{2}\log(1+\tau)-\Theta(\log n)}-2^{(n-1)\frac{1}{2}\log(1-\tau)}}\\
    \le&\sqrt{\frac{\ropt(1+ f_{11}(\varepsilon,\delta_\cS) )}{\frac{P\sigma^2(1-\varepsilon)}{(P+\sigma^2)(1+\varepsilon)}}}2^{n\left(-\frac{1}{2}\log\left(\frac{P}{\ropt}\right)+\frac{1}{2}\log\left(1+\frac{P}{\sigma^2}\right)+ f_{11}(\varepsilon,\delta_\cS) +7\varepsilon/2\right)}\frac{\Delta(\tau)(1+o(1))}{2^{(n-1)\frac{1}{2}\log(1+\tau)-\Theta(\log n)}-2^{(n-1)\frac{1}{2}\log(1-\tau)}}.
\end{align*}

\subsubsection{Reverse list-decoding}\label{sec:reverse_ld_ii}
The expected number of messages in the row (column) which can be confused by a single message is 
\begin{align*}
    &\e{}\left(\left.|\cB^n(\vbfx(m'),\sqrt{nN}+\sqrt{n\delta_\cS})\cap(\{\vbfx(m)\}_{m\in\cR}+\vsq)|\right|\wt\cE^c\cap\cJ\right)\\
    =&\e{}\left(\left.|\cB^n(\vbfx(m')-\vsq,\sqrt{nN}+\sqrt{n\delta_\cS})\cap\{\vbfx(m)\}_{m\in\cR}|\right|\wt\cE^c\cap\cJ\right)\\
    \le&\frac{\area(\cS^{n-1}(\cdot,\sqrt{nN}+\sqrt{n\delta_\cS}))}{\area(\strip^{n-1}(O_-',O_+',\sqrt{nr_-},\sqrt{nr_+}))}2^{n\varepsilon/2}\Delta(\tau)(1+o(1))\\
    \le&\sqrt{\frac{N}{\rstr}}2^{n\left(-\frac{1}{2}\log\frac{P}{N}+\frac{1}{2}\log\left(\frac{P}{\rstr}\right)+2\frac{\delta_\cS+2\sqrt{N\delta_\cS}}{N+\delta_\cS+2\sqrt{N\delta_\cS}}+\varepsilon/2\right)}\frac{\Delta(\tau)(1+o(1))}{2^{(n-1)\frac{1}{2}\log(1+\tau)-\Theta(\log n)}-2^{(n-1)\frac{1}{2}\log(1-\tau)}}\\
    \le&\sqrt{\frac{N(P+\sigma^2)(1+\varepsilon)}{P\sigma^2(1-\varepsilon)}}2^{n\left(-\frac{1}{2}\log\frac{P}{N}+\frac{1}{2}\log\left(1+\frac{P}{\sigma^2}\right)+2\frac{\delta_\cS+2\sqrt{N\delta_\cS}}{N+\delta_\cS+2\sqrt{N\delta_\cS}}+5\varepsilon/2\right)}\frac{\Delta(\tau)(1+o(1))}{2^{(n-1)\frac{1}{2}\log(1+\tau)-\Theta(\log n)}-2^{(n-1)\frac{1}{2}\log(1-\tau)}}.
\end{align*}

\subsubsection{Grid argument}
By Section~\ref{sec:blob_ld_ii} and Section~\ref{sec:reverse_ld_ii}, for any $\vzq$, $i$, $j$, $\vsq$ and $\cR$, there are at most $n^2\cdot n^2=n^4$ messages satisfying the condition in the lemma. Thus there are at most $2\cdot2^{n\varepsilon/2}\cdot n^4$ such ``bad" messages in $\bfM$, i.e., the OGS. A union bound over $\vzq$, $i$, $j$ and $\vsq$ completes the proof. 
\end{proof}

\section{Concluding remarks/future directions}\label{sec:conclusion}

In this work, we studied the capacity of a myopic adversarial channel with quadratic constraints. We did so for different amounts of common randomness, and were able
to find a complete characterization for certain regimes of the noise-to-signal ratios of Bob and James.

\begin{enumerate}
 \item For different regimes of the NSRs (Figs.~\ref{fig:rateregion_noCR}, \ref{fig:rateregion_logn}, \ref{fig:rateregion_rkey_0pt2}, and \ref{fig:rateregion_rkey_1}), we were able to characterize the capacity in the red, blue and grey regions.  We only have nonmatching upper and lower bounds on the capacity in the green and white regions.
 \item We also derived a myopic list-decoding result in the general case when Alice and Bob share a linear amount of common randomness. We believe that this is a useful technique that is worth exploring for general channels.
 \item When Alice uses a deterministic encoder, we believe that an improved converse using linear programming-type bounds might be obtained in the green and white regions.
 \item The $\vbfz$-aware symmetrization argument could also be extended to obtain Plotkin-type upper bounds on the rate in the green and white regions.
 \item We also believe that superposition codes could be used to obtain improved achievability results in the green and white regions for the case when there is no common randomness. In particular, we feel that rates exceeding $R_{\mathrm{GV}}$ should be achievable using superposition codes in the green and white regions.
 \item A natural problem is to find the minimum amount of common randomness required to achieve the capacity in Fig.~\ref{fig:rateregion_infiniteCR}. We know for certain values of the NSRs, this is achievable with no common randomness (blue and red regions in Fig.~\ref{fig:rateregion_noCR}). Even $\Theta(\log n)$ bits is sufficient to achieve $\Rld$ in the entire red region in Fig.~\ref{fig:rateregion_infiniteCR}, while the blue region can be expanded with increasing amounts ($\Omega(n)$ bits) of shared secret key. A lower bound on $\Nkey$ needed to achieve capacity along the lines of~\cite{langberg-focs2004} would be of interest.
 \item In this article, we studied the impact of an adversary who has noncausal access to a noisy version of the transmitted signal. However, in reality, James can only choose his attack vector based on a \emph{causal} observation of the transmission. Li et al.~\cite{tongxin-causal-2018} have some recent results for the quadratically constrained adversarial channel where the jammer can choose the $i$th symbol of his transmission based on  the first $i$ symbols of the transmitted codeword. An interesting direction is to look at the impact of myopia in this setup.
 \item Our work was inspired by the study of the discrete myopic adversarial channel~\cite{dey-sufficiently-2015}. A part of their work involved studying the capacity of a binary channel with a bit-flipping adversary who can flip at most $np$ bits of the transmitted codeword (for some $0<p<1/2$). The adversary can choose his attack vector based on a noncausal observation of the output of a binary symmetric channel with crossover probability $q$. Dey et al.~\cite{dey-sufficiently-2015} observed that if $q>p$ (sufficiently myopic), then the adversary is essentially ``blind,'' i.e., the capacity is equal to $1-H(p)$. This is what one would obtain when James were oblivious to the transmitted codeword. In other words, as long as the channel from Alice to Bob has capacity greater than that of the channel seen  by James, damage that James can do is minimal. What we observe in the quadratically constrained case is slightly different. A sufficient condition for our results to go through is that the \emph{list-decoding capacity} for Bob be greater than the \emph{Shannon capacity} for the channel seen by James. Even then, we can never hope to achieve the oblivious capacity $\frac{1}{2}\log (1+\frac{P}{N})$ for any finite $\sigma$. What we can achieve is the \emph{myopic list-decoding capacity}. In the bit-flipping adversarial case, the list-decoding capacity is equal to the capacity of the channel with an oblivious adversary. No amount of myopia can let us obtain a higher list-decoding capacity. However, the two capacities are different in the quadratically constrained scenario.
 \item The difference between oblivious and list-decoding capacities might explain the gap between the upper and lower bounds for general discrete myopic adversarial channels~\cite{dey-sufficiently-2015}. Our technique of using myopic list-decoding could potentially be used to close this gap in certain regimes.
 \item While the use of list-decoding as a technique for obtaining capacity of general AVCs is not new~\cite{sarwate-thesis}, we believe that myopic list-decoding and reverse list-decoding can be generalized to arbitrary AVCs to obtain results even in the case where the encoder-decoder pair do not share common randomness.
\end{enumerate}

\appendices

\section{Table of notation}\label{sec:tableofnotation}
\begin{center}
\begin{longtable}{|p{0.14\textwidth}|p{0.38\textwidth}|p{0.37\textwidth}|}
        \hline
        \textbf{Symbol} & \textbf{Description} & \textbf{Value/Range} \\ \hline
        $C$ & Capacity & $\limsup_{n\to\infty}R^{(n)}$ \\ \hline
        $\cC$ & Codebook & $\{\vx(m,k):m\in[2^{nR}],k\in[2^{\Nkey}]\}\subseteq\cB^n(0,\sqrt{nP})$ \\ \hline
        $C_{\mathrm{AWGN}}$ & Capacity of AWGN channels & $\frac{1}{2}\log\left(1+\frac{P}{N}\right)$ \\ \hline
        $\cC^{(\bfk)}$ & Codebook shared by Alice and Bob specified by common randomness $\bfk$ & $\cC^{(\bfk)}=\{\vx(m,\bfk)\}_{m=1}^{2^{nR}}$ \\ \hline
        $\cE$ & Shorthand notation for the union of several error events & $\cE=\eatyp\cup\estr\cup\eorcl$ \\ \hline
        $\wt\cE$ & Shorthand notation for the union of several error events & $\wt\cE=\cE\cup\esq(\vzq,i,j)$ \\ \hline
        $\eatyp$ & The error event that James's observation behaves atypically & See Equation~\eqref{eq:defn_eatyp} \\ \hline
        $\eerr$ & The error event that  more than $n^2+1$ codewords have large list-sizes  & See Equation~\eqref{eq:defn_eerr} \\ \hline
        $\er(m,\vsq)$ & The error event that $\vbfx(m)$  has an atypical list-decoding radius under $\vsq$ & See Equation~\eqref{eqn:def_er} \\\hline
        $\eorcl$ & The error event that the transmitted codewords falls into the last OGS in a strip & See Item~\ref{item:defn_eorcl} \\ \hline
        $\estr$ & The error event that there are not enough codewords in a strip & See Equation~\eqref{eq:defn_estr} \\ \hline
        $\esq(\vzq,i,j)$ & The error event that more than $n^2$ codewords  have atypical list-decoding radii under $\vsq$ & See Equation~\eqref{eqn:def_esq} \\\hline
        $\vbfg$ & Gaussian part of scale-and-babble attack & $\vbfg\sim\cN(0,\gamma^2\bfI_n)$ \\ \hline
        $\cJ$ & The event that $\vzq$, the strip, the OGS and $\vsq$ are instantiated & $ (\vbfzq,\vbfsq)=(\vzq ,\vsq ),\; m\in\ogs^{(j)}(\vzq ,i) $ \\\hline
        $\bfk$ & Common randomness shared by Alice and Bob & $\bfk\in\{0,1\}^{\Nkey}$ \\ \hline
        $\ell$ & Number of OGSs in a strip & $\ell=\lceil|\mstr(\vzq,i)|/2^{n\varepsilon}\rceil$ \\ \hline
        $\cL^{(\bfk)}(\vx(m),\vs)$ & List of $\vx(m)+\vs$ & $\cB^n(\vx(m)+\vs,\sqrt{nN})\cap\cC^{(\bfk)}$ \\ \hline
        $\cL^{(\bfk)}(\vx(m),\vsq)$ & List of $\vx(m)+\vsq$ & $\cB^n(\vx(m)+\vsq,\sqrt{nN}+\sqrt{n\delta_\cS})\cap\cC^{(\bfk)}$ \\ \hline
        $\bfm$ & Message held by Alice & $\bfm\sim \unif([2^{nR}])$ \\ \hline
        $\widehat\bfm$ & Bob's reconstruction & $\widehat\bfm\in\{0,1\}^{nR}$ \\ \hline
        $\mstr(\vzq,i)$ & Set of codewords in $\strip^{n-1}(\vzq,i)$ & $\vzq\in\cZ,i\in\{-\varepsilon/\delta+1,\cdots,\varepsilon/\delta\}$ \\ \hline
        $N$ & James's power constraint, i.e., $\|\vs\|_2\le\sqrt{nN}$ & $N\in\bR_{>0}$ \\ \hline
        $n$ & Blocklength/number of channel uses & $n\in\bZ_{>0}$ \\ \hline
        $\Nkey$ & Amount of common randomness & $\Nkey=n\Rkey$ \\ \hline
        $\ogs^{(j)}(\vzq,i)$ & Oracle-given set & $\vzq\in\cZ,i\in\{-\varepsilon/\delta+1,\cdots,\varepsilon/\delta\},j\in[\ell]$ \\ \hline
        $\ogs(\vzq,\vbfx)$ & The oracle-given set containing the transmitted $\vbfx$ & $\ogs(\vzq,\vbfx)=\ogs^{(\iota)}(\vzq,i)$ \\ \hline
        $P$ & Alice's power constraint, i.e., $\|\vx\|_2\le\sqrt{nP}$ & $P\in\bR_{>0}$ \\ \hline
        $R$ & Rate & $\frac{\log|\cC^{(\bfk)}|}{n}\in\bR_{\ge0}$ \\ \hline
        $\Rcode$ & Codebook rate & $\Rcode=R+\Rkey=\frac{\log|\cC|}{n}$ \\ \hline
        $R_{\mathrm{GV}}$ & Gilbert--Varshamov bound, a lower bound on capacity of quadratically constrained omniscient adversarial channels & $\frac{1}{2}\log\left(\frac{P^2}{4N(P-N)}\right)\one_{\{P\ge2N\}}$ \\ \hline
        $\Rkey$ & Key rate & $\Rkey=\Nkey/n$ \\ \hline
        $R_{\mathrm{LD}}$ & List-decoding capacity of quadratically constrained omniscient adversarial channels & $\frac{1}{2}\log\frac{P}{N}$ \\ \hline
        \multirow{2}{0.14\textwidth}{$R_{\mathrm{LP}}$} & \multirow{2}{0.37\textwidth}{Linear programming bound, an upper bound on capacity of quadratically constrained omniscient adversarial channels} & $(\alpha\log\alpha-\beta\log\beta)\one_{\{P\ge2N\}}$, \\ 
         & & where $\alpha=\frac{P+2\sqrt{N(P-N)}}{4\sqrt{N(P-N)}},\beta=\frac{P-2\sqrt{N(P-N)}}{4\sqrt{N(P-N)}}$ \\ \hline
        $\Rmyop $ & N/A & $\frac{1}{2}\log\left(\frac{(P+\sigma^2)(P+N)-2P\sqrt{N(P+\sigma^2)}}{N\sigma^2}\right)$ \\ \hline
        $R_{\mathrm{Rankin}}$ & Rankin bound, an upper bound on capacity of quadratically constrained omniscient adversarial channels & $\frac{1}{2}\log\left(\frac{P}{2N}\right)\one_{\{P\ge2N\}}$ \\ \hline
        $\bfr(m,\vsq)$ & Radius of list-decoding region $\C^{n-1}(\vbfx(m)+\vsq,\sqrt{n\bfr})=\cB^n(\vbfx(m)+\vsq,\sqrt{nN}+\sqrt{n\delta_\cS})$ for $m\in\ogs(\vzq,\vbfx)$ & See Equation~\eqref{eq:er} \\ \hline
        $\ropt(\vsq)$ & Optimal solution of optimization~\eqref{eq:opt_myop_ld} & See Equation~\eqref{eq:opt_myop_ld_sol} \\ \hline
        $\bfrstr$ & Radius $\sqrt{n\bfrstr}$ of a strip & $\bfrstr\in\frac{P\sigma^2}{P+\sigma^2}(1\pm\varepsilon)$ w.h.p. \\ \hline
        $\cS$ & An optimal covering of $\cB^n(0,\sqrt{nP})$ with quantization error at most $\sqrt{n\delta_\cS}$ & $\cS=\{\vsq^{(i)}\}_{i=1}^{|\cS|}$ \\ \hline
        $\strip^{n-1}(\vzq,i)$ & Strip & $\vzq\in\cZ,i\in\{-\varepsilon/\delta+1,\cdots,\varepsilon/\delta\}$ \\ \hline
        $\vbfs$ & James's attack vector & $\vbfs\in\cB^n(0,\sqrt{nN})$ \\ \hline
        $\vsq$ & Quantization of $\vs$ & $\vsq\in\cS$ \\ \hline
        $\vbfsz$ & AWGN to James & $\vbfsz\sim\cN(0,\sigma^2\bfI_n)$ \\ \hline
        $\vbfx$ & Alice's transmitted codeword & $\vbfx\in\cC$ \\ \hline
        $\vbfy$ & Bob's observation & $\vbfy=\vbfx+\vbfs\in\cB^n(0,\sqrt{nP}+\sqrt{nN})$ \\ \hline
        $\cZ$ & An optimal covering of $\sh^n(0,\sqrt{n(P+\sigma^2)(1\pm\varepsilon)})$ with quantization error at most $\sqrt{n\delta_\cZ}$ & $\cZ=\{\vzq^{(i)}\}_{i=1}^{|\cZ|}$ \\ \hline
        $\vbfz$ & James's noisy observation of $\vbfx$ & $\vbfz=\vbfx+\vbfsz\in\bR^n$ \\ \hline
        $\vzq$ & Quantization of $\vz$ & $\vzq\in\cZ$ \\ \hline
        $\Delta(\tau)$ & Quasi-uniformity factor & $\max_{\vz}\max_i\frac{\max_{\vx\in\strip^{n-1}(\vzq,i)}p_{\vbfx|\vbfzq(\vx|\vzq)}}{\min_{\vx\in\strip^{n-1}(\vzq,i)}p_{\vbfx|\vbfzq(\vx|\vzq)}}$ \\ \hline
        $\delta$ & Thickness of a strip (See Equation~\eqref{eq:strip_param_epsilon_delta}) & $\cO\left(\frac{\log n}{n}\right)$ \\ \hline
        $\delta_\cS$ & Quantization error parameter for $\vs$, i.e., for any $\vs\in\cB^n(0,\sqrt{nN})$, there exists $\vs'\in\cS$, such that $\|\vs-\vs'\|_2\le\sqrt{n\delta_\cS}$ & $\cO(1)$ \\\hline
        $\delta_\cZ$ & Quantization error parameter for $\vs$, i.e., for any $\vz\in\sh^n(0,\sqrt{n(P+\sigma^2)(1\pm\varepsilon)})$, there exists $\vz'\in\cZ$, such that $\|\vz-\vz'\|_2\le\sqrt{n\delta_\cZ}$ & $\cO(1)$ \\\hline
        $\varepsilon$ & N/A & $\cO(1)$ \\ \hline
        $\rho$ & N/A & $\cO(1)$ \\ \hline
        $\sigma$ & Standard deviation of channel noise to James & $\sqrt{\var(\bfsz)}$ \\ \hline
        $\tau$ & Thickness of a strip (See Equation~\eqref{eq:strip_param_rho_tau}) & $\cO\left(\frac{\log n}{n}\right)$ \\ \hline
        $\chi(\vzq,i,j,\vsq)$ & Number of codewords in an OGS with large list-sizes & See Equation~\eqref{eq:defn_no_of_cw_with_large_list} \\ \hline
        $\psi(\vzq,i,j,\vsq)$ & Number of codewords in an OGS with atypical list-decoding radii & See Equation~\eqref{eqn:def_psi} \\\hline
    \caption{Table of notation.}
    \label{tab:notation}
\end{longtable}
\end{center}

\section{Proofs of basic lemmas}\label{sec:appendix_proofs_basiclemmas}

\subsection{Proof of Lemma~\ref{lemma:beta_tail}}
Let $\veb$ denote the unit vector along $\vb$, i.e., $\veb=\vb/\|\vb\|_2$. Let $\vbfe$ denote the random unit vector along $\vbfa$ which is isotropically distributed on the unit sphere $\cS^{n-1}(0,1)$, i.e., $\vbfe=\vbfa/\|\va\|_2$. Notice that $|\langle\vbfa,\vb\rangle|>n\zeta$ if and only if $|\langle\vbfe,\veb\rangle|>\frac{n\zeta}{\|\va\|_2\|\vb\|_2}$, i.e., if and only if $\vbfe$ lies on one of two caps (shown in Figure~\ref{fig:beta_tail}) of height $1-\frac{n\zeta}{\|\va\|_2\|\vb\|_2}$. Thus we have
\begin{figure} 
	\centering
	\includegraphics[width = 0.4\textwidth]{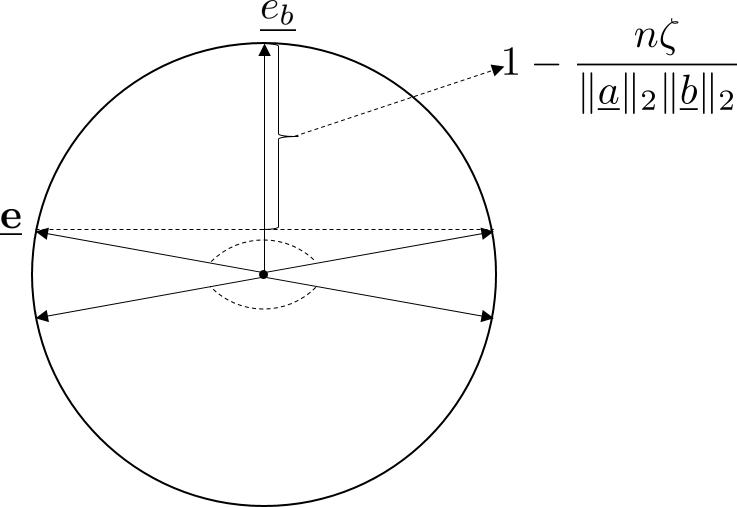}
	\caption{The geometry corresponding to the tail bound of $|\langle\vbfa,\vb\rangle|$ in Lemma~\ref{lemma:beta_tail}.}
	\label{fig:beta_tail}
\end{figure}
\begin{align*}
\p{}(|\langle\vbfa,\vb\rangle|>n\zeta)=&\p{}\left(|\langle\vbfe,\veb\rangle|>\frac{n\zeta}{\|\va\|_2\|\vb\|_2}\right)\\
=&\frac{2\area\left(\C^{n-1}\left(\frac{n\zeta}{\|\va\|_2\|\vb\|_2}\frac{\vb}{\|\vb\|_2},\sqrt{1-\frac{n^2\zeta^2}{\|\va\|_2^2\|\vb\|_2^2}},1\right)\right)}{\area(\cS^{n-1}(0,1))}\\
\le&\frac{\area\left(\cS^{n-1}\left(0,\sqrt{1-\frac{n^2\zeta^2}{\|\va\|_2^2\|\vb\|_2^2}}\right)\right)}{\area(\cS^{n-1}(0,1))}\\
=&2^{-\frac{n-1}{2}\log\left(\frac{1}{1-\frac{n^2\zeta^2}{\|\va\|_2^2\|\vb\|_2^2}}\right)}\\
\le&2^{-\frac{(n-1)n^2\zeta^2}{2\|\va\|_2^2\|\vb\|_2^2}},
\end{align*}
where the last step follows from the inequality $\log\left(\frac{1}{1-x}\right)\ge\log(1+x)\ge x$ for small enough positive $x$. \qed

\subsection{Proof of Lemma~\ref{lemma:sup_exp_ld}}
Since $ p $ is the probability that a point chosen uniformly at random from $ A $ lies in $ V $, we have
	\begin{align*}
	\p{}(|V\cap\cC|\ge cn^2)&=\sum_{i=cn^2}^{2^{nR}}\binom{2^{nR}}{i}p^i(1-p)^{2^{nR}-i}\\
	&\le \sum_{i=cn^2}^{2^{nR}}\binom{2^{nR}}{i}p^i\\
	&\le \sum_{i=cn^2}^{2^{nR}}\binom{2^{nR}}{i}2^{-in(R+\nu)}
	\end{align*}
	where the last step follows from the assumption that $ p\leq 2^{-n(R+\nu)} $. Using bounds on the binomial coefficient, the probability can be upper bounded as follows for large enough $ n $:
	\begin{align*}
	\p{}(|V\cap\cC|\ge cn^2)&\le
	\sum_{i=cn^2}^{2^{nR}}\left(\frac{2^{nR}e}{i}\right)^i2^{-in(R+\nu)}\\
	&\le
	 2^{nR}\left(\frac{2^{nR}e}{cn^2}\right)^{cn^2}2^{-cn^2\cdot n (R+\nu)}\\
	&=2^{nR+cRn^3+c(\log e)n^2-cn^2\log(cn^2)-c(R+\nu)n^3}\\
	&=2^{-c\nu n^3-2cn^2\log n+(c\log e-c\log c)n^2+nR}\\
	&\le 2^{-Cn^3}.
	\end{align*}
	This completes the proof.\qed

\section{Stochastic vs. deterministic encoding against an omniscient adversary}\label{sec:prf_omni_stoch_vs_det_enc}
Suppose we are given a sequence of $(n,R_{\mathrm{stoch}}^{(n)},P,N)$ stochastic codes $\cC_{\mathrm{stoch}}^{(n)}=\{\vx(m,k):m\in[2^{nR_{\mathrm{stoch}}^{(n)}}],k\in[2^{n\Rkey^{(n)}}]\}$ of blocklength $n$, message rate $R_{\mathrm{stoch}}^{(n)}$, bounded private key rate $\Rkey^{(n)}$, subject to maximum power constraint $P$ for Alice and maximum power constraint $N$ for James, with a deterministic decoder and average probability of error $P_{\mathrm{e,stoch}}^{(n)}\stackrel{n\to\infty}{\to}0$. Fix any $n$, we will turn $\cC_{\mathrm{stoch}}^{(n)}$ into a $(n,R_{\mathrm{det}}^{(n)},P,N)$ deterministic code $\cC_{\mathrm{det}}^{(n)}$. The deterministic decoder associated with $\cC_{\mathrm{stoch}}^{(n)}$ partitions $\bR^n$ (the space that James's observation $\vy$ lives in) into $2^{nR_{\mathrm{stoch}}^{(n)}}$ cells  $\{\cY^{(n)}(m)\subset\bR^n:m\in[2^{nR_{\mathrm{stoch}}^{(n)}}]\}$, where $\cY^{(n)}(m)\coloneq\{\vy\in\cB^n(0,\sqrt{nP}+\sqrt{nN}):\dec(\vy)=m\}$. Collect all ``good" messages into $\cM^{(n)}=\{m\in[2^{nR_{\mathrm{stoch}}^{(n)}}]:\p{}(\widehat\bfm\ne\bfm|\bfm=m)<1\}$. Assume that James's jamming strategy is deterministic. For any good message $m\in\cM^{(n)}$, there must exist at least one ``good" codeword $\vx(m,k)$ such that
\[p_{\vbfx|\bfm}(\vx(m,k)|m)>0,\quad\text{and\; } \forall\vs\in\cB^n(0,\sqrt{nN}),\;\vx(m,k)+\vs\in\cY^{(n)}(m).\]
The second condition is equivalent to $\cB^n(\vx(m,k),\sqrt{nN})\subseteq\cY^{(n)}(m)$, i.e., James does not have enough power to push $\vx(m,k)$ outside $\cY^{(n)}(m)$. For any good message, take any one of good codewords and we get a deterministic codebook $\cC_{\mathrm{det}}^{(n)}=\{\vx(m,\cdot)\in\cC_{\mathrm{stoch}}^{(n)}:\vx(m,\cdot)\text{ is good},\;m\in[2^{nR_{\mathrm{stoch}}^{(n)}}]\}$. By construction, this deterministic code with the same decoding region partition restricted to the messages in $\cM^{(n)}$ enjoys zero probability of error. We then argue that it has  asymptotically the same rate $R_{\mathrm{det}}=\lim_{n\to n}R_{\mathrm{stoch}}^{(n)}$ as $\cC_{\mathrm{stoch}}^{(n)}$. 
\begin{align*}
    P_{\mathrm{e,stoch}}^{(n)}=&\frac{1}{2^{nR_{\mathrm{stoch}}^{(n)}}}\sum_{m=1}^{2^{nR_{\mathrm{stoch}}^{(n)}}}\p{}(\widehat\bfm\ne\bfm|\bfm=m)\\
    \ge&\frac{1}{2^{nR_{\mathrm{stoch}}^{(n)}}}|\{m\in[2^{nR_{\mathrm{stoch}}^{(n)}}]:\p{}(\widehat\bfm\ne\bfm|\bfm=m)=1\}|\\
    =&\frac{1}{2^{nR_{\mathrm{stoch}}^{(n)}}}|(\cM^{(n)})^c|\\
    =&\p{}(\bfm\notin\cM^{(n)})\to0.
\end{align*}
Define $\bfe=\one_{\{\bfm\in\cM^{(n)}\}}$. We have 
\begin{align*}
    nR_{\mathrm{stoch}}^{(n)}=&H(\bfm)\\
    =&H(\bfe)+H(\bfm|\bfe)\\
    =&H(\bfe)+\p{}(\bfe=1)H(\bfm|\bfe=1)+\p{}(\bfe=0)H(\bfm|\bfe=0)\\
    =&H(\bfe)+\p{}(\bfm\in\cM^{(n)})H(\bfm|\bfm\in\cM^{(n)})+\p{}(\bfm\notin\cM^{(n)})H(\bfm|\bfm\notin\cM^{(n)}).
\end{align*}
It follows that
\begin{align*}
    R_{\mathrm{det}}^{(n)}=&\frac{1}{n}H(\bfm|\bfm\in\cM^{(n)})\\
    =&\frac{1}{n\p{}(\bfm\in\cM^{(n)})}(nR_{\mathrm{stoch}}^{(n)}-H(\bfe)-\p{}(\bfm\notin\cM^{(n)})H(\bfm|\bfm\notin\cM^{(n)}))\\
    \ge&\frac{1}{n\p{}(\bfm\in\cM^{(n)})}(nR_{\mathrm{stoch}}^{(n)}-1-P_{\mathrm{e,stoch}}^{(n)}nR_{\mathrm{stoch}}^{(n)})\\
    =&\frac{1}{1-\p{}(\bfm\notin\cM^{(n)})}\left((1-P_{\mathrm{e,stoch}}^{(n)})R_{\mathrm{stoch}}^{(n)}-\frac{1}{n}\right)\\
    \to&\lim_{n\to\infty}R_{\mathrm{stoch}}^{(n)}.
\end{align*}

\section{Quadratically constrained list-decoding capacity {with an omniscient adversary}}\label{sec:prf_omniscient_list_capacity}
\subsection{Achievability.}
We use a random spherical code, i.e., the $2^{nR},R=\frac{1}{2}\log\frac{P}{N}-\varepsilon$ codewords $\cC=\{\vx(m)\}_{m=1}^{2^{nR}}$ are chosen independently and uniformly at random from the Euclidean sphere centered at the origin of radius $\sqrt{nP}$.  Since James has a power constraint of $\sqrt{nN}$, the received vector $\vbfy=\vbfx+\vbfs$ is guaranteed to lie within the shell $\sh^n(0,\sqrt{nP}\pm\sqrt{nN})$. We will prove the following  result:
\begin{lemma}
 There exists a constant $c>0$ independent of $n$ and $\varepsilon$, but possibly on $P,N$, and $R$, such that 
  \[
  \p{}\left(\forall \vy\in\sh^n(0,\sqrt{nP}\pm\sqrt{nN}),\;\vert\cB^n(\vy,\sqrt{nN})\cap\cC\vert<c\frac{1}{\varepsilon}\log\frac{1}{\varepsilon}\right) \geq 1-2^{-\Omega(n)}.
 \]
\label{lemma:listdecod_omn_achievability}
\end{lemma}
\begin{proof}
Let $L\coloneq c\frac{1}{\varepsilon}\log\frac{1}{\varepsilon}$ be the desired list-size, for some absolute constant $c$ to be determined later. Define $\delta\coloneq N\varepsilon^2/8$. 
At first, we increase the list-decoding radius by a small amount $\sqrt{n\delta}$. As we will see later, this will be helpful when we take a union bound over
possible $\vy$'s. 
We first show that for any fixed $\vy$, the probability (over the codebook) that there are more than $L$ codewords within a distance $\sqrt{nN}+\sqrt{n\delta}$ to $\vy$ is sufficiently small.

Observe that $\vert\cB^n(\vy,\sqrt{nN}+\sqrt{n\delta})\cap\cC\vert = \vert \C^{n-1}(\vy,\sqrt{nN}+\sqrt{n\delta},\sqrt{nP})\cap \cC\vert$.
We claim that for any fixed $\vy$,
\begin{equation}
 \p{}\left(\vert\cB^n(\vy,\sqrt{nN}+\sqrt{n\delta})\cap\cC\vert>L\right) \leq c_22^{-n(L+1)\varepsilon/2}.
 \label{eq:ld_achie_fixedy}
\end{equation}
for some constant $c_2$ independent of $n$ and $\varepsilon$.

The maximal intersection of a ball $\cB^n(\vy,\sqrt{nN}+\sqrt{n\delta})$ and the Euclidean sphere $\cS^{n-1}(0,\sqrt{nP})$ is shown in Figure~\ref{fig:ld_achi} \rev{(with $ \sqrt{n\delta} $ being dropped since it is a proof artifact rather than an essential factor in the geometry of list decoding)}. It can be seen that the corresponding $\vy$ has length $\sqrt{n(P-N)}$.
\begin{figure} 
    \centering
    \includegraphics[width = 0.5\textwidth]{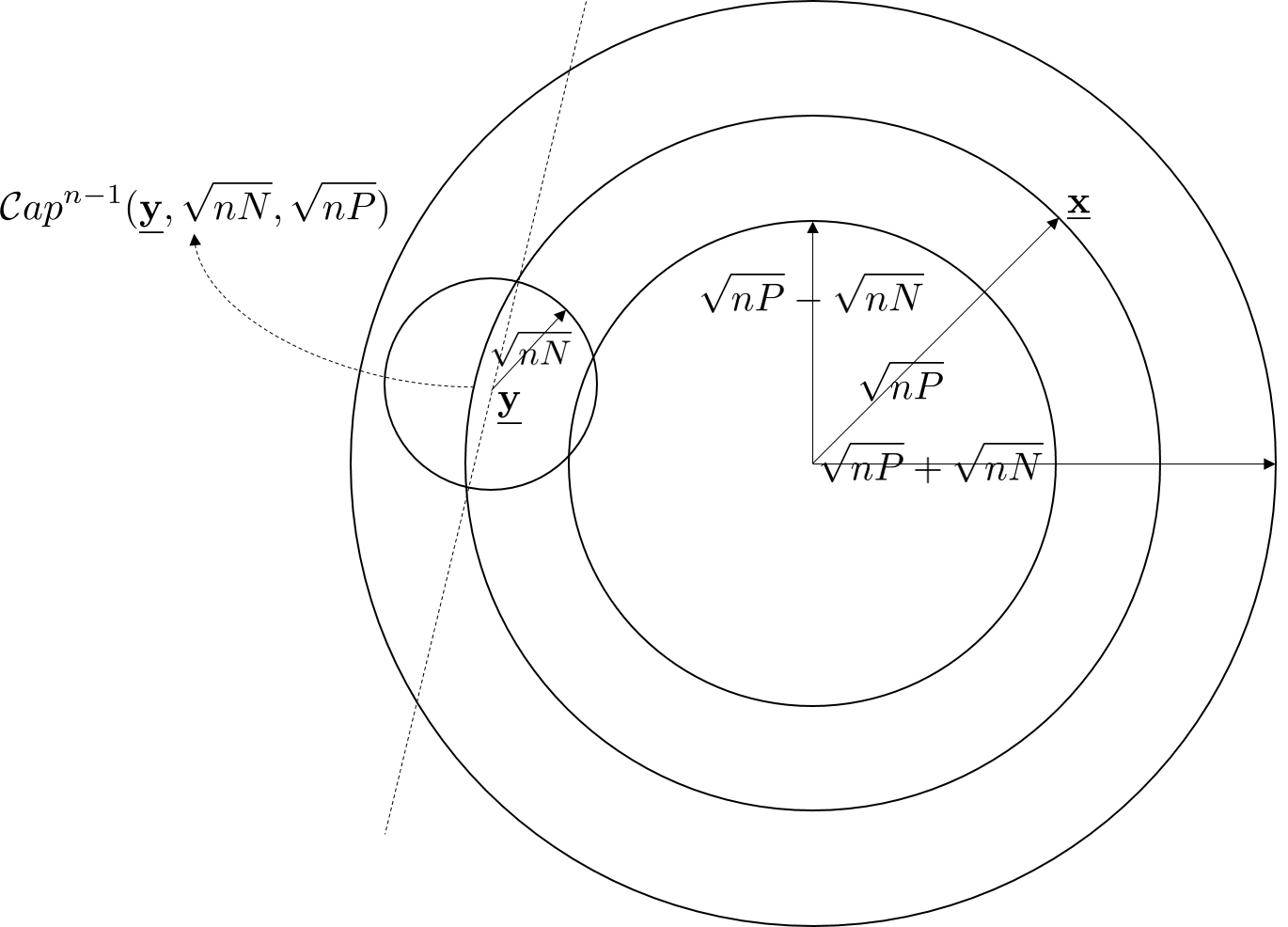}
    \caption{Maximal intersection of the decoding ball with the sphere $\cS^{n-1}(0,\sqrt{nP})$. \rev{In the figure, we omit the quantization parameter $\delta$ that goes into the actual proof, in particular dilates the radius of the noise ball by an additive factor $\sqrt{n\delta}$.}}
    \label{fig:ld_achi}
\end{figure}

The probability of a codeword falling into the cap can be upper bounded by
\begin{align*}
    p\coloneq&\p{}\left(\vbfx\in\C^{n-1}(\vy,\sqrt{nN}+\sqrt{n\delta},\sqrt{nP})\right)\\
    =&\frac{\area(\C^{n-1}(\vy,\sqrt{nN}+\sqrt{n\delta},\sqrt{nP}))}{\area(\cS^{n-1}(0,\sqrt{nP}))}\\
    \le&\frac{\area(\cS^{n-1}(\vy,\sqrt{nN}+\sqrt{n\delta}))}{\area(\cS^{n-1}(0,\sqrt{nP}))}\\
    =&\left(\frac{N}{P}\right)^{(n-1)/2}\left(1+\sqrt{\frac{\delta}{N}}\right)^{n-1}\\
    =&2^{-\frac{n-1}{2}\left(\log\left(\frac{P}{N}\right)+2\log\left(1+\sqrt{\frac{\delta}{N}}\right)\right)}\\
    \leq&c_12^{-n\left(\frac{1}{2}\log\frac{P}{N}+2\sqrt{\frac{\delta}{N}}\right)} \\
    = &c_12^{-n\left(\frac{1}{2}\log\frac{P}{N}+\frac{\varepsilon}{2}\right)},
\end{align*}
where $c_1\coloneq\sqrt{\frac{P}{N}}$. In the last step, we have used the fact that $\delta=N\varepsilon^2/8$. Now consider the left-hand side of (\eqref{eq:ld_achie_fixedy}).
\begin{align}
    &\p{}\left(|\C^{n-1}(\vy,\sqrt{nN}+\sqrt{n\delta},\sqrt{nP})\cap\cC|>L\right)\notag\\
    =&\sum_{i=L+1}^{2^{nR}}\binom{2^{nR}}{i}p^i(1-p)^{2^{nR}-i}\notag\\
    \le&2^{nR}\binom{2^{nR}}{L+1}p^{L+1}\\
    \le&2^{nR}\left(\frac{2^{nR}e}{L+1}\right)^{L+1}\left(c2^{-n\frac{1}{2}\log\frac{P}{N}+\frac{\varepsilon n}{2}}\right)^{L+1}\notag\\
    =& c_22^{nR+n(L+1)\left(R-\frac{1}{2}\log\frac{P}{N}+\frac{\varepsilon}{2}\right)}\notag\\
    =&c_22^{-n(L+1)\varepsilon/2+nR}, &\label{eq:prf_omn_lst_prob1}
\end{align}
where $c_2\coloneq\left(\frac{ec_1}{L+1}\right)^{L+1}$. Since we are interested in constant list-sizes, $c_2$ does not depend on $n$.

Define $\cY$ to be an optimal covering of $\sh^n(0,\sqrt{nP}\pm\sqrt{nN})$ by balls of radius $\sqrt{n\delta}$.
In other words, $\cY$ is a finite set of points in $\sh^n(0,\sqrt{nP}\pm\sqrt{nN})$ such that $\min_{\vy'\in \cY} \Vert \vy-\vy' \Vert\le\sqrt{n\delta}$
for all $\vy\in\sh^n(0,\sqrt{nP}\pm\sqrt{nN})$. In addition, $\cY$ is the smallest (in cardinality) over all possible coverings. 
One can achieve (for e.g., using lattice codes~\cite[Chapter 2]{conway-sloane-book})
\begin{equation}
  |\cY|\le\left(\frac{\vol(\cB^n(0,\sqrt{nP}+\sqrt{nN}+\sqrt{n\delta}))}{\vol(\cB^n(0,\sqrt{n\delta}))}\right)^{1+o(1)}=\left(\frac{\sqrt{P}+\sqrt{N}+\sqrt{\delta}}{\sqrt{\delta}}\right)^{n(1+o(1))}\eqcolon \left(\frac{c_3}{\varepsilon}\right)^n.
  \label{eq:prf_omn_lst_epsilonnetsize}
\end{equation}
We now have everything to prove Lemma~\ref{lemma:listdecod_omn_achievability}.
\begin{align*}
    &\p{}\left(\exists\vy\in\sh^n(0,\sqrt{nP}\pm\sqrt{nN}),\;|\C^{n-1}(\vy,\sqrt{nN},\sqrt{nP})\cap\cC|>L\right)\\
    \le&\p{}\left(\exists\vyq\in\cY,\;|\C^{n-1}(\rev{\vyq},\sqrt{nN}+\sqrt{n\delta},\sqrt{nP})\cap\cC|>L\right)\\
    \le& \sum_{\vyq\in\cY}\p{}\left(|\C^{n-1}(\rev{\vyq},\sqrt{nN}+\sqrt{n\delta},\sqrt{nP})\cap\cC|>L\right)
\end{align*}
Now using \eqref{eq:prf_omn_lst_epsilonnetsize} and \eqref{eq:prf_omn_lst_prob1}, we have
\begin{align*}
    &\p{}\left(\exists\vy\in\sh^n(0,\sqrt{nP}\pm\sqrt{nN}),\;|\C^{n-1}(\vy,\sqrt{nN},\sqrt{nP})\cap\cC|>L\right)\\
    \le & c_22^{-n(L+1)\varepsilon/2+nR} \left(\frac{c_3}{\varepsilon}\right)^n\\
    = & 2^{-\Omega(n)}
\end{align*}
as long as 
\begin{align*}
 (L+1)\varepsilon/2-R - \log \left(\frac{c_3}{\varepsilon}\right) >0
\end{align*}
or equivalently, $L> c\frac{1}{\varepsilon}\log\frac{1}{\varepsilon}$ for a suitable constant $c$.
This completes the proof of Lemma~\ref{lemma:listdecod_omn_achievability}.
\end{proof}

\subsection{Converse.}

Now we turn to the converse part of the list-decoding capacity theorem over quadratically constrained channels.
\begin{lemma}
If $R>\frac{1}{2}\log\frac{P}{N}$, then no sequence of codebooks of rate $R$ is  $(P,N,n^{\cO(1)})$-list-decodable.
\end{lemma}
\begin{proof}
We will show that for any code $\cC$ of rate $R=\frac{1}{2}\log\frac{P}{N}+\varepsilon$ with $2^{nR}$ codewords (not necessarily randomly) chosen from $\cS^{n-1}(0,\sqrt{nP})$, there must some $\vy$ with list-size exceeding $\cO(1/\varepsilon)$. Let's  choose $\vbfy$ uniformly at random on $\cS^{n-1}(0,\sqrt{n(P-N)})$. As we saw, such $\vbfy$'s result in the largest list-decoding regions. Define
\[p\coloneq\p{}\left(\vx\in\C^{n-1}(\vbfy,\sqrt{nN},\sqrt{nP})\right)=\p{}\left(\vbfy\in\C^{n-1}(\vx,\sqrt{nN},\sqrt{nP})\right).\]
First notice that $p$ can be lower bounded by
\begin{align*}
    p\ge&\frac{\vol(\cB^{n-1}(0,\sqrt{nN}))}{\area(\cS^{n-1}(0,\sqrt{nP}))}\\
    =&\frac{1}{2\sqrt{\pi}}\left(\frac{\sqrt{2}}{\sqrt{n}}+\cO(n^{-3/2})\right)\left(\frac{N}{P}\right)^{(n-1)/2}\\
    =&c_n\left(\frac{N}{P}\right)^{n/2}\\
    =&c_n2^{-n\frac{1}{2}\log\frac{P}{N}},
\end{align*}
where $c_n\coloneq\frac{1}{2\sqrt{\pi}}\left(\frac{\sqrt{2}}{\sqrt{n}}+\cO(n^{-3/2})\right)\sqrt{\frac{P}{N}}$. The expected number of codewords in the intersection is 
\[\e{}(|\C^{n-1}(\vbfy,\sqrt{nN},\sqrt{nP})\cap\cC|)=p2^{nR}\ge c_n2^{n\left(R-\frac{1}{2}\log\frac{P}{N}\right)}=c_n2^{n\varepsilon}.\]
Hence there must exist some $\vy$ in $\cS^{n-1}(0,\sqrt{n(P-N)})$ such that
\[|\C^{n-1}(\vy,\sqrt{nN},\sqrt{nP})\cap\cC|
    \ge\e{\vbfy\sim \unif(\cS^{n-1}(0,\sqrt{n(P-N)}))}(|\C^{n-1}(\vbfy,\sqrt{nN},\sqrt{nP})\cap\cC|)
    \ge c_n2^{n\varepsilon}.\]
By a black-box reduction from ball codes to spherical codes~\cite{zhang-vatedka-list-dec-real}, this converse holds for any code satisfying Alice's power constraint.
\end{proof}

\section{Proof of Claim~\ref{claim:cap_scaleandbabble}}\label{sec:prf_cap_scaleandbabble}
We want to prove that the scale-and-babble attack instantiates a channel whose capacity is equal to that of an equivalent AWGN channel.
  
We begin with the following observation, which follows from a simple application of the chain rule of mutual information.
\begin{lemma}
Consider any joint distribution $p_{\vbfx,\vbfy}$ on $(\vbfx,\vbfy)$ such that $\vbfx$ has differential entropy which grows as $2^{o(n)}$.
Let $\pmb\xi$ be a $\{0,1\}$-valued random variable (possibly depending on $\vbfx,\vbfy$) that takes value $0$ with probability $2^{-cn}$ for some constant $c>0$.
Then, $I(\vbfx;\vbfy) \in I(\vbfx;\vbfy|\pmb\xi=1)(1-2^{-cn}) + 2^{-cn(1-o(1))}$. 
\label{lemma:cap_scalebabble_1}
\end{lemma}
\begin{proof}
The claim follows from straightforward computation.
\begin{align*}
    I(\vbfx;\vbfy)=&\p{}(\pmb\xi=1)I(\vbfx;\vbfy|\pmb\xi=1)+\p{}(\pmb\xi=0)I(\vbfx;\vbfy|\pmb\xi=0)\\
    =&(1-2^{-cn})I(\vbfx;\vbfy|\pmb\xi=1)+2^{-cn}(H(\vbfx|\pmb\xi=0)-H(\vbfx|\vbfy,\pmb\xi=0))\\
    =&(1-2^{-cn})I(\vbfx;\vbfy|\pmb\xi=1)+2^{-cn}\cO(H(\vbfx))\\
    =&(1-2^{-cn})I(\vbfx;\vbfy|\pmb\xi=1)+2^{-cn(1-o(1))}.
\end{align*}
\end{proof}
 
Let $\widetilde{\vbfy}\coloneq(1-\alpha)\vbfx+\widetilde\vbfg$, $\widetilde{\vbfg}\coloneq\vbfg-\alpha\vbfsz$. The channel from $\vbfx$ to $\widetilde{\vbfy}$ is a standard AWGN channel with capacity \eqref{eq:cap_awgn_converse}. 
Let  
\rev{\begin{align}
\pmb\xi=\begin{cases}
\pmb\beta,&\text{if }\pmb\beta = 1 \\
0, &\text{if } 0<\pmb\beta<1
\end{cases}
=\begin{cases}
1,&\text{if } \|-\alpha\vbfz+\vbfg\|_2\le\sqrt{nN} \\
0,& \ow
\end{cases}, \notag 
\end{align}}
i.e.,   the indicator random variable  if James's power constraint is satisfied.
Clearly, $I(\vbfx;\widetilde{\vbfy}|\pmb\beta=1)=I(\vbfx;\vbfy|\pmb\beta=1)$.  
\begin{lemma}
\[
\p{}(\pmb\beta\neq 1) = 2^{-\Omega(n)}.
\]
\label{lemma:cap_scalebabble_2}
\end{lemma}
\begin{proof}
Recall that
\[\vbfg-\alpha\vbfsz\sim\cN(0,(\gamma^2+\alpha^2\sigma^2)\bfI_n)=\cN(0,(N-\alpha^2P-(N-\alpha^2(P+\sigma^2))\varepsilon)\bfI_n)\eqcolon\cN(0,(N-\alpha^2P-\varepsilon')\bfI_n).\]
Now we can bound the probability
\begin{align}
    \p{}(\pmb\beta\ne1)=&\p{}(\|-\alpha\vbfz+\vbfg\|_2>\sqrt{nN})\notag\\
    =&\p{}(\alpha^2\|\vbfx\|_2^2+\|\vbfg-\alpha\vbfsz\|_2^2-2\langle\alpha\vbfx,\vbfg-\alpha\vbfsz\rangle>nN)\notag\\
    \le&\p{}(\|\vbfg-\alpha\vbfsz\|_2^2-2\langle\alpha\vbfx,\vbfg-\alpha\vbfsz\rangle>n(N-\alpha^2P))\notag\\
    =&\p{}(\|\vbfg-\alpha\vbfsz\|_2^2-2\langle\alpha\vbfx,\vbfg-\alpha\vbfsz\rangle>n(N-\alpha^2P),\;2|\langle\alpha\vbfx,\vbfg-\alpha\vbfsz\rangle|>n\varepsilon'/2)\notag\\
    &+\p{}(\|\vbfg-\alpha\vbfsz\|_2^2-2\langle\alpha\vbfx,\vbfg-\alpha\vbfsz\rangle>n(N-\alpha^2P),\;2|\langle\alpha\vbfx,\vbfg-\alpha\vbfsz\rangle|\le n\varepsilon'/2)\notag\\
    \le&{\p{}(|\langle\alpha\vbfx,\vbfg-\alpha\vbfsz\rangle|>n\varepsilon'/4)}\label{eqn:p1}\\
    &+{\p{}(\|\vbfg-\alpha\vbfsz\|_2^2>n(N-\alpha^2P-\varepsilon'/2))}.\label{eqn:p2}
\end{align}
We bound terms \eqref{eqn:p1} and \eqref{eqn:p2} separately.
\begin{align*}
    \eqref{eqn:p1}=&\p{}(|\cN(0,\alpha^2\|\vbfx\|_2^2(N-\alpha^2P-\varepsilon'))|>n\varepsilon'/4)\\
    \le&2\exp\left(-\frac{(n\varepsilon'/4)^2}{2\alpha^2\|\vbfx\|_2^2(N-\alpha^2P-\varepsilon')}\right)\\
    \le&2\exp\left(-\frac{\varepsilon'^2}{32\alpha^2P(N-\alpha^2P-\varepsilon')}n\right)\\
    \eqcolon&2^{-g_1(\varepsilon')n},
\end{align*}
and
\begin{align*}
    \eqref{eqn:p2}=&\p{}(\|\cN(0,(N-\alpha^2P-\varepsilon')\bfI_n)\|_2^2>n(N-\alpha^2P-\varepsilon'/2))\\
    \le&\exp\left(\left\{-\frac{\varepsilon'/2}{N-\alpha^2P-\varepsilon'}+\ln\left(1+\frac{\varepsilon'/2}{N-\alpha^2P-\varepsilon'}\right)\right\}\frac{n}{2}\right)\\
    \le&\exp\left(-\frac{1}{4}\left(\frac{\varepsilon'/2}{N-\alpha^2P-\varepsilon'}\right)^2\frac{n}{2}\right)\\
    =&\exp\left(-\frac{\varepsilon'^2}{32(N-\alpha^2P-\varepsilon')^2}n\right)\\
    \eqcolon&2^{-g_2(\varepsilon')n}.
\end{align*}
\end{proof}

Using Lemmas~\ref{lemma:cap_scalebabble_1} and~\ref{lemma:cap_scalebabble_2},
we have that
\begin{align*}
    I(\vbfx;\vbfy)=& I(\vbfx;\vbfy|\pmb\beta=1)(1-2^{-\Omega(n)})+2^{-\Omega(n)}\\
    =& I(\vbfx;\widetilde\vbfy|\pmb\beta=1)(1-2^{-\Omega(n)})+2^{-\Omega(n)}\\
    =& I(\vbfx;\widetilde{\vbfy})(1+o(1)).
\end{align*}
Here we have used the fact that $\vbfx$ is power constrained, and hence has differential entropy $\Theta(n)$.
 
The rest of the proof follows along the same lines as the standard converse for the AWGN channel~\cite[Sec.\ 9.2]{cover-thomas}.
Let us briefly outline the steps involved.
 
As a first step, note that Fano's inequality still holds even in the presence of common randomness. Let $\bfk$ denote the shared secret key.
Specifically, if $\bfm$, $\widehat{\bfm}$, and $\pe$, respectively, denote the message chosen, Bob's estimate of the message, and the probability of error, then
\[
H(\bfm|\widehat{\bfm},\bfk) \leq H(\bfm|\widehat{\bfm}) \leq H(\pe) + nR \pe,  
\]
where the first step follows because conditioning reduces entropy, and the second follows from (standard) Fano's inequality. 

Now, if we demand that the probability of error be vanishingly small in $n$, then
\begin{align}
 nR &=  H(\bfm)=H(\bfm|\bfk) = I(\bfm;\widehat{\bfm}|\bfk) + H(\bfm|\widehat{\bfm},\bfk)\notag\\
    &\leq I(\bfm;\widehat{\bfm}|\bfk) + o(n) \notag \\
    & \leq I(\vbfx;\vbfy|\bfk) + o(n) \notag \\
    & \leq I(\vbfx;\widetilde{\vbfy}|\bfk)(1+o(1)) +o(n) \notag\\
    & \leq \frac{1}{2}\log\left(1+\frac{(1-\alpha)^2P}{\alpha^2\sigma^2+\gamma^2}\right)n.
\end{align}
We have skipped a number of arguments in obtaining the last step, but these follow from the standard converse proof for the AWGN channel. 
The only property of the codebook used there is that it satisfies an average power constraint (which is indeed satisfied as we have a more restrictive max power constraint).
This completes the proof of Claim~\ref{claim:cap_scaleandbabble}.

\section{Quasi-uniformity -- Proof of \rev{Lemma~\ref{lemma:quasiuniformity}}}\label{sec:prf_quasiuniformity}
\rev{As shown in Fig.~\ref{fig:quasi_unif},} define $|\overline{xO'}|\coloneq\sqrt{n\rstr}$. By the construction of the strip,
\begin{align}
|\overline{x^-O'_-}|\coloneq\sqrt{nr_-}\coloneq\sqrt{n\rstr(1-\tau)},\quad|\overline{x^+O'_+}|\coloneq\sqrt{nr_+}\coloneq\sqrt{n\rstr(1+\tau)}.
\label{eqn:def_r_plus_minus}
\end{align}
Then the quasi-uniformity factor can be computed as follows
\begin{align}
    &\left(\sup_{\vx\in\strip^{n-1}(O_-',O_+',\sqrt{nr_-},\sqrt{nr_+})}p_{\vbfx|\vbfz}(\vx|\vz)\right)\bigg/\left(\inf_{\vx\in\strip^{n-1}(O_-',O_+',\sqrt{nr_-},\sqrt{nr_+})}p_{\vbfx|\vbfz}(\vx|\vz)\right)\notag\\
    =&\frac{p_{\vbfx|\vbfz}(\vx^-|\vz)}{p_{\vbfx|\vbfz}(\vx^+|\vz)}\notag\\
    =&\frac{p_{\vbfx,\vbfz}(\vx^-,\vz)}{p_{\vbfx,\vbfz}(\vx^+,\vz)}\notag\\
    =&\frac{p_{\vbfz|\vbfx}(\vz|\vx^-)p_{\vbfx}(\vx^-)}{p_{\vbfz|\vbfx}(\vz|\vx^+)p_{\vbfx}(\vx^+)}\notag\\
    =&\frac{p_{\vbfz|\vbfx}(\vz|\vx^-)}{p_{\vbfz|\vbfx}(\vz|\vx^+)}\notag\\
    =&\exp\left(\frac{\|\vz-\vx^+\|_2^2-\|\vz-\vx^-\|_2^2}{2\sigma^2}\right)\notag\\
    {=}&\exp\left(\frac{\|\vz\|_2(\sqrt{n(P-r_-)}-\sqrt{n(P-r_+)})}{\sigma^2}\right)\label{eqn:rpm}\\
    =&\exp\left(\frac{\|\vz\|_2}{\sigma^2}\frac{2n\rstr\tau}{\sqrt{n(P-r_-)}+\sqrt{n(P-r_+)}}\right),\notag
\end{align}
where Eqn.~\eqref{eqn:rpm} holds since
\[\|\vz-\vx^{\pm}\|_2^2=nr_{\pm}+(\|\vz\|_2-\sqrt{nP-nr_{\pm}})^2=\|\vz\|_2^2+nP-2\|\vz\|_2\sqrt{nP-nr_\pm}.\]
\rev{In the above calculation, recall that as mentioned after the definition in Eqn.~\eqref{eq:defn_quasiuniformity}, the joint density $ p_{\vbfx,\vbfz} $ is given by $ \vbfx\sim \unif(\cS^{n-1}(0,\sqrt{nP})) $ and $ \vbfz = \vbfx + \vbfsz $ where $ \vbfsz\sim\cN(0,\sigma^2\bfI_n) $. }

\section{Exponentially many codewords in the strip -- proof of Lemma~\ref{lemma:exp_cw_strip}}\label{sec:prf_strip}
The expected number of codewords in a strip can be estimated as follows.
\begin{align}
    &\e{}\left(\left.|\strip^{n-1}(O_-',O_+',\sqrt{nr_-},\sqrt{nr_+})\cap\cC|\right|\eatyp^c\right)\notag\\
    \ge&\frac{\area(\C^{n-1}(O_+',\sqrt{nr_+},\sqrt{nP}))-\area(\C^{n-1}(O_-',\sqrt{nr_-},\sqrt{nP}))}{\area(\cS^{n-1}(0,\sqrt{nP}))}2^{n\Rcode}\Delta(\tau)^{-1}\notag\\
    \ge&\frac{\vol(\cB^{n-1}(O_+',\sqrt{nr_+}))-\area(\cS^{n-1}(O_-',\sqrt{nr_-}))}{\area(\cS^{n-1}(0,\sqrt{nP}))}2^{n\Rcode}\Delta(\tau)^{-1}\notag\\
    \asymp&\left[\frac{1}{\sqrt{n-1}}\left(\frac{r_+}{P}\right)^{(n-1)/2}-\left(\frac{r_-}{P}\right)^{(n-1)/2}\right]2^{n\Rcode}\Delta(\tau)^{-1}\label{eqn:vol_est}\\
    =&\left[n^{-1/2}\left(\frac{\rstr}{P}(1+\tau)\right)^{(n-1)/2}-\left(\frac{\rstr}{P}(1-\tau)\right)^{(n-1)/2}\right]2^{n\Rcode}\Delta(\tau)^{-1}\notag\\
    =&\left(2^{(n-1)\left(\frac{1}{2}\log\left(\frac{\rstr}{P}\right)+\frac{1}{2}\log(1+\tau)\right)-\frac{1}{2}\log n}-2^{(n-1)\left(\frac{1}{2}\log\left(\frac{\rstr}{P}\right)+\frac{1}{2}\log(1-\tau)\right)}\right)2^{n\Rcode}\Delta(\tau)^{-1}\notag\\
    =&\sqrt{\frac{P}{\rstr}}2^{-n\frac{1}{2}\log\left(\frac{P}{\rstr}\right)}\left(2^{(n-1)\frac{1}{2}\log(1+\tau)-\frac{1}{2}\log n}-2^{(n-1)\frac{1}{2}\log(1-\tau)}\right)2^{n\Rcode}\Delta(\tau)^{-1}\notag\\
    \le&\sqrt{\frac{(P+\sigma^2)(1+\varepsilon)}{\sigma^2(1-\varepsilon)}}2^{-n\frac{1}{2}\log\left(\frac{(P+\sigma^2)(1-\varepsilon)}{\sigma^2(1+\varepsilon)}\right)}\left(2^{(n-1)\frac{1}{2}\log(1+\tau)-\frac{1}{2}\log n}-2^{(n-1)\frac{1}{2}\log(1-\tau)}\right)2^{n\Rcode}\Delta(\tau)^{-1},\label{eqn:rstr_bound_applied}
\end{align}
where Eqn.~\eqref{eqn:vol_est} follows from the fact that
\[
    \vol(\cB^n(0,1))\asymp\frac{1}{\sqrt{\pi n}}\left(\frac{2\pi e}{n}\right)^{n/2},\quad\area(\cS^{n-1}(0,1))\asymp\sqrt{\frac{n}{\pi}}\left(\frac{2\pi e}{n}\right)^{n/2}.
\]
Eqn.~\eqref{eqn:rstr_bound_applied} follows from Eqn.~\eqref{eqn:bound_rstr}.
The factor in the parentheses is a polynomial in $n$ if we properly set $\tau=\cO((\log n)/n)$. If the coding rate is strictly above the threshold $ \frac{1}{2}\log\left(1+\frac{P}{\sigma^2}\right) $, then, in expectation, there are at least $2^{4\varepsilon n}$ codewords in every strip.

\section{Proof of Lemma~\ref{lemma:achievablerate_thetanbits_secrecy}}\label{sec:proof_achievablerate_thetanbits_secrecy}

	The coding scheme is determined by parameters $ (R,\Rkey,R_e) $, where $ \Rkey $ denotes the rate of the secret key. We generate $ 2^{n(R+\Rkey+R_e)} $ codewords uniformly at random from the sphere $ \cS^{n-1}(0,\sqrt{nP}) $. Let us index the codewords using the triple $ (i_1,i_2,i_3)\in [2^{nR}]\times[2^{n\Rkey}]\times [2^{nR_e}] $ The messages are chosen uniformly at random from  $ [2^{nR}] $. Given a message $ \bfm\in [2^{nR}] $ and key $ \bfk\in [2^{n\Rkey}] $, the encoder picks $ \bfr $ uniformly at random from $ [2^{nR_e}] $, and transmits the $ (\bfm,\bfk,\bfr) $th codeword $ \vbfx(\bfm,\bfk,\bfr) $. Bob knows $ \bfk $, and has to decode $ \bfm $ from $ \vbfy $. 

We will choose the parameters so as to satisfy:
\[
R+R_e<C_{\mathrm{myop}},
\]
and
\[
R_e =\max\left\{ 0, \;\frac{1}{2}\log\left(1+\frac{P}{\sigma^2}\right)-\Rkey - \delta  \right\}
\]
for some small $ \delta>0 $. 

As long as  $ R+R_e<C_{\mathrm{myop}} $, the probability of decoding error is $ o(1) $ from Lemma~\ref{lemma:achievablerate_thetanbits}. The mutual information rate can be bounded quite easily, and follows~\cite{wyner1975wire,kang2010wiretap}.
\begin{align}
I(\bfm,\vbfz|\cC) &= H(\bfm|\cC) - H(\bfm|\vbfz,\cC) &\notag\\
&= H(\bfm|\cC) - H(\bfm,\bfr,\bfk|\bfz,\cC) + H(\bfr,\bfk|\bfm,\bfz,\cC)&\notag\\
&= H(\bfm|\cC)  - H(\bfm,\bfr,\bfk|\cC) + I(\bfm,\bfr,\bfk;\vbfz|\cC) + H(\bfr,\bfk|\bfm,\vbfz,\cC)&\notag\\
&= -nR_e-n\Rkey +I(\vbfx;\vbfz|\cC) + H(\bfr,\bfk|\bfm,\vbfz,\cC)
\end{align}
where the last step follows from the fact that $ \bfm,\bfr,\bfk $ are mutually independent, and $ \vbfx $ is a deterministic function of $ (\bfm,\bfr,\bfk) $. We can further write
\begin{align}
I(\bfm,\vbfz|\cC) &= -nR_e-n\Rkey +I(\vbfx;\vbfz|\cC) + H(\bfm,\bfk,\bfr|\bfm,\vbfz,\cC) &\notag\\
&= -nR_e-n\Rkey +I(\vbfx;\vbfz|\cC) + H(\vbfx|\bfm,\vbfz,\cC)&\notag\\
&\leq -nR_e-n\Rkey +\frac{n}{2}\log \left(1+\frac{P}{\sigma^2}\right)(1+o(1)) + H(\vbfx|\bfm,\vbfz,\cC) &\label{eq:I_leakage_1}
\end{align}
where the last step follows from the fact that the codewords are uniformly drawn from the sphere.
If we choose $ R_e+\Rkey = \frac{1}{2}\log\left(1+\frac{P}{\sigma^2}\right) - \delta $ for some small $ \delta>0 $, then we can show that $ H(\vbfx|\bfm,\vbfz,\cC) = n o(1) $ using Fano's inequality and the fact that spherical codes achieve the capacity of the AWGN channel. Plugging into \eqref{eq:I_leakage_1}, we get that $ \frac{1}{n}I(\bfm;\vbfz)<\delta +o(1) $. This can be made arbitrarily small by choosing a small enough $ \delta $.
\qed

\bibliographystyle{IEEEtran}
\bibliography{IEEEabrv,Myopic_References} 
 
\end{document}